\documentclass[english]{article}
\usepackage[margin=1.3in]{geometry}
\usepackage[T1]{fontenc}
\usepackage[latin9]{inputenc}
\usepackage{color}
\usepackage{soul}
\usepackage{babel}
\usepackage{calc}
\usepackage{units}
\usepackage{amsmath}
\usepackage{amsthm}
\usepackage{amssymb}
\usepackage{graphicx}
\usepackage{comment}
\usepackage{cancel}
\usepackage{wrapfig}
\PassOptionsToPackage{normalem}{ulem}
\usepackage{ulem}
\usepackage[all,cmtip]{xy}
\usepackage[unicode=true,pdfusetitle,
bookmarks=true,bookmarksnumbered=true,bookmarksopen=false,
breaklinks=false,pdfborder={0 0 0},pdfborderstyle={},backref=false,colorlinks=true,psdextra]
{hyperref}
\usepackage{caption}
\usepackage[small]{titlesec}
\usepackage{bm}
\usepackage{algorithm,algpseudocode}
\usepackage{natbib}
\makeatletter
\numberwithin{equation}{section}
\numberwithin{figure}{section}

\theoremstyle{plain}

\newtheorem{thm}{\protect\theoremname}
\theoremstyle{plain}
\newtheorem{prop}[thm]{\protect\propositionname}
\theoremstyle{definition}

\theoremstyle{definition}

\theoremstyle{plain}
\newtheorem{cor}[thm]{\protect\corollaryname}
\theoremstyle{plain}

\theoremstyle{plain}
\newtheorem{lem}[thm]{\protect\lemmaname}
\theoremstyle{plain}

\theoremstyle{plain}
\newtheorem{assumption}{Assumption}

\algnewcommand{\LineComment}[1]{\Statex \hskip\ALG@thistlm \(\triangleright\) #1}

\makeatother

\newcommand{\dbeta}[1]{\frac{d {#1}}{d\beta}}
\newcommand{\partialbeta}[1]{\frac{\partial {#1}}{\partial \beta}}
\newcommand{\dbetaK}[2]{\frac{d^{#2} {#1}}{d\beta^{#2}}}
\newcommand{\partialbetaK}[2]{\frac{\partial^{#2} {#1}}{\partial \beta^{#2}}}

\providecommand{\remarkname}{Remark}
\providecommand{\lemmaname}{Lemma}
\providecommand{\conjecturename}{Conjecture}
\providecommand{\corollaryname}{Corollary}
\providecommand{\definitionname}{Definition}
\providecommand{\examplename}{Example}
\providecommand{\propositionname}{Proposition}
\providecommand{\theoremname}{Theorem}
\newcommand{\mat}[1]{\left(\begin{matrix} #1\end{matrix}\right)}

\DeclareMathOperator*{\argmin}{argmin}

\DeclareMathOperator*{\eig}{eig}
\DeclareMathOperator*{\diag}{diag}

\DeclareMathOperator*{\supp}{supp}

\newcommand{\inputmarginalSymbol}{r}
\newcommand{\inputmarginal}[1]{\inputmarginalSymbol(\hat{x}#1)}
\newcommand{\inputmarginalVect}{\bm{\inputmarginalSymbol}}
\newcommand{\intermediateencoderSymbol}{q}
\newcommand{\intermediateencoder}[2]{\intermediateencoderSymbol(\hat{x}#1| x#2)}
\newcommand{\intermediateencoderVect}{\bm{\intermediateencoderSymbol}}
\newcommand{\outputmarginal}[1]{s(\hat{x}#1)}
\newcommand{\outputmarginalVect}{\bm{s}}
\newcommand{\expectedDxWRTencoder}{\mathbb{E}_{\intermediateencoder{'}{}}\left[d(x, \hat{x}')\right]}
\newcommand{\expectedDxWRTencoderK}[1]{\mathbb{E}_{\intermediateencoder{'}{}}\left[d(x, \hat{x}')^{#1}\right]}
\newcommand{\dbar}{\bar{d}}

\algnewcommand\algorithmicinput{\textbf{Input:}}
\algnewcommand\Input{\item[\algorithmicinput]}
\algnewcommand\algorithmicoutput{\textbf{Output:}}
\algnewcommand\Output{\item[\algorithmicoutput]}

\newcommand*{\compileappendices}{}
\newcommand*{\compilefigs}{}

\begin{document}
	
\title{Root Tracking for Rate-Distortion: Approximating a Solution Curve with Higher Implicit Multivariate Derivatives}

\author{Shlomi Agmon$^1$\thanks{This work was supported by the ISF under grant 1641/21.}}

\date{%
	$^1$ School of Computer Science and Engineering, \\
	The Hebrew University of Jerusalem,
	Jerusalem,
	Israel\\
	Email: shlomi.agmon@mail.huji.ac.il \\
	[2ex]%
}

\maketitle
\global\long\def\cal#1{\mathcal{#1}}
\global\long\def\bb#1{\mathbb{#1}}
\global\long\def\bf#1{\mathbf{#1}}
\global\long\def\frak#1{\mathfrak{#1}}

\begin{abstract}
	The rate-distortion curve captures the fundamental tradeoff between compression length and resolution in lossy data compression. 
	However, it conceals the underlying dynamics of optimal source encodings or \textit{test channels}.
	We argue that these typically follow a piecewise smooth trajectory as the source information is compressed.
	These smooth dynamics are interrupted at \textit{bifurcations}, where solutions change qualitatively.
	Sub-optimal test channels may collide or exchange optimality there, for example. 
	There is typically a plethora of sub-optimal solutions, which stems from restrictions of the reproduction alphabet.
	
	We devise a family of algorithms that exploits the underlying dynamics to track a given test channel along the rate-distortion curve.
	To that end, we express implicit derivatives at the roots of a non-linear operator by higher derivative tensors.
	Providing closed-form formulae for the derivative tensors of Blahut's algorithm thus yields implicit derivatives of arbitrary order at a given test channel, thereby approximating others in its vicinity.
	Finally, our understanding of bifurcations guarantees the optimality of the root being traced, under mild assumptions, while allowing us to detect when our assumptions fail.
	
	Beyond the interest in rate distortion, this is an example of how understanding a problem's bifurcations can be translated to a numerical algorithm. 
\end{abstract}

\renewcommand*{\thefootnote}{\fnsymbol{footnote}}
\footnotetext[1]{ The author wishes to acknowledge the late Prof. Naftali Tishby for his involvement in the early stages of this work. The author is grateful to Or Ordentlich, without whose continuous support this work could not have reached its conclusion. The author thanks Shaul Zemel, Noam Agmon, Amitai Yuval and the reviewers for their helpful comments.}
\renewcommand*{\thefootnote}{\arabic{footnote}}
\setcounter{footnote}{0}

\medskip 
	

\section{Introduction}
\label{sec:introduction}

The theory of lossy data compression was introduced in the seminal works of \cite{shannon1948mathematical, shannon1959fidelity}, and has since found multiple applications beyond the obvious ones of communications and storage of information. 
Among others, clustering \citep{rose1998deterministic}, perception \citep{blau2019rethinking, sims2016rate}, and it is intimately related to the information bottleneck principle \citep{tishby1999} in learning, \citep{bachrach2003}.


Formally, let $X \sim p_X$ and $\hat{X}$ be discrete i.i.d. random variables on finite \textit{source} and \textit{reproduction} alphabets, respectively denoted $\mathcal{X}$ and $\hat{\mathcal{X}}$.
A \textit{rate distortion} (RD) problem is defined by a \textit{distortion measure} $d:\mathcal{X}\times \hat{\mathcal{X}} \to \mathbb{R}_{\geq 0}$ and a source distribution $p_X(x)$, or $p(x)$ for short. 
One seeks \citep{shannon1948mathematical, shannon1959fidelity} the minimal rate $I(X; \hat{X}) := \bb{E}_{p(\hat{x}|x) p_X(x)} \log \frac{p(\hat{x}|x)}{p(\hat{x})}$ subject to a constraint $D$ on the expected distortion.
The optimal tradeoff between the information rate per message to the distortion is encoded by the rate-distortion function
\begin{equation}		\label{eq:RD-func-def}
	R(D) :=
	\min_{p(\hat{x}|x)}
	\left\{ 
	I(X; \hat{X}): \; \bb{E}_{p(\hat{x}|x) p_X(x)} \left[ d(x, \hat{x}) \right] \leq D
	\right\} \; .
\end{equation}
Despite the interest in rate distortion, there are surprisingly few ways to calculate this tradeoff.

While the minimization problem \eqref{eq:RD-func-def} can be solved analytically in special cases, e.g., \cite[2.6]{berger71} or \cite[10.3]{Cover2006}, a solution is often obtained numerically by the iterative Blahut-Arimoto (BA) algorithm, due to \cite{blahut1972}. Using alternating minimizations, it converges to a test channel $p(\hat{x}|x)$ which achieves the minimum of \eqref{eq:RD-func-def}, \cite{csiszar1974computation}.
However, the BA algorithm suffers from \textit{critical slowing down} near critical points, \cite{agmon2021critical}, points at which the number of symbols $\hat{x}$ required for optimal reproduction decreases.
That is, there is a dramatic increase in the computational costs until convergence there.
Further, one is often interested in the entire $R(D)$ curve. However, standard computation techniques solve \eqref{eq:RD-func-def} only at specified grid points. This yields isolated samples along the curve while making little to no use of previously computed solutions.

To alleviate computational costs, one could consider more efficient BA variants or its approximations thereof, such as the related \citep{Yu10squeezing, matz2004information, sayir2000iterating} for channel capacity or \citep{sutter2015efficient} for constrained capacity.
Alternatively, the choice of initial condition could be improved. For example, deterministic annealing \citep{rose1990deterministic} uses the solution at each grid-point as the initialization for the next, in analogy to annealing in statistical physics \citep{rose1998deterministic}.
The algorithms we propose aim to improve upon this choice significantly.

In an attempt to tackle the above shortcomings of BA, we propose a new family of algorithms whose purpose is to follow the path of a known solution as some (scalar) control parameter is varied.
This is especially appealing for RD problems, as a solution at the extremities of the curve is nearly trivial to obtain. 
For, the constant encoding to $\argmin_{\hat{x}} \bb{E}[d(X, \hat{x})]$ and $x\mapsto \argmin_{\hat{x}} d(x, \hat{x})$ are respectively optimal at zero rate or when the smallest distortion is desired.
Unlike BA, our algorithm for RD provides a piecewise polynomial approximation of the path traversed by the distributions achieving \eqref{eq:RD-func-def}, with uniform convergence guarantees outside a small vicinity of the critical points.
Building on the work of \cite{agmon2021critical}, our algorithm does \textit{not} suffer from an increased computational cost near critical points, but rather from a reduced accuracy there.
Nevertheless, it admits a comprehensible tradeoff between accuracy and computational cost, permitting intelligible choices when high accuracy is desired.
The computational cost of our algorithm is comparable to that of BA with reverse deterministic annealing, with the advantage of computing the entire curve of solutions rather than only solving on a grid.

Using the method of Lagrange multipliers for \eqref{eq:RD-func-def}, with\footnote{ The normalization constraint is omitted for clarity.} $I(X; \hat{X}) + \beta \bb{E} \left[ d(x, \hat{x}) \right]$ and $\beta > 0$, one obtains a pair of equations for $p(\hat{x}|x)$ and the marginal $p(\hat{x})$.
Iterating over these equations boils down to the Blahut-Arimoto algorithm, \citep[10.7-10.8]{Cover2006}.
That is, a necessary condition for a distribution to achieve the minimum at \eqref{eq:RD-func-def} is that it is a fixed point of BA, or equivalently a root of the operator
\begin{equation}		\label{eq:RD-operator-def}
	F := Id - BA_\beta \;, 
\end{equation}
as noted by \citeauthor{agmon2021critical}. Where, $BA_\beta$ denotes a single BA iteration at the multiplier value $\beta$, and $Id$ is the identity. 
The marginal $p(\hat{x})$ may be taken as our variable (see Section \ref{part:how-and-what}.\ref{sub:high-order-deriv-tensors-of-BA}). 
However, to facilitate the discussion we shall write $\bm{x}\in\bb{R}^T$ instead, for some $T > 0$. 
Namely, we consider roots of an equation 
\begin{equation}		\label{eq:solution-as-root-of-functional-eq}
	F(\bm{x}, \beta) = \bm{0} \;,
\end{equation}
for an operator $F\left(\cdot, \beta\right)$ on $\bb{R}^T$, $F:\bb{R}^{T}\times\bb{R} \to \bb{R}^T$, and $\beta$ a real independent ``time-like'' parameter.

This work stems from the following intuition.
Suppose that the Jacobian matrix $D_{\bm{x}}F$ of $F$ is non-singular at a root $(\bm{x}_0, \beta_0)$ of \eqref{eq:solution-as-root-of-functional-eq}. Then, by the Implicit Function Theorem, there not only exists a function $\bm{x}(\beta)$ satisfying \eqref{eq:solution-as-root-of-functional-eq} through the root, $\bm{x}(\beta_0) = \bm{x}_0$, but $\bm{x}(\beta)$ also inherits the differentiability properties of $F$, \cite[I.1.7]{kielhofer2011bifurcation}. 
For example, $\bm{x}$ is analytic in $\beta$ if $F$ is analytic.
When $F$ is particularly well-behaved, then one might expect the derivatives of $\bm{x}$ with respect to $\beta$ determine the path $\bm{x}(\beta)$.
Indeed, when $F$ is real-analytic as in rate-distortion \eqref{eq:RD-operator-def}, then each coordinate of $\bm{x}(\beta)$ can be written as power-series around $\beta_0$. If the series' convergence radii happen to be large or infinite, then the entire solution path $\bm{x}(\beta)$ can be extrapolated from the point $(\bm{x}_0, \beta_0)$, at least in principle.
But even if these radii are small (but non-zero), determining $\bm{x}(\beta)$ only in some small vicinity of $\beta_0$, then one can extrapolate the path $\bm{x}(\beta)$ segment by segment, so long that $D_{\bm{x}} F$ remains non-singular.

We provide three novelties to transform this intuition into an arguably practical algorithm.
The derivatives of $\bm{x}$ with respect to $\beta$ are \textit{implied} by the requirement $F = \bm{0}$ \eqref{eq:solution-as-root-of-functional-eq}. 
To calculate implicit derivatives, we first provide a recursive formula (Theorem \ref{thm:formula-for-high-order-expansion-of-F-in-main-result-sect} in Section \ref{part:how-and-what}.\ref{sub:high-order-beta-derivatives-at-an-operator-root}) for the implicit derivatives $\tfrac{d^l}{d\beta^l} \bm{x}$ at a root in terms of the derivative tensors of $F$, for any order $l > 0$. 
While \cite{zemel2019combinatorics} already solved this for a univariate $\bm{x}$, we are unaware of such formulae when $\bm{x}$ is multivariate.
Second, we provide arbitrary-order closed-form formulae (Theorem \ref{thm:high-order-derivs-of-BA-in-main-text} in Section \ref{part:how-and-what}.\ref{sub:high-order-deriv-tensors-of-BA}) for the derivative tensors of rate-distortion \eqref{eq:RD-operator-def}.
\cite{agmon2021critical} already provided its Jacobian, while for the related channel capacity \cite{Yu10squeezing} provided the Jacobian and \cite{nakagawa2021analysis} also the Hessian.

Together, the first two components allow us to calculate numerically the implicit derivatives
\begin{equation}		\label{eq:l-th-deriv-at-beta0}
	\dbetaK{\bm{x}}{l}\bigg\rvert_{(\bm{x}_0, \beta_0)}
\end{equation}
of arbitrary order at a root of rate-distortion \eqref{eq:RD-operator-def}, under mild assumptions.
In fact, we have discovered a first-order ordinary differential equation (ODE) satisfied by RD roots (Theorem \ref{thm:beta-ODE-in-marginal-coords} in Section \ref{part:details}.\ref{sub:encoders-beta-derivatives}).
Using the implicit derivatives \eqref{eq:l-th-deriv-at-beta0} we can approximate $\bm{x}(\beta)$ nearby, for example via
\begin{equation}		\label{eq:solution-by-beta-as-taylor-approx}
	\bm{x}(\beta_0 + \Delta \beta) \cong
	\frac{1}{0!} \bm{x}_0 + 
	\frac{1}{1!} \dbeta{\bm{x}} \bigg\rvert_{(\bm{x}_0, \beta_0)} \Delta \beta +
	\frac{1}{2!} \dbetaK{\bm{x}}{2} \bigg\rvert_{(\bm{x}_0, \beta_0)} \Delta \beta^2 + \dots +
	\frac{1}{l!} \dbetaK{\bm{x}}{l} \bigg\rvert_{(\bm{x}_0, \beta_0)} \Delta \beta^l \;,
\end{equation}
where $\Delta \beta := \beta - \beta_0$. 
One can then take a step $\Delta \beta$ and recompute expansion \eqref{eq:solution-by-beta-as-taylor-approx} repeatedly.
This simple algorithm based on the Taylor method \citep{butcher2016numerical, atkinson2011numerical} gives a piecewise polynomial approximation of the path $\bm{x}(\beta)$.
After computing \eqref{eq:solution-by-beta-as-taylor-approx} along a grid, then any off-grid point is obtained by merely evaluating a polynomial.
e.g., the right of Figure \ref{fig:approx_from_point_and_RTRD}. 
With that, the above is not to be confused with the gradient flow towards a root at a fixed multiplier value $\beta_0$, which \citep{parker2010symmetry} describe in a related context by an ODE. 
In contrast, the implicit derivatives \eqref{eq:l-th-deriv-at-beta0} describe how a root evolves \textit{with} $\beta$. 

\begin{figure}[h!]
	\centering
	\hspace*{-40pt}
	\includegraphics[width=1.2\textwidth]{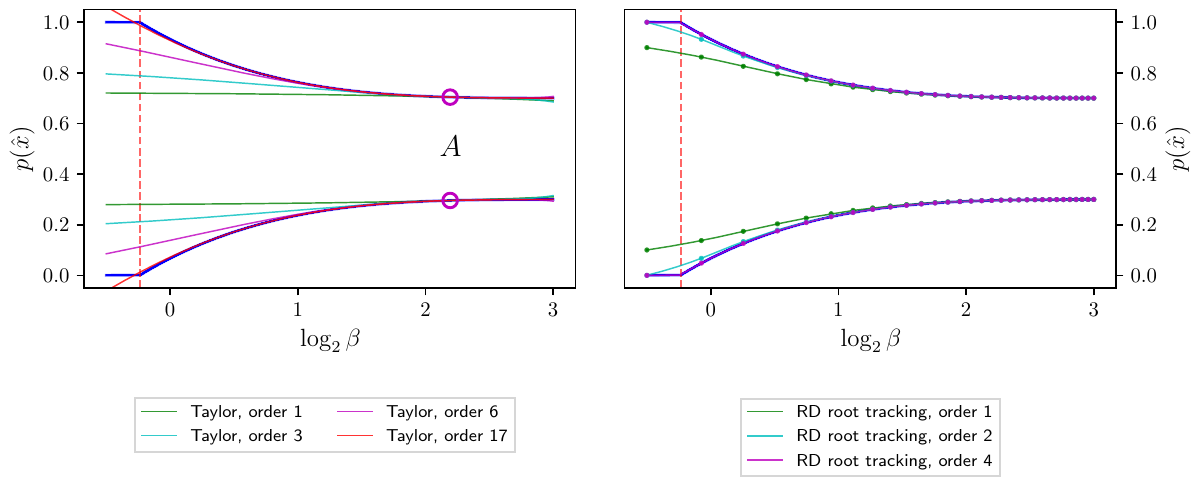}
	\caption{
		\textbf{Approximating the entire solution curve with implicit derivatives \eqref{eq:l-th-deriv-at-beta0}}, for a binary source with a Hamming distortion, compared to its analytical solution in Section \ref{part:proofs}.\ref{sec:binary-source-with-hamming-dist-appendix} (thick blue).
		Each color depicts a marginal probability distribution $p(\hat{x})$ as a function of the Lagrange multiplier $\beta$, for $\hat{x} \in \{1,2\}$; this suffices to parametrize a root (cf., Section \ref{part:how-and-what}.\ref{sub:high-order-deriv-tensors-of-BA}).
		\textbf{Left}: Taylor expansions \eqref{eq:solution-by-beta-as-taylor-approx} of several orders around the point $A$. 
		\textbf{Right}: root-tracking for RD (Algorithm \ref{algo:root-tracking-for-RD}, with $\delta = 10^{-2}$) detects and handles the bifurcation (dashed red vertical). 
		Adding grid-points bootstraps the approximation towards higher accuracy, with each grid point leveraging the computational cost invested earlier, to its right.
		Unlike BA, the entire solution curve is extrapolated here from only 31 grid points.
	}
	\label{fig:approx_from_point_and_RTRD}
\end{figure}

While the above goes a long way towards reconstructing the entire solution path, it is not enough. 
For, as we show in Section \ref{part:details}.\ref{sec:RD-bifurcations-and-root-tracking}, RD problems typically have a plethora of sub-optimal roots, \textit{not} achieving the minimum in \eqref{eq:RD-func-def}.
RD roots can collide and merge into one, or exchange optimality in some cases where the RD curve has a linear segment.
These are instances of \textit{bifurcations} --- a change in the number of roots \eqref{eq:solution-as-root-of-functional-eq}.
To ensure that the root being traced is optimal, some understanding of the solution structure is necessary. 
That is, of the bifurcations of the fixed-point equations of RD, encoded by \eqref{eq:RD-operator-def}.
Together with the two components above, the understanding established in Section \ref{part:details}.\ref{sec:RD-bifurcations-and-root-tracking} provides tools to detect and handle bifurcations (of some types), culminating in our algorithm for tracking the path of an optimal RD root (Algorithm \ref{algo:root-tracking-for-RD} in Section \ref{part:how-and-what}.\ref{sub:RD-root-tracking-near-bifrcations}), subject to mild assumptions.
In addition to the interest in rate distortion itself, this provides an example of how an understanding of a problem's bifurcations can be translated to a numerical algorithm. As the computation of the rate-distortion function is equivalent to a process of deterministic annealing \citep{rose1998deterministic}, we believe that similar tools might facilitate the numerical solution of other problems as well.

RD bifurcations were noted already by \cite{berger71} and others due to the resulting non-smoothness of the RD curve \eqref{eq:RD-func-def} at the points of bifurcation.
\cite{rose1994mapping} uses a mapping approach to provide insights for continuous source alphabets, usually assuming a squared-error distortion, $d(x, \hat{x}) = |x - \hat{x}|^2$. 
While allowing for a much broader class of distortion measures, our results in Section \ref{part:details}.\ref{sec:RD-bifurcations-and-root-tracking} suggest a slightly different picture of RD bifurcations for sources of finite alphabet.
Further, they allow a clearer view of cases where BA with reverse deterministic annealing \citep{rose1990deterministic, rose1998deterministic} follows a sub-optimal solution path rather than the optimal one.

Our algorithm is perhaps best compared to solving an ODE numerically. 
Where, one usually exploits derivatives only to the order dictated by the ODE in order to propagate the solution. 
On the other hand, the mathematical machinery we provide allows us to compute implicit multivariate derivatives $\dbetaK{\bm{x}}{l}$ \eqref{eq:l-th-deriv-at-beta0} of arbitrary order, which we specialize to RD. 
In principle, one may even change the order at will with this machinery. 
With that, we have made several deliberate concessions for the sake of simplicity. 
Among others, fixing the order and step-size results in inefficient use of computational resources; see Section \ref{part:how-and-what}.\ref{sub:efficient-RD-root-tracking} on various possible improvements and \ref{part:details}.\ref{sec:error-analysis} on the root cause of the computational difficulty. 
Nevertheless, despite these concessions, the cost of an entire solution curve appears to be roughly comparable to BA with reverse annealing, as suggested by Figure \ref{fig:err-to-computational-cost-tradeoff}. 
Furthermore, we handle only cluster-vanishing bifurcations as in \citep{agmon2021critical}, although our understanding of RD bifurcations (in Section \ref{part:details}.\ref{sec:RD-bifurcations-and-root-tracking}) permits more than that.

The error of an $l$-th order Taylor method for RD, as in expansion \eqref{eq:solution-by-beta-as-taylor-approx}, is of order $O(|\Delta \beta|^l)$ for small step sizes $|\Delta \beta|$ (Theorem \ref{thm:taylor-method-converges-for-RD-root-tracking-away-of-bifurcation} in Section \ref{part:how-and-what}.\ref{sub:taylor-method-for-RD-root-tracking}); cf., Figure \ref{fig:approx_from_point_and_RTRD}. 
Increasing $l$ or taking $|\Delta \beta|$ smaller improves the approximation in general, as one might expect. 
Though, for simplicity, we fix these parameters while computing. 
This method can be seen to have better accuracy in general than interpolating the output distributions of BA, for example, at least for orders $l > 1$. Though, the details of this are deferred to future work. 
The computational costs of Algorithm \ref{algo:root-tracking-for-RD} are only linear in source alphabet size $|\mathcal{X}|$ thanks to our choice of coordinates for $\bm{x}$ (in Section \ref{part:how-and-what}.\ref{sub:high-order-deriv-tensors-of-BA}), and are asymptotically polynomial in $|\hat{\cal{X}}|$ when $l$ is fixed. 
On the other hand, its computational costs grow (hyper-)exponentially with the order $l$. 
However, despite this rapid growth rate, higher orders $l$ make much better use of the invested computational costs, in general, when the step size $|\Delta \beta|$ is small. 
The cost-to-error tradeoff is demonstrated in Figure \ref{fig:err-to-computational-cost-tradeoff}, with details in Section \ref{part:how-and-what}.\ref{sub:costs-and-error-to-cost-tradeoff-of-RD-root-tracking-in-main-results-section}.

\begin{figure}[h!]
	\centering
	\hspace*{-40pt}
	\includegraphics[trim={3cm 0 2cm 0}, clip, width=1.2\textwidth]{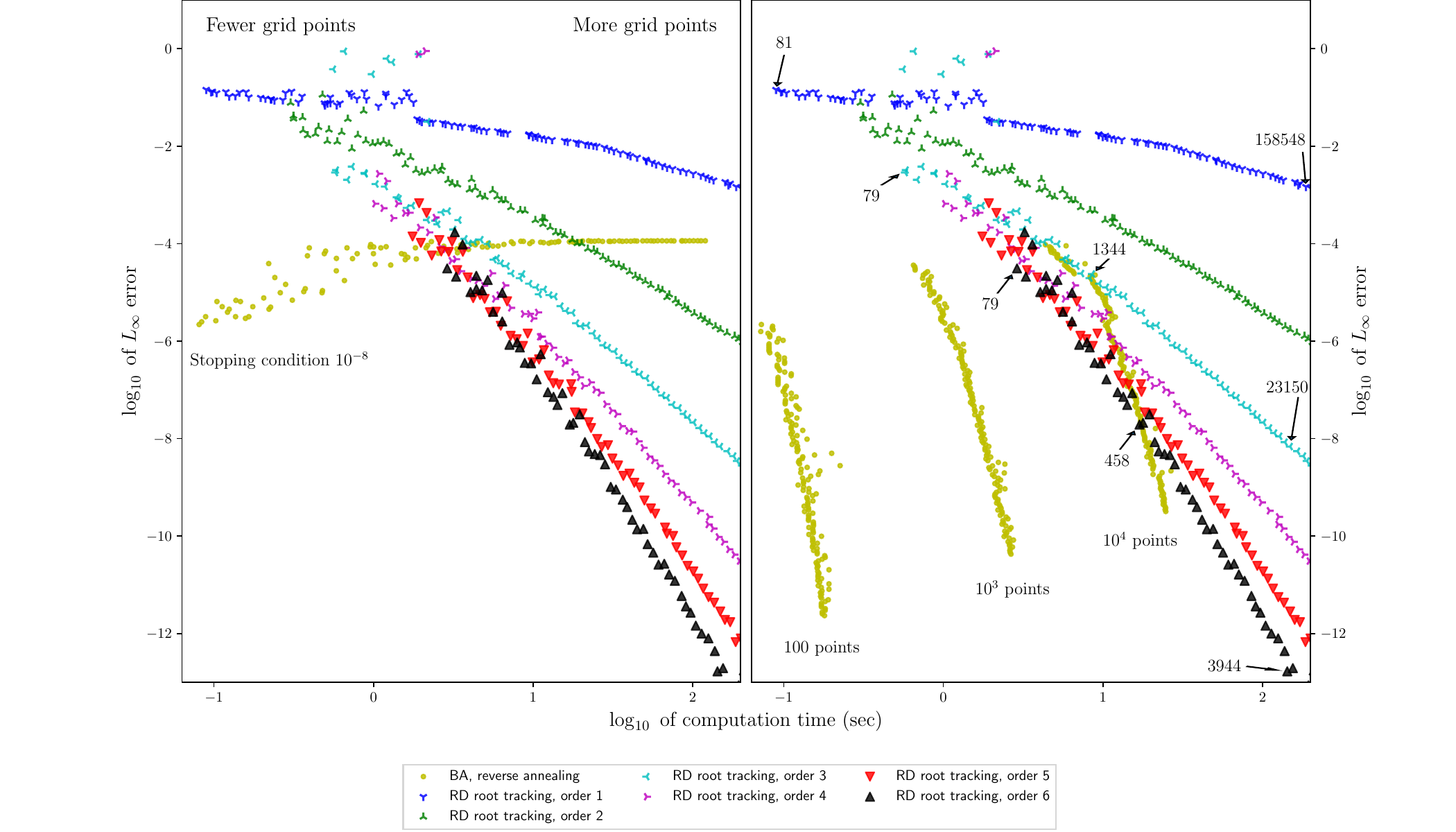}
	\caption{
		\textbf{The tradeoff between error and computational cost, by the number of grid points}.
		For each number of equally spaced grid points, the maximal error from the true solution (over all grid points) is plotted against the computation time, which serves as a rough measure for the computational complexity of the entire solution curve.
		Computation time was measured with our implementation, running single-threaded on a 1.80GHz Intel i7-8550U processor; see \ref{part:how-and-what}.\ref{sub:costs-and-error-to-cost-tradeoff-of-RD-root-tracking-in-main-results-section} on complexities and cost-to-error tradeoff. 
		Results are shown for RD root tracking (Algorithm \ref{algo:root-tracking-for-RD}) of several orders, and for Blahut-Arimoto computed in reverse annealing.
		Unlike BA, RD root tracking approximates the whole solution curve rather than just the grid points.
		The leftmost marker for each algorithm represents a grid of about 80 points, with an $\sim 8\%$ increase between consecutive markers.
		When there are too few grid points, RD root tracking is sensitive to their precise location, often failing to detect the bifurcation. This is manifested by large errors to the top-left.
		Error calculation for root tracking ignores the point of heuristic itself, where a cluster mass threshold of $\delta = 0.01$ is crossed (see Section \ref{part:how-and-what}.\ref{sub:RD-root-tracking-near-bifrcations}).
		Solutions along the grid are compared against the analytical ones for a binary source with a Hamming distortion, Section \ref{part:proofs}.\ref{sec:binary-source-with-hamming-dist-appendix}. 
		Grid points were selected uniformly for all algorithms in an attempt to avoid bias in error estimation.
		\textbf{Left}: BA's grid size is increased gradually, from only 80 grid-points at the left to $\sim \text{760,000}$ at the right, with a stopping condition of $10^{-8}$.
		As the grid becomes denser, BA computes at points closer to the bifurcation.
		Due to BA's critical slowing down, its accuracy is therefore reduced as the grid becomes denser, which is clearly noticed to the right.
		cf., the bottom of Figure \ref{fig:derivative-calculation-loses-accuracy-near-bifurcation}.
		\textbf{Right}: BA's stopping condition is varied gradually from $10^{-8}$ to $\sim 10^{-14}$, for a grid of fixed sizes $100, 10^3$ or $10^4$.
		The number of grid points is shown at several markers, for RD root tracking of orders 1, 3, and 6.
		Notice the intersection of the plot for BA with $10^4$ grid-points with RD root tracking of order 6; the latter obtains the entire solution curve with only 458 points, at the same cost and accuracy.
	}
	\label{fig:err-to-computational-cost-tradeoff}
\end{figure}

To assist the reader, this paper is divided into three. 
Part \ref{part:how-and-what} focuses on the necessary details, such as how to compute implicit derivatives for RD problems, and how these could be used to reconstruct the solution curve.
Supporting ideas are elaborated in Part \ref{part:details}: an outline of the derivation of the formulae for RD derivative tensors, a study of RD bifurcations and their relations to root-tracking, and complexity and error analyses.
Part \ref{part:proofs} provides the proofs and technical details omitted elsewhere.
At the beginning of most parts and sections, we placed a short overview to help keep track of the logical flow.
We also provide an annotated source code of our implementation at \newline
\href{https://github.com/shlomiag/RTRD}{https://github.com/shlomiag/RTRD} 

\paragraph*{Notations.}
Vectors are denoted boldface ${\bm{x}}$, its $j$-th coordinate $x_j$ and scalars $x$ in regular-font.
An initial vector condition is denoted $\bm{x}_0$. 
Also, $\bm{\alpha} = (\alpha_0, \bm{\alpha}_+)$ when considering $\alpha_0$ as the zeroth coordinate of $\bm{\alpha}$, and $\bm{\alpha}_+$ for the rest. 
$\Delta[S]$ denotes the \textit{probability simplex} or \textit{simplex} on a (finite) set $S$, $\{ p\in \bb{R}^S: \; p(s) \geq 0 \text{ for } s\in S, \text{ and } \sum_{s\in S} p(s) = 1\}$.
$T$ is the dimension of an unspecified operator $F(\bm{x}, \beta)$, $\bm{x} \in \bb{R}^T$, as in \eqref{eq:solution-as-root-of-functional-eq}.
$N := |\mathcal{X}|$ and $M := |\hat{\mathcal{X}}|$ are the source and reproduction alphabet sizes of a given RD problem $(d, p_X)$. 
A reproduction symbol $\hat{x} \in \hat{\mathcal{X}}$ is also called a \textit{cluster}.

\clearpage
\part{How to track operator roots for rate-distortion problems}
\label{part:how-and-what}

This part aims to present the details necessary for our root-tracking Algorithm \ref{algo:root-tracking-for-RD}, which reconstructs the solution curve of a rate-distortion problem.

To that end, we start in \ref{sub:beta-derivs-at-an-operator-root} with the mathematical observation underlying the calculation of higher implicit multivariate derivatives at an operator root. 
This is accompanied by direct calculations of low-orders, and a simple yet non-trivial example in \ref{subsub:lines-intersecting-parabola-example}, before diving in \ref{sub:high-order-beta-derivatives-at-an-operator-root} into the machinery for implicit derivatives of arbitrary order. 
In \ref{sub:high-order-deriv-tensors-of-BA}, we provide closed-form formulae for the derivative tensors of the Blahut-Arimoto operator. This allows us to specialize the above machinery to implicit derivatives of rate-distortion problems of arbitrary order.

Once the results for implicit derivatives at a point are in place, we use them in \ref{sub:taylor-method-for-RD-root-tracking} to reconstruct the entire solution curve. This requires an understanding of the bifurcations of the operator at hand. For rate-distortion problems, we build on the results of \cite{agmon2021critical}, which are expanded significantly in Section \ref{part:details}.\ref{sec:RD-bifurcations-and-root-tracking}. This allows RD bifurcations to be handled in \ref{sub:RD-root-tracking-near-bifrcations}, completing the algorithm.
Error guarantees for our algorithm, its computational and memory complexities, and its error-complexity tradeoff are discussed in \ref{sub:costs-and-error-to-cost-tradeoff-of-RD-root-tracking-in-main-results-section}. We conclude this part by discussing in \ref{sub:efficient-RD-root-tracking} how our choices could be improved to yield a more efficient variant of the proposed algorithm.

Part \ref{part:details} expands on these results, explaining details that are important yet not strictly necessary for understanding the algorithm.

\medskip
\section{Implicit derivatives at an operator's root, and for rate-distortion problems}
\label{sec:beta-derivs-at-an-operator-root-and-for-RD}

This section provides the tools to calculate implicit derivatives at a root $(\bm{x}_0, \beta_0)$ of an arbitrary operator $F$, culminating with the results necessary for implicit derivatives of rate-distortion problems.

\medskip
\subsection{How implicit derivatives at an operator's root can be calculated}
\label{sub:beta-derivs-at-an-operator-root}

The sequel rests upon the following observation. 
Let $F(\cdot, \beta)$ be an operator on $\bb{R}^T$, $F:\bb{R}^{T}\times\bb{R} \to \bb{R}^T$, and let $(\bm{x}_0, \beta_0)$ be its root, $F(\bm{x}_0, \beta_0) = \bm{0}$ \eqref{eq:solution-as-root-of-functional-eq}, such that
\begin{assumption}[$\bm{x}$ can be written as a function of $\beta$]
	There exists a function $\bm{x}(\beta)$ defined in some neighborhood of $\beta_0$, such that $\bm{x}(\beta)$ is a root of $F$ through $\bm{x}_0$. 		\label{assumption:operator-root-is-a-function-of-beta}
\end{assumption}
\begin{assumption}
	The function $\bm{x}(\beta)$ is differentiable at $\beta_0$ as many times as needed.
	\label{assumption:solution-is-smooth-in-beta}
\end{assumption}
We shall assume throughout that any derivative of $F$ exists, as in the case of the Blahut-Arimoto operators.
Using $\bm{x}(\beta)$ from Assumption \ref{assumption:operator-root-is-a-function-of-beta}, we have a well-defined composition
\begin{equation}		\label{eq:key-observation-composite-is-constant-path}
	\beta \longmapsto \big( \bm{x}(\beta), \beta \big) \longmapsto F\big( \bm{x}(\beta), \beta \big) 
\end{equation}
in the vicinity of $\beta_0$.
Since $\big( \bm{x}(\beta), \beta \big)$ is a root, this composition is nothing but the constant path $\beta \mapsto \bm{0}$ in $\bb{R}^T$. Hence, by Assumption \ref{assumption:solution-is-smooth-in-beta}, its derivatives with respect to $\beta$ must vanish, to any order. That is,
\begin{equation}		\label{eq:high-order-deriv-of-F-implicit}
	\dbetaK{}{l} F\big(\bm{x}(\beta), \beta\big)\Big\rvert_{\beta_0} = \bm{0},
\end{equation}
for any $l \geq 0$. 
In essence, all the results below are encoded in Equation \eqref{eq:high-order-deriv-of-F-implicit}, with the remainder of this section dedicated to extracting the information of interest.

Before specifying \emph{how} implicit derivatives $\dbetaK{\bm{x}}{l}$ \eqref{eq:l-th-deriv-at-beta0} can be calculated from \eqref{eq:high-order-deriv-of-F-implicit}, several clarifications are due.
We shall require throughout that Assumptions \ref{assumption:operator-root-is-a-function-of-beta} and \ref{assumption:solution-is-smooth-in-beta} hold at any root of $F$, except perhaps at points of bifurcation. However, we do \emph{not} require a root $\bm{x}(\beta)$ at \ref{assumption:operator-root-is-a-function-of-beta} to be unique. 
That is, there may be multiple functions $\bm{x}(\beta)$, all of which are roots of $F$ \eqref{eq:solution-as-root-of-functional-eq} through $(\bm{x}_0, \beta_0)$. 
A-priori, it is not clear that $\bm{x}$ can be written as a function of $\beta$, nor that $\bm{x}(\beta)$ is sufficiently differentiable. Both of these assumptions follow from the Implicit Function Theorem; e.g., \cite[Theorem I.1.1]{kielhofer2011bifurcation}.
For, write $D_{\bm{x}} F(\bm{x}_0, \beta_0)$ for the Jacobian matrix of $F$ with respect to its $\bm{x}$ coordinates, evaluated at $(\bm{x}_0, \beta_0)$. Recall, it is the matrix defined by $\left(D_{\bm{x}} F\right)_{i,j} := \frac{\partial F_i}{\partial x_j}$.
When it is non-singular, then the theorem not only implies that there exists a unique function $\bm{x}(\beta)$ through $(\bm{x}_0, \beta_0)$, but also that $\bm{x}(\beta)$ inherits the differentiability properties of $F$, \cite[I.1.7]{kielhofer2011bifurcation}. That is, the derivatives $\dbetaK{\bm{x}}{l}\big\rvert_{\beta_0}$ \eqref{eq:l-th-deriv-at-beta0} exist, to any order.
However, we shall usually only require that Assumptions \ref{assumption:operator-root-is-a-function-of-beta} and \ref{assumption:solution-is-smooth-in-beta} hold, rather than than the stronger condition that $D_{\bm{x}} F$ is non-singular.
This is reasonable for rate-distortion problems --- see Section \ref{sub:taylor-method-for-RD-root-tracking} for details.
In general, however, understanding the bifurcations of $F$ is necessary --- e.g., Section \ref{part:details}.\ref{sec:RD-bifurcations-and-root-tracking} for RD.
For example, both assumptions hold for the constant one-dimensional operator $F(x, \beta) := 0$, whose Jacobian vanishes everywhere despite having no bifurcations. 
While on the other hand, the Jacobian of Example \ref{subsub:lines-intersecting-parabola-example} below is singular precisely at the point of bifurcation there, where its two solution curves intersect and annihilate each other (see Figure \ref{fig:line-intersecting-parabola-problem-def}).

\medskip 
We start with a gentle low-dimensional introduction as to how the derivatives $\dbetaK{\bm{x}}{l}$ \eqref{eq:l-th-deriv-at-beta0} can be calculated.
Write $F = (F_1, \dots, F_T)$ for the coordinates of $F$. Applying the multivariate chain-rule to the first-order equation $\tfrac{d}{d\beta} F = \bm{0}$ \eqref{eq:high-order-deriv-of-F-implicit} reads,
\begin{equation}		\label{eq:first-order-expansion-of-operator-eq-implicit-written-full}
	\dbeta{}F_i = \sum_{j=1}^T \frac{\partial F_i}{\partial x_j} \frac{dx_j}{d\beta} + \frac{\partial F_i}{\partial \beta} = 0
\end{equation}
for each $i = 1, \dots, T$. 
Functions are understood henceforth to be evaluated at $(\bm{x}_0, \beta_0)$, unless otherwise stated.
Write $D_\beta F$ for the vector whose $i$-th entry is $\tfrac{\partial F_i}{\partial \beta}$. 
The calculation of $\dbeta{\bm{x}}$ thus boils down to solving the linear equation\footnote{ It suffices to find a linear pre-image under $D_{\bm{x}} F$ if it is not invertible.}
\begin{equation}		\label{eq:ODE-implicit-form}
	D_{\bm{x}} F \; \tfrac{d\bm{x}}{d\beta} = -D_\beta F \;.
\end{equation}
This is an implicit ordinary differential equation (ODE), describing how the root $\bm{x}$ evolves with $\beta$. cf., Section \ref{sub:taylor-method-for-RD-root-tracking} on the analyticity of RD roots. 
This ODE is an immediate consequence of the Implicit Function Theorem, \cite[Theorem 5]{de2014implicit}, at least when the Jacobian is non-singular.
We derive its explicit form \eqref{eq:RD-ODE-in-main-text} for RD later, in Theorem \ref{thm:beta-ODE-in-marginal-coords} (Section \ref{part:details}.\ref{sub:encoders-beta-derivatives}). 
A marginal distribution $\inputmarginalVect$ of full support on the reproduction alphabet $\hat{\cal{X}}$ which is a fixed-point of Blahut-Arimoto satisfies
\begin{equation}		\label{eq:RD-ODE-in-main-text}
	\sum_{\hat{x}'} A_{\hat{x}, \hat{x}'} \dbeta{\inputmarginal{'}} = 
	\bb{E}_{p(\hat{x}', x)} \left[\intermediateencoder{}{} d(x, \hat{x}') \right]
	- \bb{E}_{p_X} \left[\intermediateencoder{}{} d(x, \hat{x}) \right] 
\end{equation}
for every $\hat{x}$, where $A_{\hat{x}, \hat{x}'}$ is given by \eqref{eq:A-matrix-def}, and $\intermediateencoderVect$ is the test channel or \textit{encoder} corresponding to $\inputmarginalVect$ by the BA Equation \eqref{eq:encoder-eq}. 

Before proceeding to the higher-order counterparts of the implicit ODE \eqref{eq:ODE-implicit-form}, some classic material on higher derivatives is necessary. 
e.g., \cite[VIII.12]{dieudonne1969foundations} or the very readable \cite[10.3]{aguilar2021analysis} for the following. 
Let $f(\bm{x}, \beta)$ be a real-valued function on $\bb{R}^T \times \bb{R}$, $f: \bb{R}^T \times \bb{R} \to \bb{R}$. Fixing $\beta$ at $\beta_0$ for a moment, its gradient is the vector $(\tfrac{\partial f}{\partial x_1}, \dots, \tfrac{\partial f}{\partial x_T})$.
The gradient can be considered as a linear functional, mapping a vector $\bm{v} \in \bb{R}^T$ to $v_1 \cdot \tfrac{\partial f}{\partial x_1} + \dots + v_T \cdot \tfrac{\partial f}{\partial x_T}$. This is useful when considering the first-order Taylor expansion $f(\bm{x}_0, \beta_0) + \sum_j v_j \tfrac{\partial f}{\partial x_j}$ of $f(\cdot, \beta_0)$ about $\bm{x}_0$, where $\bm{v} := \bm{x} - \bm{x}_0$ is the deviation from the point of expansion.
Similarly, the Hessian matrix $\big( \tfrac{\partial^2 f}{\partial x_i \partial x_j} \big)_{i, j=1}^T$ of $f$ is a bi-linear map. Namely, it maps a pair of vectors $\bm{u}, \bm{v} \in \bb{R}^T$ to $\sum_{i, j} u_i v_j \cdot \tfrac{\partial^2 f}{\partial x_i \partial x_j}$, in a manner which is separately linear in each of the two vectors.
The second-order expansion of $f$ maps a deviation $\bm{v} \in \bb{R}^T$ from the basepoint $\bm{x}_0$ to $f(\bm{x}_0, \beta_0) + \sum_j v_j \tfrac{\partial f}{\partial x_j} + \sum_{i, j} v_i v_j \cdot \tfrac{\partial^2 f}{\partial x_i \partial x_j}$.

Next, replace the real-valued $f$ with the $i$-th component $F_i$ of the operator $F = (F_1, \dots, F_T)$ that we have started with, and allow derivatives also with respect to $\beta$. The above generalizes as follows.
For orders $b, m \geq 0$ and a coordinate $1 \leq i \leq T$ (all integers), denote by $D_{\beta^b, \bm{x}^m}^{b + m} F_i$ the following symmetric\footnote{ 
	An $m$-multilinear map $T$ is \emph{symmetric} if its value is unchanged by an arbitrary permutation of its arguments:
	$T[\bm{v}_1, \dots, \bm{v}_{m}] = T[\bm{v}_{\sigma(1)}, \dots, \bm{v}_{\sigma(m)}]$ for any permutation $\sigma$ of $\{1, 2, \dots, m\}$.
} $m$-multilinear map, which is defined on $\bm{v}_1, \dots, \bm{v}_{m} \in \bb{R}^T$ by
\begin{equation}		\label{eq:mixed-deriv-def-evaluated-applied-to-vectors}
	D_{\beta^b, \bm{x}^m}^{b + m} F_i\Big\rvert_{(\bm{x}_0, \beta_0)}[\bm{v}_1, \dots, \bm{v}_{m}] :=
	\sum_{1 \leq i_1, i_2, \dots, i_{m} \leq T} 
	v_{1, i_1} \cdots v_{m, i_{m}} \cdot 
	\frac{\partial^{m}}{\partial x_{i_1} \partial x_{i_2} \cdots \partial x_{i_{m}} } \partialbetaK{F_i}{b}\left(\bm{x}_0, \beta_0 \right) \;.
\end{equation}
Just as in the second-order Hessian, each index $i_j$ varies independently over all the coordinates $1, \dots, T$ of $\bm{v}_j$, for $j = 1, \dots, m$. 
We write $D_{\beta^b, \bm{x}^m}^{b + m} F$ for the vector whose $i$-th coordinate is $D_{\beta^b, \bm{x}^m}^{b + m} F_i$.
Evaluation at $(\bm{x}_0, \beta_0)$ shall be omitted whenever clear from the context. 
For practical purposes, the $i$-th multilinear map $D_{\beta^b, \bm{x}^m}^{b + m} F_i$ may be considered as a ``matrix'' with $m$ axes, and $D_{\beta^b, \bm{x}^m}^{b + m} F$ as a ``matrix'' with $m + 1$ axes. 
For example, the vector $D_{\bm{x}} F_i$ is the gradient of the (scalar) function $F_i$, while $D_{\bm{x}} F$ is the $T$-by-$T$ Jacobian matrix of $F$. 
The integers $b$ and $m$ are the orders of differentiation with respect to $\beta$ and the coordinates of $\bm{x}$, respectively. 
For $b > 0$, $D_{\beta^b, \bm{x}}^{b+1} F_i$ and $D_{\beta^b, \bm{x}}^{b+1} F$ are derivative tensors of higher order, although of the same respective sizes as $D_{\bm{x}} F_i$ and $D_{\bm{x}} F$.
When $b$ and $i$ are given, one can think of the $m$-tuples $\left(i_1, i_2, \dots, i_{m}\right)$ indexing $D_{\beta^b, \bm{x}^m}^{b + m} F_i$ in \eqref{eq:mixed-deriv-def-evaluated-applied-to-vectors} as differentiations in some particular order. First derive $\partialbetaK{F_i}{b}$ with respect to $x_{i_m}$, then with respect to $x_{i_{m-1}}$, and so forth up to $x_{i_1}$. This notation is redundant, in the sense that exchanging distinct indices $i_{j_1} \neq i_{j_2}$ yields the same partial derivative when $F_i$ is well-behaved.

When $i$ is understood from the context, it is sometimes convenient to index $D_{\beta^b, \bm{x}^m}^{b + m} F_{i}$ by a multi-index $\bm{\alpha} = (\alpha_0, \bm{\alpha}_+)\in \bb{N}_0\times \bb{N}_0^{T}$, with $\bb{N}_0$ the non-negative integers. $\alpha_0$ then represents the number $b$ of differentiations with respect to $\beta$, and $\alpha_j$ the number of differentiations with respect to $x_j$ for $1 \leq j \leq T$.
Define $|\bm{\alpha}| := \sum_j \alpha_j$. When setting $\bm{\alpha} = (1, 3, 2) \in \bb{N}_0^{1+2}$ for an operator $F$ on $T = 2$ dimensions, $\tfrac{\partial^6}{\partial \beta \; \partial x_1^3 \partial x_2^2 } F_i$ for example is shortened to $\frac{\partial^{|\bm{\alpha}|}}{\partial \beta^{\alpha_0} \partial \bm{x}^{\bm{\alpha}_+}} F_i = \frac{\partial^6}{\partial \beta \partial \bm{x}^{(3, 2)}} F_i$, for $i=1, 2$.
Because the order of differentiation does not matter when $F_i$ is well-behaved, then $\bm{\alpha}$ corresponds not only to the entry $(i_1, \dots, i_5) = (1, 1, 1, 2, 2)$ of $D_{\beta^1, \bm{x}^5}^{6} F_i$ \eqref{eq:mixed-deriv-def-evaluated-applied-to-vectors}, but also to its permutations. 
See also the multivariate notation definitions \eqref{eq:multivariate-notation-defs}, in Section \ref{part:details}.\ref{sec:multivariate-faa-di-brunos-formula}. 

At the next order, differentiating $\dbeta{}F_i = 0$ \eqref{eq:first-order-expansion-of-operator-eq-implicit-written-full} with respect to $\beta$ yields
\begin{equation}		\label{eq:second-order-expansion-of-operator-eq-implicit-written-full}
	\dbetaK{}{2}F_i =
	\sum_{j, k} \frac{\partial^2 F_i}{\partial x_k \partial x_j} \frac{dx_j}{d\beta} \frac{dx_k}{d\beta} +
	\sum_j \frac{\partial^2 F_i}{\partial \beta \partial x_j} \frac{dx_j}{d\beta} +
	\sum_j \frac{\partial F_i}{\partial x_j} \frac{d^2 x_j}{d\beta^2} +
	\sum_j \frac{\partial^2 F_i}{\partial x_j \partial \beta} \frac{dx_j}{d\beta} +
	\frac{\partial^2 F_i}{\partial\beta^2} = 0 \;.
\end{equation}
Unlike the first-order expansion \eqref{eq:first-order-expansion-of-operator-eq-implicit-written-full}, the second-order one \eqref{eq:second-order-expansion-of-operator-eq-implicit-written-full} contains a 3-tensor term $\sum \frac{\partial^2 F_i}{\partial x_k \partial x_j} \tfrac{dx_j}{d\beta} \tfrac{dx_k}{d\beta}$. 
By definition \eqref{eq:mixed-deriv-def-evaluated-applied-to-vectors}, the latter is $D^2_{\bm{x}, \bm{x}}F_i [\tfrac{d\bm{x}}{d\beta}, \tfrac{d\bm{x}}{d\beta}]$, while the mixed-derivatives term following it is $D^2_{\beta, \bm{x}}F_i [\tfrac{d\bm{x}}{d\beta}]$, and so forth. 
Rewriting the first few expansion orders $\dbetaK{F}{l} = \bm{0}$ \eqref{eq:high-order-deriv-of-F-implicit} in this notation,
\begin{align}
	\bm{0} &= \tfrac{d^0 F}{d\beta^0} = F \tag{\ref{eq:solution-as-root-of-functional-eq}}	\\
	\bm{0} &= \tfrac{d^1 F}{d\beta^1} = D_{\bm{x}} F[\tfrac{d\bm{x}}{d\beta}] + D_\beta F \label{eq:first-order-expansion-of-operator-eq-implicit}	\\
	\bm{0} &= \tfrac{d^2 F}{d\beta^2} = D_{\bm{x}} F[\tfrac{d^2 \bm{x}}{d\beta^2}] + D^2_{\bm{x}, \bm{x}}F[\tfrac{d\bm{x}}{d\beta}, \tfrac{d\bm{x}}{d\beta}] + 
	2 D^2_{\beta, \bm{x}}F [\tfrac{d \bm{x}}{d\beta}] + D^2_{\beta, \beta}F \label{eq:second-order-expansion-of-operator-eq-implicit}	\\
	\bm{0} &= \tfrac{d^3 F}{d\beta^3} = D_{\bm{x}} F[\tfrac{d^3 \bm{x}}{d\beta^3}] +
	D^3_{\bm{x}, \bm{x}, \bm{x}}F[\tfrac{d\bm{x}}{d\beta}, \tfrac{d\bm{x}}{d\beta}, \tfrac{d\bm{x}}{d\beta}] + 
	3 D^2_{\bm{x}, \bm{x}}F[\tfrac{d^2\bm{x}}{d\beta^2}, \tfrac{d\bm{x}}{d\beta}] 
	\nonumber \\ & \quad \quad +
	3 D^3_{\beta, \bm{x}, \bm{x}} F[\tfrac{d\bm{x}}{d\beta}, \tfrac{d\bm{x}}{d\beta}] + 
	3 D^2_{\beta, \bm{x}}F[\tfrac{d^2 \bm{x}}{d\beta^2}] + 
	3 D^3_{\beta, \beta, \bm{x}}F[\tfrac{d\bm{x}}{d\beta}] + 
	D^3_{\beta, \beta, \beta}F		\label{eq:third-order-expansion-implicit}	
\end{align}
where the third-order expansion \eqref{eq:third-order-expansion-implicit} for $\dbetaK{}{3}F$ follows from \eqref{eq:second-order-expansion-of-operator-eq-implicit} by a straightforward calculation. 

As can be seen in \eqref{eq:first-order-expansion-of-operator-eq-implicit}-\eqref{eq:third-order-expansion-implicit}, the implicit derivative $\dbetaK{\bm{x}}{l}$ \eqref{eq:l-th-deriv-at-beta0} of highest-order appears only once in the $l$-th order expansion, at the product $D_{\bm{x}} F[\tfrac{d^l \bm{x}}{d\beta^l}]$ of $\tfrac{d^l \bm{x}}{d\beta^l}$ by the linear map $D_{\bm{x}} F$. 
While, the other terms in each equation contain only derivatives $\dbetaK{\bm{x}}{k}$ of lower orders $0 < k < l$.
This holds for any order $l > 0$, as we show in Section \ref{sub:high-order-beta-derivatives-at-an-operator-root} below. 
So, in principle, all one needs to do in order to obtain the derivatives $\dbetaK{\bm{x}}{l}$ \eqref{eq:l-th-deriv-at-beta0} of \textit{any} order $l > 0$ is to solve these equations recursively:
\begin{enumerate}
	\item Suppose that the lower-order derivatives $\dbetaK{\bm{x}}{k}$ are known for $0 < k < l$.
	\item Calculate all the derivative tensors of $F$, up to order $l$.		\label{point:calculating-all-derivatives-of-F}
	\item Evaluate the multilinear forms in the expansion of $\dbetaK{}{l}F = \bm{0}$, except for $D_{\bm{x}} F[\tfrac{d^l \bm{x}}{d\beta^l}]$.
	\item Solve a linear equation in $\tfrac{d^l \bm{x}}{d\beta^l}$ with coefficients $D_{\bm{x}} F$.
\end{enumerate}
Up to the technicalities of a general-order machinery this is the heart of Algorithm \ref{algo:high-order-derivs-of-operator-roots} below, for computing implicit multivariate derivatives of arbitrary order. 
Once these technicalities are settled in Section \ref{sub:high-order-beta-derivatives-at-an-operator-root}, there are two main tasks that one needs to tackle in order to obtain the derivatives $\dbetaK{\bm{x}}{l}$ \eqref{eq:l-th-deriv-at-beta0} at an operator root, and to be able to use them.
First, one needs to calculate the derivative tensors of $F$, as in point \ref{point:calculating-all-derivatives-of-F} above. We have accomplished this for rate-distortion problems, providing in Section \ref{sub:high-order-deriv-tensors-of-BA} closed-form formulae for the derivative tensors of $Id - BA_\beta$ \eqref{eq:RD-operator-def}, of any order.
Second, one needs to tell whether the Taylor series \eqref{eq:solution-by-beta-as-taylor-approx} for $\bm{x}(\beta)$ around $\beta_0$ indeed approximates the true solution. For, operator roots may cease to be a function of $\beta$, cease to be differentiable, or even cease to exist at critical $\beta$ values. This boils down to understanding the solution structure of $F$, or equivalently its bifurcations. 
While \ref{subsub:lines-intersecting-parabola-example} below provides an example where Assumption \ref{assumption:operator-root-is-a-function-of-beta} breaks (two roots collide and annihilate), the discussion for rate-distortion is subtler. We provide guarantees in Section \ref{sub:high-order-deriv-tensors-of-BA}, with the full details postponed till Section \ref{part:details}.\ref{sec:RD-bifurcations-and-root-tracking}.

\medskip
\subsubsection{Example: implicit derivatives of line intersections with a parabola}
\label{subsub:lines-intersecting-parabola-example}

Before presenting in Section \ref{sub:high-order-beta-derivatives-at-an-operator-root} the fully fledged machinery of arbitrary order, we give a simple yet non-trivial example how implicit derivatives $\dbetaK{\bm{x}}{l}$ can be calculated at an operator root.

\medskip
\begin{figure}[h!]
		\centering
		\vspace{-5pt}
		\includegraphics[trim={0 0 0 1cm}, clip, width=0.52\textwidth]{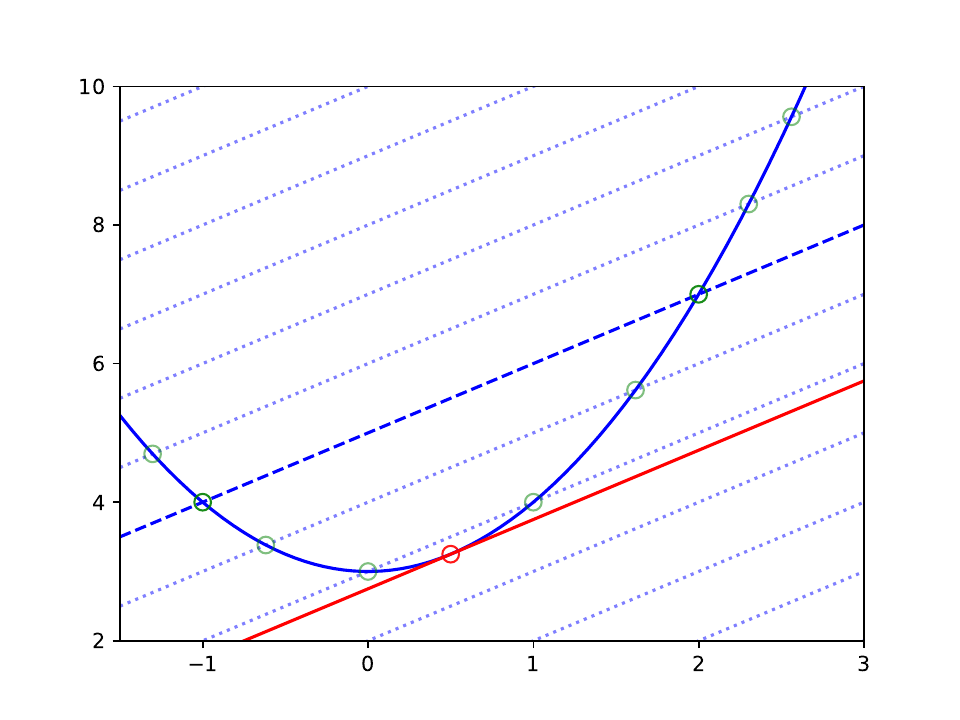}
		\caption{
			\textbf{Intersections of $y = x^2 + 3$ with the lines $y = x + \beta$.} The line at $\beta = 5$ is dashed blue, and at other integral $\beta$ values in dotted blue. 
			At $\beta_c = 2.75$ (solid red line) the system $F = \bm{0}$ \eqref{eq:line-intersecting-parabola-example} undergoes a bifurcation: there are two distinct roots above $\beta_c$, none below it, and exactly one at the bifurcation point itself.
		}
		\label{fig:line-intersecting-parabola-problem-def}
\end{figure}

A line $y = a x + \beta$ in the plane typically intersects a parabola $y = b x^2 + c x + d$ (with $b\neq 0$) at either two points or none, as in Figure \ref{fig:line-intersecting-parabola-problem-def} for instance. In the special case when it is tangent to the parabola, it intersects at a single point. Consider the problem of tracking the intersection point as the line is being translated by varying $\beta$. 

Setting
\begin{equation}		\label{eq:line-intersecting-parabola-example}
	F(x, y; \beta) := \mat{
		-y + b x^2 + c x + d \\
		-y + a x + \beta
	} \;,
\end{equation}
the intersection point is encoded by requiring $F = \bm{0}$. 
Assume that $(x_0, y_0)$ is known to be an intersection point at $\beta_0$, $F(x_0, y_0; \beta_0) = \bm{0}$, and that at the vicinity of $(x_0, y_0)$, the intersection point can be written as a function of $\beta$, $(x, y) = \big(x(\beta), y(\beta)\big)$. 
By calculating the derivative tensors of $F$ \eqref{eq:line-intersecting-parabola-example} and plugging them into the first few expansions of $\dbetaK{F}{l} = \bm{0}$ (e.g., \eqref{eq:first-order-expansion-of-operator-eq-implicit}-\eqref{eq:third-order-expansion-implicit}), we can solve for the implicit derivatives $\dbetaK{x}{l}, \dbetaK{y}{l}$ at $(x_0, y_0; \beta_0)$. See Section \ref{part:proofs}.\ref{sec:calculations-for-lines-intersecting-parabola-example-appendix} for full details. 
Doing so till fourth order yields a Taylor expansion \eqref{eq:solution-by-beta-as-taylor-approx} of the intersection point,
\begin{multline}			\label{eq:line-parabola-example-fourth-order-sol}
	\mat{x(\beta) \\ y(\beta)} \approx 
	\mat{x_0 \\ y_0} + 
	\frac{1}{\Delta} \mat{1 \\ a+\Delta} \cdot \left(\beta - \beta_0\right)
	-\frac{b}{\Delta^3} \mat{1 \\ a} \cdot \left(\beta - \beta_0\right)^2 \\ +
	\frac{2b^2}{\Delta^5} \mat{1 \\ a} \cdot \left(\beta - \beta_0\right)^3
	-\frac{5b^3}{\Delta^7} \mat{1 \\ a} \cdot \left(\beta - \beta_0\right)^4
\end{multline}
Where, $\Delta := 2 b x_0 + c - a$ vanishes precisely when the slopes $2 b x_0 + c$ of the parabola and $a$ of the line coincide. At this point a bifurcation occurs, as can be seen in Figure \ref{fig:line-intersecting-parabola-problem-def}: 
the two distinct intersection points of this problem merge then into one, and disappear beyond the critical $\beta$ value. That is, the point of intersection $(x_0, y_0)$ can no longer be written as a function of $\beta$.
Figure \ref{fig:line-intersecting-parabola-Taylor-approximations} demonstrates the first few expansion orders \eqref{eq:line-parabola-example-fourth-order-sol}, when expanded near or far of the bifurcation. 

\begin{figure}[h!]
	\vspace{-10pt}
	\centering
	\includegraphics[trim={0 0.5cm 0 0}, clip, width=1.\textwidth]{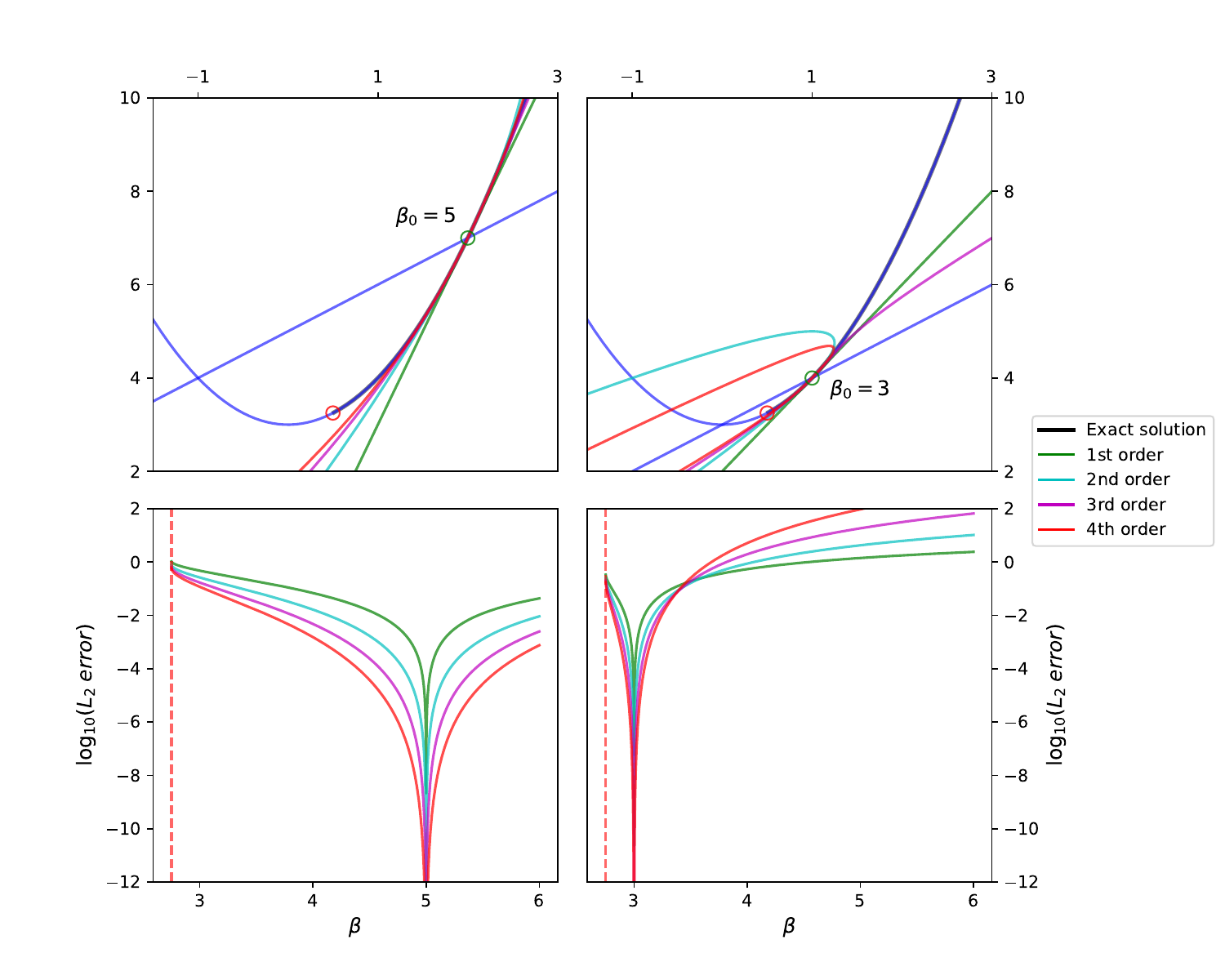}
	\caption{
		\textbf{Taylor approximations versus the exact solution \eqref{eq:line-parabola-example-exact-solution} to $F = \bm{0}$ \eqref{eq:line-intersecting-parabola-example}, by order and choice of base-point}, for the line-intersecting-parabola problem of Figure \ref{fig:line-intersecting-parabola-problem-def}.
		\textbf{Top:} the expansions \eqref{eq:line-parabola-example-fourth-order-sol} of $(x(\beta), y(\beta))$ around $\beta_0$, colored by order, and the exact solution \eqref{eq:line-parabola-example-exact-solution} in black.
		\textbf{Bottom:} the expansions' $L_2$-error from the exact solution. The comparison ends at the bifurcation point $\beta_c = 2.75$ (dashed red verticals), beyond which no solution exists. 
		\textbf{Right:} 
		the Jacobian $D_{\bm{x}} F$ becomes increasingly singular as the base-point $\beta_0$ is taken closer to $\beta_c$ (see Section \ref{part:proofs}.\ref{sec:calculations-for-lines-intersecting-parabola-example-appendix}). 
		As a result, the implicit derivatives diverge and the approximations' quality deteriorates. 
		cf., the comments after Theorem \ref{thm:formula-for-high-order-expansion-of-F-in-main-result-sect} (in Section \ref{sub:high-order-beta-derivatives-at-an-operator-root}) and Section \ref{part:details}.\ref{sub:local-error-estimates-for-beta-derivs} on the latter implication. 
	}
	\label{fig:line-intersecting-parabola-Taylor-approximations}
\end{figure}

For this simple problem, the equation $F = \bm{0}$ \eqref{eq:line-intersecting-parabola-example} can also be solved analytically, yielding 
\begin{equation}		\label{eq:line-parabola-example-exact-solution}
	x(\beta) = \frac{a - c}{2 b} \pm \sqrt{\left(\frac{a - c}{2 b}\right)^2 + \frac{\beta - d}{b}}
	\quad \text{and} \quad
	y(\beta) = \beta + a x(\beta)
\end{equation}
Our approximation \eqref{eq:line-parabola-example-fourth-order-sol} is nothing but the fourth-order Taylor expansion of the exact solution \eqref{eq:line-parabola-example-exact-solution} around $\beta_0$, as can be verified by expanding the latter directly. 
However, in contrast to the Taylor expansion of \eqref{eq:line-parabola-example-exact-solution}, the calculations leading to \eqref{eq:line-parabola-example-fourth-order-sol} can be carried out even when an exact solution is \textit{not} available, as is typical for rate-distortion problems.

\medskip
\subsection{High-order implicit derivatives at an operator's root}
\label{sub:high-order-beta-derivatives-at-an-operator-root}

We provide a machinery for calculating high-order implicit derivatives at an operator's root, built on the reasoning of Section \ref{sub:beta-derivs-at-an-operator-root}. In particular, we assume henceforth that Assumptions \ref{assumption:operator-root-is-a-function-of-beta} and \ref{assumption:solution-is-smooth-in-beta} there hold.
As before, implicit derivatives and derivative tensors are understood to be evaluated at $\left(\bm{x}_0, \beta_0\right)$.

\medskip
But first, a few definitions.
A \textit{partition} of an integer $n > 0$ is a sequence of positive integers $0 < p_1 \leq p_2 \leq \dots \leq p_m$ whose sum is $n$, $\sum_{i=1}^m p_i = n$. The integers $p_i$ are called the \textit{parts} of the partition. 
The \textit{partition function} $p(n)$ is the number of partitions of $n$. e.g., \cite{andrews1998theory}. 
When interested in the number of times $m_i$ a part $p_i$ appears (its \textit{multiplicity}), we shall sometimes write $(m_1)\cdot p_1 + \dots + (m_s)\cdot p_s$, with $0 < p_1 < \dots < p_s$ now the partition's distinct parts. The \textit{total multiplicity} is the number of parts $m := m_1 + \dots + m_s$ in a partition. 
For example, there are three partitions of the integer $n = 3$, $p(3) = 3$. Namely, $1 + 2$, $1 + 1 + 1$, and the \textit{trivial partition} $3$. Written by multiplicity, these are $(1)\cdot 1 + (1)\cdot 2$, $(3)\cdot 1$ and $(1)\cdot 3$ respectively, of total multiplicities 2, 3 and 1.

We shall need to apply the $m$-multilinear maps $D_{\beta^b, \bm{x}^m}^{b + m} F_i$ \eqref{eq:mixed-deriv-def-evaluated-applied-to-vectors} with repeated arguments. 
For a lack of better notation, if $\bm{v}_1$ appears as an argument $m_1$ times, $\bm{v}_2$ appears $m_2$ times, till $\bm{v}_s$ appearing $m_s$ times, we denote
\begin{equation}
	D_{\beta^b, \bm{x}^m}^{b + m} F_i[(\bm{v}_1)_{\times m_1}, \dots, (\bm{v}_{s})_{\times m_s}] :=
	D_{\beta^b, \bm{x}^m}^{b + m} F_i[\underset{m_1 \text{ times}}{\underbrace{\bm{v}_1, \dots, \bm{v}_1}}, \dots, 
	\underset{m_s \text{ times}}{\underbrace{\bm{v}_{s}, \dots, \bm{v}_{s}}}]
\end{equation}
With that, the arbitrary-order counterparts $\tfrac{d^l }{d\beta^l} F$ of \eqref{eq:first-order-expansion-of-operator-eq-implicit}-\eqref{eq:third-order-expansion-implicit} are as follows.

\begin{thm}			\label{thm:formula-for-high-order-expansion-of-F}
	Let $F\big(\bm{x}(\beta), \beta\big)$ be the composition of $F:\bb{R}^T\times \bb{R} \to \bb{R}^T$ with a path $\bm{x}(\beta)$ in $\bb{R}^T$, both of which are sufficiently differentiable.
	Its derivative of order $l > 0$ can be written as
	\begin{multline}		\label{eq:formula-for-high-order-expansion-of-F}
		\frac{d^l }{d\beta^l} F(\bm{x}(\beta), \beta) =
		\sum_{\text{partitions}} 
		\sum_{b=0}^{m_1 \cdot \delta(p_1 = 1)}
		\frac{l!}{b! (m_1 - b)! m_2! \cdots m_s! \cdot (p_1!)^{m_1} \cdots (p_s!)^{m_s}} \\ \cdot
		D^m_{\beta^b, \bm{x}^{m-b}} F\left[ \left( \tfrac{d^{p_1}\bm{x}}{d\beta^{p_1}} \right)_{\times (m_1 - b)}, 
		\left( \tfrac{d^{p_2}\bm{x}}{d\beta^{p_2}} \right)_{\times m_2}, 
		\dots, \left( \tfrac{d^{p_s}\bm{x}}{d\beta^{p_s}} \right)_{\times m_s} \right]
	\end{multline}
	where the outer sum is over the integer partitions $(m_1)\cdot p_1 + \dots + (m_s)\cdot p_s$ of $l$, $m := m_1 + \dots + m_s$,	and $\delta$ is Dirac's delta function.
\end{thm}
The inner summation at \eqref{eq:formula-for-high-order-expansion-of-F} goes up to the multiplicity of 1 in a partition, or zero if there is no part of size 1.
This theorem is an application of a modern version of the multivariate Fa\`a di Bruno's formula, \cite{ma2009higher}, which is elaborated in Section \ref{part:details}.\ref{sec:multivariate-faa-di-brunos-formula}.
See Section \ref{part:proofs}.\ref{sub:proof-of-formula-for-high-order-expansion-of-F} for a proof.

To illustrate Theorem \ref{thm:formula-for-high-order-expansion-of-F}, we order the summands in the expansion of $\tfrac{d^3}{d\beta^3} F$ by integer partitions. Recall the third-order expansion \eqref{eq:third-order-expansion-implicit},
\begin{multline}		\label{eq:third-order-expansion-implicit-at-thm-example}
	\tfrac{d^3 F}{d\beta^3} = 
	D_{\bm{x}} F[\tfrac{d^3 \bm{x}}{d\beta^3}] +
	D^3_{\bm{x}, \bm{x}, \bm{x}}F[\tfrac{d\bm{x}}{d\beta}, \tfrac{d\bm{x}}{d\beta}, \tfrac{d\bm{x}}{d\beta}] + 
	3 D^2_{\bm{x}, \bm{x}}F[\tfrac{d^2\bm{x}}{d\beta^2}, \tfrac{d\bm{x}}{d\beta}] 
	\\ \quad \quad +
	3 D^3_{\beta, \bm{x}, \bm{x}} F[\tfrac{d\bm{x}}{d\beta}, \tfrac{d\bm{x}}{d\beta}] + 
	3 D^2_{\beta, \bm{x}}F[\tfrac{d^2 \bm{x}}{d\beta^2}] + 
	3 D^3_{\beta, \beta, \bm{x}}F[\tfrac{d\bm{x}}{d\beta}] + 
	D^3_{\beta, \beta, \beta}F \;.
\end{multline}
Each term above corresponds to a partition and to a number $b$ of differentiations with respect to $\beta$, indexing the summations at \eqref{eq:formula-for-high-order-expansion-of-F}.
The part sizes and their multiplicities can be read off the tensors' arguments at \eqref{eq:third-order-expansion-implicit-at-thm-example}. For example, $D_{\bm{x}} F[\tfrac{d^3 \bm{x}}{d\beta^3}]$ corresponds to a partition with one part of size 3, and $b=0$ derivations with respect to $\beta$. While, $D^2_{\beta, \bm{x}}F[\tfrac{d^2 \bm{x}}{d\beta^2}]$ corresponds to $(1)\cdot 1 + (1)\cdot 2$, and $b = 1$; a first-order argument $\tfrac{d \bm{x}}{d\beta}$ corresponding to the part of size 1 is left out because $b > 0$. 
Proceeding this way, we obtain Table \ref{tab:details-of-third-order-expansion-of-F}.

\begin{table}[h!]
	\begin{center}
		\vspace{-5pt}
		\setlength{\tabcolsep}{16pt}
		\renewcommand{\arraystretch}{1.6}
		\begin{tabular}{ cccc|c }
			$l=3$	\hspace{-15pt}	&		$(3) \cdot 1$		&		$(1)\cdot 1 + (1)\cdot 2$	&		$(1)\cdot 3$	& \\ \hline 
			&		$D^3_{\bm{x},\bm{x},\bm{x}}F[\tfrac{d\bm{x}}{d\beta}, \tfrac{d\bm{x}}{d\beta}, \tfrac{d\bm{x}}{d\beta}]$	&		$3D^2_{\bm{x},\bm{x}}F[\tfrac{d^2 \bm{x}}{d\beta^2}, \tfrac{d\bm{x}}{d\beta}]$	&		$D_{\bm{x}} F[\tfrac{d^3 \bm{x}}{d\beta^3}]$	&		$b=0$	\\ 
			&		$3D^3_{\beta,\bm{x},\bm{x}}F[\tfrac{d\bm{x}}{d\beta}, \tfrac{d\bm{x}}{d\beta}]$	&	$3D^2_{\beta,\bm{x}}F[\tfrac{d^2 \bm{x}}{d\beta^2}]$		&			&		$b=1$	\\
			&		$3D^3_{\beta,\beta,\bm{x}}F[\tfrac{d\bm{x}}{d\beta}]$	&			&			&		$b=2$	\\
			&		$D^3_{\beta,\beta,\beta}F$	&			&			&		$b=3$
		\end{tabular}
		\caption{The summands of Theorem \ref{thm:formula-for-high-order-expansion-of-F} for $l = 3$. The expansion terms of $\tfrac{d^3 F}{d\beta^3}$ \eqref{eq:third-order-expansion-implicit-at-thm-example} are ordered by the integer partitions of 3 and by the number $b$ of differentiations with respect to $\beta$.}
		\label{tab:details-of-third-order-expansion-of-F}
		\vspace{-10pt}
	\end{center}
\end{table}

No part in a partition of $l$ can be of size larger than $l$. Clearly, the trivial partition $(1)\cdot l$ is the only one with a part of size $l$. Thus, among the derivatives $\tfrac{d^k\bm{x}}{d\beta^k}$ for $k > 0$, $\tfrac{d^l\bm{x}}{d\beta^l}$ is that of highest-order which appears in $\dbetaK{}{l} F$ \eqref{eq:formula-for-high-order-expansion-of-F}. It appears there once, at the term indexed by the trivial partition $(1)\cdot l$ and $b = 0$.

\begin{cor}			\label{cor:highest-order-deriv-only-with-Jacobian}
	The derivative $\tfrac{d^l\bm{x}}{d\beta^l}$ appears once in the $l$-th order expansion \eqref{eq:formula-for-high-order-expansion-of-F} of $\dbetaK{}{l} F$, at the term $D_{\bm{x}} F[\tfrac{d^l\bm{x}}{d\beta^l}]$. 
	Any other term there contains only derivatives $\tfrac{d^k\bm{x}}{d\beta^k}$ of lower orders, $k < l$.
\end{cor}

With that, the derivative of highest-order $\tfrac{d^l\bm{x}}{d\beta^l}$ can be isolated easily from $\dbetaK{}{l} F = \bm{0}$ \eqref{eq:high-order-deriv-of-F-implicit}.
As an immediate consequence of Theorem \ref{thm:formula-for-high-order-expansion-of-F}, we obtain

\begin{thm}			\label{thm:formula-for-high-order-expansion-of-F-in-main-result-sect}
	Let $l > 1$, $\bm{x} = \bm{x}(\beta)$ a root of $F = \bm{0}$ \eqref{eq:solution-as-root-of-functional-eq}, and assume that the derivatives $\tfrac{d^k \bm{x}}{d\beta^k}$ are known for all $0 < k < l$. Then, 
	\begin{multline}		\label{eq:formula-for-high-order-beta-derivatives}
		D_{\bm{x}} F\big[ \tfrac{d^l \bm{x}}{d\beta^l} \big] = 
		- \sum_{\substack{\text{non-trivial} \\ \text{partitions}}}
		\sum_{b=0}^{m_1 \cdot \delta(p_1 = 1)}
		\frac{l!}{b! (m_1 - b)! m_2! \cdots m_s! \cdot (p_1!)^{m_1} \cdots (p_s!)^{m_s}} \\ \cdot 
		D^m_{\beta^b, \bm{x}^{m-b}} F\left[ 
		\left( \tfrac{d^{p_1}\bm{x}}{d\beta^{p_1}} \right)_{\times (m_1 - b)}, 
		\left( \tfrac{d^{p_2}\bm{x}}{d\beta^{p_2}} \right)_{\times m_2}, 
		\dots, \left( \tfrac{d^{p_s}\bm{x}}{d\beta^{p_s}} \right)_{\times m_s}
		\right]
	\end{multline}
	where the $l$-th order derivative $\tfrac{d^l \bm{x}}{d\beta^l}$ appears only at the left-hand side;
	the outer sum is over the non-trivial integer partitions of $l$.
	For $l = 1$, use $D_{\bm{x}} F[\tfrac{d\bm{x}}{d\beta}] = -D_\beta F$ \eqref{eq:ODE-implicit-form} instead.
\end{thm}

While \cite{zemel2019combinatorics} provides a formula for implicit derivatives of arbitrary order when $\bm{x}$ is univariate, we are unaware of such a result in the literature for a multivariate $\bm{x}$. Unlike our Theorem \ref{thm:formula-for-high-order-expansion-of-F-in-main-result-sect} which is recursive and written as a sum over integer partitions, their formula (Theorem 15 therein) is explicit and written as a sum over vector partitions, while requiring that $D_{\bm{x}}F$ is non-singular.
These differences make Theorem \ref{thm:formula-for-high-order-expansion-of-F-in-main-result-sect} computationally more efficient for our purposes, as we compute all the derivatives up to a given order, $\tfrac{d^k \bm{x}}{d\beta^k}$ for $k = 1, \dots, l$.
Our first-order expansion \eqref{eq:first-order-expansion-of-operator-eq-implicit} coincides with that of \cite{zemel2019combinatorics} when $\bm{x}$ is univariate; a direct exercise shows that so do our second and third order expansions \eqref{eq:second-order-expansion-of-operator-eq-implicit}-\eqref{eq:third-order-expansion-implicit} (see Equations (2)-(4) there).
Verifying the equivalence directly for higher or is more challenging.

Now that we've established that $\dbetaK{}{l} F = \bm{0}$ \eqref{eq:high-order-deriv-of-F-implicit} can be solved for $\tfrac{d^l \bm{x}}{d\beta^l}$, we spell-out in Algorithm \ref{algo:high-order-derivs-of-operator-roots} the steps to solve it recursively from formula \eqref{eq:formula-for-high-order-beta-derivatives} in Theorem \ref{thm:formula-for-high-order-expansion-of-F-in-main-result-sect}.
The computational and memory complexities of Algorithm \ref{algo:high-order-derivs-of-operator-roots} are elaborated in Section \ref{part:details}.\ref{sec:computational-complexities}.
It requires an auxiliary method (\Call{Calc Deriv Tensor}{}) to compute the derivative tensors \eqref{eq:mixed-deriv-def-evaluated-applied-to-vectors} of $F$. For rate-distortion problems, these tensors are given in the next Section \ref{sub:high-order-deriv-tensors-of-BA}.
When the latter are used, we refer to it as the specialization of Algorithm \ref{algo:high-order-derivs-of-operator-roots} to rate-distortion, or simply Algorithm \ref{algo:high-order-derivs-of-operator-roots} for RD.
To keep things simple, Algorithm \ref{algo:high-order-derivs-of-operator-roots} is presented in a recursive form. Yet, implicit derivatives can be memorized (cached) at little memory cost $O(T\cdot l)$, thus avoiding the recursion. Similar comments go also to memorizing the derivative tensors of $F$ --- see Section \ref{sub:costs-and-error-to-cost-tradeoff-of-RD-root-tracking-in-main-results-section} for details.

When $F$ has multiple roots $\bm{x}(\beta)$ of through $(\bm{x}_0, \beta_0)$, then the Jacobian $D_{\bm{x}} F$ must be singular, by the Implicit Function Theorem.
In accordance with that, formula \eqref{eq:formula-for-high-order-beta-derivatives} then has multiple solutions for each derivative order $\tfrac{d^l \bm{x}}{d\beta^l}$, as they are determined up to an element in $\ker D_{\bm{x}} F$.
In the other way around, singularity of the Jacobian does \textit{not} necessarily imply that there are multiple roots $\bm{x}(\beta)$ through $(\bm{x}_0, \beta_0)$. 
For, the existence of multiple solutions $\tfrac{d^l \bm{x}}{d\beta^l}$ to \eqref{eq:formula-for-high-order-beta-derivatives} need not imply that for each there is a root $\bm{x}(\beta)$ through $(\bm{x}_0, \beta_0)$ with the given implicit derivative.
An understanding of the bifurcations of $F$ is then necessary to determine which of the possible outputs of Algorithm \ref{algo:high-order-derivs-of-operator-roots} is indeed the implicit derivative of a root $\bm{x}(\beta)$ of $F = \bm{0}$ \eqref{eq:solution-as-root-of-functional-eq}. 
We handle this in Algorithm \ref{algo:root-tracking-for-RD} (Section \ref{sub:RD-root-tracking-near-bifrcations}) for tracking an RD root by ensuring that the Jacobian is non-singular. 
Algorithm \ref{algo:high-order-derivs-of-operator-roots} then has a unique and well-defined result, which by the Implicit Function Theorem must pertain to a root. 
See also Section \ref{part:details}.\ref{sec:RD-bifurcations-and-root-tracking} on RD bifurcations.

The derivatives computed by Algorithm \ref{algo:high-order-derivs-of-operator-roots} lose their accuracy when a Jacobian eigenvalue approaches zero. e.g., when approaching a bifurcation. The effect is more pronounced as the derivative's order increases.
Figure \ref{fig:derivative-calculation-loses-accuracy-near-bifurcation} below demonstrates this nicely for rate-distortion. 
Its numerical error is negligible when far from bifurcation, and yet grows by the derivative's order when approaching it.
This loss of accuracy can be traced back to the formula of Theorem \ref{thm:formula-for-high-order-expansion-of-F-in-main-result-sect}. For, the Jacobian $D_{\bm{x}} F$ becomes ill-conditioned\footnote{ That is, its condition number is large. Where the \textit{condition number} is defined as the ratio between its largest and smallest eigenvalues, in absolute value.} when an eigenvalue gradually vanishes.
So long that $D_{\bm{x}} F$ is invertible, the $l$-th order derivative $\tfrac{d^l \bm{x}}{d\beta^l}$ essentially contains $\left(D_{\bm{x}} F\right)^{-1}$ to the $l$-th power, amplifying numerical errors as $l$ increases. 
This can be seen by multiplying both sides of \eqref{eq:formula-for-high-order-beta-derivatives} by $\left(D_{\bm{x}} F\right)^{-1}$, and then repeatedly substituting into $\tfrac{d^l \bm{x}}{d\beta^l}$ the explicit lower-order expressions for $\tfrac{d^k \bm{x}}{d\beta^k}$, for $k < l$.

\medskip
\begin{algorithm}[h]
	\caption{\hspace{-4pt}: High-order implicit derivatives of a root $\bm{x}(\beta)$ of an operator $F$}
	\label{algo:high-order-derivs-of-operator-roots}
	\begin{algorithmic}[1]
		\Function{Calculate Implicit Derivative}{$\bm{x}_0, \beta_0, l; params$}
		\Input
		\Statex Point of evaluation $(\bm{x}_0, \beta_0)$, derivative order $l > 0$, 
		\Statex additional parameters $params$ needed for computing the derivative tensors.
		\Output{$\dbetaK{\bm{x}}{l}\Big\rvert_{(\bm{x}_0, \beta_0)}$, determined uniquely if $D_{\bm{x}} F$ is non-singular (see main text).}
		\Require{}
		\Statex A method \Call{Calc Deriv Tensor}{$m, b; \bm{x_0}, \beta_0, params$} which computes $D^{m+b}_{\beta^b, \bm{x}^{m}} F\big\rvert_{(\bm{x_0}, \beta_0)}$ \eqref{eq:mixed-deriv-def-evaluated-applied-to-vectors}.
		\State Initialize $result \gets 0$.
		\For{$k = 1$ to $l-1$}
		\State $\tfrac{d^k\bm{x}}{d\beta^k} \gets$ \Call{Calculate Implicit Derivative}{$\bm{x}_0, \beta_0, k; params$}
		\begin{flushright}
			\Comment{Cacheable; see main text.}
		\end{flushright}
		\EndFor
		\For{partition $(m_1)\cdot p_1 + \dots + (m_s)\cdot p_s$ of $l$}
		\State $b_{max} \gets m_1$ if $p_1 = 1$ else 0.
		\For{$b = 0$ to $b_{max}$}
		\State $m \gets m_1 + \dots + m_s$
		\If{$m = 1$ and $b = 0$}
		\Comment{Corresponds to $D_{\bm{x}} F[\tfrac{d^l\bm{x}}{d\beta^l}]$, by Corollary \ref{cor:highest-order-deriv-only-with-Jacobian}.}
		\State Continue
		\Else
		\State $coef \gets \frac{l!}{b! (m_1 - b)! m_2! \cdots m_s! \cdot (p_1!)^{m_1} \cdots (p_s!)^{m_s}}$
		\State $tensor \gets \Call{Calc Deriv Tensor}{m-b, b; \bm{x_0}, \beta_0}$
		\Comment{Cacheable; see main text.}
		\State $result \gets result - coef \cdot tensor\left[ 
		\left( \tfrac{d^{p_1}\bm{x}}{d\beta^{p_1}} \right)_{\times (m_1 - b)}, 
		\left( \tfrac{d^{p_2}\bm{x}}{d\beta^{p_2}} \right)_{\times m_2}, 
		\dots, \left( \tfrac{d^{p_s}\bm{x}}{d\beta^{p_s}} \right)_{\times m_s}
		\right]$
		\EndIf
		\EndFor
		\EndFor
		\State $D_{\bm{x}} F \gets \Call{Calc Deriv Tensor}{1, 0; \bm{x_0}, \beta_0}$.
		\Comment{Cacheable; see main text.}
		\State $result \gets$ a linear pre-image of $result$ under $D_{\bm{x}} F$.
		\State \Return $result$
		\EndFunction
	\end{algorithmic}
\end{algorithm}

\begin{figure}[h!]
	\begin{center}
		\includegraphics[width=.9\textwidth]{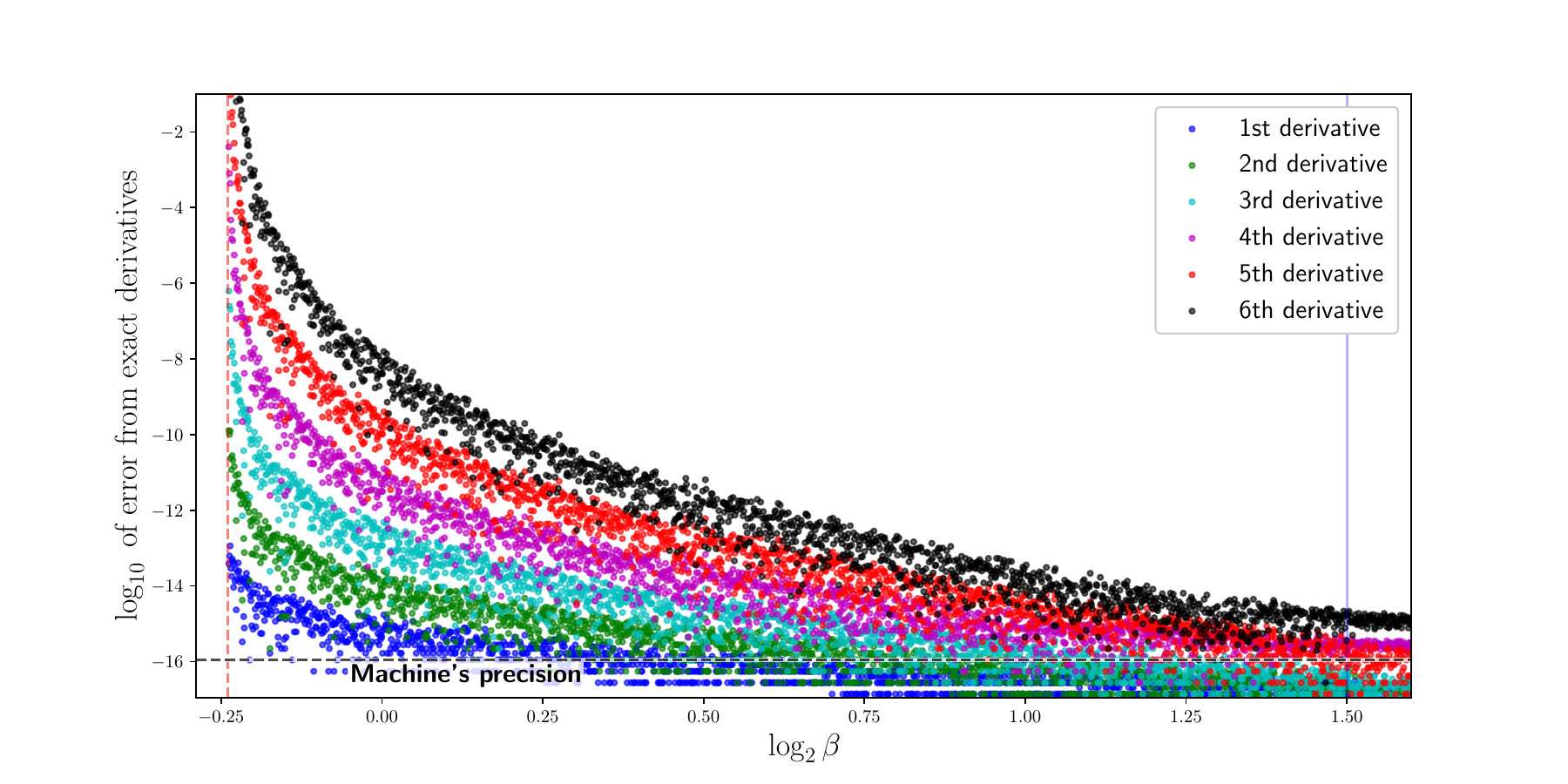}
		\includegraphics[width=.9\textwidth]{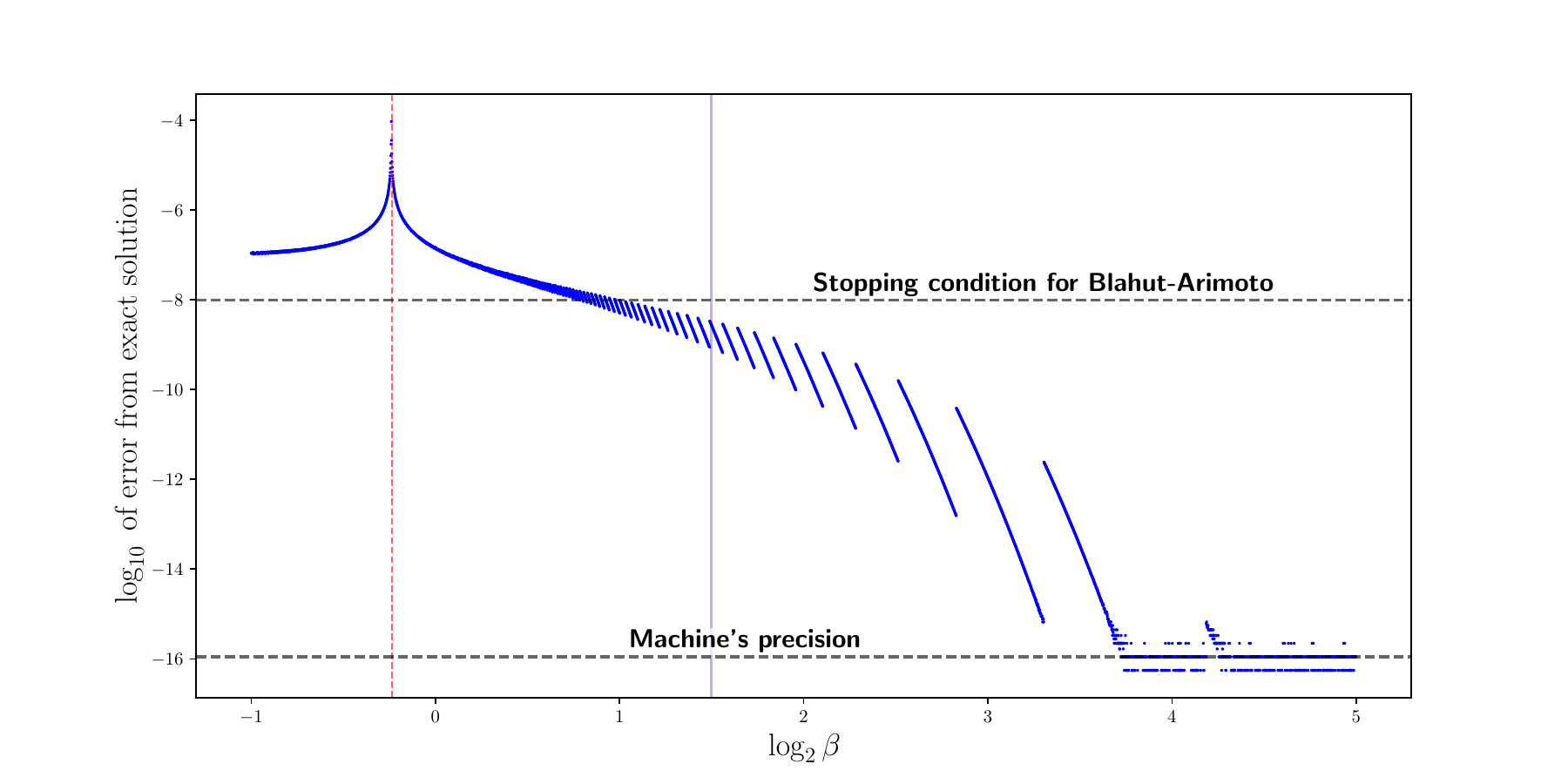}
		\caption{
			\textbf{Both Blahut-Arimoto and Algorithm \ref{algo:high-order-derivs-of-operator-roots} for RD can be remarkably accurate, yet lose their accuracy near a bifurcation} (dashed red verticals).
			To test the algorithms' accuracies, we compared their outputs to the analytical solutions of a binary source with Hamming distortion; see Section \ref{part:proofs}.\ref{sec:binary-source-with-hamming-dist-appendix} for details.
			\textbf{Top}:
			Numerical derivatives were evaluated at the exact solutions using Algorithm \ref{algo:high-order-derivs-of-operator-roots} for RD. Their $L^\infty$ distance to the analytical implicit derivatives is colored by order. 
			As noted after Theorem \ref{thm:formula-for-high-order-expansion-of-F-in-main-result-sect}, error due to numeric imprecision is amplified when approaching a bifurcation, by the derivatives' order.
			Only the range between the critical point $\beta_c$	\eqref{eq:bifurcation-beta-of-binary-source-with-Hamming-distortion} of this problem and $\log_2 \beta = 1.5$ is shown at the top (red and blue verticals).
			Outside this region, the numerical derivatives' accuracy for this example is near the machine's precision.
			\textbf{Bottom}: $L^\infty$-error of the solutions produced by Blahut-Arimoto, compared to the analytical solutions. 
			A stopping condition of $10^{-8}$ between consecutive iterates in $L^\infty$-norm was used, with uniform initial conditions for each $\beta$.
			This plot is explained by BA's critical slowing down, \citep{agmon2021critical} --- see main text, after Theorem \ref{thm:high-order-derivs-of-BA-in-main-text} in Section \ref{sub:high-order-deriv-tensors-of-BA}. 
			Note that all four axes are in a logarithmic scale.
		}
		\label{fig:derivative-calculation-loses-accuracy-near-bifurcation}
	\end{center}
\end{figure}

\clearpage

\medskip
\subsection{High-order derivative tensors of the Blahut-Arimoto operator}
\label{sub:high-order-deriv-tensors-of-BA}

We provide closed-form formulae for the derivative tensors of the Blahut-Arimoto operator of arbitrary order. 
These are necessary to compute implicit derivatives at a root with Algorithm \ref{algo:high-order-derivs-of-operator-roots} for RD (in Section \ref{sub:high-order-beta-derivatives-at-an-operator-root}). 
Our main result, in Theorem \ref{thm:high-order-derivs-of-BA-in-main-text}, requires the derivatives to be evaluated at a marginal distribution of full support. 
This shall be dealt with when reconstructing the solution curve in Section \ref{sub:taylor-method-for-RD-root-tracking}, by \textit{reducing} the RD problem to a smaller reproduction alphabet (see there).

\medskip
As noted in Section \ref{sec:introduction}, an RD problem is defined by a source distribution $p_X$, a reproduction alphabet $\hat{\mathcal{X}}$, and a distortion function $d(x, \hat{x})$.
In the sequel, we shall often assume that the distortion is finite, and also \textit{non-degenerate}: $d(\cdot, \hat{x}_1) \neq d(\cdot, \hat{x}_2)$ for all $\hat{x}_1 \neq \hat{x}_2$.
Write $N := |\mathcal{X}|$ and $M := |\hat{\mathcal{X}}|$ for the source and reproduction alphabet sizes, respectively.
When discussing rate-distortion problems, our notation shall differ from that used for root-tracking of an arbitrary operator $F$. 
Instead of $\bm{x}$, we write $\inputmarginalVect$ for our variable, which now stands for a marginal distribution on $\hat{\mathcal{X}}$. That is, $F(\cdot, \beta)$ shall be an operator on $\bb{R}^M$ (previously $\bb{R}^T$).
We further specialize $F$ to capture fixed points of the Blahut-Arimoto algorithm below.

Recall that the Blahut-Arimoto algorithm (BA) \cite{blahut1972, arimoto1972} is defined by two equations; e.g., \cite[10.8]{Cover2006}. Given a marginal distribution $\inputmarginalVect$ on the reproduction alphabet $\hat{\mathcal{X}}$, $\inputmarginal{}$, define a test channel or \textit{encoder} $\intermediateencoderVect$ by
\begin{equation}		\label{eq:encoder-eq}
	\intermediateencoder{}{} := \frac{\inputmarginal{} \; e^{-\beta d(x, \hat{x})}}{Z(x, \beta)} \;,
\end{equation}
where $Z(x, \beta) := \sum_{\hat{x}'} \inputmarginal{'} e^{-\beta d(x, \hat{x}')}$. Given an encoder $\intermediateencoderVect$, define a new marginal $\outputmarginalVect$,
\begin{equation}		\label{eq:marginal-eq}
	\outputmarginal{} := \sum_x p_X(x) \intermediateencoder{}{}
\end{equation}
A necessary (though \textit{not} sufficient) condition for $\inputmarginalVect$ to achieve the minimum $R(D)$ at \eqref{eq:RD-func-def} is that it is a fixed point of Equations \eqref{eq:encoder-eq} and \eqref{eq:marginal-eq}, $\outputmarginalVect = \inputmarginalVect$. 
e.g., \cite{berger71}. It is then called a \textit{curve achieving distribution}, or simply an \textit{achieving distribution}.

A curve-achieving marginal is determined by its test channel and vice-versa. So, a priori, it may not seem to matter which of the two is taken as the variable.
However, it turns out (Section \ref{part:details}.\ref{sub:obstructions-to-RT-assumptions-for-RD}) that only cluster-vanishing bifurcations can be detected in marginal coordinates.
This shall be dealt with later, when reconstructing the solution curve (Assumption \ref{assumption:only-cluster-vanishing-bifurcations} in Section \ref{sub:taylor-method-for-RD-root-tracking}).
For now, we define the $BA_\beta$ operator as the composition of \eqref{eq:encoder-eq} followed by \eqref{eq:marginal-eq}. 
It is a single iteration of the BA algorithm. Explicitly, 
\begin{equation}		\label{eq:BA-operator-def}
	BA_\beta\left[\inputmarginalVect\right](\hat{x}) := \inputmarginal{} \sum_x \frac{p_X(x) e^{-\beta d(x, \hat{x})}}{\sum_{\hat{x}'} \inputmarginal{'} e^{-\beta d(x, \hat{x}')}}
\end{equation}
is the operator's $\hat{x}$ coordinate when evaluated at $\inputmarginalVect$.
We could have taken instead the encoder $\intermediateencoderVect$ as our variable, defining $BA_\beta[\intermediateencoderVect]$ by the composition in reverse order, \eqref{eq:marginal-eq} followed by \eqref{eq:encoder-eq}. However, to reduce computational costs, we select the lower-dimensional cluster marginal $\inputmarginalVect$ as our variable.

\cite{csiszar1974computation} showed that the Blahut-Arimoto algorithm converges to an achieving distribution. 
An achieving distribution need not be unique, as \cite{berger71} notes for example.
Yet, it is a fixed point of $BA_\beta$ \eqref{eq:BA-operator-def}, or equivalently a root of the operator $F := Id - BA_\beta$ \eqref{eq:RD-operator-def}, as \cite{agmon2021critical} observed. Thus, to track fixed points of BA we calculate the derivative tensors of $Id - BA_\beta$ \eqref{eq:RD-operator-def}, below.
See Section \ref{part:details}.\ref{sec:RD-bifurcations-and-root-tracking} on its bifurcations. e.g., why should one expect RD problems to exhibit bifurcations, on types of RD bifurcations, and so forth. 
These shall be combined with implicit derivatives in RD when reconstructing the solution curve, in Section \ref{sub:taylor-method-for-RD-root-tracking}. 

\medskip 
A few definitions are needed to write down the derivative tensors of $F := Id - BA_\beta$ \eqref{eq:RD-operator-def}.
Define the real \emph{polynomial ring on countably many variables} as the following set of finite sums,
\begin{multline}		\label{eq:real-polynomials-in-countable-vars}
	\bb{R}[x_0, x_1, \dots] := 
	\Big\{ \sum a_i \; x_{i_1}^{d_{i_1}} x_{i_2}^{d_{i_2}} \cdots x_{i_k}^{d_{i_k}}: \Big. \\ \Big.
	a_i\in \bb{R}, \; k\in \bb{N}, \; 0 \leq i_1 < i_2 < \dots < i_k, \text{ and } d_{i_1}, d_{i_2}, \dots, d_{i_k} \in \bb{N}_0 \Big\}
\end{multline}
where $\bb{N}$ stands for the natural numbers, and $\bb{N}_0$ for the non-negative integers. 
It consists of finite sums of monomials $x_{i_1}^{d_{i_1}} x_{i_2}^{d_{i_2}} \cdots x_{i_k}^{d_{i_k}}$ with real coefficients, with each $d_{i_j} \geq 0$. 
The \emph{degree} of a polynomial in $\bb{R}[x_0, x_1, \dots]$ is the highest among its monomial degrees, which in turn is defined as $d_{i_1} + d_{i_2} + \dots + d_{i_k}$ for the above.
This slightly extends the usual definition of the (real) \emph{polynomial ring} $\bb{R}[x_0, x_1, \dots, x_n]$ in multiple variables $x_0, x_1, \dots, x_n$. e.g., \cite[9.1]{dummit2004abstract_alg}.

Next, define an $\bb{R}$-linear operator $\dbar$ on polynomials $\bb{R}[x_0, x_1, \dots]$ \eqref{eq:real-polynomials-in-countable-vars} by the Leibniz product rule $\dbar(f\cdot g) = g\cdot \dbar f + f\cdot \dbar g$. Except, that on the variables $x_0, x_1, \dots, $ it is defined by
\begin{equation}		\label{eq:variable-deriv-def-for-recursive-beta-formula}
	\dbar x_0 := 0, \quad \text{and} \quad
	\dbar x_k := x_1 \cdot x_k - x_{k+1} \quad \text{for } k > 0 \;,
\end{equation}
and $\dbar c := 0$ for constant polynomials $c \in \bb{R}$.
We note that the usual rules of polynomial differentiation such as $\dbar x_i^n = n x_i^{n-1} \dbar x_i$ follow from the Leibniz product rule. 
In algebraic context, a linear operator on a function space satisfying the Leibniz rule is known as a \emph{derivation}. e.g., \cite[C-2.6]{rotman2017advanced}.

With the operator $\dbar$ above, define as follows polynomials $P_k$ in the $k+1$ variables $x_0, x_1, \dots, x_k$, for $k = 0, 1, 2, \dots$. First, set
\begin{equation}		\label{eq:P_0-def}
	P_0(x_0) := 1 \;.
\end{equation}
Then, define $P_{k+1}$ inductively in terms of $P_k$,
\begin{equation}		\label{eq:P_k-inductive-def}
	P_{k+1}(x_0, x_1, \dots, x_{k+1}) := (x_1 - x_0) \cdot P_k + \dbar P_k
\end{equation}
With this, the first few polynomials are seen (in Section \ref{part:details}.\ref{sub:encoders-beta-derivatives}) to be
\begin{align}
	P_1(x_0, x_1) &= x_1 - x_0 	\label{eq:P_1_derived}	\\
	P_2(x_0, x_1, x_2) &= x_0^2 - 2x_0 x_1 + 2x_1^2 - x_2 	\label{eq:P_2_derived}	\\
	P_3(x_0, x_1, x_2, x_3) &= -x_0^3 + 3x_0^2 x_1 +3x_0 x_2 -6 x_0 x_1^2 + 6x_1^3 - 6x_1 x_2 + x_3 	\label{eq:P_3_derived}
\end{align}

For an encoder $\intermediateencoderVect$ and a distortion $d(x, \hat{x})$, denote
\begin{equation}			\label{eq:expected-k-th-power-distortion-def}
	\expectedDxWRTencoderK{k} := \sum_{\hat{x}'} \intermediateencoder{'}{} d(x, \hat{x}')^k
\end{equation}
for $k > 0$ and a particular coordinate $x\in \mathcal{X}$.
By abuse of notation, define for $k \geq 0$,
\begin{equation}		\label{eq:P_k-by-abuse-of-notation}
	P_k[\bm{q}; d](\hat{x}, x) := P_k\Big( d(x, \hat{x}), \expectedDxWRTencoder, \dots, \expectedDxWRTencoderK{k}\Big)
\end{equation}
That is, $P_k[\bm{q}; d]$ is a function of two variables $\hat{x}$ and $x$, defined by a pointwise evaluation of the polynomials $P_k$ \eqref{eq:P_0-def}-\eqref{eq:P_k-inductive-def} at $x_0 := d(x, \hat{x})$ and $x_k := \expectedDxWRTencoderK{k}$ for $k > 0$.
We shall write $P_k(\hat{x}, x)$ for short when the distortion $d(x, \hat{x})$ and the point $\intermediateencoderVect$ of evaluation are understood.

\medskip
We are now set to spell out the derivative tensors of $(Id - BA_\beta)[\inputmarginalVect]$ \eqref{eq:RD-operator-def}. While the entries of $(Id - BA_\beta)[\inputmarginalVect]$ are indexed by $\hat{x}$, we index the entries of its derivative tensors by a multi-index $\bm{\alpha} \in \bb{N}_0^{M+1}$, for convenience. 
See comments on indexation after definition \eqref{eq:mixed-deriv-def-evaluated-applied-to-vectors} of a derivatives tensor $D_{\beta^b, \bm{x}^m}^{b + m} F$ (in Section \ref{sub:high-order-beta-derivatives-at-an-operator-root}). 

\begin{thm}[High-order partial derivatives of $Id - BA_\beta$ \eqref{eq:RD-operator-def}]		\label{thm:high-order-derivs-of-BA-in-main-text}
	Let $p_X$ and $d(x, \hat{x})$ define an RD problem on the reproduction alphabet $\hat{\mathcal{X}}$. 
	Let $\inputmarginalVect \in \Delta[\hat{\mathcal{X}}]$, and let $\intermediateencoderVect$ be the encoder defined by it via Equation \eqref{eq:encoder-eq}. 
	Then, for any integer $\alpha_0 > 0$,
	\begin{equation}		\label{eq:repeated-beta-deriv-in-thm}
		\partialbetaK{}{\alpha_0} \big( Id - BA_\beta \big)[\inputmarginalVect](\hat{x}) = 
		-\sum_x p_X(x) \intermediateencoder{}{} P_{\alpha_0}(\hat{x}, x) \;,
	\end{equation}
	where $P_{\alpha_0}(\hat{x}, x)$ is defined by \eqref{eq:P_k-by-abuse-of-notation}.
	
	Assume further that $\inputmarginalVect$ is of full support, $\inputmarginal{} > 0$ for every $\hat{x}$. Let $\bm{\alpha} = (\alpha_0, \bm{\alpha_+}) \in \bb{N}_0^{M+1}$ be a multi-index with $\bm{\alpha_+} \neq \bm{0}$.
	Then,
	\begin{multline}		\label{eq:mixed-BA-deriv-in-thm}
		\frac{\partial^{|\bm{\alpha}|} }{\partial \beta^{\alpha_0} \partial \bm{\inputmarginalVect}^{\bm{\alpha_+}}} \left(Id - BA_\beta \right)\left[\inputmarginalVect\right](\hat{x}) \\ =
		\delta_{\bm{\alpha}, \bm{e}_{\hat{x}}} - 
		(-1)^{|\bm{\alpha_+}|-1} (|\bm{\alpha_+}|-1)! \; \bm{\alpha}! \sum_x p_X(x) \left(\frac{\intermediateencoder{'}{}}{\inputmarginal{'}}\right)^{\bm{\alpha_+} } 
		\sum_{\bm{k}\in \bb{N}_0^M: \; |\bm{k}| = \alpha_0} \left(\prod_{i\neq\hat{x}} G\big( k_i, \alpha_i; \intermediateencoderVect, d \big)_{(\hat{x}_i, x)} \right)
		\\ \cdot \left[
		\alpha_{\hat{x}} \cdot G\big( k_{\hat{x}}, \alpha_{\hat{x}}; \intermediateencoderVect, d \big)_{(\hat{x}, x)} 
		- |\bm{\alpha_+}| \cdot \left( 1 + \alpha_{\hat{x}} \right) \cdot \intermediateencoder{}{} \cdot 
		G\big( k_{\hat{x}}, 1 + \alpha_{\hat{x}}; \intermediateencoderVect, d \big)_{(\hat{x}, x)}
		\right] 
	\end{multline}
	where $G\big(k, a; \intermediateencoderVect, d \big)$ is defined on integers $k, a\geq 0$ by $G = 0$ if $a = 0 < k$, and otherwise
	\begin{equation}		\label{eq:combinatorial-G-in-terms-of-polynomials}
		G\big( k, a; \intermediateencoderVect, d \big)_{(\hat{x}, x)} := 
		\sum_{\substack{\bm{t} \in \bb{N}_0^k: \;|\bm{t}|\leq a, \\ \sum_{j} j\cdot t_{j} = k}} \frac{1}{ \bm{t}! \; (a - |\bm{t}|)!} 
		\prod_{j=1}^{k} \left( \frac{P_j(\hat{x}, x)}{j!} \right)^{t_{j}} \;.
	\end{equation}
\end{thm}

With the above derivative tensors of $Id - BA_\beta$ \eqref{eq:RD-operator-def}, Algorithm \ref{algo:high-order-derivs-of-operator-roots} for computing higher implicit derivatives (in Section \ref{sub:high-order-beta-derivatives-at-an-operator-root}) can now be specialized to RD. 
This allows us to compute the implicit derivatives $\dbetaK{\inputmarginalVect}{l}$ \eqref{eq:l-th-deriv-at-beta0} at an RD root $\inputmarginalVect $.
The Theorem's proof is outlined in Section \ref{part:details}.\ref{sec:high-order-derivs-of-BA-in-marginal-coords}.
In addition to the RD ODE (Theorem \ref{thm:beta-ODE-in-marginal-coords}), the main results there are the formulas for the encoder's repeated partial $\beta$-derivatives and its mixed ones (Propositions \ref{prop:repeated-beta-derivatives-of-encoder} and \ref{prop:mixed-high-order-enc-deriv}), which yield \eqref{eq:repeated-beta-deriv-in-thm} and \eqref{eq:mixed-BA-deriv-in-thm}. 
Section \ref{part:details}.\ref{sub:computing-high-order-derivatives-efficiently} comments how the latter can be computed efficiently.

The implicit derivatives computed by Algorithm \ref{algo:high-order-derivs-of-operator-roots} for RD lose their numerical stability when approaching a bifurcation, losing it faster if the order is higher (see after Theorem \ref{thm:formula-for-high-order-expansion-of-F-in-main-result-sect} in Section \ref{sub:high-order-beta-derivatives-at-an-operator-root}). 
Other than that, the first few implicit derivatives computed for RD are seen in Figure \ref{fig:derivative-calculation-loses-accuracy-near-bifurcation} (top) to be remarkably accurate. 
Further, note that the derivative tensors of $Id - BA_\beta$ might be numerically unstable when computed for high orders, due to the factorials at the denominators of \eqref{eq:combinatorial-G-in-terms-of-polynomials}. 
Interestingly, the bifurcation's presence also affects the classic Blahut-Arimoto algorithm, \cite{blahut1972, arimoto1972}. It too suffers from accuracy loss near bifurcations, as demonstrated by Figure \ref{fig:derivative-calculation-loses-accuracy-near-bifurcation} (bottom). This can be seen as a direct consequence of the critical slowing down near RD bifurcations, \cite{agmon2021critical}.

\medskip
\section{Reconstructing an RD solution curve from implicit derivatives}
\label{sec:RD-solution-curve-from-RD-beta-derivs}

In Section \ref{sec:beta-derivs-at-an-operator-root-and-for-RD}, we showed how to compute implicit derivatives at an RD root, using the specialization of Algorithm \ref{algo:high-order-derivs-of-operator-roots} to RD.
Now, we shall leverage these derivatives to reconstruct the entire solution curve of a given RD problem.

\medskip
In Section \ref{sub:taylor-method-for-RD-root-tracking}, we modify the vanilla fixed-step Taylor method to track an RD root numerically until a cluster-vanishing bifurcation is detected --- which is Algorithm \ref{algo:taylor-method-for-RD-root-tracking}.
We handle these bifurcations in Algorithm \ref{algo:root-tracking-for-RD} (in Section \ref{sub:RD-root-tracking-near-bifrcations}), thereby reconstructing the entire solution curve, as seen in Figure \ref{fig:reconstructing-an-RD-solution-curve}.
The algorithms' error from the true solution is of order $O(|\Delta\beta|^l)$ (Theorem \ref{thm:taylor-method-converges-for-RD-root-tracking-away-of-bifurcation}), for an $l$-th order Taylor method of step $\Delta\beta$.
Both Blahut-Arimoto and the implicit numerical derivatives we use suffer from degraded accuracy near RD bifurcations, as shown earlier in Figure \ref{fig:derivative-calculation-loses-accuracy-near-bifurcation}. 
However, while BA suffers from a hefty computational penalty near bifurcations due to its critical slowing down, \citep{agmon2021critical}, the added computational cost to our Algorithms \ref{algo:taylor-method-for-RD-root-tracking} and \ref{algo:root-tracking-for-RD} is negligible. 

\begin{figure}[h!]
	\begin{center}
		\hspace*{-55pt}
		\includegraphics[width=1.25\textwidth]{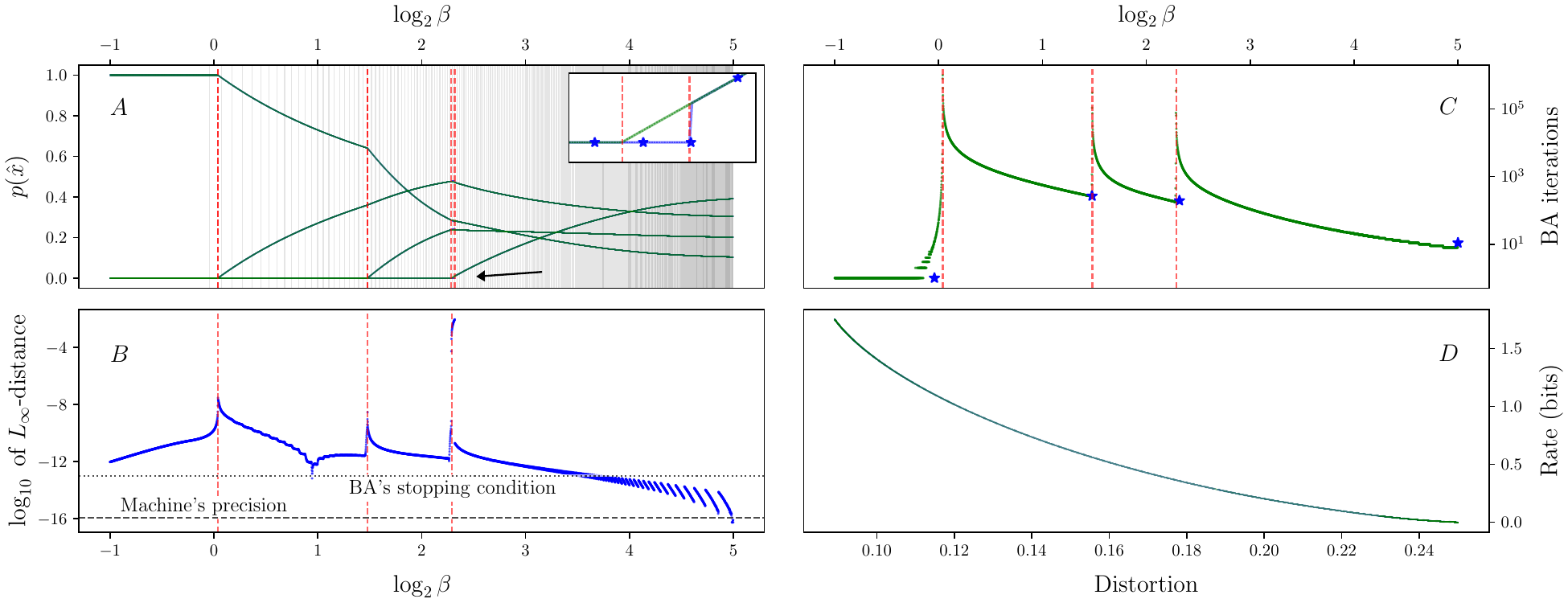}
		\caption{
			\textbf{Reconstructing a solution curve from implicit derivatives, with root-tracking for RD (Algorithm \ref{algo:root-tracking-for-RD}).}
			Reproducing the problem in \cite[Figure 1]{agmon2021critical}, defined by $d(x, \hat{x}) = \tfrac{1}{8} {\footnotesize \mat{0 & 1 & 1 & 2\\4 & 1 & 5 & 2\\4 & 5 & 1 & 2\\8 & 5 & 5 & 2}}$ and $p_X = \tfrac{1}{10} (4, 3, 2, 1)$. 
			Bifurcations are marked by dashed red verticals. 
			\newline
			\textbf{A.} Cluster marginal $p(\hat{x})$ as a function of $\beta$.
			A fixed-width grid of almost 400 points is selected along the $\beta$-axis (gray verticals), with implicit derivatives computed to the 7th order at each grid-point.
			Every point is extrapolated from the grid-point to its right using a modified Taylor method (Algorithm \ref{algo:taylor-method-for-RD-root-tracking} in Section \ref{sub:taylor-method-for-RD-root-tracking}), except for the rightmost one and near the bifurcations, where Blahut-Arimoto's algorithm is used.
			The cluster vanishing bifurcations of this problem are handled by Algorithm \ref{algo:root-tracking-for-RD} (with $\delta = 10^{-2}, \Delta \beta = \nicefrac{-31.5}{400}$), which can either overshoot or undershoot a bifurcation (Section \ref{sub:RD-root-tracking-near-bifrcations}). 
			On undershooting, the bifurcation is detected too early (redundant red vertical to the right), and the algorithm switches temporarily to a sub-optimal solution branch as a result. 
			This is magnified $\times 20$ in the top-right inset, depicting an undershooting of the rightmost bifurcation, with four nearby grid-points marked by blue stars. 
			The entire solution curve (in blue) is extrapolated from the grid.
			For comparison, Blahut-Arimoto's is in green (a $10^{-13}$ stopping condition, 5000 grid-points). The two are visually indistinguishable almost everywhere.
			\textbf{B.} The $L_\infty$-distance between the solutions produced by BA and Algorithm \ref{algo:root-tracking-for-RD} (this problem has no analytic solution). Note the localized error in Algorithm \ref{algo:root-tracking-for-RD} due to switching to the sub-optimal branch near the rightmost bifurcation.
			\textbf{C.} Number of Blahut-Arimoto iterations. BA in reverse annealing is in green; Algorithm \ref{algo:root-tracking-for-RD} invokes BA only four times (blue stars): once for the initial condition, and once per bifurcation.
			Note that Algorithm \ref{algo:root-tracking-for-RD} avoids the critical slowing down exhibited by BA, \cite{agmon2021critical}; see Section \ref{sub:RD-root-tracking-near-bifrcations}.
			\textbf{D.} The rate-distortion curve of this problem. The results of BA and Algorithm \ref{algo:root-tracking-for-RD} are indistinguishable.
		}
		\label{fig:reconstructing-an-RD-solution-curve}
	\end{center}
\end{figure}

Bounds on the algorithms' computational and memory complexities are provided in Section \ref{sub:costs-and-error-to-cost-tradeoff-of-RD-root-tracking-in-main-results-section}. 
One might expect the computational cost to increase if a smaller error is required. 
Indeed, we estimate of this tradeoff for our algorithms. 
Unlike the Blahut-Arimoto algorithm, which computes solutions only at specified grid points, when our root-tracking Algorithm \ref{algo:root-tracking-for-RD} terminates, only a fixed computational cost is needed to extrapolate any off-grid point.
In principle, the entire solution curve can be approximated from few grid points by using high-order expansions, at least in well-behaved examples. 
Although high orders may not be practical due to their computational costs, this is demonstrated in Figure \ref{fig:7-points-example}. 
Suggestions on improving our algorithms' efficiency are discussed in Section \ref{sub:efficient-RD-root-tracking}.

This section is complemented by several Sections in Part \ref{part:details}, corresponding to the subsections of Section \ref{sec:RD-solution-curve-from-RD-beta-derivs}.
In Section \ref{part:details}.\ref{sec:RD-bifurcations-and-root-tracking}, we provide several basic results on RD bifurcations, to our knowledge for the first time. This allows us to tell when does our Algorithm \ref{algo:root-tracking-for-RD} follow the optimal branch, explain why it may fail, and consider improvements.
Section \ref{part:details}.\ref{sec:error-analysis} follows standard error analysis results in the literature to provide the guarantees and error estimates for our case. In the context of numerical solutions of ODEs, a computational difficulty stems from the existence of bifurcations.
Thus, in a sense, our algorithm trades the bifurcations' hefty computational cost due to Blahut-Arimoto's critical slowing down with reduced accuracy near bifurcations.
Finally, Section \ref{part:details}.\ref{sec:computational-complexities} bounds the computational and memory costs, both of root-tracking in general, and of its specialization to rate-distortion.

\subsection{A modified Taylor method for RD root-tracking}
\label{sub:taylor-method-for-RD-root-tracking}

We parameterize a root of $Id - BA_\beta$ \eqref{eq:RD-operator-def} by its cluster marginal $\inputmarginal{}$ rather than by the encoder $\intermediateencoder{}{}$ to reduce computational costs, as noted in Section \ref{sub:high-order-deriv-tensors-of-BA}. 
Write $\tilde{\inputmarginalVect}_{\beta_n}$ or simply $\tilde{\inputmarginalVect}_{n}$ for a numerical approximation of the true solution $\inputmarginalVect_{\beta_n} \in \Delta[\hat{\mathcal{X}}]$ at $\beta_n$.
Starting at $\tilde{\inputmarginalVect}_{0} := \inputmarginalVect_{\beta_0}$, set\footnote{ We note that the approximations $\tilde{\inputmarginalVect}_{n}$ \eqref{eq:taylor-method-def-in-cluster-marginal} need \textit{not} be normalized distributions. However, Theorem \ref{thm:taylor-method-converges-for-RD-root-tracking-away-of-bifurcation} below guarantees that $\tilde{\inputmarginalVect}_{n}$ does not deviate much from the true solution $\inputmarginalVect_{\beta_n}$ when the step size $|\Delta \beta|$ is small enough. Thus, $\tilde{\inputmarginalVect}_{n}$ does not deviate much from being normalized. Note also the normalization on step \ref{algo:RD-root-tracking:normalization} of Algorithm \ref{algo:root-tracking-for-RD}. }
\begin{equation}		\label{eq:taylor-method-def-in-cluster-marginal}
	\tilde{\inputmarginalVect}_{n+1} := 
	\tilde{\inputmarginalVect}_{n} + 
	\Delta \beta \cdot T_l\big( \tilde{\inputmarginalVect}_{n}, \beta_n, \Delta \beta \big) \;.
\end{equation}
Where, $T_l(\bm{x}, \beta, \Delta \beta) := \dbeta{\bm{x}} \big \rvert_{(\bm{x}, \beta)} + \tfrac{1}{2!} \dbetaK{\bm{x}}{2} \big \rvert_{(\bm{x}, \beta)} \Delta \beta + \dots + \tfrac{1}{l!} \dbetaK{\bm{x}}{l} \big \rvert_{(\bm{x}, \beta)} \Delta \beta^{l-1}$ is the $l$-th order Taylor remainder as in expansion \eqref{eq:solution-by-beta-as-taylor-approx}, with the derivatives evaluated by Algorithm \ref{algo:high-order-derivs-of-operator-roots} for RD (Section \ref{sub:high-order-beta-derivatives-at-an-operator-root}) at $(\tilde{\inputmarginalVect}_{n}, \beta_n)$.
We extrapolate off-grid points $\tilde{\inputmarginalVect}_\beta$ from \eqref{eq:taylor-method-def-in-cluster-marginal} using the last grid point, except in some cases near a bifurcation (see Section \ref{sub:RD-root-tracking-near-bifrcations}).
This numerical method is known as the \textit{Taylor method}; e.g., \cite{atkinson2011numerical, butcher2016numerical}.
See Section \ref{part:details}.\ref{sec:error-analysis} for a recap and its error analysis.
While there are other numerical methods that use derivatives to approximate $\inputmarginalVect_\beta$, we chose the Taylor method since it is simple and well-studied. cf., Section \ref{sub:efficient-RD-root-tracking} on improvements.

Computing implicit derivatives with Algorithm \ref{algo:high-order-derivs-of-operator-roots} for RD requires the derivative tensors of $BA_\beta$, which were provided by Theorem \ref{thm:high-order-derivs-of-BA-in-main-text} (Section \ref{sub:high-order-deriv-tensors-of-BA}). 
The formulae there require a marginal $\inputmarginalVect$ to be of full support. 
To handle this and to further reduce computational costs, define the following.
For a proper nonempty subset $\hat{\cal{X}}'$ of the reproduction alphabet $\hat{\cal{X}}$, define \textit{the RD problem restricted to $\hat{\cal{X}}'$} by deleting letters $\hat{x}\in \hat{\cal{X}}$ outside $\hat{\cal{X}}'$ and the respective columns $d(\cdot, \hat{x})$ from the distortion matrix. 
For practical purposes, deleting a letter $\hat{x}$ is equivalent to allowing initializations $\inputmarginalVect_0$ of BA only if their $\hat{x}$ coordinate is zero, $\inputmarginalSymbol_0(\hat{x}) = 0$, as can be seen by the explicit form \eqref{eq:BA-operator-def} of $BA_\beta$. 
When the marginal $\inputmarginalVect$ is understood, we call the RD problem restricted to $\supp \inputmarginalVect$ \textit{the reduced problem}. 
Reducing a problem does not affect the solution, \cite[Lemma 1 in Chapter 2]{berger71}.

We require henceforth that $d(x, \hat{x})$ is finite and non-degenerate\footnote{ Namely, the ($\hat{x}$ indexed) columns of the distortion matrix $d$ are distinct --- see Section \ref{sub:high-order-deriv-tensors-of-BA}.}. 
By definition, $\inputmarginalVect$ is of full support in the reduced problem. 
Therefore, the Jacobian $D_{\inputmarginalVect}(Id - BA_\beta)\rvert_{\inputmarginalVect}$ is non-singular in the reduced problem, \citep[Theorem 1 ff.]{agmon2021critical}, when $\inputmarginalVect$ is also a fixed point of $BA_\beta$. 
In particular, the numerical derivatives $\tfrac{d^l \inputmarginalVect}{d\beta^l}$ produced by Algorithm \ref{algo:high-order-derivs-of-operator-roots} for RD are then defined uniquely, as explained after Theorem \ref{thm:formula-for-high-order-expansion-of-F-in-main-result-sect} (in Section \ref{sub:high-order-beta-derivatives-at-an-operator-root}). 
These are vectors of a lower dimension $|\supp \inputmarginalVect| \leq M$, tracing the root's path in the reduced problem; optimality of this path is discussed below. 
The explicit form of $D_{\inputmarginalVect}(Id - BA_\beta)\rvert_{\inputmarginalVect}$ at an arbitrary distribution $\inputmarginalVect$ is given later, by Corollary \ref{cor:BA-jacobian} (Section \ref{part:details}.\ref{sub:encoders-marginal-derivatives}). 
Its formula \eqref{eq:BA-jacobian} shows that the Jacobian is non-singular even when straying slightly off a stable fixed point of $BA_\beta$ (e.g., due to accumulated approximation error).

The argument at \cite{agmon2021critical} further shows that an eigenvalue of $D_{\inputmarginalVect}(Id - BA_\beta)\rvert_{\inputmarginalVect}$ vanishes if and only if $\inputmarginal{}$ vanishes for some $\hat{x}$.
This gives a simple way to detect cluster-vanishing bifurcations. Namely, an RD bifurcation where $\inputmarginal{}$ gradually vanishes. 
We take only negative steps $\Delta \beta < 0$ when considering how fixed points evolve --- see Section \ref{part:details}.\ref{sub:cluster-vanishing-bifs-are-bifs} ff. 
As shown there, cluster-vanishing bifurcations are indeed bifurcations, where two roots collide and merge.	
To detect these, one can set a cluster-mass threshold $\delta > 0$ below which $\inputmarginal{}$ is considered to have vanished.
Thus, iterating over a Taylor method step \eqref{eq:taylor-method-def-in-cluster-marginal} until a cluster vanishes allows to reconstruct the root's path till the next bifurcation. This is summarized by Algorithm \ref{algo:taylor-method-for-RD-root-tracking}, with convergence guarantees in Theorem \ref{thm:taylor-method-converges-for-RD-root-tracking-away-of-bifurcation} below. 
We handle the bifurcation later, in Section \ref{sub:RD-root-tracking-near-bifrcations}. 
A significant part of the computational difficulty stems from approaching a bifurcation (Section \ref{part:details}.\ref{sub:computational-difficulty-of-RTRD}). Thus, setting a threshold $\delta > 0$ on the cluster mass effectively restricts the problem's difficulty.

\begin{algorithm}
	\caption{\hspace{-4pt}: A modified $l$-th order Taylor method, tracking an RD root to a cluster-vanishing bifurcation}
	\label{algo:taylor-method-for-RD-root-tracking}
	\begin{algorithmic}[1]
		\Function{Track RD Root To Bifurcation}{$\inputmarginalVect_{\beta_0}, \beta_0; \Delta \beta, \delta, l, d, p_X$}
		\Input
		\Statex $\beta_0$, a root $\inputmarginalVect_{\beta_0}$ of $Id - BA_{\beta_0}$ \eqref{eq:RD-operator-def} of full support, step-size $\Delta \beta < 0$, 
		\Statex cluster mass threshold $\delta > 0$, an order $l > 0$, and an RD problem $(d, p_X)$.
		\Output
		\Statex Approximations $\tilde{\inputmarginalVect}_{i}$ of the true solution $\inputmarginalVect_{\beta_i}$ at $\beta_i$, 
		\Statex and the $l$-th order Taylor expansions $p_i$ around $(\tilde{\inputmarginalVect}_i, \beta_i)$.
		\State Initialize $\tilde{\inputmarginalVect}_0 \gets \inputmarginalVect_{\beta_0}, n \gets 0$.
		\While{$\forall \hat{x}' \; \tilde{\inputmarginalSymbol}_n(\hat{x}') > \delta$ \textbf{and} $\beta_n > 0$}	
		\Comment{Stop if too close to a bifurcation.}		\label{algo:root-tracking-on-interval:stopping-cond}
		\For{$i = 1, \dots, l$}
		\State $\dbetaK{\inputmarginalVect}{i}\Big\rvert_{(\tilde{\inputmarginalVect}_n, \beta_n)} \gets 
		\Call{Calculate Implicit Derivative}{\tilde{\inputmarginalVect}_n, \beta_n, i; d, p_X}$ 
		\begin{flushright}
			\Comment{Algorithm \ref{algo:high-order-derivs-of-operator-roots} for RD.}
		\end{flushright}
		\EndFor
		\State $p_n(\Delta ) \gets \tilde{\inputmarginalVect}_{n} +
		\frac{\Delta }{1!} \cdot \dbeta{\inputmarginalVect}\Big\rvert_{(\tilde{\inputmarginalVect}_{n}, \beta_n)} + \dots +
		\frac{\Delta^{l}}{l!} \cdot \dbetaK{\inputmarginalVect}{l}\Big\rvert_{(\tilde{\inputmarginalVect}_n, \beta_n)} $
		\State $\tilde{\inputmarginalVect}_{n+1} \gets p_n(\Delta \beta)$
		\State $\beta_{n+1} \gets \beta_n + \Delta \beta$
		\State $n \gets n + 1$
		\EndWhile
		\State \textbf{return} $\left\{ (\tilde{\inputmarginalVect}_i, \beta_i, p_i ) \right\}_{i=1}^n$
		\Comment{Approximations of the true solutions $\inputmarginalVect(\beta_n)$.}
		\EndFunction
	\end{algorithmic}
\end{algorithm}

Roots $\inputmarginalVect_\beta$ of an RD problem are very well-behaved --- namely, piecewise analytic in $\beta$.
Reduction of an RD problem mods out the kernel of the Jacobian in cluster-marginal coordinates $\inputmarginalVect$, leaving it non-singular, and yet is simple and straightforward to implement. 
By the implicit function theorem \cite[(I.1.7)]{kielhofer2011bifurcation}, there not only exists a unique root $\inputmarginalVect_\beta$ of the reduced problem (through a given $\inputmarginalVect_{\beta_0}$), but it is also (real) analytic. For, the $BA_\beta$ \eqref{eq:BA-operator-def} operator is a composition of analytic functions. 
Recall, e.g., \cite[IX.3]{dieudonne1969foundations}, that $\inputmarginalVect_\beta$ is \textit{analytic} (in some open set) if its infinite-order Taylor expansion \eqref{eq:solution-by-beta-as-taylor-approx} about any $\beta_0$ not only exists ($\inputmarginalVect_\beta$ is smooth), but also converges pointwise to the function itself, within some radius of convergence.
Every coordinate $\inputmarginalSymbol_\beta(\hat{x})$ of $\inputmarginalVect_\beta$ is nothing but an (infinite) power series in $\Delta \beta := \beta - \beta_0$.
The partial sums of a power series converge uniformly on closed intervals inside their radius of convergence. 
That is, only a finite number of summands is needed to extrapolate the solution on an interval, at a given precision. 
In principle, this allows us to extrapolate the entire solution curve by using Taylor expansions at just a few points if the convergence radii are large enough, as demonstrated by Figure \ref{fig:7-points-example}.

\begin{figure}[h!]
	\begin{center}
		\hspace*{-55pt}
		\includegraphics[width=1.25\textwidth]{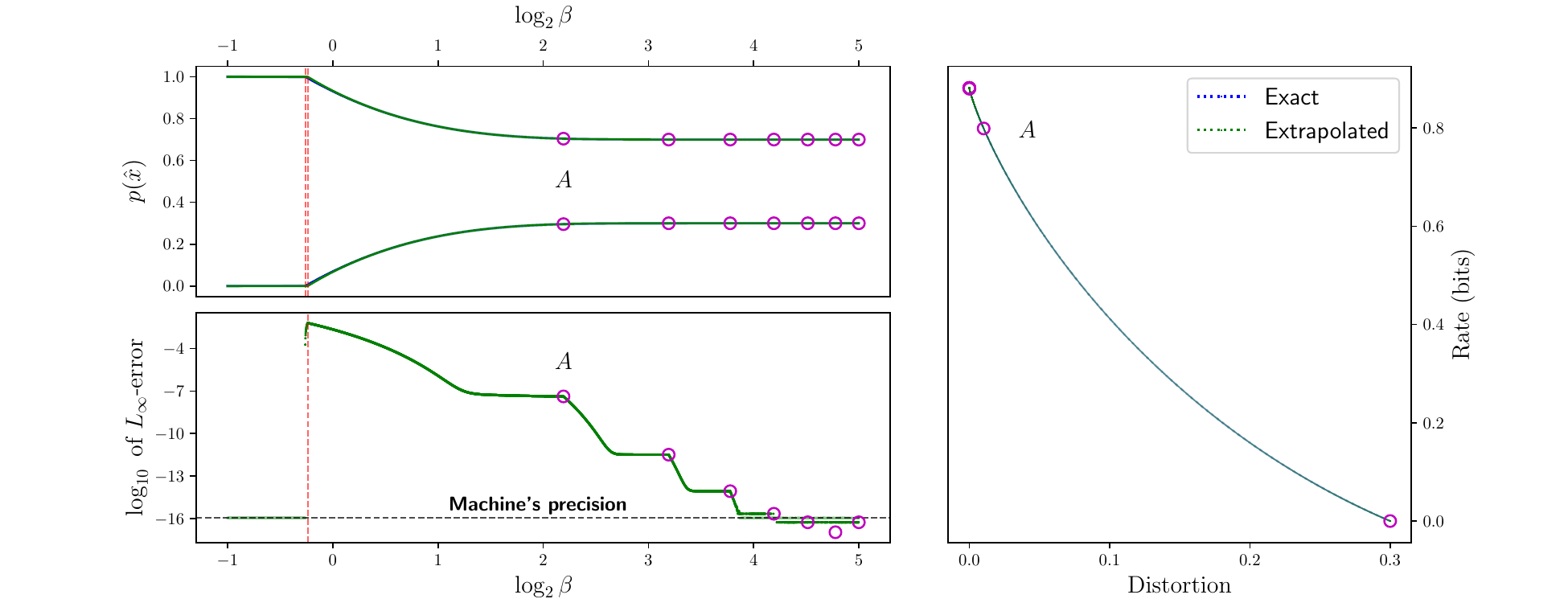}
		\caption{
			\textbf{Extrapolating the entire solution curve from 7 grid points} (magenta circles), \textbf{using Algorithm \ref{algo:root-tracking-for-RD} of order 20}, for Binary source with Hamming distortion \ref{part:proofs}.\ref{sec:binary-source-with-hamming-dist-appendix}, with $\delta = 10^{-2}$. 
			Although the error guarantees of Theorem \ref{thm:taylor-method-converges-for-RD-root-tracking-away-of-bifurcation} aim at small step-sizes, and high-order derivatives are computationally expensive, this example demonstrates the predictive power achievable by exploiting real-analyticity in well-behaved examples; see comments in the main text. 
			\newline
			\textbf{Top left}: The extrapolations (in green) from the grid points overlap the exact solutions (in blue) almost precisely, deviating visibly only near the bifurcation at $\log_2 \beta_c \approx -0.24$. 
			The algorithm overshoots the bifurcation slightly (Section \ref{sub:RD-root-tracking-near-bifrcations}), marked by an extra dashed red vertical to the left. 
			Two Blahut-Arimoto iterations were needed to compute the rightmost grid point, with the next ones extrapolated using a 20th-order Taylor polynomial with a fixed step size of $\Delta \beta \approx -4.57$. 
			\textbf{Bottom left}: $L_\infty$-error between the extrapolation generated from Algorithm \ref{algo:root-tracking-for-RD} and the analytical solution of this problem, in Section \ref{part:proofs}.\ref{sec:binary-source-with-hamming-dist-appendix}. 
			For the most part, the achieved error is near the machine's precision (note the logarithmic scale of the horizontal axis), increasing notably near the bifurcation. 
			While the point of bifurcation is an essential difficulty, Section \ref{part:details}.\ref{sec:error-analysis}, the decreased accuracy to the left of point $A$ could be improved upon, e.g., by varying the step-size; see Section \ref{sub:efficient-RD-root-tracking}. 
			\textbf{Right}: The problem's exact RD curve is visually indistinguishable from the extrapolated one. While the first few grid points cluster to the top-left, most of the curve happens to be extrapolated from the single point $A$.
		}
		\label{fig:7-points-example}
	\end{center}
\end{figure}

While the Jacobian $D_{\inputmarginalVect}(Id - BA_\beta)\rvert_{\inputmarginalVect}$ in marginal coordinates $\inputmarginalVect$ provides an appealing picture of piecewise analytic roots, which alter their course only when colliding at the simplex boundary, it does not tell the whole story. For, the Jacobian in these coordinates can detect only cluster-vanishing bifurcations, as discussed later in Section \ref{part:details}.\ref{sub:obstructions-to-RT-assumptions-for-RD}.
Indeed, RD problems have a plethora of sub-optimal solutions (Section \ref{part:details}.\ref{sub:suboptimal-RD-curves}). These may, for example, exchange optimality in another kind of bifurcation which we call \textit{support-switching}, which is demonstrated by Figure \ref{fig:Berger_example_2.7.3} (in Section \ref{part:details}.\ref{sec:RD-bifurcations-and-root-tracking}).
We show later (Section \ref{part:details}.\ref{sub:support-switching-bifurcations}) that these can explain the linear segments seen in RD curves \citep{berger71}, at least in some cases, as well as some instances of critical slowing down which are not explained by \cite{agmon2021critical} (see Section \ref{part:details}.\ref{sub:obstructions-to-RT-assumptions-for-RD} and Figure \ref{fig:Berger_example_2.7.3}). 

For our purposes, we would like to ensure that the root tracked by our algorithms remains optimal (an achieving distribution), rather than a sub-optimal one. This is done in two steps. First, ensure that Algorithm \ref{algo:taylor-method-for-RD-root-tracking} starts at an achieving distribution $\inputmarginalVect_{\beta_0}$, by invoking Blahut-Arimoto until convergence to an initial condition of full support, \cite{csiszar1974computation} (brought as Theorem \ref{thm:BA-converges-to-RD-curve-for-an-initial-cond-of-full-support} in Section \ref{part:details}.\ref{sub:suboptimal-RD-curves}).
Second, so long that Algorithm \ref{algo:taylor-method-for-RD-root-tracking} starts at an achieving distribution, then the root it tracks remains optimal. 
Similarly, for the below Algorithm \ref{algo:root-tracking-for-RD} (Section \ref{sub:RD-root-tracking-near-bifrcations}).
This follows from Assumption \ref{assumption:only-cluster-vanishing-bifurcations} below, as we show later (in Section \ref{part:details}.\ref{sub:when-does-RTRD-follow-the-optimal-path}). 

At a fixed point of $BA_\beta$ \eqref{eq:BA-operator-def}, the cluster marginal $\inputmarginalVect_\beta$ is determined by an encoder $\intermediateencoderVect_\beta$ via Equation \eqref{eq:marginal-eq}, and vice versa \eqref{eq:encoder-eq}. 
However, the $MN$-by-$MN$ Jacobian matrix $D_{\intermediateencoderVect} (Id - BA_\beta) \big\rvert_{\intermediateencoderVect_\beta}$ with respect to the encoder $\intermediateencoderVect$ contains more information than the $M$-by-$M$ Jacobian $D_{\inputmarginalVect}(Id - BA_\beta)\rvert_{\inputmarginalVect_\beta}$ with respect to the marginal $\inputmarginalVect$ --- see Section \ref{part:details}.\ref{sub:obstructions-to-RT-assumptions-for-RD}.
Where, recall that $N := |\mathcal{X}|$ and $M := |\hat{\mathcal{X}}|$. 
Indeed, the $\intermediateencoderVect$-Jacobian can detect RD bifurcations of any kind, unlike the $\inputmarginalVect$-Jacobian; see Corollary \ref{cor:multiple-RD-sols-are-detectable-by-Jacobian} (in Section \ref{part:details}.\ref{sub:obstructions-to-RT-assumptions-for-RD}), and Proposition \ref{prop:Jacobian-of-BA-in-direct-encoder-coordinates} for its explicit form. 
Therefore, we require 

\begin{assumption}
	\label{assumption:only-cluster-vanishing-bifurcations}
	The $\intermediateencoderVect$-Jacobian $D_{\intermediateencoderVect} (Id - BA_\beta) \big\rvert_{\intermediateencoderVect_\beta}$ is non-singular when evaluated at achieving distributions $\intermediateencoderVect_\beta$, except at cluster-vanishing bifurcations. 
\end{assumption}

This assumption guarantees that the RD problem is well-behaved (e.g., non-degenerate), and that it has a unique optimal solution whose path undergoes only cluster vanishing bifurcations\footnote{ This does not rule out the possibility that sub-optimal roots undergo other kinds of bifurcations.}.
It holds for most of the examples in this paper\footnote{ Except for the right bifurcation in Figure \ref{fig:Berger_example_2.7.3}, at Section \ref{part:details}.\ref{sec:RD-bifurcations-and-root-tracking}.}. 
In particular, the earlier Assumptions \ref{assumption:operator-root-is-a-function-of-beta} and \ref{assumption:solution-is-smooth-in-beta} (of Section \ref{sub:beta-derivs-at-an-operator-root}) necessary for calculating implicit derivatives follow from Assumption \ref{assumption:only-cluster-vanishing-bifurcations} (see Section \ref{part:details}.\ref{sub:when-does-RTRD-follow-the-optimal-path}).
The $\intermediateencoderVect$-Jacobian $D_{\intermediateencoderVect} (Id - BA_\beta) \big\rvert_{\intermediateencoderVect_\beta}$ is singular if the distortion matrix is degenerate (Section \ref{part:details}.\ref{sub:obstructions-to-RT-assumptions-for-RD}).
However, other than that, we find it reasonable to require that the $\intermediateencoderVect$-Jacobian is non-singular outside of cluster-vanishing bifurcations (Assumption \ref{assumption:only-cluster-vanishing-bifurcations}).
Indeed, \cite[Chapter 2]{berger71} says that ``usually, each point on the rate-distortion curve... is achieved by a unique conditional probability assignment. However, if the distortion matrix exhibits certain form of symmetry and degeneracy, there can be many choices of [a minimizer].''

One can test Assumption \ref{assumption:only-cluster-vanishing-bifurcations} directly, by calculating the eigenvalues of the $\intermediateencoderVect$-Jacobian periodically. 
These are expected to vanish only if some cluster vanishes simultaneously --- see Equation \eqref{eq:flowchart-for-different-kinds-of-RD-bifurcations} for details (Section \ref{part:details}.\ref{sub:obstructions-to-RT-assumptions-for-RD}). 
While our algorithms can be extended to handle RD bifurcations also of other kinds, using the bifurcations Section \ref{part:details}.\ref{sec:RD-bifurcations-and-root-tracking}, it is beyond the scope of this work. 

\medskip 
We conclude this subsection with convergence guarantees for our modified Taylor method at Algorithm \ref{algo:taylor-method-for-RD-root-tracking} --- see its error analysis in Section \ref{part:details}.\ref{sec:error-analysis}.
For $\delta > 0$, the \emph{closed $\delta$-interior} of the simplex $\Delta[ \hat{\mathcal{X}} ]$ consists of the distributions $\inputmarginalVect \in \Delta[ \hat{\mathcal{X}} ]$ with $\inputmarginal{'} \geq \delta$ for all $\hat{x}' \in \hat{\mathcal{X}}$. 
Note that this is not the same as the interior of the simplex, which is an open set. 

\begin{thm}[Taylor method converges uniformly between RD bifurcations, on a full support]		\label{thm:taylor-method-converges-for-RD-root-tracking-away-of-bifurcation}
	Let $\inputmarginalVect_{\beta}$ be a root of the RD problem defined by $p_X$ and a finite non-degenerate distortion measure $d$. Suppose that $\inputmarginalVect_{\beta}$ is of full support at $\beta_0 > 0$. Let $\delta > 0$ be a \emph{cluster mass threshold}, such that $\delta < \min_{\hat{x}} \inputmarginalSymbol_{\beta_0}(\hat{x})$.
	Then there exists $0 \leq \beta_f(\delta) < \beta_0$ such that,
	\begin{enumerate}
		\item $\inputmarginalVect_\beta$ is in the closed $\delta$-interior of $\Delta[\hat{\mathcal{X}}]$ for $\beta\in [\beta_f(\delta), \beta_0]$; and
		\item For $l>0$, the error of an $l$-th order Taylor method satisfies
		\begin{equation}			\label{eq:taylor-method-convergence-guarantee-for-RDRT-in-thm}
			\max_{\beta \in [\beta_f(\delta), \beta_0]} \left\| \tilde{\inputmarginalVect}_\beta - \inputmarginalVect_\beta \right\|_\infty = 
			O(|\Delta\beta|^l)
		\end{equation}
		for $|\Delta\beta| > 0$ small enough, 
		and $\tilde{\inputmarginalVect}_\beta$ the Taylor method approximations defined by \eqref{eq:taylor-method-def-in-cluster-marginal}.
	\end{enumerate}
\end{thm}

Since Algorithm \ref{algo:taylor-method-for-RD-root-tracking} tries to stop $\delta$-away from a bifurcation, its error converges uniformly for small enough step sizes, at a rate \eqref{eq:taylor-method-convergence-guarantee-for-RDRT-in-thm} proportional to its order $l$. 
All the RD roots we compute are generated by Algorithm \ref{algo:taylor-method-for-RD-root-tracking}, except for the initial point and those too close to a bifurcation. 
Bifurcations are handled by Algorithm \ref{algo:root-tracking-for-RD} below (Section \ref{sub:RD-root-tracking-near-bifrcations}). 
It will reduce the problem at hand, guaranteeing the full support required by Theorem \ref{thm:taylor-method-converges-for-RD-root-tracking-away-of-bifurcation}.
Thus, we obtain an error guarantee for nearly all the generated grid points.
The proof of Theorem \ref{thm:taylor-method-converges-for-RD-root-tracking-away-of-bifurcation} (in Section \ref{part:proofs}.\ref{sub:proof-of-thm:taylor-method-converges-for-RD-root-tracking-away-of-bifurcation}) is based on standard Taylor method error analysis, brought at Theorem \ref{thm:error-analysis-for-euler-method} (in Section \ref{part:details}.\ref{sub:error-analysis-of-Taylor-method-background}). Its crux is that implicit derivatives and the local Lipschitz constants of Taylor's method are bounded uniformly, on suitable compact subsets in the simplex interior; see Sections \ref{part:details}.\ref{sub:computational-difficulty-of-RTRD}-\ref{sub:local-error-estimates-for-beta-derivs}.

The convergence guarantees of Theorem \ref{thm:taylor-method-converges-for-RD-root-tracking-away-of-bifurcation} suggest considerations for selecting the parameters of Algorithm \ref{algo:taylor-method-for-RD-root-tracking}. 
The cluster mass threshold $\delta > 0$ obviously should not be too large, so that a bifurcation is not accidentally detected when there is none. 
On the other hand, increasing $\delta$ restricts the algorithm's computational difficulty, as discussed in Section \ref{part:details}.\ref{sec:error-analysis}. 
One would then like to select $|\Delta \beta|$ small enough and the order $l$ large enough such that Equation \eqref{eq:taylor-method-convergence-guarantee-for-RDRT-in-thm} (in Theorem \ref{thm:taylor-method-converges-for-RD-root-tracking-away-of-bifurcation}) guarantees that $\tilde{\inputmarginalVect}_\beta$ does not deviate too much from the true solution $\inputmarginalVect_\beta$. 
This guarantees the solution's accuracy and ensures that comparing a cluster mass $\tilde{\inputmarginalVect}_\beta(\hat{x})$ against the threshold $\delta$ is meaningful (on line \ref{algo:root-tracking-on-interval:stopping-cond} of Algorithm \ref{algo:taylor-method-for-RD-root-tracking}). 
On the other hand, decreasing $|\Delta \beta|$ or increasing $l$ impacts the algorithms' complexities, as elaborated in Section \ref{sub:costs-and-error-to-cost-tradeoff-of-RD-root-tracking-in-main-results-section}. 
See the Figures throughout and our implementation for sample values. 
With that, the considerations implied by Theorem \ref{thm:taylor-method-converges-for-RD-root-tracking-away-of-bifurcation} are not the only way to choose the algorithms' parameters. 
For example, motivated by the earlier discussion on analyticity, Figure \ref{fig:7-points-example} depicts high-quality approximations even though its step size $\Delta \beta$ is large.

\subsection{RD root-tracking near bifurcations}
\label{sub:RD-root-tracking-near-bifrcations}

We proceed with the solution's handling where Algorithm \ref{algo:taylor-method-for-RD-root-tracking} left off, once a cluster-vanishing bifurcation has been detected. 
The following heuristic re-gains accuracy while avoiding the computational cost due to Blahut-Arimoto's critical slowing down near the bifurcation, \citep{agmon2021critical}. 
Assumption \ref{assumption:only-cluster-vanishing-bifurcations} is required below. 

\begin{figure}[H]
	\centering
	\includegraphics[trim={0 0 0 1.5cm}, clip, width=.55\textwidth]{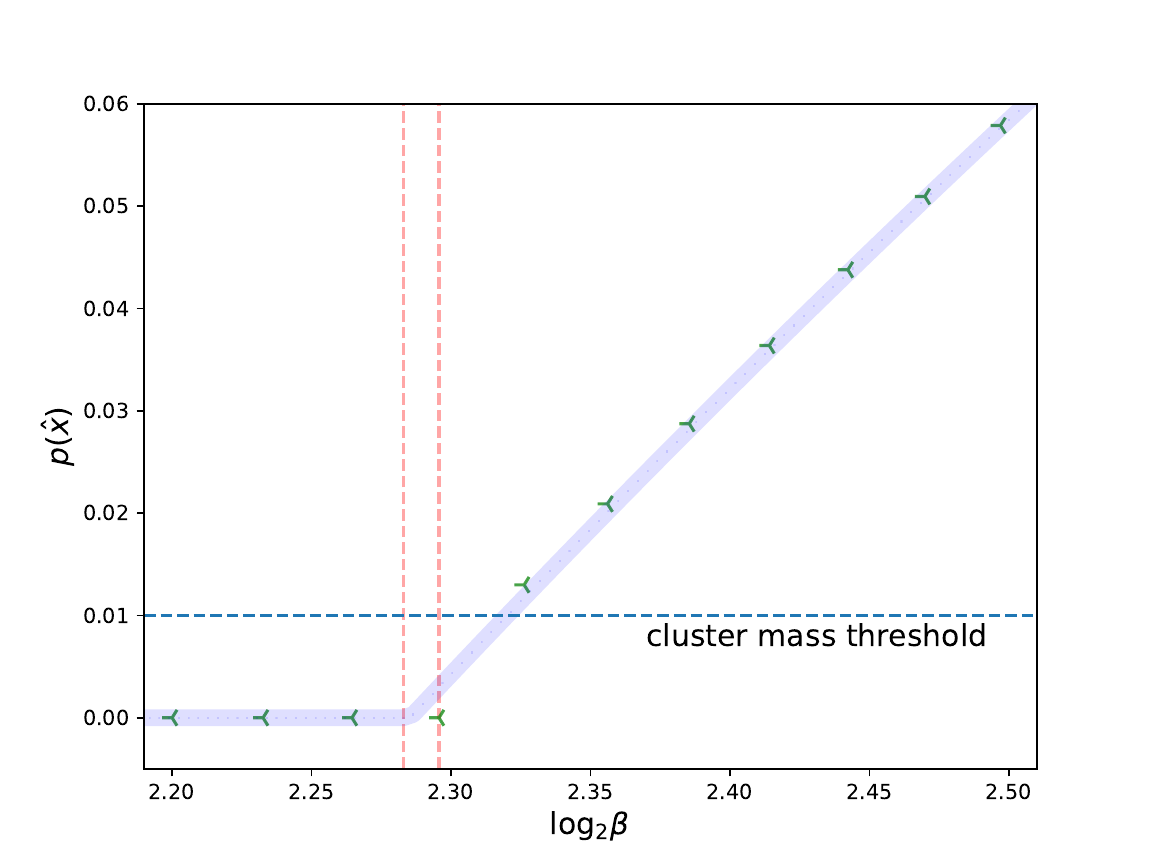}
	\caption{
		\textbf{Undershooting a bifurcation}. 
		The bifurcation is detected early (right dashed red vertical), slightly to the right of its true position (left vertical). 
		Once the threshold is crossed, our heuristic zeros the nearly-vanished cluster (at right vertical, not shown). 
		The distribution it yields (at right vertical, marked) achieves a sub-optimal branch (not shown), in which this symbol is zero. 
		However, optimality is re-gained shortly afterwards, when reaching the bifurcation; see main text. 
		The markers denote the output of Algorithm \ref{algo:root-tracking-for-RD} of second-order on a 300-point grid, with BA for comparison in blue, near the rightmost bifurcation of Figure \ref{fig:reconstructing-an-RD-solution-curve}.
	}
	\label{fig:heuristic-at-cluster-vanishing}
	\vspace*{-5pt}
\end{figure}

\medskip
Consider the first approximation $(\tilde{\inputmarginalVect}_n, \beta_n)$ after the cluster mass threshold $\delta$ has been crossed.
We zero any cluster $\hat{x}$ for which $\tilde{\inputmarginalSymbol}_n(\hat{x}) < \delta$, normalize, and then use the resulting distribution $\tilde{\inputmarginalVect}'$ as the initial condition for Blahut-Arimoto's algorithm. 
$\tilde{\inputmarginalVect}'$ has precisely those nearly-vanished clusters set to zero. 
Thus, Blahut-Arimoto converges to an achieving distribution $\tilde{\inputmarginalVect}_n''$ of the problem reduced to $\supp \tilde{\inputmarginalVect}'$, \citep{csiszar1974computation} (see Section \ref{part:details}.\ref{sub:suboptimal-RD-curves}). 
We then invoke Algorithm \ref{algo:taylor-method-for-RD-root-tracking} anew on the reduced problem, starting at the achieving distribution $(\tilde{\inputmarginalVect}_n'', \beta_n)$. 
This process repeats so long that $\beta_n > 0$ and the initial condition $\tilde{\inputmarginalVect}_n''$ is non-trivial.
We summarize this heuristic in Algorithm \ref{algo:root-tracking-for-RD}.

The first approximation $(\tilde{\inputmarginalVect}_n, \beta_n)$ after crossing the threshold may either be earlier than\footnote{ Note that the step size $\Delta \beta$ is negative. } the true point of bifurcation $\beta_c$, $\beta_c < \beta_n$, called an \textit{undershooting}. 
Or the bifurcation may be detected too late, $\beta_n \leq \beta_c$, an \textit{overshooting}. 
Either of these must happen when a bifurcation is detected, usually with strict inequality. 
A sub-optimal root exists at a right vicinity of a bifurcation, $\beta > \beta_c$, as shown in Section \ref{part:details}.\ref{sub:cluster-vanishing-bifs-are-bifs}. 
It is a solution of the reduced problem to $\supp \tilde{\inputmarginalVect}'$, with $\tilde{\inputmarginalVect}'$ the normalized zeroization as above. 
The optimal and sub-optimal roots intersect at $\beta_c$, merging there into one.
On an overshoot, $\beta_n \leq \beta_c$, the initial condition $\tilde{\inputmarginalVect}_n''$ produced by the heuristic will remain on the optimal branch. 
On an undershoot, $\beta_c < \beta_n$, $\tilde{\inputmarginalVect}_n''$ will be on the sub-optimal branch of smaller support $\supp \tilde{\inputmarginalVect}'$ while re-gaining optimality shortly thereafter, as demonstrated by Figure \ref{fig:heuristic-at-cluster-vanishing}. 

Invoking Blahut-Arimoto before starting anew with Algorithm \ref{algo:taylor-method-for-RD-root-tracking} re-gains accuracy, compensating for the approximation errors accumulated so far, while also ensuring that the root being tracked is optimal on the non-vanished clusters.
Further, this allows us to \emph{avoid} Blahut-Arimoto's critical slowing down.
For, as shown by \cite{agmon2021critical}, the latter is exhibited due to a Jacobian eigenvalue that gradually vanishes when approaching a bifurcation. 
This eigenvalue corresponds to a cluster of vanishing marginal and so is removed once that cluster is removed, before invoking Blahut-Arimoto on line \ref{algo:RD-root-tracking:BA-after-bif} of Algorithm \ref{algo:root-tracking-for-RD}.
The lack of critical slowing down is consistent with our numerical results.

\begin{algorithm}
	\caption{\hspace{-4pt}: Root-tracking for RD, with $l$-th order Taylor method.}
	\label{algo:root-tracking-for-RD}
	\begin{algorithmic}[1]
		\Input
		\Statex $\beta_0$, a root $\inputmarginalVect_{\beta_0}$ of $Id - BA_{\beta_0}$ \eqref{eq:RD-operator-def}, a step-size $\Delta \beta < 0$, a cluster mass threshold $0 < \delta < 1$, 
		\Statex an order $l > 0$, and an RD problem $(d, p_X)$.
		\Output
		\Statex Approximations $\tilde{\inputmarginalVect}_{i}$ of the true solution $\inputmarginalVect_{\beta_i}$ at $\beta_i$, 
		\Statex and the $l$-th order Taylor expansions $p_i$ around $(\tilde{\inputmarginalVect}_i, \beta_i)$.
		\State Initialize $\tilde{\inputmarginalVect} \gets \inputmarginalVect_{\beta_0}, \beta \gets \beta_0, results \gets \{\}$.
		\While{$|\supp \tilde{\inputmarginalVect}| > 1$ \textbf{and} $\beta > 0$}
		\Comment{Stop if solution has a trivial support.}
		\State $\tilde{\inputmarginalVect}, \bar{d} \gets \text{the reduction of } \tilde{\inputmarginalVect}, d \text{ to } \supp \tilde{\inputmarginalVect}$.			\label{algo:RD-root-tracking:reduction}
		\State $solution \; path \gets \Call{Track RD Root To Bifurcation}{\tilde{\inputmarginalVect}, \beta; \Delta \beta, \delta, l, \bar{d}, p_X}$.
		\begin{flushright}
			\vspace{-9pt}
			\LineComment{Algorithm \ref{algo:taylor-method-for-RD-root-tracking}.}
		\end{flushright}
		\State Append items in $solution \; path$ to $results$, except for the last.
		\State $(\tilde{\inputmarginalVect}, \beta) \gets \text{last item in } solution \; path$.
		\State $\tilde{\inputmarginalSymbol}(\hat{x}) \gets 0$ for any $\hat{x}$ with $\tilde{\inputmarginalSymbol}(\hat{x}) < \delta$.
		\State $\tilde{\inputmarginalVect} \gets $ normalize $\tilde{\inputmarginalVect}$.		\label{algo:RD-root-tracking:normalization}
		\State $\tilde{\inputmarginalVect} \gets \Call{Blahut-Arimoto}{\tilde{\inputmarginalVect}, \beta_n}$.
		\Comment{Iterate $BA_\beta$ \eqref{eq:BA-operator-def} until convergence.}		\label{algo:RD-root-tracking:BA-after-bif}
		\EndWhile
		\State \textbf{return} $results$
		\Comment{Approximations of the true solutions $\inputmarginalVect(\beta_n)$.}
	\end{algorithmic}
\end{algorithm}

In practice, this heuristic works nicely so long that the bifurcation is not missed altogether due to large approximation errors in Algorithm \ref{algo:taylor-method-for-RD-root-tracking}.
Off-grid points are later extrapolated using a Taylor step \eqref{eq:taylor-method-def-in-cluster-marginal} at the last grid point, unless the extrapolation has a negative coordinate due to an over-shot bifurcation. In this case, the point is extrapolated using the expansion at the next grid point.
As an alternative heuristic, the bifurcation could be handled by extrapolating it linearly from the last point before the threshold is crossed and then using Blahut-Arimoto, exploiting the accuracy of first-order derivatives; cf., Figure \ref{fig:derivative-calculation-loses-accuracy-near-bifurcation}. This also works nicely in practice, on grids dense enough.

\subsection{Computational costs and cost-to-error tradeoff}
\label{sub:costs-and-error-to-cost-tradeoff-of-RD-root-tracking-in-main-results-section}

We provide bounds on the computational and memory costs of root-tracking for RD (Algorithm \ref{algo:root-tracking-for-RD}), and estimate the tradeoff between its error and computational costs. 
While the cost of implicit derivatives grows rapidly with the order, counter-intuitively, higher orders usually make much better use of the computational cost.
Blahut-Arimoto makes little to no use of the computational effort invested at previous grid points and does not yield off-grid information.
In contrast, our algorithm leverages the derivatives at a point to improve the accuracy at subsequent ones, allowing also cheap extrapolation of off-grid points. 
Despite the many improvements possible to root-tracking RD (in Section \ref{sub:efficient-RD-root-tracking} below), the cost of an entire solution curve is already roughly comparable to BA as is, as seen in Figure \ref{fig:err-to-computational-cost-tradeoff} (Section \ref{sec:introduction}). 
See also Section \ref{part:details}.\ref{sec:computational-complexities}, on the complexities of root-tracking in general, with and without tensor memorization, and the complexities of RD derivative tensors.

\medskip
The computational cost of RD root-tracking (Algorithm \ref{algo:root-tracking-for-RD}) is comprised of the cost of implicit derivatives and of Blahut-Arimoto near bifurcations. However, the cost due to BA is typically negligible since our heuristic avoids its critical slowing down near bifurcations (see Section \ref{sub:RD-root-tracking-near-bifrcations}).
Memorizing (caching) each computed tensor $D^{b+m}_{\beta^b, \inputmarginalVect^m} (Id - BA_\beta)[\inputmarginalVect]$ (at Theorem \ref{thm:high-order-derivs-of-BA-in-main-text} in Section \ref{sub:high-order-deriv-tensors-of-BA}) upon its first appearance reduces the computational costs drastically, as shown in Section \ref{part:details}.\ref{sec:computational-complexities}. 
Similarly, also for memorizing implicit derivatives when calculating them recursively in Algorithm \ref{algo:high-order-derivs-of-operator-roots} (Section \ref{sub:high-order-beta-derivatives-at-an-operator-root}). 
Our implementation also memorizes some of the quantities en route, further reducing costs.
For the most part, these can be discarded once done computing at a point, so the memory costs do \textit{not} scale with the grid size.
Thus, the computational cost of Algorithm \ref{algo:root-tracking-for-RD} is essentially determined by the implicit derivatives' cost at a point (times the number of grid points). 
We provide the following (loose) bounds.

\begin{thm}[Complexity bounds for RD root tracking]			\label{thm:complexity-of-RD-root-tracking}
	For $l \geq 2$, the computational complexity of Algorithm \ref{algo:root-tracking-for-RD} of order $l$ with tensor memorization is bounded by
	\begin{equation}		\label{eq:computational-complexity-of-RD-root-tracking-for-L-geq-2-in-thm}
		O\left( N \cdot (M + l)^{(3l + \nicefrac{5}{2})} e^{(2M+l) H_e(\rho) + \pi \sqrt{\nicefrac{2l}{3}} } \right) 
	\end{equation}
	times the number of grid-points. Where, $N$ and $M$ are the source and reproduction alphabet sizes (respectively), $\rho := \tfrac{2M}{2M + l}$, and $H_e(\cdot)$ is the binary entropy in nats. For $l = 1$, it is instead
	\begin{equation}		\label{eq:computational-complexity-of-RD-root-tracking-for-L-=1-in-thm}
		O\left( M^3 N \right) \;.
	\end{equation}
	
	The algorithm's memory complexity is bounded by
	\begin{equation}		\label{eq:memory-complexity-of-RD-root-tracking-in-thm}
		O\left( l^{(l + \nicefrac{5}{2}) } \ln l \right) + O\left(MNl^2\right) +
		O\left( M^l l \right) 
	\end{equation}
\end{thm}

The bounds provided by Theorem \ref{thm:complexity-of-RD-root-tracking} are rather loose in the derivative's order $l$.
The term $(M + l)^{(3l + \nicefrac{5}{2})}$ in the computational complexity \eqref{eq:computational-complexity-of-RD-root-tracking-for-L-geq-2-in-thm} is hyper-exponential in $l$ primarily due to inefficient copy operations used to simplify our implementation, and due to the loose bounds used in the theorem's proof (in Section \ref{part:proofs}.\ref{sub:complexity-of-RD-root-tracking}). 
Similarly, the first term in the memory complexity \eqref{eq:memory-complexity-of-RD-root-tracking-in-thm} is in practice only of order $O(l \cdot 1.73^l)$, at least for $l\leq 25$; see Table \ref{tab:complexities-of-RD-deriv-tensors} in Section \ref{part:details}.\ref{sec:computational-complexities}. 
Nevertheless, the complexity of higher implicit multivariate derivatives is high even when the cost of RD derivative tensors is set aside (Proposition \ref{prop:total-computational-complexity-with-memorization} in Section \ref{part:details}.\ref{sec:computational-complexities}). 
This arguably stems from\footnote{ cf., Theorem \ref{thm:mFDB-from-Ma} (Section \ref{part:details}.\ref{sec:multivariate-faa-di-brunos-formula}), which is used to prove the formula of Theorem \ref{thm:formula-for-high-order-expansion-of-F-in-main-result-sect} (Section \ref{sub:high-order-beta-derivatives-at-an-operator-root}) for higher implicit multivariate derivatives. } the combinatorics of partial derivatives, \citep{hardy2006combinatorics}.
With that, we note that the complexities of Algorithm \ref{algo:root-tracking-for-RD} (both in practice and the bounds \eqref{eq:computational-complexity-of-RD-root-tracking-for-L-geq-2-in-thm}-\eqref{eq:memory-complexity-of-RD-root-tracking-in-thm}) do \textit{not} depend on the problem's details\footnote{ e.g., examine the dependence of Equations \eqref{eq:expected-k-th-power-distortion-def}-\eqref{eq:combinatorial-G-in-terms-of-polynomials} (in Section \ref{sub:high-order-deriv-tensors-of-BA}) on the problem definition $(d, p_X)$.}, but only on its dimensions $N$ and $M$, and on the order $l$; denote by $\text{cost-per-grid-point}(N, M, l)$ its computational cost at a grid-point. 
On the other hand, the algorithm's accuracy does depend on the problem's details. 

Although the algorithm's complexities grow rapidly with the order $l$, higher orders generally provide a much better cost-to-error tradeoff, as shown below. 
Thanks to our choice of cluster-marginal coordinates (in Section \ref{sub:high-order-deriv-tensors-of-BA}), the complexities are only linear in the source alphabet size $N$. 
This is useful when computing with large source alphabets, $|M| \ll |N| $. 
For any fixed order $l$, $e^{(2M+l) H_e(\rho)}$ is asymptotically linear in $M$, and so the costs are polynomial in $M$. 
Further, the reproduction alphabet size $M$ in these bounds can be replaced by the solution's support size, which varies along the grid. For, Algorithm \ref{algo:root-tracking-for-RD} reduces the problem after each bifurcation (in step \ref{algo:RD-root-tracking:reduction}).

\begin{figure}[h!]
	\centering
	\vspace{-5pt}
	\includegraphics[trim={0 0 0 1cm}, clip, width=0.55\textwidth]{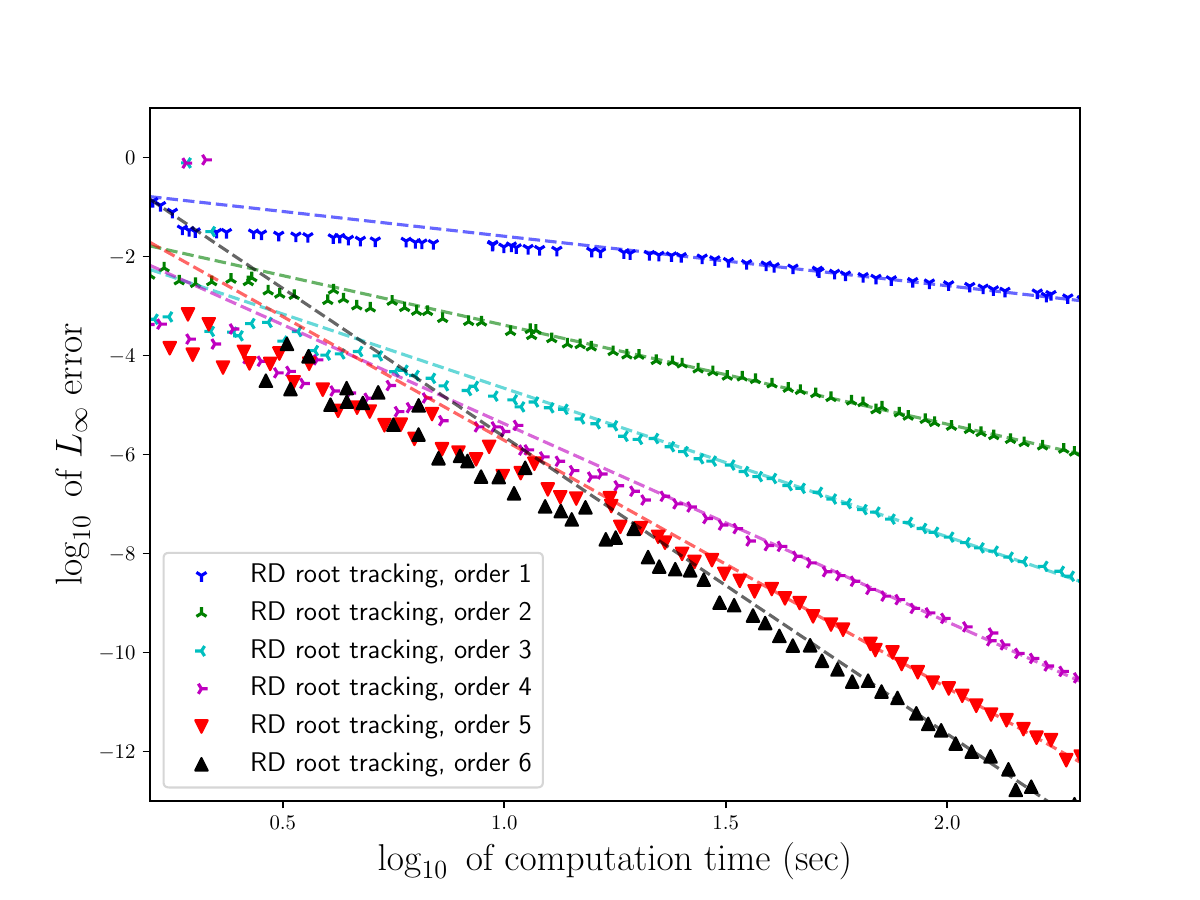}
	\caption{
		\textbf{Higher order methods make better use of computational costs when small error is required.}
		To demonstrate Equation \eqref{eq:error-goes-like-cost-to-minus-order}, 
		a linear regression of slope $(-l)$ (dashed line) is matched to the rightmost markers per method order $l$.
		Plot details are as in Figure \ref{fig:err-to-computational-cost-tradeoff}.
	}
	\label{fig:error-is-power-of-computational-cost-for-RDRT}
	\vspace{-5pt}
\end{figure}

Ideally, it would be possible to tell the tradeoff between accuracy and computational cost in advance, with better results expected as the computational effort is increased. 
Indeed, except for $\delta$-close to a bifurcation, an $l$-th order RD root-tracking with step-size $|\Delta\beta|$ converges to the true solution at a rate of $O(|\Delta\beta|^l)$, by Theorem \ref{thm:taylor-method-converges-for-RD-root-tracking-away-of-bifurcation} (in Section \ref{sub:taylor-method-for-RD-root-tracking}). 
So, write $error \propto |\Delta\beta|^l$, when the step-size $|\Delta\beta|$ is small. 
As the number of grid points is inversely proportional to the step size, the total computational cost is of order $|\Delta \beta|^{-1} \cdot \text{cost-per-grid-point}(N, M, l)$. This yields
\begin{equation}		\label{eq:error-goes-like-cost-to-minus-order}
	\log error \propto l \cdot \log \left(\text{cost-per-grid-point}(N, M, l)\right) - l \cdot \log (\text{total-cost}) \;.
\end{equation}
Carefully note that, with other parameters fixed, increasing the computational effort reduces the error at a rate of $(\text{total-cost})^{-l}$. 
Thus, despite their costs, derivatives of higher orders $l$ generally provide a much better tradeoff between the error and computational costs. 
This off-hand tradeoff is demonstrated by Figure \ref{fig:error-is-power-of-computational-cost-for-RDRT}; it is increasingly accurate at higher costs (smaller $|\Delta\beta|$), as expected. 
Similar plots are typical to the error analysis of the Taylor method; e.g., \cite[Figure 223(i)]{butcher2016numerical}.

Unlike Algorithm \ref{algo:root-tracking-for-RD}, Blahut-Arimoto makes little use (if any) of the computational cost invested in adjacent grid points.
For, BA is usually computed independently along a grid (computing each point anew) or in reverse annealing \cite{rose1990deterministic}, where the solution at one point is taken as the initialization for the next. 
While reverse annealing reduces the computational costs noticeably, it may follow sub-optimal solution branches. 
In a sense, reverse annealing is similar to root tracking of zeroth order.

Since Algorithm \ref{algo:root-tracking-for-RD} leverages each point to calculate the next, executing it on denser grids or at higher orders will usually improve the overall accuracy (though not always -- see Figure \ref{fig:err-to-computational-cost-tradeoff}).
In contrast, adding grid points to BA usually degrades its overall accuracy due to critical slowing down, \citep{agmon2021critical}; cf., Figure \ref{fig:derivative-calculation-loses-accuracy-near-bifurcation} (bottom). 
This is demonstrated by Figure \ref{fig:err-to-computational-cost-tradeoff}.
Judging by our implementations of both, it is easy to find parameters where either algorithm outperforms the other.
With that, various improvements can be made to Algorithm \ref{algo:root-tracking-for-RD}, as we discuss next.

\subsection{Possible improvements to RD root-tracking}
\label{sub:efficient-RD-root-tracking}

We discuss several approaches which we believe could significantly improve the cost-to-error tradeoff of RD root tracking (Algorithm \ref{algo:root-tracking-for-RD}). Most notably, by making the derivative order and step size adaptive.
cf., the related Sections \ref{sub:costs-and-error-to-cost-tradeoff-of-RD-root-tracking-in-main-results-section} and \ref{part:details}.\ref{sec:error-analysis} on the cost-to-error tradeoff and error analysis.

\medskip
Tracking an operator's root belong in general to a family of hard-to-solve numerical problems called \textit{stiff} (see Section \ref{part:details}.\ref{sub:computational-difficulty-of-RTRD}). 
While we acknowledge that much literature exists on stiff differential equations, we have chosen to ignore it to avoid straying off our main line of discussion.
Instead, we focus on conspicuous improvements, which this work has not attended for the sake of simplicity.

Algorithm \ref{algo:root-tracking-for-RD} uses both the classic Taylor method for solving ordinary differential equations and Blahut-Arimoto's algorithm.
Numerical methods for solving $\tfrac{d\bm{x}}{d\beta} = f(\bm{x}, \beta)$ usually exploit the values of $f$ or its derivatives to approximate $\bm{x}(\beta)$, but often cannot directly reduce the error accumulated from the true solution $\bm{x}(\beta_0)$ at a point $\beta_0$. cf., Section 22 in \cite{butcher2016numerical} for example.
On the other hand, Blahut-Arimoto's algorithm reduces the error from the true solution, but does not follow its path $\bm{x}(\beta)$ as $\beta$ varies.
We note that these building blocks could have been combined in ways other than those we have chosen. 
However, to our understanding, it is necessary to use a component that directly reduces the error accumulated at a point, whether BA or another algorithm. 
Otherwise, either the error accumulated until a bifurcation would get out of control, or the computational costs. For, following an operator's root is inherently stiff; see Section \ref{part:details}.\ref{sub:computational-difficulty-of-RTRD}.

Rather than using implicit derivatives to construct the Taylor expansion around a point, they could be used to find the rational fraction which best matches the expansion, or \textit{Pad\'{e} approximant}; e.g., the classic \cite{baker1975essentials}.
These are often superior to Taylor series, \cite{wolframPadeApprox}.

However, even with the Taylor series method, neither fixing its step size nor its order are optimal choices.
For, there may be large spans along the $\beta$-axis where the solution changes very little. For example, in Figure \ref{fig:7-points-example}, most of the computational power is spent where the solution barely changes (till point $A$ there). Instead, big step sizes or low orders could have been more cost-effective there, while using small step sizes or high orders where the solution changes rapidly (to the left of $A$ there).
On the other hand, numerical derivatives of higher orders lose their accuracy faster when approaching a bifurcation (Figure \ref{fig:derivative-calculation-loses-accuracy-near-bifurcation} top).
An estimate of the local error is needed to make the method adaptive, \cite[226]{butcher2016numerical}.
This can be achieved by estimating the local Lipschitz constants of the Taylor expansion \eqref{eq:taylor-method-def-in-cluster-marginal}, by computing $D_{\inputmarginalVect} T_l$. 
These can be computed easily using only the results presented so far and Proposition \ref{prop:Jacobian-of-high-order-beta-derivative} (in Section \ref{part:details}.\ref{sub:local-error-estimates-for-beta-derivs}). 

Last, our implementation for computing the derivative tensors could be optimized.
In addition to many non-algorithmic optimizations, note that it copies out each partial derivative to multiple memory locations, which correspond to permutations of the derivatives' order, as elaborated in Section \ref{part:proofs}.\ref{subsub:complexity-of-RD-deriv-tensor}.

\newpage
\part{The ideas underlying root-tracking for rate-distortion}
\label{part:details}

This part elaborates on the main ideas supporting the results of Part \ref{part:how-and-what}. 

Section \ref{sec:multivariate-faa-di-brunos-formula} elaborates on the mathematical prerequisites of arbitrary-order derivative calculations. 
These are necessary for the formula of Theorem \ref{thm:formula-for-high-order-expansion-of-F-in-main-result-sect} for implicit derivatives of arbitrary-order (Section \ref{part:how-and-what}.\ref{sub:high-order-beta-derivatives-at-an-operator-root}), and for Blahut-Arimoto's derivative tensors (Section \ref{part:how-and-what}.\ref{sub:high-order-deriv-tensors-of-BA}). 
Section \ref{sec:high-order-derivs-of-BA-in-marginal-coords} outlines the proof of the latter, provides related results such as the RD ODE, and comments how to compute RD derivative tensors efficiently. 

To reconstruct an RD solution curve from its implicit derivatives, Section \ref{sec:RD-bifurcations-and-root-tracking} provides some understanding of RD bifurcations. 
While its goal to show that Algorithms \ref{algo:taylor-method-for-RD-root-tracking} and \ref{algo:root-tracking-for-RD} follow the optimal solution path subject to Assumption \ref{assumption:only-cluster-vanishing-bifurcations}, we provide several basic results on RD bifurcations, and the tools to detect and distinguish between two types of bifurcations.

In Section \ref{sec:error-analysis} we analyze the error accumulated by Algorithm \ref{algo:taylor-method-for-RD-root-tracking} (in Section \ref{part:how-and-what}.\ref{sub:taylor-method-for-RD-root-tracking}). The analysis not only reveals the source of the computational difficulty, but also provides a tool that could be used to mitigate it (see Section \ref{part:how-and-what}.\ref{sub:efficient-RD-root-tracking} on improvements).
Finally, Section \ref{sec:computational-complexities} provides results supporting the computational and memory complexity bounds presented at Section \ref{part:how-and-what}.\ref{sub:costs-and-error-to-cost-tradeoff-of-RD-root-tracking-in-main-results-section}: for root tracking in general, for RD root tracking (Algorithm \ref{algo:root-tracking-for-RD}), and for computing RD derivative tensors.

\medskip
\section{Preliminaries: the multivariate Fa\`a di Bruno's formula}
\label{sec:multivariate-faa-di-brunos-formula}

We review the preliminaries needed to calculate the higher derivative tensors of $Id - BA_\beta$ \eqref{eq:RD-operator-def} (in Section \ref{sec:high-order-derivs-of-BA-in-marginal-coords}), and for the proof of the arbitrary-order expansion of $\tfrac{d^l}{d\beta^l} F$ of Theorem \ref{thm:formula-for-high-order-expansion-of-F} (in Section \ref{part:proofs}.\ref{sub:proof-of-formula-for-high-order-expansion-of-F}).
These are based chiefly on the multivariate Fa\`a di Bruno's formula from \cite{ma2009higher}, in Theorem \ref{thm:mFDB-from-Ma} below.

\medskip 
The $n$-th derivative of a product $g(t) \cdot f(t)$ of (scalar) functions $f$ and $g$ is given by the well-known Leibniz rule,
\begin{equation}		\label{eq:Leibniz-rule-for-product}
	\frac{d^n}{dt^n} \big(g f\big) = \sum_{k=0}^n \binom{n}{k} \frac{d^k g}{dt^k} \cdot \frac{d^{n-k} f}{dt^{n-k}} \;.
\end{equation}
In comparison, the $n$-th derivative of the composite $g\left(f(t)\right)$ is given by the less familiar Fa\`a di Bruno's formula, \cite{faa1855sullo, faadi1857note}. If all the necessary derivatives of $g(t)$ and $f(t)$ are defined, it can be written \cite[Theorem 2]{roman1980formula} as
\begin{equation}		\label{eq:univariate-faa-di-bruno-formula}
	\frac{d^n}{dt^n} g\big(f(t)\big) = 
	\sum \frac{n!}{k_1! \; 1!^{k_1} \cdots k_n! \; n!^{k_n}} 
	\cdot g^{(k_1 + \dots + k_n)}\left(f(t)\right) \cdot
	\left( f^{(1)}(t) \right)^{k_1} \cdots \left( f^{(n)}(t) \right)^{k_n} ,
\end{equation}
where the sum is over all the non-negatives integers $k_1, \dots, k_n$ satisfying $k_1 + 2k_2 + \dots + n k_n = n$. 
These represent a \emph{partition} (as in Section \ref{part:how-and-what}.\ref{sub:high-order-beta-derivatives-at-an-operator-root}) of an $n$-sized set to $k_1$ subsets of size 1, $k_2$ subsets of size 2, etc., hence $k_1 + k_2 + \dots + k_n$ subsets in total.
cf., \cite{hardy2006combinatorics} for a lucid combinatorial interpretation of the formula's coefficients, \cite{johnson2002curious} for further details and historical notes.

For our purposes, we shall also need a multivariate version of Fa\`a di Bruno's formula \eqref{eq:univariate-faa-di-bruno-formula}.
While its multivariate form is often attributed to\footnote{ cf., literature survey at the introduction of \cite{ma2009higher}.} \citet[Theorem 2.1]{constantine1996multivariate}, we present a more modern form by \cite{ma2009higher}, which generalizes the combinatorial arguments of \cite{hardy2006combinatorics} from the univariate case. 

Following \cite{ma2009higher}, we first recall some multivariate notation. 
Denote by $\bb{N}_0$ the non-negative integers, and let $\bm{x}\in \bb{R}^\nu$.
For a \emph{multi-index} $\bm{\alpha} = \left( \alpha_1, \dots, \alpha_\nu\right)\in \bb{N}_0^\nu$ of \emph{length} $\nu$, define
\begin{equation}		\label{eq:multivariate-notation-defs}
	|\bm{\alpha}| := \sum_{j=1}^\nu \alpha_j, \quad
	\bm{\alpha}! := \prod_{j=1}^\nu \alpha_j!, \quad
	\bm{x}^{\bm{\alpha}} := \prod_{j=1}^\nu x_j^{\alpha_j}, \quad
	\frac{\partial^{|\bm{\alpha}|} z}{\partial \bm{x}^{\bm{\alpha}}} := \prod_{j=1}^\nu \left(\frac{\partial}{\partial x_j}\right)^{\alpha_j} z \;,
\end{equation}
with $\theta^0$ defined to be 1. A multi-index $\bm{\alpha} \in \bb{N}_0^\nu$ is said to be \emph{decomposed} into $s$ \emph{parts} $\bm{p}_1, \dots, \bm{p}_s \in \bb{N}_0^\nu$ with (respective) \emph{multiplicities} $\bm{m}_1, \dots, \bm{m}_s\in \bb{N}_0^\mu$ if the \emph{decomposition equation}
\begin{equation}		\label{eq:mFDB-decomposition-equation}
	\bm{\alpha} = |\bm{m}_1| \bm{p}_1 + |\bm{m}_2| \bm{p}_2 + \dots + |\bm{m}_s| \bm{p}_s 
\end{equation}
holds, and the parts are distinct. Note that the parts $\bm{p}$'s and the multiplicities $\bm{m}$'s are multi-indices of lengths $\nu$ and $\mu$, respectively. The \emph{total multiplicity} is defined to be
\begin{equation}		\label{eq:mFDB-total-multiplicity-is-sum-of-multiplicities}
	\bm{m} := \bm{m}_1 + \bm{m}_2 + \dots + \bm{m}_s \;.
\end{equation}
The list $\left( s, \bm{p}, \bm{m} \right)$ is a \emph{$\mu$-decomposition} of $\alpha$, or simply a decomposition. One way to ensure that its parts are distinct is by requiring $0 \ll \bm{p}_1 \ll \bm{p}_2 \ll \dots \ll \bm{p}_s$, \cite{ma2009higher}, where $\ll$ is the total order defined by $\bm{\alpha} \ll \bm{\beta}$ if there is $j \leq \nu$ such that $\alpha_1 = \beta_1, \dots, \alpha_{j-1} = \beta_{j-1}$, but $\alpha_j < \beta_j$.

\begin{thm}[Multivariate Fa\`a di Bruno's, \cite{ma2009higher}]		\label{thm:mFDB-from-Ma}
	Let $\bm{x}\in \bb{R}^\nu \overset{f}{\longrightarrow} \bm{y} \in \bb{R}^\mu \overset{g}{\longrightarrow} z\in \bb{R}$, with $f, g$ sufficiently smooth functions. Write $\bm{x} = \left( x_1, \dots, x_\nu \right)$ and $\bm{y} = \left(y_1, \dots, y_\mu\right)$ for the variables' coordinates. Then,
	\begin{equation}		\label{eq:mFDB-Ma-concise-form}
		\frac{\partial^{|\bm{\alpha}|} z}{\partial \bm{x}^{\bm{\alpha}}} = 
		\bm{\alpha}! \sum_{(s, \bm{p}, \bm{m}) \in \cal{D}} \frac{\partial^{|\bm{m}|}z}{\partial \bm{y}^{\bm{m}}} \;
		\prod_{k=1}^s \frac{1}{\bm{m}_k!} \left[ \frac{1}{\bm{p}_k!} \frac{\partial^{|\bm{p}_k|} \bm{y}}{\partial \bm{x}^{\bm{p}_k}} \right]^{\bm{m}_k}
	\end{equation}
	where $\cal{D}$ is the set of all $\mu$-decompositions of $\bm{\alpha}$. 
	Explicitly, the summands to the right are given by
	\begin{equation}		\label{eq:mFDB-Ma-explicit-expansion}
		\bm{\alpha}! \; \frac{\partial^{r_1 + \dots + r_\mu} z }{\partial y_1^{r_1} \cdots \partial y_\mu^{r_\mu}}
		\prod_{k=1}^{s} \prod_{i=1}^{\mu} \frac{1}{m_{k,i}!} \left[
		\frac{1}{p_{k,1}! \cdots p_{k,\nu}!}
		\frac{\partial^{p_{k,1} + \cdots + p_{k,\nu}} y_i }{ \partial x_1^{p_{k,1}} \cdots \partial x_\nu^{p_{k,\nu}} }
		\right]^{m_{k,i}}
	\end{equation}
	where $\bm{p}_k = \left( p_{k,1}, \cdots, p_{k,\nu} \right), \bm{m}_k = \left( m_{k,1}, \cdots, m_{k,\mu} \right), r_i := m_{1,i} + \cdots + m_{s,i}$, and hence $\bm{m} = \left( r_1, \cdots, r_\mu \right)$.
\end{thm}

The multivariate Fa\`a di Bruno's formula \cite{ma2009higher} at \eqref{eq:mFDB-Ma-concise-form} simplifies when the outer composite $g$ in $g\circ f$ is linear. The first term $\tfrac{\partial^{|\bm{m}|}z}{\partial \bm{y}^{\bm{m}}}$ at \eqref{eq:mFDB-Ma-concise-form} then vanishes, except for decompositions of a total multiplicity $|\bm{m}| = 1$. From definition \eqref{eq:mFDB-total-multiplicity-is-sum-of-multiplicities}, there is only one part of non-zero multiplicity in such decompositions. They may be taken with $s = 1$, $\bm{p}_1 = \bm{\alpha}$ and total multiplicity $\bm{m} = \bm{m}_1$ (compare to \eqref{eq:mFDB-Ma-concise-form} and the decomposition equation \eqref{eq:mFDB-decomposition-equation}). That is, such $\mu$-decompositions are determined by picking a unit vector $\bm{m}_1 = \bm{e}_j\in \bb{N}_0^\mu$, for $j = 1, \dots, \mu$. We obtain,

\begin{cor}		\label{cor:mFDB-from-Ma-when-outer-func-is-linear}
	Under the conditions of Theorem \ref{thm:mFDB-from-Ma}, suppose that $g:\bb{R}^\mu\to \bb{R}$ is linear.
	Then,
	\begin{equation}		\label{eq:mFDB-Ma-concise-form-with-outer-func-is-linear}
		\frac{\partial^{|\bm{\alpha}|} z}{\partial \bm{x}^{\bm{\alpha}}} = 
		\sum_{j = 1}^\mu \frac{\partial z}{\partial y_j} \cdot \frac{\partial^{|\bm{\alpha}|} y_j}{\partial \bm{x}^{\bm{\alpha}}}
	\end{equation}
\end{cor}

\medskip 
\section{High order derivatives of the Blahut-Arimoto operator in cluster-marginal coordinates}
\label{sec:high-order-derivs-of-BA-in-marginal-coords}

In this section we outline the proof of the closed-form formulae for the higher derivative tensors of $Id - BA_\beta$ \eqref{eq:RD-operator-def} at Theorem \ref{thm:high-order-derivs-of-BA-in-main-text} (in Section \ref{part:how-and-what}.\ref{sub:high-order-deriv-tensors-of-BA}). 
See Section \ref{part:proofs}.\ref{sec:appendix-derivations-of-BA-high-order-deriv-in-marginal-coords} for details. 
On top of intermediate results, Theorem \ref{thm:beta-ODE-in-marginal-coords} below (Section \ref{sub:encoders-beta-derivatives}) provides the specialization to RD of the implicit differential equation $D_{\bm{x}} F \; \tfrac{d\bm{x}}{d\beta} = -D_\beta F$ \eqref{eq:ODE-implicit-form}.
The efficient computations of the derivatives tensors of $Id - BA_\beta$ is discussed in Section \ref{sub:computing-high-order-derivatives-efficiently}. 
The results here are built on top of Section \ref{sec:multivariate-faa-di-brunos-formula}, primarily on the multivariate Fa\`a di Bruno's formula at \cite{ma2009higher} (Theorem \ref{thm:mFDB-from-Ma}).

\medskip
Recall the notations of Section \ref{part:how-and-what}.\ref{sub:high-order-deriv-tensors-of-BA}.
Let $p_X$ and $d(x, \hat{x})$ define an RD problem. Write $N := |\mathcal{X}|$ and $M := |\hat{\mathcal{X}}|$ for the source and reproduction alphabet sizes, respectively. The $BA_\beta$ operator in marginal coordinates \eqref{eq:BA-operator-def} is defined there as the composition of
\begin{align}
	\intermediateencoder{}{} &:= \frac{\inputmarginal{} \; e^{-\beta d(x, \hat{x})}}{Z(x, \beta)}
	\quad \text{and}	\tag{\ref{eq:encoder-eq}}	\\
	\outputmarginal{} &:= \sum_x p_X(x) \intermediateencoder{}{} \;,		\tag{\ref{eq:marginal-eq}}
\end{align}
in that order, where $Z(x, \beta) := \sum_{\hat{x}'} \inputmarginal{'} e^{-\beta d(x, \hat{x}')}$. 
We write $\outputmarginalVect$ rather than $\inputmarginalVect$ at the cluster-marginal equation \eqref{eq:marginal-eq} to better distinguish input from output marginals.
While our eventual goal is to track RD solutions, calculations in this section do \textit{not} assume $\inputmarginalVect$ to be a fixed point of $BA_\beta$, $\outputmarginalVect := BA_\beta[\inputmarginalVect] = \inputmarginalVect$, unless stated otherwise.

The encoder equation \eqref{eq:encoder-eq} can be considered as $M\times N$ real functions $\intermediateencoder{}{}$ in the $1 + M$ variables $\beta$ and $\inputmarginal{}$, while the cluster marginal equation \eqref{eq:marginal-eq} are $M$ functions $\outputmarginal{}$ in the $M\times N$ variables $\intermediateencoder{}{}$. That is, we view a Blahut-Arimoto iteration \eqref{eq:BA-operator-def} as the composition
\begin{equation}		\label{eq:BA-as-func-composition}
	\xymatrix@C=3pc{
		\left(\beta, \inputmarginalVect \right) \in \bb{R}^{1+M} \ar[r]^(.51){\eqref{eq:encoder-eq}}	&
		\intermediateencoderVect \in \bb{R}^{M\times N} \ar[r]^(.58){\eqref{eq:marginal-eq}}		&
		\outputmarginalVect \in \bb{R}^M
	} .
\end{equation}
Of these two steps, the marginal equation \eqref{eq:marginal-eq} is linear. That renders the derivatives of $\outputmarginalVect$ with respect to $\intermediateencoderVect$ rather straightforward, in Section \ref{sub:marginals-derivative}. As the encoder equation \eqref{eq:encoder-eq} is more complicated, we untangle its high-order derivatives gradually. In Section \ref{sub:encoders-marginal-derivatives} we tackle the repeated derivatives of $\intermediateencoderVect$ with respect to $\inputmarginalVect$ alone, and in Section \ref{sub:encoders-beta-derivatives} with respect to $\beta$ alone. Only in \ref{sub:encoders-mixed-derivatives} do we combine derivatives of both types, yielding formula \eqref{eq:mixed-BA-deriv-in-thm} for mixed derivatives of arbitrary order (Theorem \ref{thm:high-order-derivs-of-BA-in-main-text} in Section \ref{part:how-and-what}.\ref{sub:high-order-deriv-tensors-of-BA}).

\subsection{Deriving the marginal equation \eqref{eq:marginal-eq}}

\label{sub:marginals-derivative}

Denote high-order derivatives using multi-index notation, as in the multivariate Fa\`a di Bruno's formula (in Section \ref{sec:multivariate-faa-di-brunos-formula}). 
A derivative of $BA_\beta$ \eqref{eq:BA-operator-def} with respect to its $1+M$ variables $(\beta, \inputmarginalVect)$ is denoted by $\bm{\alpha} \in \bb{N}_0^{1+M}$; cf., the dependencies depicted in Equation \eqref{eq:BA-as-func-composition}. 
Considering $\beta$ as the zeroth coordinate, $\alpha_0$ stands for the number of differentiations with respect to $\beta$, and $\alpha_j$ for that with respect to $\inputmarginal{_j}$, where $1 \leq j \leq M$.
We write $\bm{\alpha}_+$ for $\bm{\alpha}$ with its zeroth coordinate removed, $\left(\alpha_0, \bm{\alpha}_+\right) := \left(\alpha_0, \alpha_1, \dots, \alpha_M\right)$. 

Next, exploiting the linearity of the marginal's equation \eqref{eq:marginal-eq}, we apply Corollary \ref{cor:mFDB-from-Ma-when-outer-func-is-linear} to the multivariate Fa\`a di Bruno's formula (of Section \ref{sec:multivariate-faa-di-brunos-formula}). 
As the $\hat{x}'$ output coordinate of $BA_\beta$ is a real function, we have
\begin{equation}		\label{eq:high-order-BA-deriv-in-mFDB-terms-implicit}
	\frac{\partial^{|\bm{\alpha}|} BA_\beta\left[\inputmarginalVect\right](\hat{x}')}{\partial \beta^{\alpha_0} \partial \inputmarginalVect^{\bm{\alpha}+}} \overset{\eqref{eq:mFDB-Ma-concise-form-with-outer-func-is-linear}}{=}
	\sum_{x, \hat{x}} \frac{\partial \outputmarginal{'}}{\partial \intermediateencoder{}{}} 
	\frac{\partial^{|\bm{\alpha}|} \intermediateencoder{}{}}{\partial \beta^{\alpha_0} \partial \inputmarginalVect^{\bm{\alpha}_+}} \;.
\end{equation}
For the marginal equation's \eqref{eq:marginal-eq} first order derivative,
\begin{equation}		\label{eq:first-order-deriv-of-cluster-marginal-eq}
	\frac{\partial}{\partial \intermediateencoder{}{}} \outputmarginal{'} = 
	\sum_{x'} p_X(x') \frac{\partial}{\partial \intermediateencoder{}{}} \intermediateencoder{'}{'} =
	\delta_{\hat{x}, \hat{x}'} \cdot p_X(x)
\end{equation}
As expected, this is constant in the variables $\intermediateencoder{}{}$, and so all the higher derivatives of the output marginal in \eqref{eq:marginal-eq} vanish. Thus,
\begin{multline}		\label{eq:high-order-BA-deriv-only-enc-deriv-implicit}
	\frac{\partial^{|\bm{\alpha}|} BA_\beta\left[\inputmarginalVect\right](\hat{x}')}{\partial \beta^{\alpha_0} \partial \inputmarginalVect^{\bm{\alpha}+}}  \overset{\eqref{eq:high-order-BA-deriv-in-mFDB-terms-implicit}}{=}
	\sum_{x, \hat{x}} \frac{\partial \outputmarginal{'}}{\partial \intermediateencoder{}{}} 
	\frac{\partial^{|\bm{\alpha}|} \intermediateencoder{}{}}{\partial \beta^{\alpha_0} \partial \inputmarginalVect^{\bm{\alpha}_+}}
	\\ \overset{\eqref{eq:first-order-deriv-of-cluster-marginal-eq}}{=} 
	\sum_{x, \hat{x}} \delta_{\hat{x}, \hat{x}'} \cdot p_X(x) \;
	\frac{\partial^{|\bm{\alpha}|} \intermediateencoder{}{}}{\partial \beta^{\alpha_0} \partial \inputmarginalVect^{\bm{\alpha}_+} } =
	\sum_{x} p_X(x) \; \frac{\partial^{|\bm{\alpha}|} \intermediateencoder{'}{}}{\partial \beta^{\alpha_0} \partial \inputmarginalVect^{\bm{\alpha}_+} }
\end{multline}
To complete the calculation, we calculate the encoder's derivatives in the following subsections.

\subsection{Encoder's \eqref{eq:encoder-eq} derivatives with respect to the marginal}

\label{sub:encoders-marginal-derivatives}

Repeatedly deriving the encoder \eqref{eq:encoder-eq} with respect to the coordinates of the input marginal $\inputmarginalVect$ yields the following; it is proved in Section \ref{part:proofs}.\ref{sub:proof-of-prop:repeated-encoder-deriv-wrt-coords} by induction on the order $|\bm{\alpha}_+|$ of differentiation.

\begin{prop}			\label{prop:repeated-encoder-deriv-wrt-coords}
	Let $\beta > 0$, $\inputmarginalVect\in \bb{R}^M$ a distribution, and $\bm{\alpha}$ a multi-index as above. 
	Let $\intermediateencoderVect$ be defined in terms of $\inputmarginalVect$ by the encoder Equation \eqref{eq:encoder-eq}.
	Its derivative of order $\bm{0} \neq \bm{\alpha}_+ \in \bb{N}_0^{M}$ with respect to $\inputmarginalVect$ is,
	\begin{equation}		\label{eq:repeated-encoder-deriv-wrt-coords-in-prop}
		\frac{\partial^{|\bm{\alpha}_+|}}{\partial \inputmarginalVect^{\bm{\alpha}_+}} \intermediateencoder{'}{}\Big\rvert_{\inputmarginalVect} =
		\frac{(-1)^{|\bm{\alpha}_+|-1} (|\bm{\alpha}_+|-1)! \; e^{-\beta \left<\bm{\alpha}_+, d(x, \hat{x})\right> }  }{Z^{|\bm{\alpha}_+|}(x, \beta)}\cdot \Big[
		\left< \bm{\alpha}_+, \bm{e}_{\hat{x}'} \right> 
		- |\bm{\alpha}_+| \cdot\intermediateencoder{'}{}
		\Big]
	\end{equation}
	where $\left<\cdot, \cdot\right>$ is the usual scalar product on $\bb{R}^M$, 
	$\bm{e}_{\hat{x}'}$ is the standard basis vector in $\bb{R}^M$ at the $\hat{x}'$ entry,
	and $d(x, \hat{x})$ is considered as an $\hat{x}$-indexed vector for $x$ fixed.
\end{prop}

To clarify \eqref{eq:repeated-encoder-deriv-wrt-coords-in-prop}, when $x\in \cal{X}$ is fixed, then $d(x, \hat{x})$ is merely a vector in $\bb{R}^M$. As $\bm{\alpha}_+$ is in $\bb{N}_0^M \subset \bb{R}_0^M$, there is sense in taking the scalar product $\left<\bm{\alpha}_+, d(x, \hat{x})\right>$. For example, when $\bm{\alpha}_+$ is the $j$-th standard basis vector $\bm{e}_j$, then
\begin{equation}		\label{eq:scalar-product-example-with-multi-index}
	\left<\bm{e}_j, d(x, \hat{x})\right> =
	d(x, \hat{x}_j) \;.
\end{equation}
The first-order derivatives $|\bm{\alpha}_+| = 1$ are an important special case, 

\begin{cor}		\label{cor:BA-jacobian}
	Outside the simplex boundary, $\forall \hat{x}\; \inputmarginal{} > 0$, the Jacobian of $Id - BA_\beta$ \eqref{eq:RD-operator-def} in marginal coordinates when evaluated at $\inputmarginalVect$ is given by
	\begin{equation}		\label{eq:BA-jacobian}
		\Big(D_{\inputmarginalVect} \left( Id - BA_\beta \right)\big\rvert_{\inputmarginalVect}\Big)_{\hat{x}_i, \hat{x}_j } =
		\sum_{x} p_X(x) \frac{\intermediateencoder{_j}{} \intermediateencoder{_i}{} }{\inputmarginal{_j}} 
		+ \Delta[\inputmarginalVect]_{\hat{x}_i, \hat{x}_j }
	\end{equation}
	where $\intermediateencoderVect$ is defined in terms of $\inputmarginalVect$ by the encoder Equation \eqref{eq:encoder-eq},
	and $\Delta[\inputmarginalVect] := \diag \left(\frac{\inputmarginal{_j} - BA_\beta[\inputmarginalVect](\hat{x}_j)}{\inputmarginal{_j}}\right)_j$.
\end{cor}

The term $\Delta[\inputmarginalVect]$ vanishes precisely at fixed points of $BA_\beta$, $\inputmarginalVect = BA_\beta[\inputmarginalVect]$. It can be considered as a perturbation due to evaluating the Jacobian outside of fixed points. 
We shall need the general form \eqref{eq:BA-jacobian} for the proof of the convergence guarantees Theorem \ref{thm:taylor-method-converges-for-RD-root-tracking-away-of-bifurcation} of Taylor method for RD.
Otherwise, when evaluated at fixed points, this result agrees with \cite{agmon2021critical} ($A^\intercal$ at Equation (7) there). 
While the Jacobian can also be calculated at the simplex boundary, that shall not be useful to us. 
Rather than calculating \eqref{eq:BA-jacobian} directly, we prove it using Proposition \ref{prop:repeated-encoder-deriv-wrt-coords} to illustrate the multivariate notation.

\begin{proof}[Proof of Corollary \ref{cor:BA-jacobian}]
	A first order derivative $\nicefrac{\partial}{\partial\inputmarginal{_j}}$ is represented by the $j$-th standard basis vector, $\bm{\alpha}_+ = \bm{e}_j$. Differentiating its $(\hat{x}_i, x)$ entry,
	\begin{equation}		\label{eq:first-order-encoder-deriv-wrt-marginal-explicit}
		\frac{\partial^{|\bm{\alpha}_+|}}{\partial \inputmarginalVect^{\bm{\alpha}_+}} \intermediateencoder{_i}{} =
		\frac{\partial}{\partial\inputmarginal{_j}} \intermediateencoder{_i}{} 
		\underset{\eqref{eq:scalar-product-example-with-multi-index}}{\overset{\eqref{eq:repeated-encoder-deriv-wrt-coords-in-prop}}{=}} 
		\frac{ e^{-\beta d(x, \hat{x}_j) }  }{Z(x, \beta)}\cdot \Big[
		\left< \bm{e}_j, \bm{e}_i \right> 
		- \intermediateencoder{_i}{}
		\Big] 
		\overset{\eqref{eq:encoder-eq}}{=}
		\frac{\intermediateencoder{_j}{}}{\inputmarginal{_j}} \left[ \delta_{i, j} - \intermediateencoder{_i}{} \right]
	\end{equation}
	Setting $\alpha_0 = 0$ and plugging this back into formula \eqref{eq:high-order-BA-deriv-only-enc-deriv-implicit} for the derivative of $BA_\beta$ yields \eqref{eq:BA-jacobian},
	\begin{multline}
		\frac{\partial BA_\beta\left[\inputmarginalVect\right](\hat{x}_i)}{\partial r(\hat{x}_j)}
		\overset{\eqref{eq:high-order-BA-deriv-only-enc-deriv-implicit}}{=}
		\sum_{x} p_X(x) \; \frac{\partial \intermediateencoder{_i}{}}{\partial \inputmarginal{_j} }
		\overset{\eqref{eq:first-order-encoder-deriv-wrt-marginal-explicit}}{=}
		\sum_{x} p_X(x) \frac{\intermediateencoder{_j}{}}{\inputmarginal{_j}} \left[ \delta_{i, j} - \intermediateencoder{_i}{} \right] \\ \overset{\eqref{eq:marginal-eq}}{=}
		\delta_{i, j} \cdot \frac{BA_\beta[\inputmarginalVect](\hat{x}_j)}{\inputmarginal{_j}} - \sum_{x} p_X(x) \frac{\intermediateencoder{_j}{} \intermediateencoder{_i}{} }{\inputmarginal{_j}} 
	\end{multline}
	Where at the last equality, $BA_\beta[\inputmarginalVect](\hat{x}_j)$ is the $j$-th output coordinate $\outputmarginal{_j}$ \eqref{eq:marginal-eq} of a Blahut-Arimoto iteration calculated at $\inputmarginalVect$.
\end{proof}

At a fixed point of $BA_\beta$, the Jacobian's properties relevant to us are the following, due to Benger.

\begin{thm}[\cite{agmon2021critical}]		\label{thm:properties-of-BA-jacobian-from-ISIT-CSD-paper}
	Let an RD problem be defined by $p_X$ and a finite non-degenerate distortion $d(x, \hat{x})$ (defined in Section \ref{part:how-and-what}.\ref{sub:high-order-deriv-tensors-of-BA}). Let $\inputmarginalVect$ be a fixed point of $BA_\beta$ outside the simplex boundary, $\forall \hat{x} \; \inputmarginal{} > 0$. 
	Then, $D_{\inputmarginalVect} \left( Id - BA_\beta \right)\rvert_{\inputmarginalVect}$ \eqref{eq:BA-jacobian} is non-singular, diagonalizable, and with real non-negative eigenvalues. 
\end{thm}

\subsection{Encoder's \eqref{eq:encoder-eq} partial derivative with respect to $\beta$}

\label{sub:encoders-beta-derivatives}

For derivatives with respect to $\beta$ we need the polynomials $P_k$ \eqref{eq:P_0-def}-\eqref{eq:P_k-inductive-def} defined in Section \ref{part:how-and-what}.\ref{sub:high-order-deriv-tensors-of-BA}. 
Recall the derivation $\dbar$ defined there on the infinite polynomial ring $\bb{R}\left[x_0, x_1, \dots\right]$,
\begin{equation}		\tag{\ref{eq:variable-deriv-def-for-recursive-beta-formula}}
	\dbar x_0 := 0, \quad \text{and} \quad
	\dbar x_k := x_1 \cdot x_k - x_{k+1} \quad \text{for } k > 0 \;.
\end{equation}
To calculate $\dbar$ on a monomial, apply the Leibniz rule $\dbar x_i^j = jx_i^{j-1} \dbar x_i$, for $j > 0$, 
\begin{align}		\label{eq:differentiating-a-monomial-in-P_k}
	\dbar &\left(x_0^{i_0} x_1^{i_1} x_2^{i_2} \cdots x_k^{i_k} \right) \nonumber \\ &\;\;=	
	\left(i_0 x_0^{i_0 - 1} \dbar x_0\right) x_1^{i_1} x_2^{i_2} \cdots x_k^{i_k} +
	x_0^{i_0} \left(i_1 x_1^{i_1 - 1} \dbar x_1\right) x_2^{i_2} \cdots x_k^{i_k} +
	\cdots +
	x_0^{i_0} x_1^{i_1} x_2^{i_2} \cdots \left( i_k x_k^{i_k - 1} \dbar x_k \right) \nonumber
	\\ &\overset{\eqref{eq:variable-deriv-def-for-recursive-beta-formula}}{=}
	0 + i_1 x_0^{i_0} x_1^{i_1 - 1} \left( x_1^2 - x_2 \right) x_2^{i_2} \cdots x_k^{i_k} +
	\cdots +
	i_k x_0^{i_0} x_1^{i_1} x_2^{i_2} \cdots x_k^{i_k - 1} \left( x_1 x_k - x_{k+1} \right)
\end{align}
where $i_0, i_1, \dots, i_k \geq 0$, and the $j$-th summand is understood to vanish if $i_j = 0$. The first equality follows from the usual rules of differentiation, the second from the definition \eqref{eq:variable-deriv-def-for-recursive-beta-formula} of $\dbar$.
With this, one can calculate $\dbar P_k$ for a polynomial $P_k$, by linearity of $\dbar$. Starting at $P_0(x_0) = 1$ \eqref{eq:P_0-def}, use the inductive definition \eqref{eq:P_k-inductive-def} to calculate $P_{k+1}$ from $P_k$.
e.g., $P_1, P_2$ and $P_3$ are given at \eqref{eq:P_1_derived}-\eqref{eq:P_3_derived}.

The polynomials $P_k$ are defined by simple algebraic formulae \eqref{eq:P_0-def}-\eqref{eq:P_k-inductive-def} which can be calculated easily, as in \eqref{eq:differentiating-a-monomial-in-P_k}. They encapsulate the algebra involved in the encoder's high-order partial derivatives with respect to $\beta$.
Plugging the distortion $d(x, \hat{x})$ and the expectations $\expectedDxWRTencoderK{k}$ \eqref{eq:expected-k-th-power-distortion-def} in place of the variables $x_0, x_1, \dots, x_k$ of $P_k$, one can calculate the $M$-by-$N$ matrices $P_k[\bm{q}; d](\hat{x}, x)$ \eqref{eq:P_k-by-abuse-of-notation}, also denoted $P_k(\hat{x}, x)$ for short. 
This yields the following formula for the encoder's repeated partial $\beta$-derivatives --- see Section \ref{part:proofs}.\ref{sub:proof-of-prop:repeated-beta-derivatives-of-encoder} for proof.

\begin{prop}		\label{prop:repeated-beta-derivatives-of-encoder}
	Let $\beta > 0$, $\inputmarginalVect\in \bb{R}^M$ a distribution, and $\intermediateencoderVect$ defined in terms of $\inputmarginalVect$ by the encoder Equation \eqref{eq:encoder-eq}.
	Then for $k > 0$,
	\begin{equation}		\label{eq:formula-for-repeated-beta-deriv-in-prop-recursive-form}
		\partialbetaK{\intermediateencoder{}{}}{k}\Big\rvert_{\inputmarginalVect} = 
		\intermediateencoder{}{} \cdot P_k(\hat{x}, x)
	\end{equation}
\end{prop}

Plugging this into the expansion \eqref{eq:high-order-BA-deriv-only-enc-deriv-implicit} of $\tfrac{\partial^{|\bm{\alpha}|} }{\partial \beta^{\alpha_0} \partial \inputmarginalVect^{\bm{\alpha}+}} BA_\beta\left[\inputmarginalVect\right](\hat{x}')$ (in Section \ref{sub:marginals-derivative}) immediately yields formula \eqref{eq:repeated-beta-deriv-in-thm} of Theorem \ref{thm:high-order-derivs-of-BA-in-main-text} (Section \ref{part:how-and-what}.\ref{sub:high-order-deriv-tensors-of-BA}) for the repeated $\beta$-derivative of $Id - BA_\beta$ \eqref{eq:RD-operator-def}.
With $P_1$ written explicitly, the first order can be written explicitly as follows (see Section \ref{part:proofs}.\ref{sub:proof-of-cor:beta-deriv-of-BA-op} for proof),

\begin{cor}		\label{cor:BA-beta-deriv}
	With the notation of Proposition \ref{prop:repeated-beta-derivatives-of-encoder},
	\begin{equation}
		\partialbeta{} \left( Id - BA_\beta \right)\Big\rvert_{\inputmarginalVect}(\hat{x}) =
		\bb{E}_{p_X(x)} \left[\intermediateencoder{}{} d(x, \hat{x}) \right]
		- \bb{E}_{\intermediateencoder{'}{}p_X(x)} \left[\intermediateencoder{}{} d(x, \hat{x}') \right]
	\end{equation}
\end{cor}

Now that Corollaries \ref{cor:BA-jacobian} and \ref{cor:BA-beta-deriv} provide all the first-order derivative tensors, we can specialize the implicit differential equation $D_{\bm{x}} F \; \tfrac{d\bm{x}}{d\beta} = -D_\beta F$ \eqref{eq:ODE-implicit-form} to RD.

\begin{thm}[RD ODE]		\label{thm:beta-ODE-in-marginal-coords}
	Let $\inputmarginalVect$ be a fixed point of $BA_\beta$ outside the simplex boundary, $\forall \hat{x}\; \inputmarginal{} > 0$, for which Assumptions \ref{assumption:operator-root-is-a-function-of-beta} and \ref{assumption:solution-is-smooth-in-beta} in Section \ref{part:how-and-what}.\ref{sub:beta-derivs-at-an-operator-root} hold. Then
	\begin{equation}		\label{eq:beta-ODE-in-marginal-coord}
		\sum_{\hat{x}'} A_{\hat{x}, \hat{x}'} \dbeta{\inputmarginal{'}} = 
		\bb{E}_{\intermediateencoder{'}{} p_X(x)} \left[\intermediateencoder{}{} d(x, \hat{x}') \right]
		- \bb{E}_{p_X(x)} \left[\intermediateencoder{}{} d(x, \hat{x}) \right]
	\end{equation}
	where $\intermediateencoderVect$ is defined in terms of $\inputmarginalVect$ by the encoder Equation \eqref{eq:encoder-eq}, and
	\begin{equation}		\label{eq:A-matrix-def}
		A_{\hat{x}, \hat{x}'} := \sum_x p_X(x) \frac{\intermediateencoder{'}{} \intermediateencoder{}{} }{\inputmarginal{'}} \;.
	\end{equation}
\end{thm}

\subsection{Encoder's \eqref{eq:encoder-eq} mixed derivatives}

\label{sub:encoders-mixed-derivatives}

The previous Sections \ref{sub:encoders-marginal-derivatives} and \ref{sub:encoders-beta-derivatives} enabled us to calculate derivatives with respect to the coordinates $\inputmarginalVect$ alone or with respect to $\beta$ alone. 
cf., the dependencies Equation \eqref{eq:BA-as-func-composition}. 
To derive simultaneously with respect to both, recall $G\big( k, a; \intermediateencoderVect, d \big)$ \eqref{eq:combinatorial-G-in-terms-of-polynomials} from Section \ref{part:how-and-what}.\ref{sub:high-order-deriv-tensors-of-BA}. It is defined on non-negative integers $k, a$, and its values are $M\times N$ matrices. We set $G = 0$ if $a = 0 < k$, and otherwise
\begin{equation}		\tag{\ref{eq:combinatorial-G-in-terms-of-polynomials}}
	G\big( k, a; \intermediateencoderVect, d \big)_{(\hat{x}, x)} := 
	\sum_{\substack{\bm{t}: \;|\bm{t}|\leq a, \\ \sum_{j} j\cdot t_{j} = k}} \frac{1}{ \bm{t}! \; (a - |\bm{t}|)!} 
	\prod_{j=1}^{k} \left( \frac{P_j(\hat{x}, x)}{j!} \right)^{t_{j}} \;,
\end{equation}
where $\bm{t} \in \bb{N}_0^k$, and $P_j(\hat{x}, x)$ are the matrices defined above (Equation \eqref{eq:P_k-by-abuse-of-notation} in Section \ref{part:how-and-what}.\ref{sub:high-order-deriv-tensors-of-BA}). 
Since $\bm{t}$ represents the multiplicities of an integer partition of $k$, we may take $\bm{t} \in \bb{N}_0^l$ for any convenient $l\geq k$.
With that, 

\begin{prop}		\label{prop:mixed-high-order-enc-deriv}
	Let $d(x, \hat{x})$ and $p_X$ define an RD problem, $\inputmarginalVect \in \Delta[\hat{\mathcal{X}}]$ be an input marginal outside the simplex boundary, $\forall \hat{x} \; \inputmarginal{} > 0$, and $\intermediateencoderVect$ the encoder its defines by the encoder Equation \eqref{eq:encoder-eq}.
	For $\bm{\alpha}\in \bb{N}_0^{M + 1}$ a multi-index with $\bm{\alpha}_+ \neq \bm{0}$,
	\begin{multline}			\label{eq:mixed-high-order-enc-deriv-in-prop}
		\frac{\partial^{|\bm{\alpha}|}}{\partial \beta^{\alpha_0} \; \partial \inputmarginalVect^{\bm{\alpha}_+}} \intermediateencoder{'}{}\Big\rvert_{\inputmarginalVect} =
		(-1)^{|\bm{\alpha}_+|-1} (|\bm{\alpha}_+|-1)! \left(\frac{\intermediateencoder{}{}}{\inputmarginal{}}\right)^{\bm{\alpha}_+ } \bm{\alpha}! 
		\sum_{\bm{k}: \; |\bm{k}| = \alpha_0} \left(\prod_{i\neq\hat{x}'} G\big( k_i, \alpha_i; \intermediateencoderVect, d \big)_{(\hat{x}_i, x)} \right)
		\\ \cdot \left[
		\alpha_{\hat{x}'} \cdot G\big( k_{\hat{x}'}, \alpha_{\hat{x}'}; \intermediateencoderVect, d \big)_{(\hat{x}', x)} 
		- |\bm{\alpha}_+| \cdot \left( 1 + \alpha_{\hat{x}'} \right) \cdot \intermediateencoder{'}{} \cdot 
		G\big( k_{\hat{x}'}, 1 + \alpha_{\hat{x}'}; \intermediateencoderVect, d \big)_{(\hat{x}', x)}
		\right] 
	\end{multline}
	where $\bm{k} \in \bb{N}_0^M$, $i = 1, \dots, M$, and $G\big(k, a; \intermediateencoderVect, d \big)$ is defined by \eqref{eq:combinatorial-G-in-terms-of-polynomials}. \newline
\end{prop}
This Proposition is proven in Section \ref{part:proofs}.\ref{sub:proof-of-prop:mixed-high-order-enc-deriv} by differentiating the encoder's $\inputmarginalVect$-derivative \eqref{eq:repeated-encoder-deriv-wrt-coords-in-prop} with respect to $\beta$, followed by an application of formula \eqref{eq:formula-for-repeated-beta-deriv-in-prop-recursive-form} for the encoder's repeated partial $\beta$-derivative, using the tools of Section \ref{sec:multivariate-faa-di-brunos-formula}.
Plugging \eqref{eq:mixed-high-order-enc-deriv-in-prop} back into the expansion \eqref{eq:high-order-BA-deriv-only-enc-deriv-implicit} of $\tfrac{\partial^{|\bm{\alpha}|}}{\partial \beta^{\alpha_0} \partial \inputmarginalVect^{\bm{\alpha}+}} BA_\beta$, we finally obtain formula \eqref{eq:mixed-BA-deriv-in-thm} of Theorem \ref{thm:high-order-derivs-of-BA-in-main-text} for the mixed derivative of $Id - BA_\beta$ \eqref{eq:RD-operator-def}, stated as Corollary \ref{cor:mixed-deriv-for-RD-operator} below.
Together with Proposition \ref{prop:repeated-beta-derivatives-of-encoder} from the previous Subsection, this concludes the proof of Theorem \ref{thm:high-order-derivs-of-BA-in-main-text} in Section \ref{part:how-and-what}.\ref{sub:high-order-deriv-tensors-of-BA}.

\begin{cor}			\label{cor:mixed-deriv-for-RD-operator}
	Under the conditions of Proposition \ref{prop:mixed-high-order-enc-deriv},
	\begin{multline}
		\frac{\partial^{|\bm{\alpha}|} }{\partial \beta^{\alpha_0} \partial \inputmarginalVect^{\bm{\alpha}+}} \left(Id - BA_\beta \right)\left[\inputmarginalVect\right](\hat{x}') \\ =
		\delta_{\bm{\alpha}, \bm{e}_{\hat{x}'}} - 
		(-1)^{|\bm{\alpha}_+|-1} (|\bm{\alpha}_+|-1)! \; \bm{\alpha}! \sum_x p_X(x) \left(\frac{\intermediateencoder{}{}}{\inputmarginal{}}\right)^{\bm{\alpha}_+ } 
		\sum_{\bm{k}: \; |\bm{k}| = \alpha_0} \left(\prod_{i\neq\hat{x}'} G\big( k_i, \alpha_i; \intermediateencoderVect, d \big)_{(\hat{x}_i, x)} \right)
		\\ \cdot \left[
		\alpha_{\hat{x}'} \cdot G\big( k_{\hat{x}'}, \alpha_{\hat{x}'}; \intermediateencoderVect, d \big)_{(\hat{x}', x)} 
		- |\bm{\alpha}_+| \cdot \left( 1 + \alpha_{\hat{x}'} \right) \cdot \intermediateencoder{'}{} \cdot 
		G\big( k_{\hat{x}'}, 1 + \alpha_{\hat{x}'}; \intermediateencoderVect, d \big)_{(\hat{x}', x)}
		\right] 
	\end{multline}
\end{cor}

As a sanity check, when differentiating only with respect to the marginal $\inputmarginalVect$, $\alpha_0 = 0$, then Proposition \ref{prop:mixed-high-order-enc-deriv} reduces to formula \eqref{eq:repeated-encoder-deriv-wrt-coords-in-prop} for $\tfrac{\partial^{|\bm{\alpha}_+|}}{\partial \inputmarginalVect^{\bm{\alpha}_+}} \intermediateencoder{}{}$, as one might expect. Indeed, $\bm{k} = \bm{0}$ is then the only vector in $\bb{N}_0^M$ satisfying $|\bm{k}| = \alpha_0 = 0$. 
The sum $G(0, a)$ at \eqref{eq:combinatorial-G-in-terms-of-polynomials} then reduces to $\nicefrac{1}{a!}$, and so the summation at \eqref{eq:mixed-high-order-enc-deriv-in-prop} over $\bm{k}$ simplifies to $\frac{1}{\bm{\alpha}!} \cdot \big[	\alpha_{\hat{x}'} - |\bm{\alpha}_+| \cdot \intermediateencoder{'}{} \big]$.
Plugging this back into \eqref{eq:mixed-high-order-enc-deriv-in-prop} yields formula \eqref{eq:repeated-encoder-deriv-wrt-coords-in-prop}, as expected.
The other special case when differentiating with respect to $\beta$ alone cannot be verified in a similar manner, since Proposition \ref{prop:mixed-high-order-enc-deriv} requires $\bm{\alpha}_+ \neq 0$.

\subsection{A note on how to compute high-order derivatives of $Id - BA_\beta$ efficiently}

\label{sub:computing-high-order-derivatives-efficiently}

Some quantities in the formulas for the derivative tensors $D^l_{\beta^b, \bm{x}^{l-b}} (Id - BA_\beta)[\inputmarginalVect]$ (Theorem \ref{thm:high-order-derivs-of-BA-in-main-text} in Section \ref{part:how-and-what}.\ref{sub:high-order-deriv-tensors-of-BA}) are shared, especially when partial derivatives with respect to $\beta$ are present. 
While this was purposefully implied by our presentation, we now discuss how these derivative tensors can computed in an algorithmically efficient manner.
cf., Section \ref{part:proofs}.\ref{sub:complexity-of-deriv-tensors-for-RD} on the complexities of the derivative tensors.

\medskip
The form of $G$ \eqref{eq:combinatorial-G-in-terms-of-polynomials} suggests a multi-step approach to computing the derivatives tensors at Theorem \ref{thm:high-order-derivs-of-BA-in-main-text}.
First, compute the $M$-by-$N$ matrices $P_k(\hat{x}, x)$ \eqref{eq:P_k-by-abuse-of-notation}. While the matrices $P_k(\hat{x}, x)$ depend both on $\intermediateencoderVect$ and $d(x, \hat{x})$, their polynomial form $P_k$ depends on $k$ alone, as it is defined algebraically by  \eqref{eq:P_0-def}-\eqref{eq:P_k-inductive-def}. Thus, the $P_k$ polynomials can be computed once an for all, while the $P_k(\hat{x}, x)$ matrices need to be computed anew for every encoder $\intermediateencoderVect$. Yet, the expectations $\expectedDxWRTencoderK{k}$ \eqref{eq:expected-k-th-power-distortion-def} are shared among $P_k(\hat{x}, x)$'s for distinct $k$ values.
With these computed, the derivative tensors \eqref{eq:repeated-beta-deriv-in-thm} with respect to $\beta$ alone can be computed readily.

Second, compute the $M$-by-$N$ matrices $G\big( k, a\big)$ according to \eqref{eq:combinatorial-G-in-terms-of-polynomials}, when the $P_k(\hat{x}, x)$ matrices are given. When the range of admissible $k$ and $a$ values is known in advance, then it is possible to iterate over the partition vectors $\bm{t}$ at \eqref{eq:combinatorial-G-in-terms-of-polynomials} only once. For, the summand computed there for a particular partition $\bm{t}$ can be added to the matrices $G(k, a)$ at $k = \sum_j j\cdot t_j$ and at all the $a$ values with $a \geq |\bm{t}|$. 

Third, once the $G\big( k, a\big)$ matrices are computed, then one can compute the mixed derivatives of $(Id - BA_\beta)[\inputmarginalVect]$ by \eqref{eq:mixed-BA-deriv-in-thm}.
Each $\bm{\alpha}$-indexed partial derivative there may correspond to multiple tensor entries. 
cf., the comments on indexation, after definition \eqref{eq:mixed-deriv-def-evaluated-applied-to-vectors} in Section \ref{part:how-and-what}.\ref{sub:beta-derivs-at-an-operator-root}.
Thus, partial derivatives can be computed only once per $\bm{\alpha}$ and then distributed to the various tensor entries. Alternatively, each $\bm{\alpha}$-indexed derivative can be re-used at tensor evaluation. 
See also Section \ref{part:proofs}.\ref{subsub:complexity-of-RD-deriv-tensor} on the tensors' complexities.

If one wishes to compute implicit derivatives $\dbetaK{\inputmarginalVect}{l}$ \eqref{eq:l-th-deriv-at-beta0} of all orders $l$ up to $L > 0$, then inspecting \eqref{eq:mixed-BA-deriv-in-thm} shows that it is enough to compute the $P_k(\hat{x}, x)$ matrices for all $k \leq L$, and  $G\big( k, a\big)$ on the integral grid $0 \leq k \leq L, \; 0 \leq a \leq 1 + L$. Once these are computed, then any $\bm{\alpha}$-indexed partial derivative with $|\bm{\alpha}|\leq L$ can be computed.

\medskip
\section{On RD bifurcations and root tracking for RD}
\label{sec:RD-bifurcations-and-root-tracking}

As described in Section \ref{sec:introduction}, the goal of this work is to track the path of an \textit{optimal} solution. In RD context, this means tracking an \textit{achieving distribution}, rather than any fixed point of the Blahut-Arimoto algorithm. 
As shown below, there are typically many fixed points of $BA_\beta$ \eqref{eq:BA-operator-def} which do \textit{not} achieve the rate-distortion curve, and so are \textit{sub-optimal}.
Ensuring that the root being tracked is indeed an optimal one requires an understanding of the solutions' structure or equivalently of RD bifurcations.

\medskip
We start by showing in Section \ref{sub:suboptimal-RD-curves} that an RD problem typically has a plethora of suboptimal solutions, stemming from the various \textit{restrictions} of a given problem (defined in Section \ref{part:how-and-what}.\ref{sub:taylor-method-for-RD-root-tracking}).
Our main case of interest is that of a \textit{cluster vanishing}. Namely, when the marginal probability of a cluster $\hat{x}$ vanishes gradually as $\beta$ varies, resulting in an optimal solution on a smaller support. 
In Section \ref{sub:cluster-vanishing-bifs-are-bifs} we show that two solution branches must then collide and merge into one, so these are indeed bifurcations. 
This type of bifurcations is handled by Algorithms \ref{algo:taylor-method-for-RD-root-tracking} and \ref{algo:root-tracking-for-RD} in Section \ref{part:how-and-what}.\ref{sec:RD-solution-curve-from-RD-beta-derivs}.
Besides bifurcations in which the support shrinks, we are also aware of bifurcations where the optimal solution switches support. An explanation of these is deferred to Section \ref{sub:support-switching-bifurcations}.
Both of the latter are \textit{local} bifurcations, as they can be detected by the relevant Jacobian. 
``There are also bifurcations that cannot be detected by looking at small neighborhoods of fixed points'', \citep{kuznetsov2004elements}, known as \emph{global bifurcations}. 
Such bifurcations could break the continuity of $\bm{x}(\beta)$, violating Assumption \ref{assumption:solution-is-smooth-in-beta}.
It turns out that there are no global bifurcations in rate-distortion (Section \ref{sub:obstructions-to-RT-assumptions-for-RD}). However, considering only the cluster-marginal of a root might cause a support-switching bifurcation to appear as if it is a global bifurcation.
With that, while we do \textit{not} argue to classify all RD bifurcations, we do believe that the tools brought here are a significant step in that direction.

To track an operator's root $\bm{x}$ it is necessary that it can be written as a function $\bm{x}(\beta)$ of $\beta$ (Assumption \ref{assumption:operator-root-is-a-function-of-beta} in \ref{part:how-and-what}.\ref{sub:beta-derivs-at-an-operator-root}) which is smooth (Assumption \ref{assumption:solution-is-smooth-in-beta}), except at bifurcations.
In Section \ref{sub:obstructions-to-RT-assumptions-for-RD} we consider the obstructions to these assumptions in RD, building on the classic results of \cite{berger71}. Our main result is Theorem \ref{thm:achieving-distributions-of-beta-are-convex}, which allows easy detection of non-uniqueness of the achieving distribution, in terms of Blahut-Arimoto's Jacobian with respect to the encoder $\intermediateencoderVect$. 
The subtle differences between the Jacobian in encoder $\intermediateencoderVect$ and in marginal coordinates $\inputmarginalVect$ allows to detect the bifurcation's type in some cases (Equation \eqref{eq:flowchart-for-different-kinds-of-RD-bifurcations}), and extends the argument of \citeauthor{agmon2021critical} on critical slowing down.
This not only allows us to argue in Section \ref{sub:when-does-RTRD-follow-the-optimal-path} that root tracking for RD does indeed track the optimal solution, subject to Assumption \ref{assumption:only-cluster-vanishing-bifurcations}, but also gives us a straightforward tool to detect failures, and so to consider how they might be corrected.


We note that \cite{rose1994mapping} considers RD bifurcations of continuous source alphabets, usually assuming a squared-error distortion measure.
However, it is assumed there \cite[IV.C]{rose1994mapping} that the distortion varies continuously with $\beta$ and that the solution's support grows monotonically with $\beta$. 
As a result, bifurcations of continuous sources are classified there as either ``split'' or ``mass growing'' bifurcations.
However, both of these assumptions need \textit{not} hold for finite RD problems, as seen by the rightmost bifurcation of Figure \ref{fig:Berger_example_2.7.3}. 
cf., Sections \ref{sub:obstructions-to-RT-assumptions-for-RD} and \ref{sub:support-switching-bifurcations}.
In related contexts, \cite{parker2010symmetry} consider bifurcations using the Lyapunov-Schmidt reduction, a general-purpose tool for handling singularities. In contrast, we exploit the structure of RD problems by using the reductions defined in Section \ref{part:how-and-what}.\ref{sub:taylor-method-for-RD-root-tracking}. 
This simplifies the work with RD roots, facilitating the results below.

\subsection{Suboptimal RD curves}
\label{sub:suboptimal-RD-curves}

We proceed with the discussion in Section \ref{part:how-and-what}.\ref{sub:high-order-deriv-tensors-of-BA} around the definition \eqref{eq:BA-operator-def} of $BA_\beta$. The Blahut-Arimoto algorithm converges to a curve-achieving distribution, \cite[Theorem 1]{csiszar1974computation}, yet in a manner which depends on the choice of initial condition $\inputmarginalVect_0$, as hinted by \citeauthor{csiszar1974computation}. When $\inputmarginalVect_0$ is of full support, we have the following.

\begin{thm}[\cite{csiszar1974computation}]		\label{thm:BA-converges-to-RD-curve-for-an-initial-cond-of-full-support}
	Let $p_X$ and $d(x, \hat{x})$ define an RD problem on a finite reproduction alphabet $\hat{\cal{X}}$, and let $\beta > 0$. Let $\inputmarginalVect_0\in \Delta[\hat{\mathcal{X}}]$ be an initial condition of full support, $\inputmarginalSymbol_0(\hat{x}) > 0$ for all $\hat{x}\in \hat{\cal{X}}$. 
	Then, there exists a curve-achieving distribution $\inputmarginalVect^*$ such that $BA_\beta^n[\inputmarginalVect_0] \overset{n \to \infty}{\longrightarrow} \inputmarginalVect^*$.
\end{thm}

Let $\inputmarginalVect_0$ be an initial condition of full support, and denote by $\inputmarginalVect^*$ its limit under $BA_\beta$, an achieving distribution. 
If $\supp \inputmarginalVect^* := \{\hat{x}\in \hat{\mathcal{X}}: \inputmarginalSymbol^*(\hat{x}) > 0\}$ is of size 2 at least, then we can choose a non-empty proper subset $\hat{\cal{X}}'$ of $\supp \inputmarginalVect^*$. 
Next, invoke Theorem \ref{thm:BA-converges-to-RD-curve-for-an-initial-cond-of-full-support} on the RD problem restricted to $\hat{\cal{X}}'$, starting at some initial condition $\inputmarginalVect'_0$ of full-support on $\hat{\cal{X}}'$, $\inputmarginalVect'_0 \in \Delta[\hat{\mathcal{X}}']$. This yields an achieving distribution $\inputmarginalVect^{*\prime}$ which obviously differs from $\inputmarginalVect^*$, as their supports differ.
While $\inputmarginalVect^*$ achieves the rate-distortion curve of the problem we have started with, $\inputmarginalVect^{*\prime}$ achieves that of the restricted problem. 
These two RD curves may differ, as demonstrated by Figure \ref{fig:optimal-and-suboptimal-RD-solution-for-CSD-problem}.
cf., Lemma 1 at \cite[Section 2.5]{berger71}. 

In light of the above, we have the following refinement of \citeauthor{csiszar1974computation}'s Theorem \ref{thm:BA-converges-to-RD-curve-for-an-initial-cond-of-full-support},

\begin{thm}					\label{thm:BA-converges-to-RD-curve-for-an-initial-cond-of-arbitrary-support}
	Under the conditions of Theorem \ref{thm:BA-converges-to-RD-curve-for-an-initial-cond-of-full-support},
	let $\inputmarginalVect_0\in \Delta[\hat{\mathcal{X}}]$ be an initial condition, not necessarily of full support. 
	Then, there exists a distribution $\inputmarginalVect^*$ which achieves the rate-distortion curve of reduced RD problem to $\supp \inputmarginalVect_0$, 
	such that $BA_\beta^n[\inputmarginalVect_0] \overset{n \to \infty}{\longrightarrow} \inputmarginalVect^*$.
\end{thm}

This shows that RD problems have sub-optimal solution branches. Namely, distributions that obtain the RD curve of a restricted problem, but not that of the unrestricted problem. 
Proceeding with the argument above Theorem \ref{thm:BA-converges-to-RD-curve-for-an-initial-cond-of-arbitrary-support}, one might expect that many sub-optimal branches typically exist at a given $\beta$ value.
Indeed, this is often the case, as demonstrated by Figure \ref{fig:optimal-and-suboptimal-RD-solution-for-CSD-problem}.

\begin{figure}[h!]
	\begin{center}
		\vspace*{15pt}
		\hspace*{-12pt}
		\includegraphics[width=1.05\textwidth]{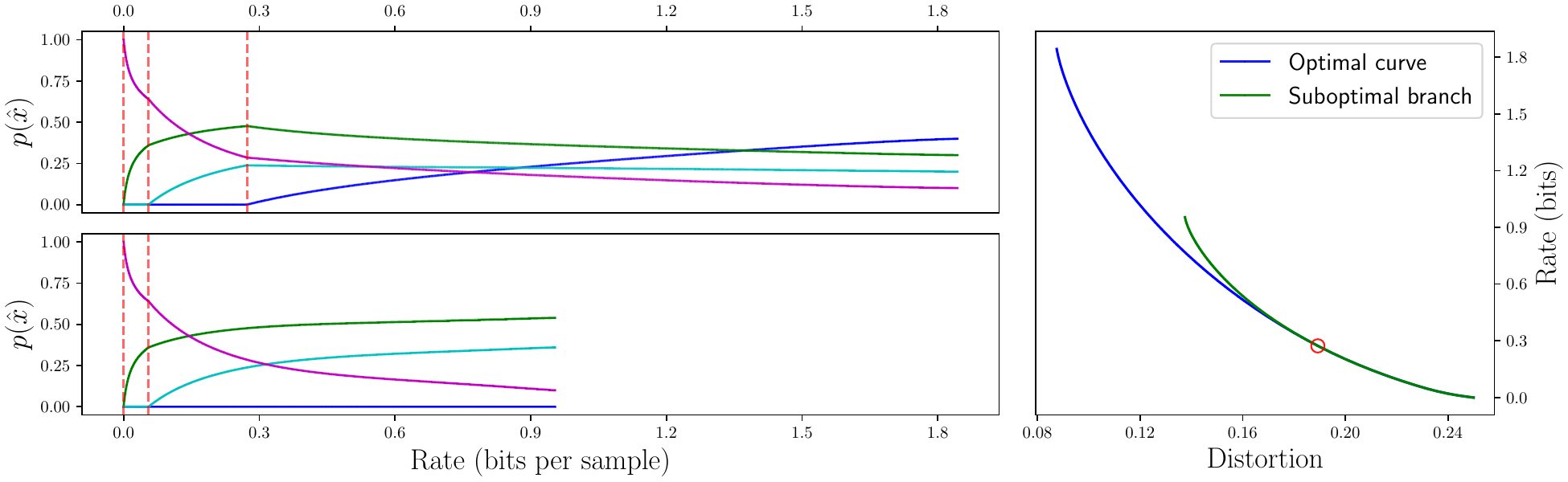}
		\caption{
			\textbf{An optimal and a sub-optimal branch merge at a cluster-vanishing bifurcation.} 
			\textbf{Top left}: The optimal solution by rate for the problem in Figure \ref{fig:reconstructing-an-RD-solution-curve}; BA with uniform initial conditions.
			\textbf{Bottom left}: The solution of the restricted problem, with the blue cluster deleted (equivalently, initialized to zero).
			\textbf{Left}: The solutions of the restricted and unrestricted problems differ above the bifurcation at $R_c \approx 0.273$ (rightmost dashed red vertical), merging into one at the bifurcation point and to its left.
			\textbf{Right}: The rate-distortion curves of both branches intersect at this point, marked by a red circle.
		}
		\vspace*{-10pt}
		\label{fig:optimal-and-suboptimal-RD-solution-for-CSD-problem}
	\end{center}
\end{figure}

\subsection{Cluster-vanishing bifurcations (support shrinking)}
\label{sub:cluster-vanishing-bifs-are-bifs}

The support of an optimal RD solution typically changes with $\beta$. 
A priori, it may not be clear that if the support of an optimal solution changes then there must be a bifurcation. 
Namely, that the number of fixed points of $BA_\beta$ must change.
We prove that if the support of an optimal solution shrinks, then two roots must intersect and merge into one. 
The argument for support-switching bifurcations is subtler and deferred to Section \ref{sub:support-switching-bifurcations}.

\medskip 
When the smallest distortion possible is desired, then the channel $x\mapsto \argmin_{\hat{x}} d(x, \hat{x})$ is clearly optimal; it is often of full support. 
While for zero rate, the constant encoding to a letter in $\argmin_{\hat{x}} \bb{E}[d(X, \hat{x})]$ is optimal.
In between, the support of achieving distributions $\inputmarginalVect^*_\beta$ usually shrinks gradually as a larger distortion $D$ is allowed ($\beta$ decreased). 
Suppose that the support of $\inputmarginalVect^*_\beta$ shrinks at $\beta_c$, $\supp \inputmarginalVect^*_{\beta^-} \subsetneqq \supp \inputmarginalVect^*_{\beta^+}$. 
At $\beta_c$ itself, $\inputmarginalVect^*_{\beta_c}$ may be considered as a solution of the reduced problem on $\supp \inputmarginalVect^*_{\beta^-}$ by deleting letters $\hat{x}$ outside its support, \cite[Lemma 1 in Section 2.5]{berger71}. 
By Theorem \ref{thm:properties-of-BA-jacobian-from-ISIT-CSD-paper} in Section \ref{sub:encoders-marginal-derivatives}, the Jacobian $D_{\inputmarginalVect} \left( Id - BA_\beta \right)\rvert_{\inputmarginalVect^*_{\beta_c}}$ at the reduced solution is non-singular. 
Thus, by the Implicit Function Theorem, $Id - BA_\beta$ on the reduced problem has a unique root $\inputmarginalVect'_\beta$ through $(\inputmarginalVect^*_{\beta_c}, \beta_c)$. 
While below $\beta_c$ it coincides with $\inputmarginalVect^*_{\beta^-}$, above $\beta_c$ it must differ from $\inputmarginalVect^*_{\beta^+}$ as their supports differ, although both are fixed points of $BA_\beta$ \eqref{eq:BA-operator-def}. This shows that two roots of $Id - BA_\beta$ must intersect at $\beta_c$, as demonstrated to the left of Figure \ref{fig:optimal-and-suboptimal-RD-solution-for-CSD-problem}.

As a side note, we comment that if an achieving distribution $\inputmarginalVect^*_\beta$ is known to be unique, then $\inputmarginalVect_\beta'$ is unstable above $\beta_c$. For, adding a small non-zero perturbation to the coordinates of $\inputmarginalVect_\beta'$ will then result in an initial condition which converges to $\inputmarginalVect^*_\beta$ under $BA_\beta$, by \citeauthor{csiszar1974computation}'s Theorem \ref{thm:BA-converges-to-RD-curve-for-an-initial-cond-of-full-support}.

\subsection{Obstructions to the root-tracking assumptions for RD}
\label{sub:obstructions-to-RT-assumptions-for-RD}

We discuss the obstructions to Assumptions \ref{assumption:operator-root-is-a-function-of-beta} and \ref{assumption:solution-is-smooth-in-beta} (from Section \ref{part:how-and-what}.\ref{sub:beta-derivs-at-an-operator-root}) in RD context. Namely, that the achieving distributions $\inputmarginalVect$ can be written as a smooth function $\inputmarginalVect_\beta$ of $\beta$. 
We also discuss a simple method to detect such obstructions.

\medskip
To that end, we dive into the subtle differences between three possible choices: of fixing a multiplier value $\beta > 0$, a point $(D, R)$ on the rate-distortion curve, or a distribution $\intermediateencoder{}{}$ which achieves the RD curve.
By \cite[Theorem 2.5.2]{berger71}, for every point $(D, R)$ on the RD curve there is some $\beta$ value such that $(D, R)$ can be generated parametrically\footnote{ By Equations (2.5.15) and (2.5.16) there; see also the discussion following Equation (2.5.19) there.} from $\beta$. This justifies talking about ``the curve points $(D, R)$ of $\beta$'', or ``the achieving distributions $\intermediateencoderVect$ of $\beta$''. 
However, to a value of $\beta$ there may correspond more than a single point on the rate-distortion curve; and a point on the rate-distortion curve may be achieved by more than a single distribution. 
We described these two kinds of non-uniquenesses schematically by,
\begin{equation}		\label{eq:schema-for-non-uniquness-in-RD}
	\xymatrix{
		\beta \quad \ar@{~>}[r]	&
		\quad \text{curve points } (D, R) \quad  \ar@{~>}[r]		&
		\quad \text{achieving distributions } \intermediateencoder{}{}
	}
\end{equation}

For non-uniqueness to the left of \eqref{eq:schema-for-non-uniquness-in-RD}, the multiplier $-\beta$ is the slope of the RD curve (Theorem 2.5.1 in \cite{berger71}), and is a continuous function $\beta(D)$ of the distortion $D$, on the open interval $(D_{min}, D_{max})$ (Theorem 2.5.5). 
As $R(D)$ is convex, its slope $-\beta$ is monotonically non-decreasing in $D$.
Thus, $D \in (D_{min}, D_{max})$ can be written as a function of $\beta$ if and only if $-\beta(D)$ is strictly increasing, in which case $D(\beta)$ is continuous. When $\beta(D)$ is constant in $D$, the RD curve has an entire linear segment which corresponds to a single $\beta$ value. 
In particular, the achieving distributions cannot be written then as a function of $\beta$, breaking Assumption \ref{assumption:operator-root-is-a-function-of-beta}. e.g., the right bifurcation in Figure \ref{fig:Berger_example_2.7.3}.
See \cite[Section 2.7]{berger71} for further examples.

However, even if the RD curve has no linear segments, a distribution $\intermediateencoderVect$ to the right of \eqref{eq:schema-for-non-uniquness-in-RD} which achieves a particular curve point $(D, R)$ need not be unique.
For example, suppose that a column in the distortion matrix is duplicate ($d$ is degenerate): there are $\hat{x}_1 \neq \hat{x}_2$ such that $d(\cdot, \hat{x}_1) = d(\cdot, \hat{x}_2)$. The clusters $\hat{x}_1$ and $\hat{x}_2$ are then indistinguishable for any practical purpose. e.g., if $\inputmarginalVect$ is an achieving distribution with $\inputmarginalSymbol_{1 2} := \inputmarginal{_1} + \inputmarginal{_2} > 0$, then it is not difficult to see that dividing $r_{12}$ arbitrarily between $\inputmarginal{_1}$ and $\inputmarginal{_2}$ also yields an achieving distribution $\inputmarginalVect'$: for $\lambda \in [0, 1]$, set $\inputmarginalSymbol'(\hat{x}_1) := \lambda \cdot \inputmarginalSymbol_{1 2}, \; \inputmarginalSymbol'(\hat{x}_2) := (1 - \lambda) \cdot \inputmarginalSymbol_{1 2}$ and $\inputmarginalSymbol'(\hat{x}) := \inputmarginalSymbol(\hat{x})$ otherwise.
Indeed, \citeauthor{berger71} notes that ``if the distortion matrix exhibits certain form of symmetry and degeneracy, there can be many choices of [a minimizer]''. 
While this breaks neither of our assumptions, the solution may then have multiple parameterizations in terms of $\beta$.

Non-uniqueness to the right of \eqref{eq:schema-for-non-uniquness-in-RD} (multiple distributions achieving a single curve point) implies that the simplex contains an entire line of achieving distributions,

\begin{thm}[Theorem 2.4.2 at \cite{berger71}]		\label{thm:achieving-distributions-at-curve-point-are-convex-berger}
	If the conditional probability distributions $\intermediateencoderVect'$ and $\intermediateencoderVect''$ both achieve a point $(D, R(D))$ on the rate-distortion curve, then so do their convex combinations $\lambda \intermediateencoderVect' + (1 - \lambda) \intermediateencoderVect''$, for any $\lambda \in [0, 1]$. 
\end{thm}

We generalize this as follows, to account for any non-uniqueness at \eqref{eq:schema-for-non-uniquness-in-RD}.

\begin{thm}		\label{thm:achieving-distributions-of-beta-are-convex}
	For any $\beta > 0$ value, the set of achieving distributions which corresponds to $\beta$ is convex.
\end{thm}

\begin{proof}[Proof of Theorem \ref{thm:achieving-distributions-of-beta-are-convex}]
	For a conditional probability distribution $\intermediateencoderVect$ and $\beta > 0$, denote by $\mathcal{L}_\beta[\intermediateencoderVect]$ the value of the RD Lagrangian there, 
	$I(\intermediateencoderVect; p_X) + \beta \; \bb{E}_{\intermediateencoder{}{} p_X(x)} \left[ d(x, \hat{x}) \right]$.
	
	Suppose that $\intermediateencoderVect' \neq \intermediateencoderVect''$ both achieve the RD curve \eqref{eq:RD-func-def} at the same $\beta$ value. 
	Therefore, $\mathcal{L}_\beta[\intermediateencoderVect']$ and $\mathcal{L}_\beta[\intermediateencoderVect'']$ must both equal the minimal Lagrangian value $\mathcal{L}_\beta^*$.
	Let $\intermediateencoderVect_\lambda := \lambda \intermediateencoderVect' + (1 - \lambda) \intermediateencoderVect''$ for $\lambda \in [0, 1]$. We need to prove that $\intermediateencoderVect_\lambda$ also achieves a point on the RD curve.
	Since mutual information is convex in $\intermediateencoderVect$ and the expectation term in $\mathcal{L}$ linear, we have that
	\begin{equation}
		\mathcal{L}_\beta[\intermediateencoderVect_\lambda] \leq
		\lambda \mathcal{L}_\beta[\intermediateencoderVect'] + (1 - \lambda) \mathcal{L}_\beta[\intermediateencoderVect''] =
		\mathcal{L}_\beta^* \;.
	\end{equation}
	But $\mathcal{L}_\beta^*$ is the minimal value of the Lagrangian, and so equality follows.
	
	Finally, the RD curve is the envelope of lines of slope $-\beta$ and intercept $\mathcal{L}_\beta^*$ along the $R$-axis, \cite{berger71}. Thus, $\intermediateencoderVect_\lambda$ indeed achieves the curve \eqref{eq:RD-func-def}.
\end{proof}

Theorem \ref{thm:achieving-distributions-of-beta-are-convex} yields two important corollaries. 
Recall the definition \eqref{eq:BA-operator-def} of $BA_\beta[\inputmarginalVect]$ (Section \ref{part:how-and-what}.\ref{sub:high-order-deriv-tensors-of-BA}), where it is considered as an operator in the marginals $\inputmarginalVect$. 
By abuse of notation, write $BA_\beta[\intermediateencoderVect]$ for the evaluation of $BA_\beta$ at $\intermediateencoderVect$, now considered as an operator in the encoders $\intermediateencoderVect$ (see comments there).
If for a particular $\beta$ value there is more than one achieving distribution $\intermediateencoderVect' \neq \intermediateencoderVect''$, then by Theorem \ref{thm:achieving-distributions-of-beta-are-convex}, the entire line section connecting $\intermediateencoderVect'$ to $\intermediateencoderVect''$ is comprised of fixed points of $BA_\beta$, or equivalently roots of $Id - BA_\beta$. Therefore, its Jacobian with respect to these coordinates must vanish along the vector pointing from $\intermediateencoderVect'$ to $\intermediateencoderVect''$.

\begin{cor}[Non-uniqueness of RD solutions is detectable by $\intermediateencoderVect$-Jacobian]
	\label{cor:multiple-RD-sols-are-detectable-by-Jacobian}
	If at $\beta > 0$ there is more than one achieving distribution $\intermediateencoderVect$, then $\ker D_{\intermediateencoderVect}(Id - BA_\beta)[\intermediateencoderVect] \neq \{\bm{0}\}$.
\end{cor}

The Jacobian in these coordinates is given explicitly by Proposition \ref{prop:Jacobian-of-BA-in-direct-encoder-coordinates} below.
An important practical aspect of this Corollary is that the kernel may be calculated at \textit{any} of the problem's achieving distributions $\intermediateencoderVect$ which correspond to $\beta$; it does not matter when testing for uniqueness.
If $\inputmarginalVect^*$ is a fixed point of BA when initialized at an arbitrary marginal $\inputmarginalVect_0$ (\textit{not} necessarily of full support) then, by Theorem \ref{thm:BA-converges-to-RD-curve-for-an-initial-cond-of-arbitrary-support}, it achieves the curve of the reduced problem to $\inputmarginalVect_0$. 
Thus, $\ker D_{\intermediateencoderVect}(Id - BA_\beta)[\intermediateencoderVect^*]$ gives a simple tool to test whether there might be additional distributions achieving a reduced problem at the corresponding encoder $\intermediateencoderVect^*$. 

We note that Corollary \ref{cor:multiple-RD-sols-are-detectable-by-Jacobian} is not merely a logical negation of the Implicit Function Theorem, as the latter is local in nature. Compare for example to the lines intersecting parabola Example in Section \ref{part:how-and-what}.\ref{subsub:lines-intersecting-parabola-example}. 
So long that the base $\beta_0$ of the expansion there is above the critical point, the problem has two distinct solutions, yet the Jacobian there is non-singular at each. 
Unlike Corollary \ref{cor:multiple-RD-sols-are-detectable-by-Jacobian} for RD, knowledge of one solution in that example does not allow us to detect that the other exits.
Instead, this discussion boils down to the following. cf., \cite[Section 2.3]{kuznetsov2004elements}.

\begin{cor}		\label{cor:there-are-no-global-bifs-in-RD}
	There are no global bifurcations in rate-distortion problems.
\end{cor}

We note, however, that bifurcations may not be detectable unless encoder coordinates $\intermediateencoderVect$ are used, as explained below.
The Jacobian's explicit form in these coordinates is given below (proof in Section \ref{part:proofs}.\ref{sub:proof-of-prop:Jacobian-of-BA-in-direct-encoder-coordinates}). 
cf., the Jacobian $D_{\inputmarginalVect} \left( Id - BA_\beta \right)[\inputmarginalVect]$ in marginal coordinates (Corollary \ref{cor:BA-jacobian} in Section \ref{sub:encoders-marginal-derivatives}). 

\begin{prop}			\label{prop:Jacobian-of-BA-in-direct-encoder-coordinates}
	Let $\intermediateencoderVect$ be a conditional distribution, $\outputmarginalVect$ the marginal defined by it via the marginal Equation \eqref{eq:marginal-eq}.
	Then, the Jacobian of $BA_\beta$ in encoder coordinates is given by the matrix 
	\begin{equation}			\label{eq:BA-jacobian-wrt-encoder}
		(D_{\intermediateencoderVect} BA_\beta[\intermediateencoderVect])_{(\hat{x}, x), (\hat{x}', x')} =
		\frac{e^{-\beta d(x, \hat{x}')}}{\sum_{\hat{x}''} \outputmarginal{''} e^{-\beta d(x, \hat{x}'')}} 
		\left[ \delta_{\hat{x}, \hat{x}'} - BA_\beta[\intermediateencoderVect](\hat{x}|x) \right] p_X(x')
	\end{equation}
	whose rows and columns are indexed by $(\hat{x}, x)$ and $(\hat{x}', x') \in \hat{\mathcal{X}}\times \mathcal{X}$, respectively. When $\forall \hat{x} \; \outputmarginal{} \neq 0$, this simplifies to
	\begin{equation}
		\frac{ BA_\beta[\intermediateencoderVect](\hat{x}'|x) }{ \outputmarginal{'} } 
		\left[ \delta_{\hat{x}, \hat{x}'} - BA_\beta[\intermediateencoderVect](\hat{x}|x) \right] p_X(x') \;.
		\label{eq:RD-jacobian-wrt-direct-enc-coords}
	\end{equation}
\end{prop}

There is a subtle difference between the Jacobians of $Id - BA_\beta$ with respect to the cluster marginal $\inputmarginalVect$ and encoder $\intermediateencoderVect$ coordinates.
For, the $\inputmarginalVect$-Jacobian can only detect cluster-vanishing bifurcations, as explained below, while the $\intermediateencoderVect$-Jacobian can be used to detect \textit{any} RD bifurcation, by Corollary \ref{cor:multiple-RD-sols-are-detectable-by-Jacobian}.
As a result, it is easy to mis-detect bifurcations other than cluster-vanishing when using marginal coordinates $\inputmarginalVect$. e.g., the support-switching bifurcation to the right of Figure \ref{fig:Berger_example_2.7.3} appears as a discontinuity in marginal coordinates $\inputmarginalVect$ (panel A), with no Jacobian eigenvalue vanishing to indicate its appearance (panel C). 
Nevertheless, its encoder coordinates $\intermediateencoderVect$-Jacobian has an eigenvalue vanishing precisely there (panel D); cf., Section \ref{sub:support-switching-bifurcations} on support-switching bifurcations. 

By the Implicit Function Theorem, if $D_{\intermediateencoderVect} (Id - BA_{\beta})\rvert_{\intermediateencoderVect}$ has no kernel at its root $\intermediateencoderVect$ then there is a unique function $\intermediateencoderVect_\beta$ through it. 
As discussed in Section \ref{part:how-and-what}.\ref{sub:beta-derivs-at-an-operator-root}, this implies that Assumptions \ref{assumption:operator-root-is-a-function-of-beta} and \ref{assumption:solution-is-smooth-in-beta} hold.
In particular, if $\intermediateencoderVect$ is an achieving distribution, then by Corollary \ref{cor:multiple-RD-sols-are-detectable-by-Jacobian} there is no other curve-achieving distribution at the same $\beta$ value.
For, that would be detectable by a non-trivial kernel. 
cf., the discussion around Equation \eqref{eq:schema-for-non-uniquness-in-RD}.

\begin{cor}		\label{cor:RT-assumptions-hold-for-RD}
	Let $\beta > 0$, and let $\intermediateencoderVect$ be a fixed point of $BA_\beta$ such that 
	\begin{equation}		\label{eq:trivial-kernel-as-sufficient-cond}
		\ker D_{\intermediateencoderVect} (Id - BA_{\beta})\rvert_{\intermediateencoderVect} = \{\bm{0}\} \;.
	\end{equation}
	Then, Assumptions \ref{assumption:operator-root-is-a-function-of-beta} and \ref{assumption:solution-is-smooth-in-beta} hold there for RD.
	
	Further, if $\intermediateencoderVect$ is also an achieving distribution, then \eqref{eq:trivial-kernel-as-sufficient-cond} implies that there is no other curve-achieving distribution at that $\beta$ value.
\end{cor}

In contrast, the Jacobian $D_{\inputmarginalVect} (Id - BA_{\beta})\rvert_{\inputmarginalVect}$ in marginal coordinates can only detect bifurcations where the support shrinks to a proper subset, at least in non-degenerate problems. For, by \cite[Lemma 2 ff.]{agmon2021critical}, its kernel is determined only by $\supp \inputmarginalVect$ (corresponding precisely to clusters outside the support), regardless of how many other achieving distributions there may be.
By reducing the problem to $\supp \inputmarginalVect$ (as in Section \ref{sub:cluster-vanishing-bifs-are-bifs}), its eigenvalues can be seen to be continuous in $\beta$, vanishing if and only if a cluster vanishes.
Therefore, the $\inputmarginalVect$-Jacobian cannot detect a bifurcation that switches between two distinct supports, even if they are of the same size. e.g., the right bifurcation of Figure \ref{fig:Berger_example_2.7.3}. 

The above gives a way to distinguish numerically between bifurcations of different types when tracking RD roots, summarized at \eqref{eq:flowchart-for-different-kinds-of-RD-bifurcations} below. 
Since eigenvalues are continuous in the choice of matrix, and we assume that $D_{\intermediateencoderVect} (Id - BA_{\beta})\rvert_{\intermediateencoderVect}$ is usually of full rank (Assumption \ref{assumption:only-cluster-vanishing-bifurcations}), then one can put a small threshold on eigenvalues, below which the Jacobian is considered singular. This is a necessary condition for bifurcation, of any kind. 
For simplicity, suppose that the total algebraic multiplicity of the vanishing eigenvalues is 1. 
If any cluster vanishes simultaneously, then the fixed point is approaching a cluster-vanishing bifurcation. 
Otherwise, the fixed point may be approaching a bifurcation of some other kind, or the distortion matrix may misbehave, as noted after \eqref{eq:schema-for-non-uniquness-in-RD}. 
While several approaches may come to mind on how bifurcations other than cluster-vanishing can be handled, that is beyond the scope of this work. 
\begin{equation}			\label{eq:flowchart-for-different-kinds-of-RD-bifurcations}
	\xymatrix@C=.4pc{
		\text{$|\lambda|$ approaches 0, for any }\lambda \in \eig D_{\intermediateencoderVect} (Id - BA_{\beta})\rvert_{\intermediateencoderVect} \text{?} \ar[r]^(.65){\text{No}} \ar[d]_{\text{Yes}}	&	\text{No bifurcation}	\\
		\inputmarginal{} \text{ approaches 0, for any $\hat{x}$?}\ar[d]_{\text{Yes}} \ar[dr]^{\text{No}}	&& \\
		\text{Cluster-vanishing bifurcation}	&	\text{Possibly a support-switching bifurcation}	&
	}
\end{equation}

\medskip 
The fact that there are RD bifurcations that are captured by the kernel of the $\intermediateencoderVect$-Jacobian but \textit{not} by that of the $\inputmarginalVect$-Jacobian can be understood intuitively by the relation
\begin{eqnarray}
	(D_{\inputmarginalVect}BA_\beta[\inputmarginalVect])_{\hat{x}, \hat{x}'} =
	\sum_{x = x'} (D_{\intermediateencoderVect}BA_\beta[\intermediateencoderVect])_{(\hat{x}, x), (\hat{x}', x')}
\end{eqnarray}
at a fixed point of full-support, which can be verified directly. That is, $D_{\inputmarginalVect}BA_\beta[\inputmarginalVect]$ is the blockwise trace of $D_{\intermediateencoderVect}BA_\beta[\intermediateencoderVect]$, and so it contains only a ``summary'' of the information in the latter. In particular, it may be that $D_{\intermediateencoderVect}BA_\beta[\intermediateencoderVect]$ has an eigenvector of eigenvalue 1 while $D_{\inputmarginalVect}BA_\beta[\inputmarginalVect]$ does not.

As a side note, the argument of \citeauthor{agmon2021critical} for critical slowing down of Blahut-Arimoto was based on the Jacobian $D_{\inputmarginalVect} (Id - BA_{\beta})\rvert_{\inputmarginalVect}$ with respect to cluster-marginal coordinates. However, their Theorem 5 makes no use of the choice of coordinates system, and so implies critical slowing down also when an eigenvalue of the $\intermediateencoderVect$-Jacobian vanishes gradually, even if no eigenvalue of the $\inputmarginalVect$-Jacobian is vanishing. 
This is demonstrated by Figure \ref{fig:Berger_example_2.7.3} (panels B and D, right bifurcation).

\begin{figure}[h]
	\hspace{-110pt}
	\ifdefined\compilefigs
	\includegraphics[width=1.5\textwidth]{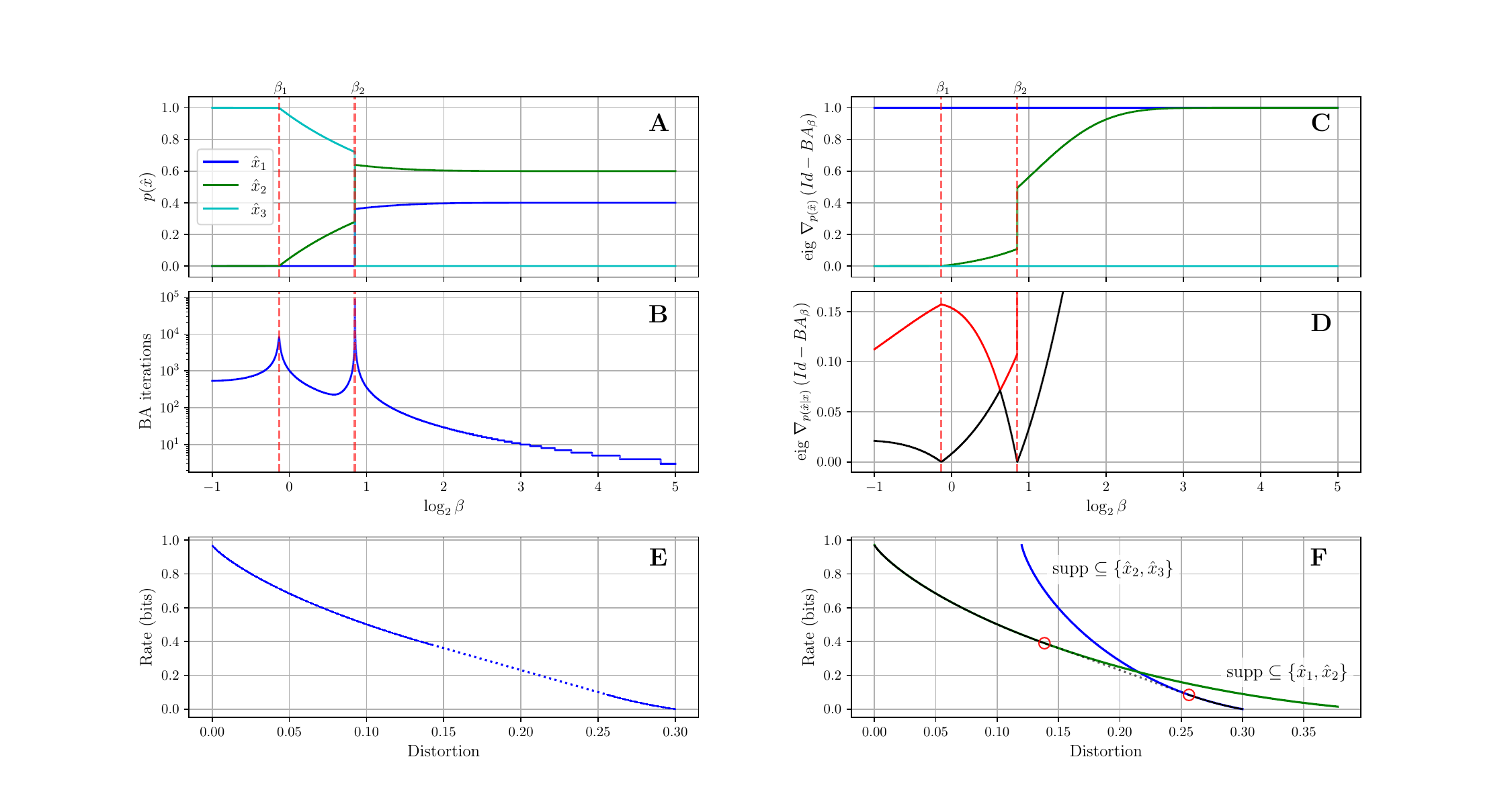}
	\else
	\includegraphics[width=1.5\textwidth]{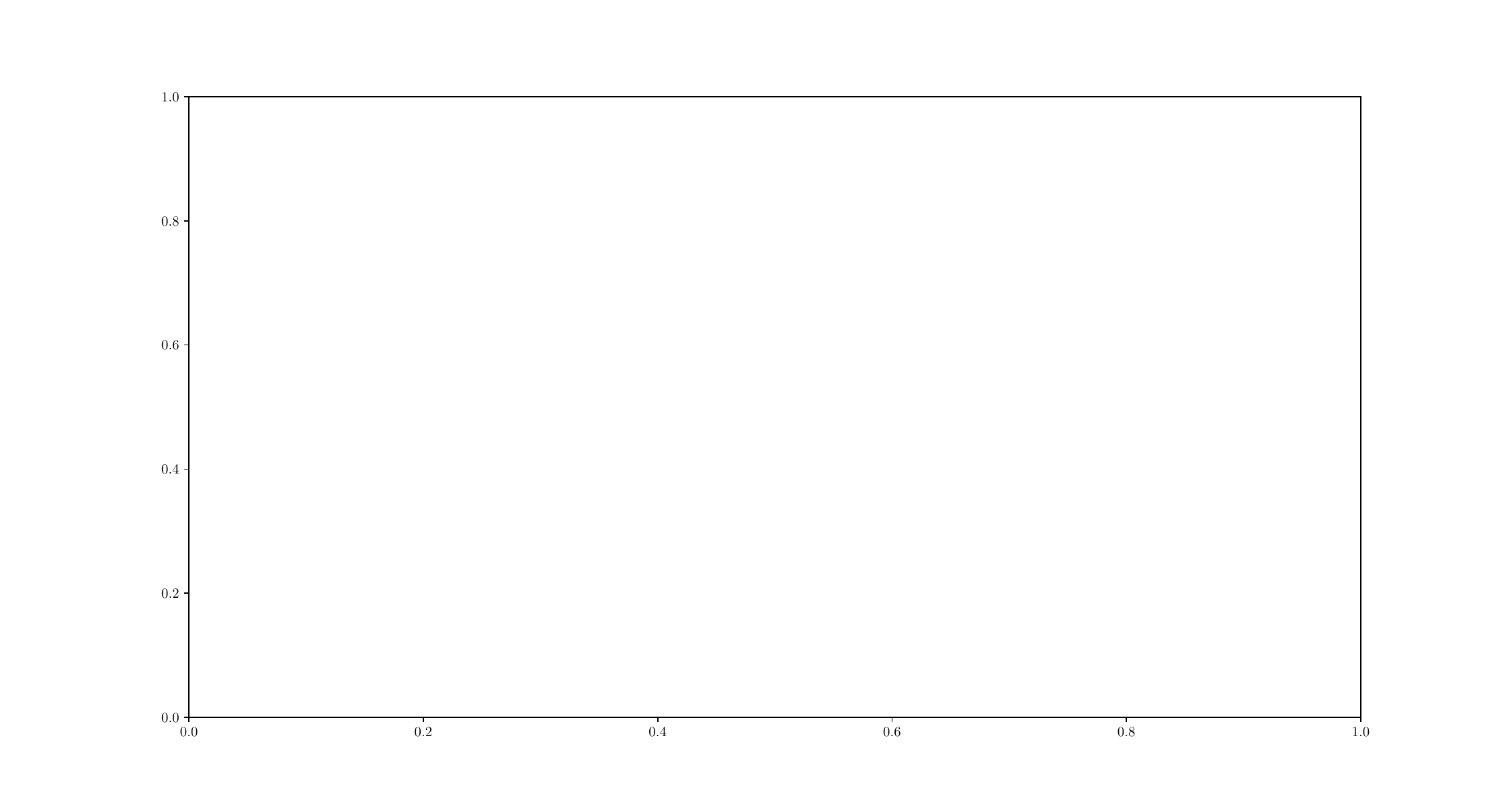}
	\fi
	\caption{
		\label{fig:Berger_example_2.7.3}
		\textbf{A support-switching and a cluster-vanishing bifurcation.}
		Reproducing \cite[Figure 2.7.6]{berger71}, defined by $d(x, \hat{x}) = \left(\protect\begin{matrix} 1 & 0 & 0.3 \\ 0 & 1 & 0.3 \protect\end{matrix}\right)$ and $p_X = (0.4, 0.6)$. This problem exhibits two bifurcations, at $\beta_1$ and $\beta_2$ (dashed red verticals). 
		\textbf{A.} Cluster marginal $p_\beta(\hat{x})$ as a function of $\beta$. While at $\beta_1$ a cluster vanishes, at $\beta_2$ the support switches between two distinct subsets of size 2. At $\beta_2$, an entire line section of distributions is optimal. This results in an apparent discontinuity at panels A through D, as they are plotted by $\beta$ value. 
		\textbf{B.} Blahut-Arimoto iterations until convergence (stopping condition is $10^{-7}$), initialized with uniform initial conditions at each $\beta$ value. Critical slowing down is clearly noticed at both sides of the two bifurcations, even though panel C has a vanishing eigenvalue only to the right of $\beta_1$. This is explained by the vanishing eigenvalues of the $\intermediateencoderVect$-Jacobian with respect to encoder coordinates, at panel D. See Section \ref{sub:obstructions-to-RT-assumptions-for-RD}.
		\textbf{C.} Eigenvalues of the $\inputmarginalVect$-Jacobian with respect to cluster-marginal coordinates $\inputmarginal{}$. An eigenvalue approaches zero at the cluster-vanishing bifurcation to the left. In contrast, the solution's support size does not change in the vicinity of $\beta_2$. Therefore, the Jacobian's rank there remains unchanged \cite[Theorem 1]{agmon2021critical}.
		\textbf{D.} Eigenvalues of the $\intermediateencoderVect$-Jacobian with respect to the encoder's coordinates $\intermediateencoder{}{}$. 
		Multiple achieving distributions at a single value $\beta_2$ are detectable by a non-trivial kernel direction, in accordance with Corollary \ref{cor:multiple-RD-sols-are-detectable-by-Jacobian}.
		\textbf{E.} The rate-distortion curve has a linear segment (dotted), corresponding to $\beta = \beta_2$.
		\textbf{F.} The linear curve segment can be explained in terms of a support-switching bifurcation between two suboptimal RD curves, of the problems restricted to $\{\hat{x}_1, \hat{x}_2\}$ and $\{\hat{x}_2, \hat{x}_3\}$. 
		The optimal RD curve in black (dashed) alternately coincides with the two suboptimal curves (blue and green), as explained in Section \ref{sub:support-switching-bifurcations}.
		Both our Algorithm \ref{algo:root-tracking-for-RD} and BA with reverse annealing miss the right bifurcation, and so follow the sub-optimal root in green; see Section \ref{sub:when-does-RTRD-follow-the-optimal-path}.
	}
\end{figure}

\clearpage

\subsection{Why does root-tracking for RD follow the optimal solution path?}

\label{sub:when-does-RTRD-follow-the-optimal-path}

As seen in previous subsections, RD problems typically have a plethora of sub-optimal solutions, which do \textit{not} achieve the problem's rate-distortion curve. 
Under Assumption \ref{assumption:only-cluster-vanishing-bifurcations} (of Section \ref{part:how-and-what}.\ref{sub:taylor-method-for-RD-root-tracking}), the convexity of achieving distributions (Theorem \ref{thm:achieving-distributions-of-beta-are-convex}) implies that an optimal root tracked by our algorithms does indeed remain optimal (namely, achieving). We elaborate on this below.

\medskip 
First, consider Algorithm \ref{algo:taylor-method-for-RD-root-tracking} for tracking a root between cluster-vanishing bifurcations.
An achieving distribution tracked by it remains achieving so long that the $\intermediateencoderVect$-Jacobian is non-singular.
Indeed, by Corollary \ref{cor:RT-assumptions-hold-for-RD} in Section \ref{sub:obstructions-to-RT-assumptions-for-RD}, Assumptions \ref{assumption:operator-root-is-a-function-of-beta} and \ref{assumption:solution-is-smooth-in-beta} which are necessary for root-tracking \ref{part:how-and-what}.\ref{sub:high-order-beta-derivatives-at-an-operator-root} hold.
Further, as the root being tracked is an achieving distribution, then that Corollary also implies that there is no other achieving distribution. 
We note that two distinct roots cannot exchange the property of being curve achieving without both being achieving simultaneously, because the rate and distortion functionals are continuous in the encoder $\intermediateencoderVect$, \footnote{ Denote $I(\intermediateencoderVect) := I(\intermediateencoderVect; p_X)$, $D(\intermediateencoderVect) := \bb{E}_{\intermediateencoderVect \; p_X}[d(x, \hat{x})]$, $F_\beta(\intermediateencoderVect) := I(\intermediateencoderVect) + \beta D(\intermediateencoderVect)$, and write $F_\beta := \min_{\intermediateencoderVect} F_\beta(\intermediateencoderVect)$ for its minimal value. Suppose that $\intermediateencoderVect_\beta'$ and $\intermediateencoderVect_\beta''$ are (continuous) paths such that $F_{\beta_1}(\intermediateencoderVect_{\beta_1}') < F_{\beta_1}(\intermediateencoderVect_{\beta_1}'')$ but $F_{\beta_2}(\intermediateencoderVect_{\beta_2}') > F_{\beta_2}(\intermediateencoderVect_{\beta_2}'')$, for some $\beta_1 < \beta_2$. Then by continuity, there must exist $\beta_3 \in (\beta_1, \beta_2)$ where $F_\beta(\cdot)$ obtains the same value on both paths, $F_{\beta_3}(\intermediateencoderVect_{\beta_3}') = F_{\beta_3}(\intermediateencoderVect_{\beta_3}'')$. The argument follows by setting $F_{\beta_1} = F_{\beta_1}(\intermediateencoderVect_{\beta_1}')$ and $F_{\beta_2} = F_{\beta_2}(\intermediateencoderVect_{\beta_2}'')$. }.
The above arguments hold so long that $D_{\intermediateencoderVect}(Id - BA_\beta)[\intermediateencoderVect_{\beta}]$ remains non-singular, which by Assumption \ref{assumption:only-cluster-vanishing-bifurcations} is true until the next cluster-vanishing bifurcation is reached.

Second, when an achieving distribution $\intermediateencoderVect_\beta$ approaches a cluster-vanishing bifurcation, then by continuity of the RD curve, \cite[Theorems 2.4.1, 2.5.4]{berger71}, it achieves the curve also at the point of bifurcation.
As shown in Section \ref{sub:cluster-vanishing-bifs-are-bifs}, a sub-optimal root must exist to the right of the bifurcation, merging with the optimal one at the point of bifurcation. Thus, while the heuristic used by Algorithm \ref{algo:root-tracking-for-RD} to handle the bifurcation may temporarily follow a sub-optimal branch, it achieves the curve once the point of bifurcation is reached, re-gaining optimality.

Third, we need to guarantee that Algorithm \ref{algo:taylor-method-for-RD-root-tracking} is indeed initialized at an achieving distribution.
When starting at $\beta_0 \gg 0$ with an initial condition of full support, iterate with Blahut-Arimoto until convergence. By \citeauthor{csiszar1974computation}'s Theorem \ref{thm:BA-converges-to-RD-curve-for-an-initial-cond-of-full-support}, the distribution $\intermediateencoderVect_{\beta_0}$ obtained this way is an achieving distribution. 
As shown above, the heuristic used by Algorithm \ref{algo:root-tracking-for-RD} always ends at an achieving distribution.
Thus, Algorithm \ref{algo:root-tracking-for-RD} will initialize Algorithm \ref{algo:taylor-method-for-RD-root-tracking} at achieving distribution the next time it is invoked.

\medskip 
The assumption that the $\intermediateencoderVect$-Jacobian is non-singular outside cluster-vanishing bifurcations is necessary for Algorithm \ref{algo:taylor-method-for-RD-root-tracking} to keep track of the optimal root. For, this algorithm tracks the root of the reduced problem, essentially using only the $\inputmarginalVect$-Jacobian as its stopping condition.
Therefore, it would only detect cluster-vanishing bifurcations, as explained in Section \ref{sub:obstructions-to-RT-assumptions-for-RD}.
e.g., at the support-switching bifurcation to the right of Figure \ref{fig:Berger_example_2.7.3}, Algorithm \ref{algo:taylor-method-for-RD-root-tracking} would continue to track the root past the bifurcation, even after it had lost optimality. 
So will BA with reverse annealing, as the root continues to exist beyond the bifurcation.
However, while missing bifurcations other than cluster vanishing, Algorithm \ref{algo:taylor-method-for-RD-root-tracking} would suffer no computational penalty or accuracy loss near them as Blahut-Arimoto does, since no eigenvalue of the $\inputmarginalVect$-Jacobian vanishes there (cf., panels B and C of Figure \ref{fig:Berger_example_2.7.3}).

Nevertheless, the discussion around \eqref{eq:flowchart-for-different-kinds-of-RD-bifurcations} provides a simple method to detect bifurcations other than cluster-vanishing ones, up to pathologies of the distortion matrix. Once detected, we expect that they could be handled using tools similar to those developed here. 
As an alternative approach, RD derivatives in Algorithm \ref{algo:high-order-derivs-of-operator-roots} can be calculated with respect to $\intermediateencoderVect$ rather than $\inputmarginalVect$ coordinates, changing the stopping condition of Algorithm \ref{algo:taylor-method-for-RD-root-tracking} accordingly. 
We have not chosen this approach due the computational costs of having a higher-dimensional variable.

\subsection{Linear curve segments as support-switching bifurcations}

\label{sub:support-switching-bifurcations}

As a side note to this Section's main line of discussion, we offer an explanation of linear curve segments in terms of support-switching bifurcations between sub-optimal RD curves. This explains the right bifurcation in Figure \ref{fig:Berger_example_2.7.3}.

\medskip
Suppose that at $D_1$ the curve is achieved by a solution of support $A \subseteq \hat{\mathcal{X}}$, and at $D_2 > D_1$ by a solution of support $B \subseteq \hat{\mathcal{X}}$, where neither of these subsets contains the other, $A \nsubseteq B$ and $B \nsubseteq A$. The case of a shrinking support at bifurcation was already handled in Section \ref{sub:cluster-vanishing-bifs-are-bifs}.
Without loss of generality, suppose that $\hat{\cal{X}} = A \cup B$.
For simplicity, suppose that outside bifurcations, the curve is achieved by a unique distribution. 
Denote by $R^A(D)$ and $R^B(D)$ the RD curves of the restricted problems to $A$ and $B$, respectively, and suppose that their respective ranges $(D_{min}^A, D_{max}^A)$ and $(D_{min}^B, D_{max}^B)$ of distortion values has a large enough intersection for the below to be meaningful. 

Since both $R^A(D)$ and $R^B(D)$ are continuous, and only $R^A$ is optimal at $D_1$ (only $R^B$ at $D_2$), then the two curves must intersect somewhere in $(D_1, D_2)$. 
Recall that a differentiable function is convex if and only if all its tangents lie below its graph.
Thus, the tangent to $R^A$ at a point $D$ is below the curve of $R^A$. Pick the smallest $D_1'$ such that the tangent to $R^A$ at $D_1'$ also intersects the curve $R^B$, say at $D_2'$. There must be such a point $D_1'$, as the two curves $R^A$ and $R^B$ intersect and are continuous. 
$D_1' > D_1$ because $R^A$ is below $R^B$ at $D_1$, by assumption.
By convexity, $D_1 < D_1' < D_2' < D_2$.
If $R^A$ still coincides with the problem's RD curve at $D_1'$ and $R^B$ with that at $D_2'$, then by Theorem \ref{thm:achieving-distributions-of-beta-are-convex}, the entire line section connecting $(D_1', R^A(D_1'))$ to $(D_2', R^A(D_2'))$ is achievable, at the $\beta$ value corresponding to its slope. 
Other than its vertices at $D_1'$ and $D_2'$, which are obtained by distributions supported on $A$ or $B$ alone, any other point along this section is obtained by distributions supported on the entire $\hat{\mathcal{X}}$. Points in the section's interior lie on the RD curve of the original problem, but on neither of the suboptimal curves $R^A$ or $R^B$.
See panel F in Figure \ref{fig:Berger_example_2.7.3} for example.

\medskip
\section{Error analysis for root-tracking for RD}
\label{sec:error-analysis}

We analyze the error of Algorithm \ref{algo:taylor-method-for-RD-root-tracking}, our specialization of Taylor's method to RD. While the highlight of this section is its convergence guarantees (Theorem \ref{thm:taylor-method-converges-for-RD-root-tracking-away-of-bifurcation} in Section \ref{part:how-and-what}.\ref{sub:taylor-method-for-RD-root-tracking}; see Section \ref{part:proofs}.\ref{sub:proof-of-thm:taylor-method-converges-for-RD-root-tracking-away-of-bifurcation} for proof), the results comprising it are of interest on their own right.
In Section \ref{sub:error-analysis-of-Taylor-method-background} we briefly recap standard error analysis, highlighting relevant subtleties.
In Section \ref{sub:computational-difficulty-of-RTRD} we show that tracking an operator's root in the presence of bifurcations generally belongs to a family of ``stiff'' problems, which are harder to solve with standard numerical methods such as Taylor's.
Stopping towards a bifurcation effectively restricts the problem's difficulty, allowing RD derivative tensors to be bounded uniformly.
To mitigate this difficulty, we suggest in Section \ref{sub:local-error-estimates-for-beta-derivs} a tool for local error estimation, which could be used to improve the cost-to-error tradeoff of Algorithm \ref{algo:taylor-method-for-RD-root-tracking}.

\subsection{Preliminaries: error analysis of Taylor methods}
\label{sub:error-analysis-of-Taylor-method-background}

We succinctly recap the error analysis of the Euler and Taylor methods for solving ordinary differential equations. While these are necessary for the sequel, the reader well versed in this material is advised to skip to the next subsection.

\medskip
We follow standard definitions of numerical approximations, as in \cite{atkinson2011numerical} or \cite{butcher2016numerical}. A first-order initial value problem is defined by
\begin{equation}		\label{eq:IVP-definition}
	\dbeta{\bm{x}} = f\left( \bm{x}, \beta \right),		\quad
	\bm{x}(\beta_0) = \bm{x}_0 \;.
\end{equation}
By the implicit ODE \eqref{eq:ODE-implicit-form}, tracking an operator root $(\bm{x}_0, \beta_0)$ is of this form. 
Where, RD roots satisfy the ODE of Theorem \ref{thm:beta-ODE-in-marginal-coords} in Section \ref{sub:encoders-beta-derivatives}.

Write $\bm{x}_n$ for a numerical approximation of the true solution $\bm{x}(\beta_n)$ at $\beta := \beta_n$, and $\bm{e}_n := \bm{x}(\beta_n) - \bm{x}_n$ for the approximation's error, known as the \emph{global truncation error} in the context of numerical approximation. A numerical approximation can be obtained by setting
\begin{equation}		\label{eq:Euler-method-def}
	\bm{x}_{n+1} := \bm{x}_n + \Delta \beta \cdot f\left( \bm{x}_n, \beta_n \right)
\end{equation}
where $\Delta \beta$ is the \emph{step size}, often fixed, and $\beta_{n+1} := \beta_n + \Delta \beta$. This approximation method \eqref{eq:Euler-method-def} is known as the \emph{Euler method}. When its right-hand side is replaced by a Taylor polynomial of degree $l > 0$, this is known as the \emph{Taylor method},
\begin{equation}		\label{eq:Taylor-method-def}
	\bm{x}_{n+1} := \bm{x}_n + \Delta \beta \cdot T_l\big( \bm{x}_n, \beta_n, \Delta \beta \big) \;.
\end{equation}
Euler's method is a first-order Taylor method. In root-tracking context, we write
\begin{equation}		\label{eq:Taylor-poly-def-for-Taylor-method}
	T_l\big( \bm{x}_n, \beta_n, \Delta \beta \big) := 
	\frac{1}{1!} \cdot \dbeta{\bm{x}}\Big\rvert_{(\bm{x}_n, \beta_n)} +
	\frac{\Delta \beta}{2!} \cdot \dbetaK{\bm{x}}{2}\Big\rvert_{(\bm{x}_n, \beta_n)} + \dots +
	\frac{\Delta \beta^{l-1}}{l!} \cdot \dbetaK{\bm{x}}{l}\Big\rvert_{(\bm{x}_n, \beta_n)} 
\end{equation}
Where, $\dbetaK{\bm{x}}{k}$ are the implicit derivatives computed by Algorithm \ref{algo:high-order-derivs-of-operator-roots} (Theorem \ref{thm:formula-for-high-order-expansion-of-F-in-main-result-sect}). This gives a numerical algorithm for tracking operator roots \eqref{eq:solution-as-root-of-functional-eq}, under Assumptions \ref{assumption:operator-root-is-a-function-of-beta} and \ref{assumption:solution-is-smooth-in-beta} (in Section \ref{part:how-and-what}.\ref{sub:beta-derivs-at-an-operator-root}).
Where, in RD context the step $\Delta \beta$ is negative (see Section \ref{part:how-and-what}.\ref{sub:taylor-method-for-RD-root-tracking}).

Error analysis of Euler's method \eqref{eq:Euler-method-def} is a standard result in numerical solution of ODEs. e.g., \cite[Theorem 2.4]{atkinson2011numerical} or \cite[Theorem 212A]{butcher2016numerical}. Its generalization to an $l$-th order Taylor method is straightforward, though usually not given explicitly in textbooks. Following the notes of \cite{gottlieb2006}, its crux is subtracting the numerical approximation $\bm{x}_{n+1}$ \eqref{eq:Taylor-method-def} from the true solution 
\begin{equation}
	\bm{x}(\beta_{n+1}) = 
	\bm{x}(\beta_{n}) + \Delta \beta \cdot T_l\big(\bm{x}(\beta_n), \beta_n, \Delta\beta\big) + 
	\Delta \beta \cdot \underset{LTE}{\underbrace{\frac{\Delta \beta^{l}}{(l+1)!} \dbetaK{\bm{x}}{l+1}}} \;,
\end{equation}
where the last term is the \emph{local truncation error}, which is simply the Taylor remainder at $(\bm{x}(\beta'), \beta')$, for some intermediate $\beta' \in [\beta_n, \beta_{n+1}]$. This yields
\begin{equation}			\label{eq:local-to-global-truncation-error-for-Taylor-method}
	\| \bm{e}_{n+1} \|_\infty \leq \| \bm{e}_n \|_\infty + 
	|\Delta \beta| \cdot \| T_l\big(\bm{x}(\beta_n), \beta_n, \Delta\beta\big) - T_l\big(\bm{x}_n, \beta_n, \Delta\beta\big) \|_\infty +
	|\Delta \beta| \cdot \| LTE \|_\infty \;,
\end{equation}
where $\|\cdot \|_\infty$ denotes the supremum norm. 
If $T_l$ satisfies the \emph{Lipschitz condition with a constant $L_l$} with respect to $\bm{x}$, $\| T_l(\bm{x}, \beta) - T_l(\bm{x}', \beta) \|_\infty \leq L_l \; \| \bm{x} - \bm{x}'\|_\infty$ for every $\bm{x}, \bm{x}'$ and $\beta\in \left[ \beta_0, \beta_f \right]$, then 
\begin{equation}		\label{eq:single-step-error-increase-in-Taylor-method}
	\| \bm{e}_{n+1} \|_\infty \leq 
	\left(1 + |\Delta \beta| L_l\right) \cdot \| \bm{e}_n \|_\infty +
	|\Delta \beta| \cdot \| LTE \|_\infty \;,
\end{equation}
and so the bound grows exponentially.
This allows one to show \cite[Equation (5.11)]{atkinson2011numerical} that the global truncation error of an $l$-th order Taylor method is of order $O(|\Delta \beta|^l)$:

\begin{thm}[Error analysis of Taylor method]		\label{thm:error-analysis-for-euler-method}
	For an initial-value problem \eqref{eq:IVP-definition} on $[\beta_0, \beta_f]$, the \textit{global truncation error} obtained by a Taylor method of order $l$ with a step size of $|\Delta \beta|$ at most satisfies
	\begin{equation}			\label{eq:GTE-of-Taylor-method}
		\max_{\beta_0 \leq \beta_n \leq \beta_f} \| \bm{x}(\beta_n) - \bm{x}_n \|_\infty \leq
		e^{(\beta_f - \beta_0)L_l} \|\bm{e}_0\|_\infty + \frac{e^{(\beta_f - \beta_0)L_l} - 1}{L_l} \cdot \tfrac{1}{(l+1)!} |\Delta \beta|^l \max_{\beta_0 \leq \beta \leq \beta_f} \left\| \tfrac{d^{l+1} \bm{x}(\beta)}{d\beta^{l+1}} \right\|_\infty
	\end{equation}
	where $L_l$ is the Lipschitz constant of $T_l$ \eqref{eq:Taylor-poly-def-for-Taylor-method}, and $\bm{e}_0 := \bm{x}(\beta_0) - \bm{x}_0$ is the initial error.
\end{thm}

We note that the Lipschitz constant $L_l$ may be taken \cite[Equation (3.9)]{atkinson2011numerical} as the supremum of the problem's linearization, $\sup \|D_{\bm{x}} T_l\|_\infty$, over the relevant domain in $\bm{x}$ and $\beta$. The matrix norm $\|D_{\bm{x}} T_l\|_\infty$ at a point can be used to estimate the \textit{local Lipschitz constant} of $T_l$; namely, its Lipschitz constant over an arbitrarily small neighborhood.

\subsection{The computational difficulty in root tracking for RD}
\label{sub:computational-difficulty-of-RTRD}

There is a computational difficulty in tracking operator roots with Taylor's method, a difficulty that stems from the presence of bifurcations.
For Taylor's method error analysis (Theorem \ref{thm:error-analysis-for-euler-method}), this is manifested in general via local Lipschitz constants. 
For, its error bounds explode when approaching a bifurcation.
The below Lemma \ref{lem:uniform-bound-on-BA-derivative-tensors} bounds RD derivative tensors, and Proposition \ref{prop:Jacobian-of-high-order-beta-derivative} (in Section \ref{sub:local-error-estimates-for-beta-derivs}) can be used to bound local Lipschitz constants. 
Using both, one can show that setting a cluster-vanishing threshold (as in Algorithm \ref{algo:taylor-method-for-RD-root-tracking}) restricts the computational difficulty.

\medskip
In its implicit ODE form, tracking an operator's root is defined by the initial value problem consisting of
\begin{equation}		\tag{\ref{eq:ODE-implicit-form}}
	\tfrac{d\bm{x}}{d\beta} = - \left(D_{\bm{x}} F \right)^{-1} D_\beta F \;,
\end{equation}
at a given root $(\bm{x}_0, \beta_0)$ of $F = \bm{0}$ \eqref{eq:solution-as-root-of-functional-eq}, so long that $D_{\bm{x}} F$ is non-singular.
Suppose that an eigenvalue of $D_{\bm{x}} F$ vanishes gradually as $\beta \to \beta_c$. e.g., when approaching a bifurcation.
The Jacobian $D_{\bm{x}} F$ then usually becomes ill-conditioned as a result\footnote{ Unless all its eigenvalues vanish at the same rate when $\beta \to \beta_c$.}.
The linearization of this differential equation would then in general be ill-conditioned\footnote{ Write $D_{\bm{x}} F^{-1}$ for the inverse of the Jacobian matrix, and differentiate $D_{\bm{x}} F \; D_{\bm{x}} F^{-1} = I$ with respect to $\bm{x}$, to obtain $D_{\bm{x}} F \; D_{\bm{x}}(D_{\bm{x}} F^{-1}) = -D^2_{\bm{x}, \bm{x}} F \; D_{\bm{x}} F^{-1}$, or equivalently $D_{\bm{x}}(D_{\bm{x}} F^{-1}) = -D_{\bm{x}} F^{-1} \; D^2_{\bm{x}, \bm{x}} F \; D_{\bm{x}} F^{-1}$. This shows that the linearization $-D_{\bm{x}} \left(D_{\bm{x}} F^{-1} D_\beta F\right)$ of \eqref{eq:ODE-implicit-form} can be written as $-D_{\bm{x}} F^{-1} M$ for $M := D^2_{\beta, \bm{x}}F - D^2_{\bm{x}, \bm{x}}F \; D_{\bm{x}} F^{-1} D_\beta F$ a matrix. So, it would be ill-conditioned if an eigenvalue $\lambda_\beta$ of $D_{\bm{x}} F$ vanishes as $\beta\to \beta_c$, unless no column of $M$ has a component in the $\lambda_\beta$-eigenspace of $D_{\bm{x}} F$.}, 
implying that it is \textit{stiff} \cite[Chapter 8]{atkinson2011numerical}.
While there is no widely accepted definition of stiff equations, their \emph{``most important common feature [...] is that when such equations are being solved with standard numerical methods, the step-size $|\Delta\beta|$ must be extremely small in order to maintain stability --- far smaller than would appear to be necessary from a consideration of the truncation error''}, \cite[Chapter 8]{atkinson2011numerical}.
See also \cite[Section 6]{butcher2000numerical}.

Indeed, while a finite Lipschitz constant is required for Taylor method's error analysis (Theorem \ref{thm:error-analysis-for-euler-method} above), 
it need not be bounded near a bifurcation\footnote{
	By the argument above, the linearization of \eqref{eq:ODE-implicit-form} need not have a finite matrix norm when an eigenvalue of $D_{\bm{x}} F$ vanishes gradually.
}. cf., its error bound \eqref{eq:single-step-error-increase-in-Taylor-method}.
The existence of a bifurcation is not just a technical hurdle in proving that the Theorem's conditions hold, but an essential one, impeding algorithms' performance there.
For Blahut-Arimoto's algorithm, this is manifested by critical slowing down near bifurcations, \cite{agmon2021critical}, while for tracking an operator's root it is manifested in the stiffness of the implicit ODE \eqref{eq:ODE-implicit-form}.
Both BA and our Algorithm \ref{algo:high-order-derivs-of-operator-roots} for RD suffer from reduced accuracy when approaching a bifurcation, as depicted by Figure \ref{fig:derivative-calculation-loses-accuracy-near-bifurcation}.

While there is much literature on stiff differential equations, stopping at a cluster mass threshold $\delta > 0$ as in Algorithm \ref{algo:taylor-method-for-RD-root-tracking} is a straightforward solution for guaranteeing convergence of Taylor's method.
The proof of Theorem \ref{thm:taylor-method-converges-for-RD-root-tracking-away-of-bifurcation} bounds the local Lipschitz constants so long that the bifurcation is at least $\delta$-far.
To show this, Lemma \ref{lem:uniform-bound-on-BA-derivative-tensors} below guarantees that RD derivative tensors are then bounded uniformly. While Proposition \ref{prop:Jacobian-of-high-order-beta-derivative} in Section \ref{sub:local-error-estimates-for-beta-derivs} allows to bound not only the implicit derivatives, but also the local Lipschitz constants.
Writing $\inputmarginalVect$ for a distribution in $\Delta[\hat{\mathcal{X}}]$, we have the following.

\begin{lem}[RD derivative tensors are bounded uniformly on compact subsets in simplex interior]	\label{lem:uniform-bound-on-BA-derivative-tensors}
	For any $\delta > 0$ small enough, the derivative tensors of $Id - BA_\beta$ \eqref{eq:RD-operator-def} (\eqref{eq:repeated-beta-deriv-in-thm} and \eqref{eq:mixed-BA-deriv-in-thm} in Theorem \ref{thm:high-order-derivs-of-BA-in-main-text}) are bounded uniformly on the closed $\delta$-interior of the simplex, under the supremum norm. 
	
	Explicitly, let an RD problem be defined by $p_X$ and $d(x, \hat{x})$ ($d$ for short), and let $\inputmarginalVect$ be a distribution in the closed $\delta$-interior of $\Delta[ \hat{\mathcal{X}} ]$. Then, for any orders $b, m \geq 0$ of differentiation (other than $b = m = 0$),
	\begin{equation}		\label{eq:uniform-upper-bound-on-deriv-tensor}
		\left| \left(D_{\beta^{b}, \inputmarginalVect^{m}}^{b + m} \left(Id - BA_\beta \right)\left[\inputmarginalVect\right](\hat{x}) \right)_{(i_1, i_2, \dots, i_m)} \right| \leq
		1 + \frac{1}{\delta^{m}} \cdot C(b, m; d, |\hat{\mathcal{X}}|)
	\end{equation}
	where 
	\begin{equation}
		C(b, m; d, M) := 2b \cdot (m + 1)! \; \binom{b + M - 1}{b} \left[ m! \; p(b) \cdot \left( 2^{b} b! \cdot d_{max}^{b^2} \right)^{1 + m} \right]^{M} \;, 
	\end{equation}
	with $d_{max} := \max \left\{ 1, \max_{x, \hat{x}} d(x, \hat{x}) \right\}$, and $p(m)$ is the partition function (see Section \ref{part:how-and-what}.\ref{sub:high-order-beta-derivatives-at-an-operator-root}).
\end{lem}

While the bound at \eqref{eq:uniform-upper-bound-on-deriv-tensor} can be improved, its purpose is to bound the derivative tensors of $Id - BA_\beta$ \eqref{eq:RD-operator-def} uniformly, with constants involving only the orders $b$ and $m$ of differentiation, and the problem's definition. 
See Section \ref{part:proofs}.\ref{sub:proof-of-lem:uniform-bound-on-BA-derivative-tensors} for its proof.
Note that the bound does \emph{not} depend on the variable $\beta$, even though $BA_\beta$ and its fixed points do depend on $\beta$.
For a derivative with respect to $\beta$ alone, $m = 0$, the bound is uniform on the entire simplex.

\subsection{Local Lipschitz constants of high-order implicit derivatives}
\label{sub:local-error-estimates-for-beta-derivs}

The previous Subsection \ref{sub:computational-difficulty-of-RTRD} shows that it need not be possible to bound Taylor method's local Lipschitz constants uniformly. In the same token, using a fixed order and step size in Algorithm \ref{algo:taylor-method-for-RD-root-tracking} is computationally inefficient. 
cf., Figure \ref{fig:7-points-example}, and Section \ref{part:how-and-what}.\ref{sub:efficient-RD-root-tracking} for improvements. 
For both purposes, it is useful to have an estimate of local Lipschitz constants, which we provide below.

\medskip
The computational inefficiency due to fixed order and step size (Section \ref{part:how-and-what}.\ref{sub:efficient-RD-root-tracking}) can be traced back to the variations in local Lipschitz constants, which explode at cluster-vanishing bifurcations.
Fitting a single value to all $\beta$ and $\bm{x}$ values is too conservative.
Instead, one could use $\|D_{\bm{x}} T_l \|_\infty \cdot \| \bm{x} - \bm{x}'\|_\infty$ to estimate an upper bound to $\| T_l(\bm{x}, \beta) - T_l(\bm{x}', \beta) \|_\infty$, if $\bm{x}$ and $\bm{x}'$ close enough.
This leads to a local estimate of the error's growth rate as in \eqref{eq:single-step-error-increase-in-Taylor-method}, up to replacing $L_l$ there with the matrix norm $\|D_{\bm{x}} T_l \rvert_{(\bm{x}_n, \beta_n)}\|_\infty$.
From the definition \eqref{eq:Taylor-poly-def-for-Taylor-method} of $T_l$, it is a sum of implicit derivative vectors $\tfrac{d^m \bm{x}}{d\beta^m}$ for $m = 1, \dots, l$. And so, to calculate the latter matrix norm it suffices to calculate the Jacobian matrices $D_{\bm{x}} \tfrac{d^m \bm{x}}{d\beta^m}$.
A direct calculation in Section \ref{part:proofs}.\ref{sub:proof-of-prop:Jacobian-of-high-order-beta-derivative} yields the formula below. 
While it involves many more summands than formula \eqref{eq:formula-for-high-order-beta-derivatives} for implicit derivatives (Theorem \ref{thm:formula-for-high-order-expansion-of-F-in-main-result-sect}), all the ingredients needed for $D_{\bm{x}} \tfrac{d^l \bm{x}}{d\beta^l}$ were already computed when calculating $\tfrac{d^l \bm{x}}{d\beta^l}$ (if $l \geq 2$). 

\begin{prop}			\label{prop:Jacobian-of-high-order-beta-derivative}
	Under the conditions of Theorem \ref{thm:formula-for-high-order-expansion-of-F-in-main-result-sect} (Section \ref{part:how-and-what}.\ref{sub:high-order-beta-derivatives-at-an-operator-root}), suppose further that the Jacobian matrix $D_{\bm{x}} F$ is invertible. Then, the Jacobian of the $l$-th order derivative is given by,
	\begin{equation}
		\begin{aligned}			\label{eq:formula-for-jacobian-of-high-order-beta-deriv}
			D_{\bm{x}} \tfrac{d^l \bm{x}}{d\beta^l} &= 
			- \left(D_{\bm{x}} F\right)^{-1} \left( D^2_{\bm{x}, \bm{x}} F \right) \tfrac{d^l \bm{x}}{d\beta^l} \\
			&- \left(D_{\bm{x}} F\right)^{-1} \sum_{\substack{\text{non-trivial} \\ \text{partitions}}}
			\sum_{b=0}^{m_1 \cdot \delta(p_1 = 1)}
			\frac{l!}{b! (m_1 - b)! m_2! \cdots m_s! \cdot (p_1!)^{m_1} \cdots (p_s!)^{m_s}} \\ &\cdot \Bigg\{ 
			D^m_{\beta^b, \bm{x}^{m-b {+1}}} F\Big[ 
			\left( \tfrac{d^{p_1}\bm{x}}{d\beta^{p_1}} \right)_{\times (m_1 - b)}, 
			\left( \tfrac{d^{p_2}\bm{x}}{d\beta^{p_2}} \right)_{\times m_2}, 
			\dots, \left( \tfrac{d^{p_s}\bm{x}}{d\beta^{p_s}} \right)_{\times m_s} \Big] 
			\\ &+
			{ (m_1 - b)} \cdot D^m_{\beta^b, \bm{x}^{m-b}} F\Big[ 
			{  D_{\bm{x}} \tfrac{d^{p_1}\bm{x}}{d\beta^{p_1}}},
			\left( \tfrac{d^{p_1}\bm{x}}{d\beta^{p_1}} \right)_{\times (m_1 - b {  - 1})}, 
			\left( \tfrac{d^{p_2}\bm{x}}{d\beta^{p_2}} \right)_{\times m_2}, 
			\dots, \left( \tfrac{d^{p_s}\bm{x}}{d\beta^{p_s}} \right)_{\times m_s} \Big] 
			\\ &+
			{ m_2} \cdot D^m_{\beta^b, \bm{x}^{m-b}} F\Big[ 
			\left( \tfrac{d^{p_1}\bm{x}}{d\beta^{p_1}} \right)_{\times (m_1 - b)}, 
			{ D_{\bm{x}} \tfrac{d^{p_2}\bm{x}}{d\beta^{p_2}} },
			\left( \tfrac{d^{p_2}\bm{x}}{d\beta^{p_2}} \right)_{\times (m_2 { - 1})}, 
			\dots, \left( \tfrac{d^{p_s}\bm{x}}{d\beta^{p_s}} \right)_{\times m_s} \Big] 
			\\ &+ \dots \\ &+
			{ m_s} \cdot D^m_{\beta^b, \bm{x}^{m-b}} F\Big[ 
			\left( \tfrac{d^{p_1}\bm{x}}{d\beta^{p_1}} \right)_{\times (m_1 - b)}, 
			\left( \tfrac{d^{p_2}\bm{x}}{d\beta^{p_2}} \right)_{\times m_2}, 
			\dots, { D_{\bm{x}} \tfrac{d^{p_s}\bm{x}}{d\beta^{p_s}}}, 
			\left( \tfrac{d^{p_s}\bm{x}}{d\beta^{p_s}} \right)_{\times (m_s { - 1})} \Big]
			\Bigg\}
		\end{aligned}
	\end{equation}
	where the summation is over non-trivial integer partitions $(m_1)\cdot p_1 + \dots + (m_s)\cdot p_s$ of $l$, and $m := m_1 + \dots + m_s$. 
\end{prop}

This proposition is the last building block needed for the proof of Theorem \ref{thm:taylor-method-converges-for-RD-root-tracking-away-of-bifurcation}, whose proof is brought at Section \ref{part:proofs}.\ref{sub:proof-of-thm:taylor-method-converges-for-RD-root-tracking-away-of-bifurcation}.
As with the implicit derivatives $\tfrac{d^l \bm{x}}{d\beta^l}$, the Jacobians $D_{\bm{x}} \tfrac{d^l \bm{x}}{d\beta^l}$ essentially contain $\left(D_{\bm{x}} F\right)^{-1}$ to the $(l+1)$-st power, due to the first term at \eqref{eq:formula-for-jacobian-of-high-order-beta-deriv}, and so lose their accuracy when approaching a bifurcation.
cf., the notes after Theorem \ref{thm:formula-for-high-order-expansion-of-F-in-main-result-sect} (in Section \ref{part:how-and-what}.\ref{sub:high-order-beta-derivatives-at-an-operator-root}).

\medskip
\section{Complexity of root-tracking and root-tracking for RD }

\label{sec:computational-complexities}

In this section, we present the main complexity results: of root-tracking and root-tracking for RD. 
We provide bounds for root-tracking both with and without tensor memorization. For RD, $F = Id - BA_\beta$ \eqref{eq:RD-operator-def}, the complexities of the derivatives tensors are broken down to their components, in Table \ref{tab:complexities-of-RD-deriv-tensors} below. 
Adding these to the complexities of root-tracking yields the complexity bounds of RD root tracking (Theorem \ref{thm:complexity-of-RD-root-tracking} in Section \ref{part:how-and-what}.\ref{sub:costs-and-error-to-cost-tradeoff-of-RD-root-tracking-in-main-results-section}). 
See Section \ref{part:proofs}.\ref{sec:proofs-for-complexity-at-appendix} for proofs of the below.
\medskip

Recall, $p(n)$ is the number of partitions of an integer $n$, with no restriction. One may restrict the number of parts of which a partition is comprised (its total multiplicity). We write $p_{\leq k}(n)$ for the number of partitions when no more than $k$ parts are allowed. e.g., $p_{\leq n}(n) = p(n)$ follows directly from the definition. 
See Equation \eqref{eq:recurrence-formula-for-partition-w-exact-part-num} ff. in Section \ref{sub:complexity-of-high-order-derivs-for-root-tracking} for details. 
Write $C(b, m)$ for the computational complexity of a derivatives tensor $D^{m}_{\beta^b, \bm{x}^{m-b}} F$. 
Where, $F(\cdot, \beta)$ is an unspecified operator on $\bb{R}^T$, as in Section \ref{part:how-and-what}.\ref{sub:beta-derivs-at-an-operator-root}. 
With this, the complexity of Algorithm \ref{algo:high-order-derivs-of-operator-roots} for computing implicit derivatives at an operator's root (Section \ref{part:how-and-what}.\ref{sub:high-order-beta-derivatives-at-an-operator-root}) is as follows (proofs in Section \ref{sub:complexity-of-high-order-derivs-for-root-tracking}).

\begin{prop}[Complexity of $l$-th order implicit derivative]
	\label{prop:computational-cost-of-high-order-derivs-when-all-tensors-are-evaluated}
	Assume that the derivatives $\tfrac{d^k \bm{x}}{d\beta^k}$ are known for all $0 < k < l$, and let $C(b, m)$ the complexity of a derivatives tensor $D^{m}_{\beta^b, \bm{x}^{m-b}} F$.
	Then, the computational complexity of the $l$-th order implicit derivative formula \eqref{eq:formula-for-high-order-beta-derivatives} for $\tfrac{d^l \bm{x}}{d\beta^l}$ is 
	\begin{equation}		\label{eq:computational-complexity-of-root-tracking-no-memorization}
		O(T^3) +
		O(T\cdot \sum_{j=1}^l p(j)) +
		\sum_{m=0}^l \sum_{b=0}^m p_{\leq m-b}(l-m) \Big[ O((m-b+1) T^{m-b+1}) + C(b, m) \Big]
	\end{equation}
\end{prop}
The last summand in \eqref{eq:computational-complexity-of-root-tracking-no-memorization} with $C(b, m)$ stands for the cost of computing the derivative tensors, the third for evaluating the multilinear forms they define, the second for summing over the evaluated forms, and the first for finding a linear pre-image under $D_{\bm{x}}F$.

As seen already by the first few expansion orders \eqref{eq:first-order-expansion-of-operator-eq-implicit}-\eqref{eq:third-order-expansion-implicit} of $\tfrac{d^k F}{d\beta^k} = 0$ (in Section \ref{part:how-and-what}.\ref{sub:beta-derivs-at-an-operator-root}), derivative tensors are often re-used after their first appearance. 
Hence, it makes sense to memorize computed tensors so that they are computed only once. 
This is especially true when the computational costs $C(b, m)$ of derivative tensors $D^{m}_{\beta^b, \bm{x}^{m-b}} F$ are high, as in rate distortion problems.
cf., the example in Section \ref{part:how-and-what}.\ref{subsub:lines-intersecting-parabola-example} in contrast.

\begin{prop}[Complexity of $l$-th order implicit derivative, with tensor memorization]
	\label{prop:computational-cost-of-high-order-derivs-when-tensors-are-memorized}
	Assume that the derivatives $\tfrac{d^k \bm{x}}{d\beta^k}$ are known for all $0 < k < l$, and let $C(b, m)$ the complexity of calculating a derivative tensor $D^{m}_{\beta^b, \bm{x}^{m-b}} F$. Assume further that all the derivative tensors $D^{m}_{\beta^b, \bm{x}^{m-b}} F$ with $m < l$ have already been computed.
	Then, the computational complexity of formula \eqref{eq:formula-for-high-order-beta-derivatives} for $\tfrac{d^l \bm{x}}{d\beta^l}$ is 
	\begin{equation}		\label{eq:computational-costs-in-prop-for-recursive-formula}
		O(T^3) +
		O(T\cdot \sum_{j=1}^l p(j)) +
		\sum_{m=0}^l \sum_{b=0}^m p_{\leq m-b}(l-m) O((m-b+1) T^{m-b+1}) +
		\sum_{b=0}^l C(b, l)
	\end{equation}
	The memory complexity of storing all the derivative tensors $D^{m}_{\beta^b, \bm{x}^{m-b}} F$ with $0 \leq b \leq m < l$ is
	\begin{equation}		\label{eq:memory-costs-in-prop-for-recursive-formula}
		O\Big( l \cdot T^l \Big)
	\end{equation}
\end{prop}

Comparing the computational complexity \eqref{eq:computational-costs-in-prop-for-recursive-formula} to its counterpart \eqref{eq:computational-complexity-of-root-tracking-no-memorization} without tensor memorization (Proposition \ref{prop:computational-cost-of-high-order-derivs-when-all-tensors-are-evaluated}), the number of tensors computed for the $l$-th order derivative is sliced from $\sum_{j=1}^l p(j)$ to just the $l+1$ newly needed tensors. cf., Corollary \ref{cor:two-forms-for-the-total-number-of-summands-at-recursive-formula-for-beta-derivs} in Section \ref{part:proofs}.\ref{sub:complexity-of-high-order-derivs-for-root-tracking}.

Finally, the computational complexity of \textit{all} the implicit derivatives up to order $L$ (including) can be bounded as following:

\begin{prop}[Cumulative complexity of implicit derivatives, with tensor memorization]
	\label{prop:total-computational-complexity-with-memorization}
	Under the conditions of Proposition \ref{prop:computational-cost-of-high-order-derivs-when-tensors-are-memorized}, the computational complexity of all the implicit derivatives $\tfrac{d^k \bm{x}}{d\beta^k}$ for $0 < k \leq L$ is
	\vspace*{-9pt}
	\begin{equation}		\label{eq:total-computational-complexity-in-prop}
		O\left( e^{\nicefrac{9}{4} \cdot \ln L + (L+1) \ln T + \pi\sqrt{\nicefrac{2L}{3}}} \right) + 
		\sum_{l=1}^L \sum_{b=0}^l C(b, l) \;,
		\vspace*{-3pt}
	\end{equation}
	when $L \geq 2$. When $L = 1$, it is
	\begin{equation}		\label{eq:total-computational-complexity-in-prop-for-first-order}
		O\left( T^3 \right) + 
		C(0, 1) + C(1, 1) \;.
	\end{equation}
	The memory complexity is as at \eqref{eq:memory-costs-in-prop-for-recursive-formula}.
\end{prop}

For RD, our operator is the $M$-dimensional $Id - BA_\beta$ \eqref{eq:RD-operator-def} (namely, $T = M$), defined via $BA_\beta$ \eqref{eq:BA-operator-def} in marginal coordinates (Section \ref{part:how-and-what}.\ref{sub:high-order-deriv-tensors-of-BA}). 
Where, $N := |\mathcal{X}|$ and $M := |\hat{\mathcal{X}}|$ are the source and reproduction alphabet sizes of the given RD problem. 
Now that we have complexity results for root tracking, Table \ref{tab:complexities-of-RD-deriv-tensors} summarizes the complexity of computing RD derivative tensors (see Section \ref{part:proofs}.\ref{sub:complexity-of-deriv-tensors-for-RD}). 
Combining the complexities of both (in Section \ref{part:proofs}.\ref{sub:complexity-of-RD-root-tracking}) yields the complexity bounds for RD root tracking at Theorem \ref{thm:complexity-of-RD-root-tracking} (Section \ref{part:how-and-what}.\ref{sub:costs-and-error-to-cost-tradeoff-of-RD-root-tracking-in-main-results-section}). 
We comment that the hyper-exponential terms $2^k k!$ in Table \ref{tab:complexities-of-RD-deriv-tensors} are only exponential in practice. 
These result from the loose bounds we have provided for the algebraic properties of the $P_k$ polynomials \eqref{eq:P_0-def}-\eqref{eq:P_k-inductive-def}. 
In particular, the bounds provided here depend the dimensions of the RD problem at hand, but \textit{not} on its details.

\begin{table}[h!]
	\begin{center}
		\setlength{\tabcolsep}{8pt}
		\renewcommand{\arraystretch}{1.6}
		\begin{tabular}{ccc}
			Quantity	&	Computations	&	Memory	\\
			\hline
			$P_k$ \eqref{eq:P_0-def}-\eqref{eq:P_k-inductive-def}	&
			Irrelevant	&
			$O\left( 2^k k! (k+1) \log_2 (k+1) \right)$  	\\[7pt]
			$\expectedDxWRTencoderK{k}$ \eqref{eq:expected-k-th-power-distortion-def}	&
			$O(MN)$		&		$O(N)$			\\[7pt]
			$P_k[\intermediateencoderVect; d]$ \eqref{eq:P_k-by-abuse-of-notation}	&
			$O(MNk \; 2^k k!)$		&
			$O(MN)$								\\[7pt]
			$G\big(k, a\big)$ \eqref{eq:combinatorial-G-in-terms-of-polynomials}	&
			$O(MNL \sum_{k=0}^L p(k))$	&
			$O(MN (L + 1) (L + 2) )$			\\[7pt]
			$D^b_{\beta^b} (Id - BA_\beta)[\inputmarginalVect]$ \eqref{eq:repeated-beta-deriv-in-thm} 	&
			$O(MN)$		&		$O(M)$	\\[7pt]
			$D^{b+m}_{\beta^b, \inputmarginalVect^m} (Id - BA_\beta)[\inputmarginalVect]$ \eqref{eq:mixed-BA-deriv-in-thm} &
			$\begin{aligned}
				O(\genfrac{(}{)}{0pt}{1}{m + M - 1}{m} \genfrac{(}{)}{0pt}{1}{b + M - 1}{b} \cdot M^2 N) \\ 
				+ O(\genfrac{(}{)}{0pt}{1}{m + M - 1}{m} m! M) \\
			\end{aligned}$	&
			$O(M^{m+1})$
		\end{tabular}
	\end{center}
	\vspace{-5pt}
	\caption{
		\textbf{Breakdown of the complexity of RD derivative tensors}. 
		For implicit derivatives up to order $L$ (including), one needs to compute $P_k$, $\expectedDxWRTencoderK{k}$ and $P_k[\intermediateencoderVect; d]$ for $k \leq L$; and the matrices $G(k, a)$. The derivative tensors for $b \leq L$ and for $b+m \leq L$ can then be computed. 
		The hyper-exponential terms $2^k k!$ at $P_k$ and at $P_k[\intermediateencoderVect; d]$ are very loose, and are roughly exponential in practice. 
		For $k \leq 25$, the memory needed to store $P_k$ is roughly $O(1.73^k)$, while the computational complexity of $P_k[\intermediateencoderVect; d]$ is roughly $O(MNk \; 1.56^k)$; see \ref{subsub:problem-ind-initial-computations} and \ref{subsub:initial computations-at-a-point} for details.
	}
	\label{tab:complexities-of-RD-deriv-tensors}
\end{table}

\ifdefined\compileappendices
\newpage
\appendix

\part{Proofs and technical details}
\label{part:proofs}

\medskip
\section{Calculations for the line-intersecting-parabola example}
\label{sec:calculations-for-lines-intersecting-parabola-example-appendix}

We elaborate on the calculations for the example in \ref{subsub:lines-intersecting-parabola-example}, in Part \ref{part:how-and-what}.

\medskip
Write $\bm{p} := \left(x, y\right)^t$ for the coordinates, $\tfrac{d\bm{p}}{d\beta}$ for the vector of derivatives $(\tfrac{dx}{d\beta}, \tfrac{dy}{d\beta})^t$. For $i = 1, 2$ and vectors $\bm{u} = (u_x, u_y), \bm{v} = (v_x, v_y)$,
\begin{equation}			\label{eq:line-parabola-example-calc-2nd-order-term}
	D_{\bm{p}, \bm{p}}F_i[\bm{u}, \bm{v}] = 
	\frac{\partial^2 F_i}{\partial x \partial y} \cdot \left( u_x v_y + u_y v_x\right) + 
	\frac{\partial^2 F_i}{\partial x^2} \cdot u_x v_x + 
	\frac{\partial^2 F_i}{\partial y^2} \cdot u_y v_y
\end{equation}
From the definition \eqref{eq:line-intersecting-parabola-example}, $\frac{\partial^2 F_1}{\partial x^2} = 2b$ is the only derivative of $F$ at \eqref{eq:line-parabola-example-calc-2nd-order-term} which does not vanish, and so 
\begin{equation}
	D_{\bm{p}, \bm{p}}F[\tfrac{d\bm{p}}{d\beta}, \tfrac{d\bm{p}}{d\beta}] = 
	\left(\dbeta{x}\right)^2 \cdot \mat{2 b \\ 0} \;,
\end{equation}
where $\tfrac{dx}{d\beta}$ is the first coordinate of $\tfrac{d\bm{p}}{d\beta}$. Similarly, $D_{\bm{p}, \bm{p}}F[\tfrac{d^2\bm{p}}{d\beta^2}, \tfrac{d\bm{p}}{d\beta}] = \dbeta{x} \dbetaK{x}{2} \cdot (2 b, 0)^t $.
Thus, the first few expansion orders \eqref{eq:first-order-expansion-of-operator-eq-implicit}-\eqref{eq:third-order-expansion-implicit} around $(x_0, y_0; \beta_0)$ are,
\begin{align}
	0 &= \mat{2 b x_0 + c &  -1 \\ a & -1} \frac{d\bm{p}}{d\beta} + \mat{0 \\ 1}		\\
	0 &= \mat{2 b x_0 + c &  -1 \\ a & -1} \frac{d^2\bm{p}}{d\beta^2} + 
	\left(\dbeta{x}\right)^2 \cdot \mat{2 b \\ 0} \\
	0 &= \underset{D_{\bm{p}}F}{\underbrace{\mat{2 b x_0 + c &  -1 \\ a & -1}}} \frac{d^3\bm{p}}{d\beta^3} + 
	3 \dbeta{x} \dbetaK{x}{2} \cdot \mat{2 b \\ 0} 
\end{align}
while the other derivative tensors at \eqref{eq:first-order-expansion-of-operator-eq-implicit}-\eqref{eq:third-order-expansion-implicit} vanish.

The Jacobian $D_{\bm{p}}F$ is invertible whenever its determinant does not vanish, which is to say that the slope $2 b x_0 + c$ of the parabola at the intersection point differs from the slope $a$ of the line. Its inverse is then $(D_{\bm{p}}F)^{-1} = \frac{1}{\Delta} \mat{1 & -1 \\ a & -2 b x_0 - c}$, where $\Delta := 2 b x_0 + c - a$. A straightforward calculation yields, 
\begin{align}
	\frac{d\bm{p}}{d\beta}		&= \frac{1}{\Delta} \mat{1 \\ a +\Delta}, \label{eq:line-parabola-example-1st-ord-deriv} \\
	\frac{d^2\bm{p}}{d\beta^2}	&= -\frac{2 b}{\Delta^3} \mat{1 \\ a}, \label{eq:line-parabola-example-2nd-ord-deriv} \\
	\frac{d^3\bm{p}}{d\beta^3}	&= \frac{12 b^2}{\Delta^5} \mat{1 \\ a}, \quad \text{and} \label{eq:line-parabola-example-3rd-ord-deriv} \\
	\frac{d^4\bm{p}}{d\beta^4}	&= -\frac{120 b^3}{\Delta^7} \mat{1 \\ a} \;,		\label{eq:line-parabola-example-4th-ord-deriv}
\end{align}
where the fourth-order derivative \eqref{eq:line-parabola-example-4th-ord-deriv} follows by a similar calculation. 
Combining these yields the fourth-order Taylor expansion \eqref{eq:line-parabola-example-fourth-order-sol}, in Section \ref{subsub:lines-intersecting-parabola-example}. 


By requiring $x(\beta_0) = x_0$, one can see that $\Delta^2$ is the discriminant of the polynomial which defines the exact solution \eqref{eq:line-parabola-example-exact-solution}. Therefore, $\Delta$ vanishes if and only if \eqref{eq:line-parabola-example-exact-solution} has exactly one solution, which is to say that $F = 0$ \eqref{eq:line-intersecting-parabola-example} undergoes a bifurcation. For this example, this is also equivalent the Jacobian of $F$ being singular. 

	%

\medskip
\section{Proofs for high-order implicit derivatives of an operator's root}
\label{sec:proof-of-high-ord-beta-deriv-of-F-appendix}

\medskip
\subsection{Proof of the formula for an operator's high-order $\beta$-expansion, Theorem \ref{thm:formula-for-high-order-expansion-of-F}}
\label{sub:proof-of-formula-for-high-order-expansion-of-F}

We prove formula \eqref{eq:formula-for-high-order-expansion-of-F} (of Theorem \ref{thm:formula-for-high-order-expansion-of-F}) for a the expansion of $\tfrac{d^l }{d\beta^l} F(\bm{x}(\beta), \beta)$. The preliminaries for this subsection are provided in Section \ref{part:details}.\ref{sec:multivariate-faa-di-brunos-formula}.

\medskip
To get a grip on the proof, we calculate $\dbetaK{F_i}{l}$ for $l=1$ and $2$ directly from Fa\`a di Bruno's formula \cite{ma2009higher}, recovering the first two expansion orders \eqref{eq:first-order-expansion-of-operator-eq-implicit}-\eqref{eq:second-order-expansion-of-operator-eq-implicit} in Section \ref{part:how-and-what}.\ref{sub:beta-derivs-at-an-operator-root}. 
Using the notation of Section \ref{part:details}.\ref{sec:multivariate-faa-di-brunos-formula}, write $\cal{D}$ for $(T+1)$-decompositions of an integer $l > 0$.
Recall, a $(T+1)$-decomposition $(s, \bm{p}, \bm{m})$ of $l$ to $s$ parts, $1 \leq s \leq l$, is comprised of parts $0 < p_1 < \dots < p_s \in \bb{N}$ and multiplicities $\bm{m}_1, \dots, \bm{m}_s \in \bb{N}_0^{T+1}$, satisfying the decomposition equation \eqref{eq:mFDB-decomposition-equation},
\begin{equation}		\label{eq:mFDB-decomposition-eq-for-F-proof}
	l = |\bm{m}_1| p_1 + |\bm{m}_2| p_2 + \dots + |\bm{m}_s| p_s \;.
\end{equation}
We denote the $j$-th coordinate of the vector $\bm{m}_k$ by $m_{k, j}$, where $k = 1, \dots, s$ and $j = 0, \dots, T$.
Write $\bm{m} := \bm{m}_1 + \dots + \bm{m}_s \in \bb{N}_0^{T+1}$ for the total multiplicity \eqref{eq:mFDB-total-multiplicity-is-sum-of-multiplicities}.
Considering $F_i$ as a function of $\big(\beta, \bm{x}\big)\in \bb{R}^{T+1}$,
we index its $\beta$-coordinate by zero, and those of $\bm{x}$ by $1, \dots, T$. For a multi-index $\bm{m}\in \bb{N}_0^{T+1}$, write $m_0$ for its zeroth coordinate and $\bm{m}_+$ for its other $T$ entries, as above.
With this, the proof below of Theorem \ref{thm:formula-for-high-order-expansion-of-F} starts with
\begin{equation}		\tag{\ref{eq:high-order-beta-deriv-of-F-in-proof}}
	\dbetaK{}{l}F_i(\bm{x}(\beta), \beta) =
	l! \sum_{(s, \bm{p}, \bm{m})\in\cal{D}} \frac{\partial^{|\bm{m}|} F_i}{\partial \beta^{m_0} \partial \bm{x}^{\bm{m}_+}}
	\prod_{k=1}^s \frac{1}{\bm{m}_k!} \left[ \frac{1}{p_k!} \dbetaK{\beta}{p_k} \right]^{m_{k,0}}
	\left[ \frac{1}{p_k!} \dbetaK{\bm{x}}{p_k} \right]^{\bm{m}_{k+}} 
\end{equation}

Set $l = 1$. By the decomposition equation \eqref{eq:mFDB-decomposition-eq-for-F-proof}, there is only one part of size $p_1 = 1$, and multiplicity $|\bm{m}_1| = 1$. In particular, we must have $\bm{m}_1 = \bm{e}_j$, for $\bm{e}_j$ a standard basis vector. Where, either $\bm{m}_1 = \bm{e}_0$ points to the $\beta$-coordinates, or $\bm{m}_1 = \bm{e}_j$ points to one of the $\bm{x}$-coordinates, $j = 1, \dots, T$.	Since these are all the $(T + 1)$-decompositions for $l = 1$, 
\begin{equation}
	\dbeta{F_i} \overset{\eqref{eq:high-order-beta-deriv-of-F-in-proof}}{=}
	\sum_{\substack{\bm{m} = \bm{e}_j: \\ j=0, 1, \dots, T}} \frac{\partial F_i}{\partial \beta^{m_{0}} \partial \bm{x}^{\bm{m}_{+}}}
	\left[ \dbeta{\beta} \right]^{m_{0}} \left[ \dbeta{\bm{x}} \right]^{\bm{m}_{+}} =
	\partialbeta{F_i} + \sum_{j=1}^T \frac{\partial F_i}{\partial x_j} \dbeta{x_j} \;.
\end{equation}
This recovers the implicit first-order expansion $\tfrac{d^1 F}{d\beta^1} = D_{\bm{x}} F[\tfrac{d\bm{x}}{d\beta}] + D_\beta F$ \eqref{eq:first-order-expansion-of-operator-eq-implicit}.

For $l = 2$, setting $p_1 = 2$ again necessitates $\bm{m}_1 = \bm{e}_j$ (one part of size 2). Since $p_1 > 1$, the term $\left[ \frac{1}{p_1!} \dbetaK{\beta}{p_1} \right]^{m_{1,0}}$ at \eqref{eq:high-order-beta-deriv-of-F-in-proof} vanishes unless $m_{1,0} = 0$, and so only decompositions with $\bm{m}_1 = \bm{e}_j$ for $j > 0$ contribute. This yields the term $D_x F_i[\tfrac{d^2 \bm{x}}{d\beta^2}]$ in the implicit second-order expansion \eqref{eq:second-order-expansion-of-operator-eq-implicit},
\begin{equation}
	2! \sum_{j=1}^T \frac{\partial F_i}{\partial x_j} \left[ \frac{1}{2!} \dbetaK{x_j}{2} \right] \;.
\end{equation}
When $p_1 = 1$, the decomposition equation \eqref{eq:mFDB-decomposition-eq-for-F-proof} implies that there is only $s = 1$ part, of multiplicity $|\bm{m}_1| = 2$. Since $p_1 = 1$ we do not have the above restriction on the value of $m_{1,0}$, and so $\bm{m}_1 = \bm{e}_j + \bm{e}_k$, for $0 \leq j \leq k \leq T$. Further sub-dividing to the cases $j = k = 0, j = 0 < k$ and $0 < j \leq k$, we obtain 
\begin{multline}
	2! \partialbetaK{F_i}{2} \frac{1}{2!} \left[ 1 \right]^2 +
	2! \sum_{k=1}^T \frac{\partial^2 F_i}{\partial \beta \partial x_k} \dbeta{x_k} \\ +
	2! \sum_{0 < j < k \leq T} \frac{\partial^2 F_i}{\partial x_j \partial x_k} \dbeta{x_j} \dbeta{x_k} +
	2! \sum_{0 < j = k \leq T} \frac{1}{2!} \frac{\partial^2 F_i}{\partial x_j \partial x_k} \dbeta{x_j} \dbeta{x_k}
\end{multline}
The first two terms account respectively for $D^2_{\beta, \beta} F_i$ and $2 D^2_{\beta, x} F_i [\tfrac{d\bm{x}}{d\beta}]$ at \eqref{eq:second-order-expansion-of-operator-eq-implicit}, while the last line sums up to the remaining term $D^2_{x, x} F_i[\tfrac{d\bm{x}}{d\beta}, \tfrac{d\bm{x}}{d\beta}] = \sum_{j, k} \tfrac{\partial^2 F_i}{\partial x_j \partial x_k} \dbeta{x_j} \dbeta{x_k}$ there. This can either be seen directly, or as a special case of the following lemma, on rearranging summation order.

\begin{lem}		\label{lem:summation-reoder-M-ball-on-grid-vs-all-M-tuples}
	Let $f$ be a function on $\bb{N}_0^T$, and $0 < M \in \bb{N}$. Then,
	\begin{equation}
		\sum_{\bm{m}\in \bb{N}_0^T: \; |\bm{m}| = M} \frac{1}{\bm{m}!} f(\bm{m}) = 
		\frac{1}{M!} \sum_{1 \leq i_1, \dots, i_M \leq T} f(\bm{e}_{i_1} + \dots + \bm{e}_{i_M})
	\end{equation}
\end{lem}

\begin{proof}[Proof of Lemma \ref{lem:summation-reoder-M-ball-on-grid-vs-all-M-tuples}]
	We would like to rewrite the sum over $\{\bm{m}\in \bb{N}_0^T: |\bm{m}| = M\}$ as a sum over all $M$-tuples $1 \leq i_1, \dots, i_M \leq T$. Write $\bm{m}\in \bb{N}_0^T$ with $|\bm{m}| = M$ as a sum
	\begin{equation}		\label{eq:multi-index-as-sum-of-standard-basis-vectors}
		\bm{m} = \bm{e}_{i_1} + \dots + \bm{e}_{i_M}
	\end{equation}
	of standard basis vectors $\bm{e}_i$, with each index $1 \leq i_j \leq T$ for $j = 1, \dots, M$. When the indices $i_1, \dots, i_M$ are distinct, then permuting them yields a different $M$-tuple of indices, without affecting their sum $\bm{m}$ at \eqref{eq:multi-index-as-sum-of-standard-basis-vectors}. A single $\bm{m}$ value then corresponds to $M!$ distinct permutation of $(i_1, \dots, i_M)$. The indices are distinct if and only if no coordinate of $\bm{m}$ is larger than 1, $m_k \leq 1 \; \forall k=1, \dots, T$. So,
	\begin{equation}
		\sum_{\substack{\bm{m}\in \bb{N}_0^T: \; |\bm{m}| = M,\\ \forall k, \; m_k \leq 1}} M!\; f(\bm{m}) = 
		\sum_{\substack{1 \leq i_1, \dots, i_M \leq T:\\ i_1, \dots, i_M \text{ are distinct} }} f(\bm{e}_{i_1} + \dots + \bm{e}_{i_M})
	\end{equation}
	When exactly two indices $i_{j_1}$ and $i_{j_2}$ for $j_1 \neq j_2$ are the same $i_{j_1} = i_{j_2}$, with all the others distinct, then swapping $i_{j_1}$ with $i_{j_2}$ leaves the $M$-tuple $(i_1, \dots, i_M)$ unchanged. Hence, there are only $\nicefrac{M!}{2!}$ distinct $M$-tuples corresponding to the sum $\bm{m}$ \eqref{eq:multi-index-as-sum-of-standard-basis-vectors}. Since $i_{j_1} = i_{j_2}$ are the only identical indices, then $m_{k} = 2$ for $k := i_{j_1}$ and otherwise $m_k \leq 1$.
	Proceeding in this manner, one has in general 
	\begin{equation}
		\sum_{\bm{m}\in \bb{N}_0^T: \; |\bm{m}| = M} \frac{M!}{\bm{m}!} f(\bm{m}) =
		\sum_{1 \leq i_1, \dots, i_M \leq T} f(\bm{e}_{i_1} + \dots + \bm{e}_{i_M})
	\end{equation}
\end{proof}

Lemma \ref{lem:summation-reoder-M-ball-on-grid-vs-all-M-tuples} will allow us to exchange summations over high-dimensional decompositions $(s, \bm{p}, \bm{m})$ with summations over integer partitions and multi-linear differentials of $F_i$.

\begin{proof}[Proof of Theorem \ref{thm:formula-for-high-order-expansion-of-F}]
	The $i$-th coordinate of $F$ can be written as a composition
	\begin{equation}
		\xymatrix@C=3pc{
			\beta \in \bb{R} \ar[r] 								&
			\big( \beta, \bm{x}(\beta) \big)\in \bb{R}^{T + 1} \ar[r]	&
			F_i\big(\bm{x}(\beta), \beta\big) \in \bb{R} \;.
		}
	\end{equation}
	We shall fully expand the $l$-th order derivative $\dbetaK{F_i}{l}$ of $F_i\big(\bm{x}(\beta), \beta\big)$ with respect to $\beta$, using the multivariate Fa\`a di Bruno's formula \cite{ma2009higher} (Theorem \ref{thm:mFDB-from-Ma} in Section \ref{part:details}.\ref{sec:multivariate-faa-di-brunos-formula}).
	By the formula,
	\begin{multline}		\label{eq:high-order-beta-deriv-of-F-in-proof}
		\dbetaK{F_i}{l} \overset{\eqref{eq:mFDB-Ma-concise-form}}{=}
		l! \sum_{(s, \bm{p}, \bm{m})\in\cal{D}} \frac{\partial^{|\bm{m}|} F_i}{\partial \beta^{m_0} \partial \bm{x}^{\bm{m}_+}}
		\prod_{k=1}^s \frac{1}{\bm{m}_k!} \left[ \frac{1}{p_k!} \dbetaK{(\beta, \bm{x})}{p_k} \right]^{\bm{m}_k} 
		\\ \overset{\eqref{eq:multivariate-notation-defs}}{=}
		l! \sum_{(s, \bm{p}, \bm{m})\in\cal{D}} \frac{\partial^{|\bm{m}|} F_i}{\partial \beta^{m_0} \partial \bm{x}^{\bm{m}_+}}
		\prod_{k=1}^s \frac{1}{\bm{m}_k!} \left[ \frac{1}{p_k!} \dbetaK{\beta}{p_k} \right]^{m_{k,0}}
		\left[ \frac{1}{p_k!} \dbetaK{\bm{x}}{p_k} \right]^{\bm{m}_{k+}} 
	\end{multline}
	where $\bm{m} := \bm{m}_1 + \dots + \bm{m}_s \in \bb{N}_0^{T+1}$ is the total multiplicity \eqref{eq:mFDB-total-multiplicity-is-sum-of-multiplicities}, and the last equality since we index the $\beta$-coordinate of $\bb{R}^{T+1}$ by 0, and those of $\bm{x}$ by $1, \dots, T$; see \eqref{eq:mFDB-decomposition-eq-for-F-proof} ff.
	
	With the notation preceding Theorem \ref{thm:mFDB-from-Ma}, the parts $p_k$ of a decomposition are positive integers, strictly increasing in $k = 1, \dots, s$. Hence, only the first part can be of size 1, $p_1 \geq 1$, while the others are of size 2 at least: $p_k > 2$ for $k>1$.
	Note that the same differential operator $\dbetaK{}{p_k}$ appears for all the $T+1$ coordinates of $\big(\beta, \bm{x}(\beta)\big)$ at \eqref{eq:high-order-beta-deriv-of-F-in-proof}. At the zeroth entry of the $k$-th part, the multiplicand $\left[ \frac{1}{p_k!} \dbetaK{\beta}{p_k} \right]^{m_{k,0}}$ vanishes unless either $m_{k,0} = 0$ or $p_k = 1$. Since $p_k$ can be 1 only for $k=1$, the latter is to say $p_1 = 1$. In particular, for the $k$-th summand to contribute, $m_{k,0}$ must vanish when $k > 1$.
	
	To simplify notation, write $b$ instead of $m_{1, 0}$. When a summand does not vanish, $b$ is the number of times $F_i$ is differentiated with respect to $\beta$ (since $m_{k,0}=0$ except perhaps for $k = 1$).
	With these observations, decompositions that contribute to \eqref{eq:high-order-beta-deriv-of-F-in-proof} are as follows:
	\begin{enumerate}
		\item Decompositions with $p_1 > 1$ and $b = 0$.						\label{type:decompositions-wout-beta}
		\item Decompositions with $p_1 = 1$ and any $b = 0, \dots, |\bm{m}_1|$.		\label{type:decompositions-with-beta}
	\end{enumerate}
	
	To proceed, write $M_1 := |\bm{m}_1|, \dots, M_s := |\bm{m}_s|$. The decomposition equation \eqref{eq:mFDB-decomposition-eq-for-F-proof} can then be read as a partition of the integer $l$, to $M_1$ sets of sizes $p_1$ up to $M_s$ sets of size $p_s$ (and no sets of other sizes),
	\begin{equation}		\label{eq:decomposition-as-partition-eq-in-proof}
		l = M_1 \cdot p_1 + \dots + M_s \cdot p_s \;.
	\end{equation}
	With this, the summation over $(s, \bm{p}, \bm{m})\in \cal{D}$ in \eqref{eq:high-order-beta-deriv-of-F-in-proof} can be rewritten as a sum over integer partitions of $l$. Given a partition \eqref{eq:decomposition-as-partition-eq-in-proof} to parts $p_1, \dots, p_s$ of respective multiplicities $M_1, \dots, M_s$, we need to sum over all the integral vectors $\bm{m}_1, \dots, \bm{m}_s\in \bb{N}_0^{T+1}$ with $|\bm{m}_1| = M_1, \dots, |\bm{m}_s| = M_s$. This yields,
	\begin{multline}		\label{eq:rewriting-decomposition-as-sum-over-partits-in-proof}
		\dbetaK{F_i}{l} \overset{\eqref{eq:high-order-beta-deriv-of-F-in-proof}}{=}
		l! \sum_{(s, \bm{p}, \bm{m})\in\cal{D}} \frac{\partial^{|\bm{m}|} F_i}{\partial \beta^{m_0} \partial \bm{x}^{\bm{m}_+}}
		\prod_{k=1}^s \frac{1}{\bm{m}_k!} \left[ \frac{1}{p_k!} \dbetaK{\beta}{p_k} \right]^{m_{k,0}}
		\left[ \frac{1}{p_k!} \dbetaK{\bm{x}}{p_k} \right]^{\bm{m}_{k+}} \\ =
		l! \sum_{\substack{\text{partitions} \\ \eqref{eq:decomposition-as-partition-eq-in-proof} \text{ of } l }} \; \sum_{\bm{m}_1: |\bm{m}_1| = M_1} \dots \sum_{\bm{m}_s: |\bm{m}_s| = M_s}
		\frac{\partial^{|\bm{m}|} F_i}{\partial \beta^{b} \; \partial \bm{x}^{\bm{m}_+}}
		\prod_{k=1}^s \frac{1}{\bm{m}_k! \; p_k!^{|\bm{m}_k|}} \left[ \dbetaK{\beta}{p_k} \right]^{m_{k,0}}
		\left[ \dbetaK{\bm{x}}{p_k} \right]^{\bm{m}_{k+}} 
	\end{multline}
	where in the last line we replaced $m_{1, 0}$ by $b$, using the fact that $m_{k, 0}$ must vanish for $k > 1$, and so $m_0 = m_{1, 0} + m_{2, 0} + \dots + m_{s, 0} = b$.
	
	The summation over $\bm{m}_1$ can be broken to two, isolating its zeroth component $b$ from the rest, as $|\bm{m}_1| = b + |\bm{m}_{1+}|$ by definition. 
	Recall that decompositions with $b > 0$ contribute only when $p_1 = 1$, as in \ref{type:decompositions-with-beta}, so
	\begin{equation}		\label{eq:breaking-sum-over-m1}
		\sum_{\bm{m}_1: |\bm{m}_1| = M_1} \dots \; = \;
		\sum_{b = 0}^{M_1 \cdot \delta(p_1 = 1)} \; \sum_{\bm{m}_{1+}: |\bm{m}_{1+}| = M_1 - b} \dots
	\end{equation}
	Other than $m_{1,0} =: b$, $m_{k, 0}$ always vanishes. So, we lose nothing by replacing our $(T+1)$-dimensional multiplicity vectors $\bm{m}_k$ with smaller ones $\bm{m}_{k+} \in \bb{N}_0^T$. So,
	\begin{multline}			\label{eq:breaking-summation-over-beta-deriv-from-m1-in-proof}
		\dbetaK{F_i}{l} \underset{\eqref{eq:breaking-sum-over-m1}}{\overset{\eqref{eq:rewriting-decomposition-as-sum-over-partits-in-proof}}{=}}
		\sum_{\text{partitions of } l} \sum_{b = 0}^{M_1 \cdot \delta(p_1 = 1)} 
		\frac{l!}{b! \; (p_1!)^{m_1} \cdots (p_s!)^{m_s}}
		\sum_{\bm{m}_{1+}: |\bm{m}_{1+}| = M_1 - b} \;
		\sum_{\bm{m}_{2+}: |\bm{m}_{2+}| = M_2} 
		\cdots \\ \cdots \sum_{\bm{m}_{s+}: |\bm{m}_{s+}| = M_s}
		\frac{\partial^{|\bm{m}|} F_i}{\partial \beta^{b} \; \partial x^{\bm{m}_+}}
		\prod_{k=1}^s \frac{1}{\bm{m}_{k+}! } \left[ \dbetaK{\bm{x}}{p_k} \right]^{\bm{m}_{k+}} 
	\end{multline}
	where the $\nicefrac{1}{b!}$ coefficient in the first line is due to $\bm{m}_1! = b! \cdot \bm{m}_{1+}!$, by definition \eqref{eq:multivariate-notation-defs}.
	
	Next, for $s > 1$, the derivative term $\partial^{|\bm{m}|}F_i$ in the last line of \eqref{eq:breaking-summation-over-beta-deriv-from-m1-in-proof} can be written as
	\begin{equation}		\label{eq:isolating-irrelevant-derivs-in-proof}
		\frac{\partial^{|\bm{m}|} F_i}{\partial \beta^{b} \; \partial \bm{x}^{\bm{m}_+}} =
		\frac{\partial^{M_s}}{\partial \bm{x}^{\bm{m}_{s+}}} \frac{\partial^{|\bm{m}| - M_s} F_i}{\partial \beta^{b} \; \partial \bm{x}^{\bm{m}_{1+}} \cdots \partial \bm{x}^{\bm{m}_{(s-1)+}}} =: 
		\frac{\partial^{M_s}}{\partial \bm{x}^{\bm{m}_{s+}}} G \;.
	\end{equation}
	While $G$ itself does \textit{not} depend on the multi-index $\bm{m}_{s+}$, we consider $\frac{\partial^{M_s}}{\partial \bm{x}^{\bm{m}_{s+}}} G$ as a function of $\bm{m}_{s+}$.
	This allows us to invoke Lemma \ref{lem:summation-reoder-M-ball-on-grid-vs-all-M-tuples} on the last summation at \eqref{eq:breaking-summation-over-beta-deriv-from-m1-in-proof}, as follows:
	\begin{multline}		\label{eq:applying-the-lemma-to-change-summation-order}
		\sum_{\bm{m}_{s+}: |\bm{m}_{s+}| = M_s}
		\frac{\partial^{|\bm{m}|} F_i}{\partial \beta^{b} \; \partial \bm{x}^{\bm{m}_+}}
		\prod_{k=1}^s \frac{1}{\bm{m}_{k+}! } \left[ \dbetaK{\bm{x}}{p_k} \right]^{\bm{m}_{k+}} 
		\\ \overset{\eqref{eq:isolating-irrelevant-derivs-in-proof}}{=}
		\left(\prod_{k=1}^{s-1} \frac{1}{\bm{m}_{k+}! } \left[ \dbetaK{\bm{x}}{p_k} \right]^{\bm{m}_{k+}} \right) \cdot
		\sum_{\bm{m}_{s+}: |\bm{m}_{s+}| = M_s} \frac{1}{\bm{m}_{s+}! } \cdot 
		\frac{\partial^{M_s} G}{\partial \bm{x}^{\bm{m}_{s+}}} \left[ \dbetaK{\bm{x}}{p_s} \right]^{\bm{m}_{s+}} 
		\\ \overset{\text{Lemma } \ref{lem:summation-reoder-M-ball-on-grid-vs-all-M-tuples}}{=} \;
		\left(\prod_{k=1}^{s-1} \frac{1}{\bm{m}_{k+}! } \left[ \dbetaK{\bm{x}}{p_k} \right]^{\bm{m}_{k+}} \right) \cdot
		\frac{1}{M_S!} \sum_{1 \leq i_1, \dots, i_{M_s} \leq T} 
		\frac{\partial^{M_s} G}{\partial x_{i_1} \cdots \partial x_{i_{M_s}}} \cdot
		\dbetaK{x_{i_1}}{p_s} \cdots \dbetaK{x_{i_{M_s}}}{p_s} \\ \overset{\eqref{eq:mixed-deriv-def-evaluated-applied-to-vectors}}{=}
		\frac{1}{M_S!} \left(\prod_{k=1}^{s-1} \frac{1}{\bm{m}_{k+}! } \left[ \dbetaK{\bm{x}}{p_k} \right]^{\bm{m}_{k+}} \right) \cdot
		D^{M_s}_{\underset{M_s \text{ times}}{\underbrace{\bm{x}, \dots, \bm{x}}}} \; G\big[\underset{M_s \text{ times}}{\underbrace{\tfrac{d^{p_s} \bm{x}}{d\beta^{p_s}}, \dots, \tfrac{d^{p_s} \bm{x}}{d\beta^{p_s}}}}\big]
	\end{multline}
	where the last line is by the definition of a multivariate derivative tensor, \eqref{eq:mixed-deriv-def-evaluated-applied-to-vectors} in Section \ref{part:how-and-what}.\ref{sub:beta-derivs-at-an-operator-root}.
	
	The exact same manipulations as in \eqref{eq:applying-the-lemma-to-change-summation-order} can be applied to all the other partition parts $k = 1, \dots, s-1$, with one caveat. For $k=1$, $|\bm{m}_{1+}|$ is $M_1 - b$ rather than $M_1$, while the remaining $b$ degrees are consumed by the derivative with respect to $\beta$, as can be seen from \eqref{eq:isolating-irrelevant-derivs-in-proof}. 
	Thus, the coefficient at the last line of \eqref{eq:applying-the-lemma-to-change-summation-order} ends being $\left( (M_1 - b)! \cdot M_2! \cdots M_s! \right)^{-1}$, and the differentiation corresponding to $k=1$ is $b$-times with respect to $\beta$, and only $(M_1 - b)$ times with respect to $\bm{x}$. In particular, it is then an $(M_1 - b)$-multilinear form, involving $\tfrac{d^{p_1} \bm{x}}{d\beta^{p_1}}$ only $(M_1 - b)$ times. 
	Gathering these back to \eqref{eq:breaking-summation-over-beta-deriv-from-m1-in-proof} for all parts $k = 1, \dots, s$ completes the proof, yielding \eqref{eq:formula-for-high-order-expansion-of-F}.
\end{proof}

\medskip
\subsection{Proof of the formula for the derivative's Jacobian, Proposition \ref{prop:Jacobian-of-high-order-beta-derivative}}
\label{sub:proof-of-prop:Jacobian-of-high-order-beta-derivative}

\begin{proof}[Proof of Proposition \ref{prop:Jacobian-of-high-order-beta-derivative}]
	To avoid clutter, we write $D_{\bm{x}} F^{-1}$ in this proof for the inverse of the Jacobian matrix $D_{\bm{x}} F$.
	Rewrite formula \eqref{eq:formula-for-high-order-beta-derivatives} (Theorem \ref{thm:formula-for-high-order-expansion-of-F-in-main-result-sect}) compactly as,
	\begin{equation}		\label{eq:short-notation-for-high-order-beta-deriv}
		D_{\bm{x}} F \; \tfrac{d^l \bm{x}}{d\beta^l} = - S
	\end{equation}
	where for short,
	\begin{equation}		\label{eq:short-notation-for-summand-over-partits-in-formula-for-high-ord-beta-derivs}
		S := \sum_{\substack{\text{non-trivial} \\ \text{partitions}}}
		\sum_{b=0}^{m_1 \cdot \delta(p_1 = 1)} C \cdot 
		D^m_{\beta^b, \bm{x}^{m-b}} F\Big[ 
		\left( \tfrac{d^{p_1}\bm{x}}{d\beta^{p_1}} \right)_{\times (m_1 - b)}, 
		\left( \tfrac{d^{p_2}\bm{x}}{d\beta^{p_2}} \right)_{\times m_2}, 
		\dots, \left( \tfrac{d^{p_s}\bm{x}}{d\beta^{p_s}} \right)_{\times m_s} \Big] \;,
	\end{equation}
	and $C$ stands for the coefficient $\frac{l!}{b! (m_1 - b)! m_2! \cdots m_s! \cdot (p_1!)^{m_1} \cdots (p_s!)^{m_s}}$ at \eqref{eq:formula-for-high-order-beta-derivatives}. Differentiating both sides of \eqref{eq:short-notation-for-high-order-beta-deriv} with respect to the coordinates $\bm{x}$,
	\begin{equation}
		D^2_{\bm{x}, \bm{x}} F \; \tfrac{d^l \bm{x}}{d\beta^l} +
		D_{\bm{x}} F \; D_{\bm{x}} \tfrac{d^l \bm{x}}{d\beta^l} = 
		- D_{\bm{x}} S
	\end{equation}
	That is,
	\begin{equation}		\label{eq:after-handling-deriv-of-inverse-in-proof}
		D_{\bm{x}} \tfrac{d^l \bm{x}}{d\beta^l} =
		- D_{\bm{x}} F^{-1} \left( D^2_{\bm{x}, \bm{x}} F \right) \tfrac{d^l \bm{x}}{d\beta^l}
		- D_{\bm{x}} F^{-1} D_{\bm{x}} S
	\end{equation}
	To complete the proof, it suffices to calculate $D_{\bm{x}} S$.
	
	The differentiation of a single addend in $S$ \eqref{eq:short-notation-for-summand-over-partits-in-formula-for-high-ord-beta-derivs} is a sum of 2-tensors (matrices). Each addend in this sum involves extra $\bm{x}$-differentiations: once of $F$ itself, and once for each of its $m$ arguments. Differentiating a single addend,
	\begin{multline}
		D_{\bm{x}} \left( D^m_{\beta^b, \bm{x}^{m-b}} F\Big[ 
		\left( \tfrac{d^{p_1}\bm{x}}{d\beta^{p_1}} \right)_{\times (m_1 - b)}, 
		\left( \tfrac{d^{p_2}\bm{x}}{d\beta^{p_2}} \right)_{\times m_2}, 
		\dots, \left( \tfrac{d^{p_s}\bm{x}}{d\beta^{p_s}} \right)_{\times m_s} \Big] \right) 
		\\ =
		D^m_{\beta^b, \bm{x}^{m-b {+1}}} F\Big[ 
		\left( \tfrac{d^{p_1}\bm{x}}{d\beta^{p_1}} \right)_{\times (m_1 - b)}, 
		\left( \tfrac{d^{p_2}\bm{x}}{d\beta^{p_2}} \right)_{\times m_2}, 
		\dots, \left( \tfrac{d^{p_s}\bm{x}}{d\beta^{p_s}} \right)_{\times m_s} \Big] 
		\\ +
		D^m_{\beta^b, \bm{x}^{m-b}} F\Big[ { D_{\bm{x}} \tfrac{d^{p_1}\bm{x}}{d\beta^{p_1}}}, 
		\underset{(m_1 - b - 1) \text{ times}}{\underbrace{\tfrac{d^{p_1}\bm{x}}{d\beta^{p_1}}, \dots, \tfrac{d^{p_1}\bm{x}}{d\beta^{p_1}}}},
		\left( \tfrac{d^{p_2}\bm{x}}{d\beta^{p_2}} \right)_{\times m_2}, 
		\dots, \left( \tfrac{d^{p_s}\bm{x}}{d\beta^{p_s}} \right)_{\times m_s} \Big] 
		\\ +
		D^m_{\beta^b, \bm{x}^{m-b}} F\Big[ \; \underset{\text{once}}{\underbrace{\tfrac{d^{p_1}\bm{x}}{d\beta^{p_1}}}}, {  D_{\bm{x}} \tfrac{d^{p_1}\bm{x}}{d\beta^{p_1}}}, 
		\underset{(m_1 - b - 2) \text{ times}}{\underbrace{\tfrac{d^{p_1}\bm{x}}{d\beta^{p_1}}, \dots, \tfrac{d^{p_1}\bm{x}}{d\beta^{p_1}}}},
		\left( \tfrac{d^{p_2}\bm{x}}{d\beta^{p_2}} \right)_{\times m_2}, 
		\dots, \left( \tfrac{d^{p_s}\bm{x}}{d\beta^{p_s}} \right)_{\times m_s} \Big]
		+ \dots \\ +
		D^m_{\beta^b, \bm{x}^{m-b}} F\Big[ \underset{(m_1 - b - 1) \text{ times}}{\underbrace{\tfrac{d^{p_1}\bm{x}}{d\beta^{p_1}}, \dots, \tfrac{d^{p_1}\bm{x}}{d\beta^{p_1}}}}, {  D_{\bm{x}} \tfrac{d^{p_1}\bm{x}}{d\beta^{p_1}}}, 
		\left( \tfrac{d^{p_2}\bm{x}}{d\beta^{p_2}} \right)_{\times m_2}, 
		\dots, \left( \tfrac{d^{p_s}\bm{x}}{d\beta^{p_s}} \right)_{\times m_s} \Big] 
		\\ + \dots \\ +
		D^m_{\beta^b, \bm{x}^{m-b}} F\Big[ 
		\left( \tfrac{d^{p_1}\bm{x}}{d\beta^{p_1}} \right)_{\times (m_1 - b)}, 
		\left( \tfrac{d^{p_2}\bm{x}}{d\beta^{p_2}} \right)_{\times m_2}, 
		\dots, { D_{\bm{x}} \tfrac{d^{p_s}\bm{x}}{d\beta^{p_s}}}, \underset{m_s - 1 \text{ times}}{\underbrace{\tfrac{d^{p_s}\bm{x}}{d\beta^{p_s}}, \dots, \tfrac{d^{p_s}\bm{x}}{d\beta^{p_s}}}} \Big] 
		\\ +
		D^m_{\beta^b, \bm{x}^{m-b}} F\Big[ 
		\left( \tfrac{d^{p_1}\bm{x}}{d\beta^{p_1}} \right)_{\times (m_1 - b)}, 
		\left( \tfrac{d^{p_2}\bm{x}}{d\beta^{p_2}} \right)_{\times m_2}, 
		\dots, \underset{\text{once}}{\underbrace{\tfrac{d^{p_s}\bm{x}}{d\beta^{p_s}}}},
		{ D_{\bm{x}} \tfrac{d^{p_s}\bm{x}}{d\beta^{p_s}}}, \underset{m_s - 2 \text{ times}}{\underbrace{\tfrac{d^{p_s}\bm{x}}{d\beta^{p_s}}, \dots, \tfrac{d^{p_s}\bm{x}}{d\beta^{p_s}}}} \Big] 
		+ \dots \\ +
		D^m_{\beta^b, \bm{x}^{m-b}} F\Big[ 
		\left( \tfrac{d^{p_1}\bm{x}}{d\beta^{p_1}} \right)_{\times (m_1 - b)}, 
		\left( \tfrac{d^{p_2}\bm{x}}{d\beta^{p_2}} \right)_{\times m_2}, 
		\dots, \underset{m_s - 1 \text{ times}}{\underbrace{\tfrac{d^{p_s}\bm{x}}{d\beta^{p_s}}, \dots, \tfrac{d^{p_s}\bm{x}}{d\beta^{p_s}}}}, {  D_{\bm{x}} \tfrac{d^{p_s}\bm{x}}{d\beta^{p_s}}} \Big] 
	\end{multline}
	Since derivative tensors of high-order are symmetric (e.g., \cite[Section 10.3]{aguilar2021analysis}), permuting its arguments has no effect. So, the above simplifies to
	\begin{multline}		\label{eq:x-deriv-of-summand-in-proof}
		D_{\bm{x}} \left( D^m_{\beta^b, \bm{x}^{m-b}} F\Big[ 
		\left( \tfrac{d^{p_1}\bm{x}}{d\beta^{p_1}} \right)_{\times (m_1 - b)}, 
		\left( \tfrac{d^{p_2}\bm{x}}{d\beta^{p_2}} \right)_{\times m_2}, 
		\dots, \left( \tfrac{d^{p_s}\bm{x}}{d\beta^{p_s}} \right)_{\times m_s} \Big] \right) 
		\\ =
		D^m_{\beta^b, \bm{x}^{m-b {+1}}} F\Big[ 
		\left( \tfrac{d^{p_1}\bm{x}}{d\beta^{p_1}} \right)_{\times (m_1 - b)}, 
		\left( \tfrac{d^{p_2}\bm{x}}{d\beta^{p_2}} \right)_{\times m_2}, 
		\dots, \left( \tfrac{d^{p_s}\bm{x}}{d\beta^{p_s}} \right)_{\times m_s} \Big] 
		\\ +
		{ (m_1 - b)} \cdot D^m_{\beta^b, \bm{x}^{m-b}} F\Big[ 
		{  D_{\bm{x}} \tfrac{d^{p_1}\bm{x}}{d\beta^{p_1}}},
		\left( \tfrac{d^{p_1}\bm{x}}{d\beta^{p_1}} \right)_{\times (m_1 - b {  - 1})}, 
		\left( \tfrac{d^{p_2}\bm{x}}{d\beta^{p_2}} \right)_{\times m_2}, 
		\dots, \left( \tfrac{d^{p_s}\bm{x}}{d\beta^{p_s}} \right)_{\times m_s} \Big] 
		\\ +
		{ m_2} \cdot D^m_{\beta^b, \bm{x}^{m-b}} F\Big[ 
		\left( \tfrac{d^{p_1}\bm{x}}{d\beta^{p_1}} \right)_{\times (m_1 - b)}, 
		{ D_{\bm{x}} \tfrac{d^{p_2}\bm{x}}{d\beta^{p_2}} },
		\left( \tfrac{d^{p_2}\bm{x}}{d\beta^{p_2}} \right)_{\times (m_2 { - 1})}, 
		\dots, \left( \tfrac{d^{p_s}\bm{x}}{d\beta^{p_s}} \right)_{\times m_s} \Big] 
		\\ + \dots \\ +
		{ m_s} \cdot D^m_{\beta^b, \bm{x}^{m-b}} F\Big[ 
		\left( \tfrac{d^{p_1}\bm{x}}{d\beta^{p_1}} \right)_{\times (m_1 - b)}, 
		\left( \tfrac{d^{p_2}\bm{x}}{d\beta^{p_2}} \right)_{\times m_2}, 
		\dots, { D_{\bm{x}} \tfrac{d^{p_s}\bm{x}}{d\beta^{p_s}}}, 
		\left( \tfrac{d^{p_s}\bm{x}}{d\beta^{p_s}} \right)_{\times (m_s { - 1})} \Big] 
	\end{multline}
	Combining the latter \eqref{eq:x-deriv-of-summand-in-proof} with the definition \eqref{eq:short-notation-for-summand-over-partits-in-formula-for-high-ord-beta-derivs} of $S$ and Equation \eqref{eq:after-handling-deriv-of-inverse-in-proof} yields the required result \eqref{eq:formula-for-jacobian-of-high-order-beta-deriv}.
\end{proof}

\medskip
\section{Derivations of high-order derivatives of the Blahut-Arimoto operator}
\label{sec:appendix-derivations-of-BA-high-order-deriv-in-marginal-coords}

In this section, we calculated derivatives of the Blahut-Arimoto operators, mainly those presented in Section \ref{sec:high-order-derivs-of-BA-in-marginal-coords}.

\medskip
\subsection{Proof of Proposition \ref{prop:repeated-encoder-deriv-wrt-coords}, formula for the encoder's repeated marginal derivatives}

\label{sub:proof-of-prop:repeated-encoder-deriv-wrt-coords}

To prove Proposition \ref{prop:repeated-encoder-deriv-wrt-coords}, we rewrite it in an equivalent form which is more convenient for proofs by induction. Rather than writing high-order derivatives in multi-index notation, one could write it as a sequence of differentiations. Setting $M = 2$ for example, a third-order derivative represented by $\bm{\alpha}_+ = \left(2, 1\right)$ can be written equivalently as
\begin{equation}		\label{eq:equivalent-forms-to-consider-a-high-order-deriv}
	\frac{\partial^{3}}{\partial \inputmarginalVect^{(2, 1)}} \intermediateencoder{}{} =
	\frac{\partial^{3}}{\partial \inputmarginal{_1} \partial \inputmarginal{_1} \partial \inputmarginal{_2}} \intermediateencoder{}{} \;.
\end{equation}
While the left-hand side of \eqref{eq:equivalent-forms-to-consider-a-high-order-deriv} is understood \eqref{eq:multivariate-notation-defs} as an application of the differential operator $\left(\frac{\partial}{\partial\inputmarginal{_1}}\right)^2 \cdot \left(\frac{\partial}{\partial\inputmarginal{_2}}\right)$, its right-hand side can be considered as a sequence of differentiations: first differentiate with respect to $\inputmarginal{_2}$, and then twice with respect to $\inputmarginal{_1}$. 

Write $\left<\cdot, \cdot\right>$ for the usual scalar product on $\bb{R}^M$, $\bm{e}_{\hat{x}_i}$ for the $i$-th standard basis vector. Then, $\left<\bm{\alpha}_+, \bm{e}_{\hat{x}_i}\right> = \alpha_i$ is the number differentiations with respect to the $i$-th coordinate $\inputmarginal{_i}$, and $|\bm{\alpha}_+|$ the number of differentiations in total. 
When differentiating with respect to $\inputmarginal{_{i_1}}, \inputmarginal{_{i_2}}$ up to $\inputmarginal{_{i_k}}$, the total number $\alpha_{\hat{x}'}$ of differentiations with respect to a particular coordinate $\hat{x}'$ can be written as $\sum_{j=1}^k \delta_{\hat{x}', \hat{x}_{i_j}} = \delta_{\hat{x}', \hat{x}_{i_1}} + \dots + \delta_{\hat{x}', \hat{x}_{i_k}}$.
Thus, Proposition \ref{prop:repeated-encoder-deriv-wrt-coords} is equivalent to the following.

\begin{prop}			\label{prop:repeated-encoder-deriv-wrt-coords-in-appendix}
	For $k > 0$, the repeated encoder \eqref{eq:encoder-eq} derivative with respect to $\inputmarginalVect$ is,
	\begin{equation}		\label{eq:repeated-encoder-deriv-wrt-coords-in-prop-in-appendix}
		\frac{\partial^k}{\partial \inputmarginal{_{i_1}} \cdots \partial \inputmarginal{_{i_k}}} \intermediateencoder{}{} =
		\frac{(-1)^{k-1} (k-1)! \; e^{-\beta \sum_{j=1}^k d(x, \hat{x}_{i_j})}  }{Z^k(x, \beta)}\cdot \left[
		\sum_{j=1}^k \delta_{\hat{x}, \hat{x}_{i_j}} 
		- k \cdot\intermediateencoder{}{}
		\right]
	\end{equation}
	where $1 \leq i_1, \dots, i_k\leq M$ need \emph{not} be distinct. 
\end{prop}

\begin{proof}[Proof of Proposition \ref{prop:repeated-encoder-deriv-wrt-coords-in-appendix}]
	We prove \eqref{eq:repeated-encoder-deriv-wrt-coords-in-prop-in-appendix} by induction on $k > 0$. First, note that
	\begin{equation}		\label{eq:first-partial-deriv-of-partition-func}
		\frac{\partial}{\partial \inputmarginal{}} Z(x, \beta) =
		\frac{\partial}{\partial \inputmarginal{}} \sum_{\hat{x}'} \inputmarginal{'} e^{-\beta d(x, \hat{x}')} =
		e^{-\beta d(x, \hat{x})}
	\end{equation}
	So, for $k > 0$ we have
	\begin{equation}		\label{eq:deriv-of-partition-func-inv-power}
		\frac{\partial}{\partial \inputmarginal{}} \frac{1}{Z^k(x, \beta)} \overset{\eqref{eq:first-partial-deriv-of-partition-func}}{=}
		- \frac{k e^{-\beta d(x, \hat{x})}}{Z^{k+1}(x, \beta)}
	\end{equation}
	Hence,
	\begin{equation}		\label{eq:deriv-of-partition-func-inv-with-coord}
		\frac{\partial}{\partial \inputmarginal{'}} \left( \frac{\inputmarginal{} }{Z^k(x, \beta)} \right) =
		\frac{\delta_{\hat{x}, \hat{x}'} }{Z^k(x, \beta)} - \frac{k e^{-\beta d(x, \hat{x}')}}{Z^{k+1}(x, \beta)} \cdot \inputmarginal{}
	\end{equation}
	Thus, for the first derivative of the encoder \eqref{eq:encoder-eq} we have
	\begin{multline}		\label{eq:first-order-deriv-of-encoder-by-marginal}
		\frac{\partial}{\partial \inputmarginal{_{i_1}}} \intermediateencoder{}{} \overset{\eqref{eq:encoder-eq}}{=}
		e^{-\beta d(x, \hat{x})}\cdot \frac{\partial}{\partial \inputmarginal{_{i_1}}} \left( \frac{\inputmarginal{} }{Z(x, \beta)} \right) 
		\\ \overset{\eqref{eq:deriv-of-partition-func-inv-with-coord}}{=} 
		e^{-\beta d(x, \hat{x})}\cdot \left\{ \frac{\delta_{\hat{x}, \hat{x}_{i_1}} }{Z(x, \beta)} -
		\frac{e^{-\beta d(x, \hat{x}_{i_1})}}{Z^2(x, \beta)} \cdot \inputmarginal{} \right\} \\ =
		\frac{e^{-\beta d(x, \hat{x}_{i_1})}}{Z(x, \beta)}\cdot \bigg\{ e^{-\beta d(x, \hat{x})}\cdot \frac{\delta_{\hat{x}, \hat{x}_{i_1}} }{e^{-\beta d(x, \hat{x}_{i_1})}} -
		\frac{e^{-\beta d(x, \hat{x})} }{Z(x, \beta)} \cdot \inputmarginal{} \bigg\} 
		\\ \overset{\eqref{eq:encoder-eq}}{=}
		\frac{e^{-\beta d(x, \hat{x}_{i_1})}}{Z(x, \beta)}\cdot \bigg\{ \delta_{\hat{x}, \hat{x}_{i_1}} -
		q(\hat{x}|x) \bigg\}
	\end{multline}
	This is the induction basis $k = 1$ for \eqref{eq:repeated-encoder-deriv-wrt-coords-in-prop-in-appendix}.
	
	Next, assume that \eqref{eq:repeated-encoder-deriv-wrt-coords-in-prop-in-appendix} holds for any derivative up to order $k$. Then, for $k + 1$ we have,
	\begin{multline}
		\frac{\partial}{\partial \inputmarginal{_{i_{k+1}}}} \frac{\partial^k}{\partial \inputmarginal{_{i_1}} \cdots \partial \inputmarginal{_{i_k}}} \intermediateencoder{}{} \\ \overset{\eqref{eq:repeated-encoder-deriv-wrt-coords-in-prop-in-appendix}}{=}
		\frac{\partial}{\partial \inputmarginal{_{i_{k+1}}}} \left[ 
		\frac{(-1)^{k-1} (k-1)! \; e^{-\beta \sum_{j=1}^k d(x, \hat{x}_{i_j})} }{Z^k(x, \beta)}\cdot \left\{ 
		\sum_{j=1}^k \delta_{\hat{x}, \hat{x}_{i_j}}
		- k \cdot \intermediateencoder{}{}
		\right\} \right]
		\\ \overset{\eqref{eq:deriv-of-partition-func-inv-power}}{\underset{\eqref{eq:first-order-deriv-of-encoder-by-marginal}}{=}} 
		\frac{(-1)^{k} \; k! \; e^{-\beta \sum_{j=1}^k d(x, \hat{x}_{i_j})} }{Z^{k+1}(x, \beta)} \cdot e^{-\beta d(x, \hat{x}_{i_{k+1}})} \cdot  \left\{ 
		\sum_{j=1}^k \delta_{\hat{x}, \hat{x}_{i_j}}
		- k \cdot \intermediateencoder{}{}
		\right\} \\ +
		\frac{(-1)^{k-1} (k-1)! \; e^{-\beta \sum_{j=1}^k d(x, \hat{x}_{i_j})} }{Z^k(x, \beta)} \cdot (-k) \cdot 
		\frac{e^{-\beta d(x, \hat{x}_{i_{k+1}})}}{Z(x, \beta)}\cdot \bigg\{ \delta_{\hat{x}, \hat{x}_{i_{k+1}}} -
		q(\hat{x}|x) \bigg\} \\ =
		\frac{(-1)^{k} \; k! \; e^{-\beta \sum_{j=1}^{k+1} d(x, \hat{x}_{i_j})} }{Z^{k+1}(x, \beta)} \cdot \left\{ 
		\sum_{j=1}^{k+1} \delta_{\hat{x}, \hat{x}_{i_j}}
		- (k + 1) \cdot \intermediateencoder{}{}
		\right\}
	\end{multline}
	This completes the proof of the induction step.
\end{proof}

\medskip
\subsection{Proof of Proposition \ref{prop:repeated-beta-derivatives-of-encoder}, formula for the encoder's repeated $\beta$-derivatives}

\label{sub:proof-of-prop:repeated-beta-derivatives-of-encoder}

\begin{proof}[Proof of Proposition \ref{prop:repeated-beta-derivatives-of-encoder}]
	Note that,
	\begin{multline}		\label{eq:Z-deriv-wrt-beta}
		\partialbeta{} e^{-\beta d(x, \hat{x})} = -d(x, \hat{x}) e^{-\beta d(x, \hat{x})} \\ \Longrightarrow \quad
		\partialbeta{} Z(x, \beta) = \partialbeta{} \sum_{\hat{x}'} \inputmarginal{'} e^{-\beta d(x, \hat{x}')} =
		-\sum_{\hat{x}'} d(x, \hat{x}') \inputmarginal{'} e^{-\beta d(x, \hat{x}')}
	\end{multline}
	So, 
	\begin{multline}		\label{eq:Z-inverse-deriv-wrt-beta}
		\partialbeta{} \frac{1}{Z(x, \beta)} \overset{\eqref{eq:Z-deriv-wrt-beta}}{=}
		\frac{1}{Z(x, \beta)} \cdot \sum_{\hat{x}'} d(x, \hat{x}') \frac{\inputmarginal{'} e^{-\beta d(x, \hat{x}')}}{Z(x, \beta)} \overset{\eqref{eq:encoder-eq}}{=}
		\frac{1}{Z(x, \beta)} \cdot \sum_{\hat{x}'} \intermediateencoder{'}{} d(x, \hat{x}') 
	\end{multline}
	Thus, for the encoder's first-order $\beta$-derivative,
	\begin{multline}		\label{eq:first-beta-deriv-of-encoder-in-proof}
		\partialbeta{\intermediateencoder{}{}} \overset{\eqref{eq:encoder-eq}}{=} 
		\inputmarginal{}\cdot \partialbeta{} \frac{e^{-\beta d(x, \hat{x})}}{Z(x, \beta)}
		\underset{\eqref{eq:Z-inverse-deriv-wrt-beta}}{\overset{\eqref{eq:Z-deriv-wrt-beta}}{=}}
		\inputmarginal{}\cdot \left[
		\frac{e^{-\beta d(x, \hat{x})}}{Z(x, \beta)} \cdot \sum_{\hat{x}'} \intermediateencoder{'}{} d(x, \hat{x}') -
		d(x, \hat{x})\cdot \frac{e^{-\beta d(x, \hat{x})}}{Z(x, \beta)} 
		\right]
		\\ \overset{\eqref{eq:encoder-eq}}{=}
		\intermediateencoder{}{}\cdot \sum_{\hat{x}'} \Big( \intermediateencoder{'}{} - \delta_{\hat{x}, \hat{x}'} \Big) d(x, \hat{x}') =
		\intermediateencoder{}{}\cdot \Big(
		\expectedDxWRTencoder - d(x, \hat{x})
		\Big)
	\end{multline}
	where we have denoted $\expectedDxWRTencoderK{k} := \sum_{\hat{x}'} q(\hat{x}'|x) d(x, \hat{x}')^k$, for $k > 0$. Writing $x_0$ for $d(x, \hat{x})$ and $x_k$ for $\expectedDxWRTencoderK{k}$ when $k > 0$, this proves the first-order version of formula \eqref{eq:formula-for-repeated-beta-deriv-in-prop-recursive-form}, with $P_1(x_0, x_1) := x_1 - x_0$ \eqref{eq:P_1_derived}.
	
	Unlike the encoder $\intermediateencoderVect$, the distortion $d(x, \hat{x})$ does \emph{not} depend on $\beta$. So, for $k > 0$,
	\begin{multline}		\label{eq:beta-deriv-of-high-order-d-moments}
		\partialbeta{} \expectedDxWRTencoderK{k} =
		\sum_{\hat{x}'} d(x, \hat{x}')^k \; \partialbeta{\intermediateencoder{'}{}} \\
		\overset{\eqref{eq:first-beta-deriv-of-encoder-in-proof}}{=}
		\sum_{\hat{x}'} d(x, \hat{x}')^k \; \intermediateencoder{'}{} \cdot 
		\sum_{\hat{x}''} \big( \intermediateencoder{''}{} - \delta_{\hat{x}', \hat{x}''}\big) d(x, \hat{x}'') \\ =
		\sum_{\hat{x}', \hat{x}''} d(x, \hat{x}')^k \; d(x, \hat{x}'')
		\intermediateencoder{'}{} \big( \intermediateencoder{''}{} - \delta_{\hat{x}', \hat{x}''}\big) \\ =
		\expectedDxWRTencoder \cdot \expectedDxWRTencoderK{k} -\expectedDxWRTencoderK{k+1}
	\end{multline}
	That is, when writing $x_k := \expectedDxWRTencoderK{k}$ for $k > 0$, the $x_k$'s satisfy the relations $\dbar x_k = x_1\cdot x_k - x_{k+1}$ for $k > 0$ and $\dbar x_0 = 0$ \eqref{eq:variable-deriv-def-for-recursive-beta-formula}, where $\dbar$ is written in place of $\nicefrac{\partial}{\partial \beta}$ to emphasize the algebraic properties of this definition. 
	
	Next, assuming that the derivative formula \eqref{eq:formula-for-repeated-beta-deriv-in-prop-recursive-form} holds for $k$, we prove it for $k + 1$. With $x_0$ and $x_k$ for $k >0$ as before,
	\begin{multline}
		\partialbetaK{\intermediateencoder{}{}}{k+1} =
		\partialbeta{} \partialbetaK{\intermediateencoder{}{}}{k} \overset{\eqref{eq:formula-for-repeated-beta-deriv-in-prop-recursive-form}}{=} 
		\partialbeta{} \Big[\intermediateencoder{}{} \cdot P_k\Big( x_0, x_1, \dots, x_k\Big)\Big] \\ =
		\partialbeta{\intermediateencoder{}{}} \cdot P_k\Big( x_0, x_1, \dots, x_k\Big) +
		\intermediateencoder{}{} \cdot \tfrac{\partial}{\partial\beta} P_k\Big( x_0, x_1, \dots, x_k\Big) 
		\\ \overset{\eqref{eq:first-beta-deriv-of-encoder-in-proof}}{=}
		\intermediateencoder{}{} \cdot (x_1 - x_0) \cdot P_k\Big( x_0, x_1, \dots, x_k\Big) +
		\intermediateencoder{}{} \cdot \dbar P_k\Big( x_0, x_1, \dots, x_k\Big) \\ =
		\intermediateencoder{}{} \cdot \Big( (x_1 - x_0) \cdot P_k  + \dbar P_k \Big)
	\end{multline}
	Where we have replaced $\nicefrac{\partial}{\partial \beta}$ by $\dbar$, as before.
	So, setting $P_{k+1} := (x_1 - x_0) \cdot P_k + \dbar P_k$ \eqref{eq:P_k-inductive-def} completes the proof.
\end{proof}

\medskip
\subsection{Proof of Corollary \ref{cor:BA-beta-deriv}, for the partial $\beta$-derivative of $BA_\beta$}

\label{sub:proof-of-cor:beta-deriv-of-BA-op}

\begin{proof}[Proof of Corollary \ref{cor:BA-beta-deriv}]
	The encoder's first partial $\beta$-derivative is 
	\begin{equation}		\label{eq:first-beta-deriv-of-encoder}
		\partialbeta{\intermediateencoder{}{}} 
		\underset{\eqref{eq:P_1_derived}}{\overset{\eqref{eq:formula-for-repeated-beta-deriv-in-prop-recursive-form}}{=}}
		\intermediateencoder{}{} \cdot \big( \expectedDxWRTencoder - d(x, \hat{x}) \big)
	\end{equation}
	Plugging this back into formula \eqref{eq:high-order-BA-deriv-only-enc-deriv-implicit} for the derivative of $BA_\beta$,
	\begin{multline}
		\frac{\partial BA_\beta\left[r\right](\hat{x})}{\partial \beta}
		\overset{\eqref{eq:high-order-BA-deriv-only-enc-deriv-implicit}}{=}
		\sum_{x} p_X(x) \; \frac{\partial \intermediateencoder{}{}}{\partial \beta }
		\overset{\eqref{eq:first-beta-deriv-of-encoder}}{=}
		\sum_{x} p_X(x) \; \left[ 
		\sum_{\hat{x}'} \intermediateencoder{}{} \intermediateencoder{'}{} d(x, \hat{x}') -
		\intermediateencoder{}{} d(x, \hat{x})
		\right] \\ =
		\bb{E}_{\intermediateencoder{'}{}p_X(x)} \left[\intermediateencoder{}{} d(x, \hat{x}') \right]
		- \bb{E}_{p_X(x)} \left[\intermediateencoder{}{} d(x, \hat{x}) \right]
	\end{multline}
	Since the identity operator does not depend on $\beta$, this yields the result.
\end{proof}

\medskip
\subsection{Proof of Proposition \ref{prop:mixed-high-order-enc-deriv}, formula for mixed high-order encoder derivatives}

\label{sub:proof-of-prop:mixed-high-order-enc-deriv}

\begin{proof}[Proof of Proposition \ref{prop:mixed-high-order-enc-deriv}]
	Let $\bm{\alpha} = (\alpha_0, \bm{\alpha}_+) \in \bb{N}_0^{M+1}$ with $\bm{\alpha}_+ \neq 0$ be given, and an input marginal $\inputmarginalVect$ outside the simplex boundary, $\forall \hat{x} \; \inputmarginal{} > 0$.
	Let $\intermediateencoderVect$ be the encoder defined by $\inputmarginalVect$ \eqref{eq:encoder-eq}.
	For a fixed $x$ coordinate, consider $d(x, \hat{x})$ and $\intermediateencoder{}{}$ as $\hat{x}$-indexed vectors. 
	Thus, by the multivariate vector-power notation in Section \ref{part:details}.\ref{sec:multivariate-faa-di-brunos-formula}, 
	\begin{equation}			\label{eq:encoder-over-marginal-in-multivariate-notation}
		\frac{e^{-\beta d(x, \hat{x}')}}{Z(x, \beta)} \overset{\eqref{eq:encoder-eq}}{=}
		\frac{\intermediateencoder{'}{} }{\inputmarginal{'} }
		\quad \Longrightarrow \quad
		\left(\frac{e^{-\beta d(x, \hat{x})}}{Z(x, \beta)} \right)^{\bm{\alpha}_+} \overset{\eqref{eq:multivariate-notation-defs}}{=}
		\frac{e^{-\beta \left< \bm{\alpha}_+, d(x, \hat{x}) \right>}}{Z^{|\bm{\alpha}_+|}(x, \beta)} =
		\frac{\intermediateencoder{}{}^{\bm{\alpha}_+} }{\inputmarginalVect^{\bm{\alpha}_+} }
	\end{equation}
	Hence, using Proposition \ref{prop:repeated-encoder-deriv-wrt-coords},
	\begin{multline}		\label{eq:enc-mixed-repeated-deriv-midway}
		\frac{\partial^{|\bm{\alpha}|}}{\partial \beta^{\alpha_0} \; \partial \inputmarginalVect^{\bm{\alpha}_+}} \intermediateencoder{'}{} 
		\\ \overset{\eqref{eq:repeated-encoder-deriv-wrt-coords-in-prop}}{=}
		\partialbetaK{}{\alpha_0} \left\{
		\frac{(-1)^{|\bm{\alpha}_+|-1} (|\bm{\alpha}_+|-1)! \; e^{-\beta \left<\bm{\alpha}_+, d(x, \hat{x})\right> }  }{Z^{|\bm{\alpha}_+|}(x, \beta)}\cdot \Big[
		\left< \bm{\alpha}_+, \bm{e}_{\hat{x}'} \right>
		- |\bm{\alpha}_+| \cdot\intermediateencoder{'}{}
		\Big]
		\right\}
		\\ \overset{\eqref{eq:encoder-over-marginal-in-multivariate-notation}}{=}
		\frac{(-1)^{|\bm{\alpha}_+|-1} (|\bm{\alpha}_+|-1)!}{ \inputmarginalVect^{\bm{\alpha}_+} } \cdot \partialbetaK{}{\alpha_0} \left\{
		\intermediateencoder{}{}^{\bm{\alpha}_+}
		\Big[
		\alpha_{\hat{x}'} - |\bm{\alpha}_+| \cdot\intermediateencoder{'}{}
		\Big]
		\right\}
	\end{multline}
	Where, $\intermediateencoder{}{}^{\bm{\alpha}_+}$ in the last line is considered as an $\hat{x}$-indexed vector for $x$ fixed, $\bm{\alpha}_+ \in \bb{N}_0^M$.
	
	To proceed, we need a generalization of the Leibniz rule \eqref{eq:Leibniz-rule-for-product} to multiple factors $f_1, \dots, f_m$. It is a direct exercise by induction \cite{thaheem2003classroom} to see that \newline
	\begin{equation}		\label{eq:generalized-Leibniz-rule-for-prodcut}
		\left(f_1 f_2 \cdots f_m\right)^{(n)} = 
		\sum_{|\bm{k}| = n} \frac{n!}{\bm{k}!} \;
		\prod_{i=1}^{m} f_i^{(k_i)}
	\end{equation}
	where the sum is over all $m$-tuples $\bm{k} := (k_1, \dots, k_m)$ of non-negative integers with $|\bm{k}| = \sum_{i=1}^{m} k_i = n$, and $\nicefrac{n!}{\bm{k}!}$ is the multinomial coefficient,
	\begin{equation}
		\frac{n!}{\bm{k}!} = \frac{n!}{k_1 ! \; k_2 ! \; \cdots \; k_m !} := 
		\binom{n}{k_1, k_2, \dots, k_m} \;.
	\end{equation}
	Next, break the derivative of a product  $\partialbetaK{}{\alpha_0} \intermediateencoder{}{}^{\bm{\alpha}_+}$ \eqref{eq:enc-mixed-repeated-deriv-midway} to a product of derivatives,
	\begin{equation}		\label{eq:high-order-beta-deriv-of-encoder-power-intermid-step}
		\partialbetaK{}{\alpha_0} \intermediateencoder{}{}^{\bm{\alpha}_+} \overset{\eqref{eq:multivariate-notation-defs}}{=}
		\partialbetaK{}{\alpha_0} \Big( \prod_{i=1}^M \intermediateencoder{_i}{}^{\alpha_i} \Big) \overset{\eqref{eq:generalized-Leibniz-rule-for-prodcut}}{=}
		\sum_{\bm{k}: \; |\bm{k}| = \alpha_0} \frac{\alpha_0!}{\bm{k}!}
		\prod_{i=1}^{M} \partialbetaK{}{k_i} \big( \intermediateencoder{_i}{}^{\alpha_i} \big)
	\end{equation}
	where we have set $n := \alpha_0$ and $m := M$ at \eqref{eq:generalized-Leibniz-rule-for-prodcut}, so that $\bm{k} \in \bb{N}_0^M$.
	Note that the inner multiplicand to the right vanishes whenever $\alpha_i = 0$ but $k_i > 0$.
	
	To calculate a multiplicand $\partialbetaK{}{k_i} \intermediateencoder{_i}{}^{\alpha_i}$ at \eqref{eq:high-order-beta-deriv-of-encoder-power-intermid-step}, we apply the univariate Fa\`a di Bruno's formula \eqref{eq:univariate-faa-di-bruno-formula} to $x\mapsto x^{\alpha_i}$ composed after $\beta\mapsto \intermediateencoder{_i}{}$. When $k_i > 0$, we have for $\bm{t} \in \bb{N}_0^{k_i}$ 
	\begin{equation}		\label{eq:repeated-deriv-of-polynomial-in-proof}
		\frac{d^{|\bm{t}|}}{dx^{|\bm{t}|}} x^{\alpha_i} =
		\delta_{|\bm{t}|\leq \alpha_i} \cdot \frac{\alpha_i!}{(\alpha_i-|\bm{t}|)!} x^{\alpha_i-|\bm{t}|}
	\end{equation}
	where $\delta_{|\bm{t}|\leq \alpha_i}$ is the Kronecker delta.
	So, summing over all integer partitions $\bm{t} = \left(t_1, \dots, t_{k_i}\right)$ of $k_i$ to at most $|\bm{t}|\leq \alpha_i$ subsets, $\sum_{j=1}^{k_i} j\cdot t_j = k_i$, we obtain
	\begin{multline}		\label{eq:high-order-beta-deriv-of-coord-of-encoder-power-intermid-step}
		\partialbetaK{}{k_i} \intermediateencoder{_i}{}^{\alpha_i} \underset{\eqref{eq:repeated-deriv-of-polynomial-in-proof}}{\overset{\eqref{eq:univariate-faa-di-bruno-formula}}{=}}
		\sum \frac{k_i!}{\bm{t}! \cdot \left(1!^{t_1} \cdots k_i!^{t_{k_i}}\right)} 
		\cdot \frac{\alpha_i!}{(\alpha_i - |\bm{t}|)!} \; \intermediateencoder{_i}{}^{\alpha_i - |\bm{t}|} \cdot
		\prod_{j=1}^{k_i} \left( \partialbetaK{\intermediateencoder{_i}{}}{j} \right)^{t_j} \\
		\overset{\eqref{eq:formula-for-repeated-beta-deriv-in-prop-recursive-form}}{=}
		\sum \frac{k_i!}{\bm{t}! \cdot \left(1!^{t_1} \cdots k_i!^{t_{k_i}}\right)} 
		\cdot \frac{\alpha_i!}{(\alpha_i - |\bm{t}|)!} \; \intermediateencoder{_i}{}^{\alpha_i - |\bm{t}|} \cdot
		\prod_{j=1}^{k_i} \intermediateencoder{_i}{}^{t_j} \cdot P_j(\hat{x}_i, x)^{t_j} \\ =
		\intermediateencoder{_i}{}^{\alpha_i } \cdot \sum 
		\frac{k_i!}{\bm{t}! } \cdot \frac{\alpha_i!}{(\alpha_i - |\bm{t}|)!} 
		\prod_{j=1}^{k_i} \left( \frac{P_j(\hat{x}_i, x)}{j!} \right)^{t_j}
	\end{multline}
	Where, in the second equality we used formula \eqref{eq:formula-for-repeated-beta-deriv-in-prop-recursive-form} for the encoder's repeated partial $\beta$-derivative (Proposition \ref{prop:repeated-beta-derivatives-of-encoder}), and the last simplifies thanks to $\prod_j \intermediateencoder{_i}{}^{t_j} = {\intermediateencoder{_i}{}}^{\sum_j t_j} = {\intermediateencoder{_i}{}}^{|\bm{t}|}$.
	Note that if $k_i = 0$, then $\bm{t} = \bm{0}$ is the only integer partition of $k_i$, and so the end result of \eqref{eq:high-order-beta-deriv-of-coord-of-encoder-power-intermid-step} correctly reduces to $\intermediateencoder{_i}{}^{\alpha_i}$.
	Plugging the latter back into \eqref{eq:high-order-beta-deriv-of-encoder-power-intermid-step},
	\begin{multline}		\label{eq:high-order-beta-deriv-of-encoder-power-intermid-explicit}
		\partialbetaK{}{\alpha_0} \intermediateencoder{}{}^{\bm{\alpha}_+} \overset{\eqref{eq:high-order-beta-deriv-of-encoder-power-intermid-step}}{=}
		\sum_{\bm{k}: \; |\bm{k}| = \alpha_0} \frac{\alpha_0!}{\bm{k}!}
		\prod_{i=1}^{M} \partialbetaK{}{k_i} \big( \intermediateencoder{_i}{}^{\alpha_i} \big) 
		\\ \overset{\eqref{eq:high-order-beta-deriv-of-coord-of-encoder-power-intermid-step}}{=}
		\sum_{\bm{k}: \; |\bm{k}| = \alpha_0} \frac{\alpha_0!}{\bm{k}!}
		\prod_{i=1}^{M} \intermediateencoder{_i}{}^{\alpha_i } \sum_{\substack{\bm{t}: \;|\bm{t}|\leq \alpha_i, \\ \sum_{j} j\cdot t_j = k_i}} 
		\frac{k_i!}{\bm{t}! } \cdot \frac{\alpha_i!}{(\alpha_i - |\bm{t}|)!} 
		\prod_{j=1}^{k_i} \left( \frac{P_j(\hat{x}_i, x)}{j!} \right)^{t_j} \\ =
		\bm{\alpha}! \; \intermediateencoder{}{}^{\bm{\alpha}_+ } 
		\sum_{\bm{k}: \; |\bm{k}| = \alpha_0} \; \prod_{i=1}^{M} \; \sum_{\substack{\bm{t}: \;|\bm{t}|\leq \alpha_i, \\ \sum_{j} j\cdot t_{j} = k_i}} \frac{1}{ \bm{t}! \; (\alpha_i - |\bm{t}|)!} 
		\prod_{j=1}^{k_i} \left( \frac{P_j(\hat{x}_i, x)}{j!} \right)^{t_{j}} 
		\\ =: 
		\bm{\alpha}! \; \intermediateencoder{}{}^{\bm{\alpha}_+ } \sum_{\bm{k}: \; |\bm{k}| = \alpha_0} \; \prod_{i=1}^{M} 
		G\big( k_i, \alpha_i; \intermediateencoderVect, d \big)_{(\hat{x}_i, x)}
	\end{multline}
	Where, the second equality follows by taking $k_i, \alpha_i$ and $\intermediateencoder{_i}{}^{\alpha_i }$ out of the product over $i$, using the multivariate notation \eqref{eq:multivariate-notation-defs}, and the last equality defines $G$.
	To summarize,
	\begin{equation}		\label{eq:high-order-beta-deriv-of-powered-product-in-terms-of-F}
		\partialbetaK{}{\alpha_0} \intermediateencoder{}{}^{\bm{\alpha}_+} =
		\bm{\alpha}! \; \intermediateencoder{}{}^{\bm{\alpha}_+ } 
		\sum_{\bm{k}: \; |\bm{k}| = \alpha_0} \; \prod_{i=1}^{M} G\big( k_i, \alpha_i; \intermediateencoderVect, d \big)_{(\hat{x}_i, x)}
	\end{equation}
	where $G\big( k, a; \intermediateencoderVect, d \big)$ is a function on integers $0 \leq k, a$, whose values are $M\times N$ matrices. We set $G(k, a; \intermediateencoderVect, d) = 0$ if $a = 0 < k$ due to the comment after \eqref{eq:high-order-beta-deriv-of-encoder-power-intermid-step}, and otherwise
	\begin{equation}		\label{eq:combinatorial-F-in-terms-of-polynomials-in-proof}
		G\big( k, a; \intermediateencoderVect, d \big)_{(\hat{x}, x)} := 
		\sum_{\substack{\bm{t}: \;|\bm{t}|\leq a, \\ \sum_{j} j\cdot t_{j} = k}} \frac{1}{ \bm{t}! \; (a - |\bm{t}|)!} 
		\prod_{j=1}^{k} \left( \frac{P_j(\hat{x}, x)}{j!} \right)^{t_{j}} 
	\end{equation}
	where $\bm{t} \in \bb{N}_0^k$. We write below $G(k, a)$ for short. It depends on $\intermediateencoderVect$ and the distortion $d$ since so does $P_j(\hat{x}, x)$ \eqref{eq:P_k-by-abuse-of-notation}.
	It might sometimes be more convenient computationally to have $\bm{t} \in \bb{N}_0^l$ for some pre-fixed value $l$. As $\bm{t}$ represents a partition of $k$, any integer $l \geq k$ would do. 
	
	To complete the calculation at \eqref{eq:enc-mixed-repeated-deriv-midway}, denote $\bm{\alpha}'_+ := \bm{\alpha}_+ +  \bm{e}_{\hat{x}'}$ for $\bm{\alpha}$ with 1 added to its $\hat{x}'$-coordinate,
	where $\bm{e}_{\hat{x}'} \in \bb{R}^M$ is the standard basis vector at $\hat{x}'$. 
	In coordinates, $\alpha'_i = \alpha_i + \delta_{i, \hat{x}'}$.
	This does not affect the $\beta$ coordinate $\alpha_0 = \alpha'_0$. 
	Also, note that $\intermediateencoder{'}{} \cdot \intermediateencoder{}{}^{\bm{\alpha}_+} = \intermediateencoder{}{}^{\bm{\alpha}_+ \; + \; \bm{e}_{\hat{x}'}} = \intermediateencoder{}{}^{\bm{\alpha}'_+}$, where $\intermediateencoder{}{}^{\bm{\alpha}_+}$ is considered as an $\hat{x}$-indexed vector for $x$ fixed, as before. Thus,
	\begin{multline}
		\frac{\partial^{|\bm{\alpha}|}}{\partial \beta^{\alpha_0} \; \partial \inputmarginalVect^{\bm{\alpha}_+}} \intermediateencoder{'}{} 
		\overset{\eqref{eq:enc-mixed-repeated-deriv-midway}}{=} 
		\frac{(-1)^{|\bm{\alpha}_+|-1} (|\bm{\alpha}_+|-1)!}{ \inputmarginalVect^{\bm{\alpha}_+} } \cdot 
		\partialbetaK{}{\alpha_0} \left\{ \intermediateencoder{}{}^{\bm{\alpha}_+}
		\Big[ \alpha_{\hat{x}'} - |\bm{\alpha}_+| \cdot\intermediateencoder{'}{} \Big] \right\}
		\\ \overset{\eqref{eq:high-order-beta-deriv-of-powered-product-in-terms-of-F}}{=}
		\frac{(-1)^{|\bm{\alpha}_+|-1} (|\bm{\alpha}_+|-1)!}{ \inputmarginalVect^{\bm{\alpha}_+} } \cdot \left\{
		\alpha_{\hat{x}'} \cdot \bm{\alpha}! \; \intermediateencoder{}{}^{\bm{\alpha}_+ } 
		\sum_{\bm{k}: \; |\bm{k}| = \alpha_0} \; \prod_{i=1}^{M} G\big( k_i, \alpha_i \big)_{(\hat{x}_i, x)}
		\right. \\ \left. - 
		|\bm{\alpha}_+| \cdot \bm{\alpha}'! \; \intermediateencoder{}{}^{\bm{\alpha}'_+ } 
		\sum_{\bm{k}: \; |\bm{k}| = \alpha_0} \; \prod_{i=1}^{M} G\big( k_i, \alpha'_i \big)_{(\hat{x}_i, x)}
		\right\}
	\end{multline}
	Where we applied formula \eqref{eq:high-order-beta-deriv-of-powered-product-in-terms-of-F} for $\partialbetaK{}{\alpha_0} \intermediateencoder{}{}^{\bm{\alpha}_+}$ once to $\intermediateencoder{}{}^{\bm{\alpha}_+}$ and once to $\intermediateencoder{'}{} \cdot \intermediateencoder{}{}^{\bm{\alpha}_+}$.
	From the definition \eqref{eq:multivariate-notation-defs}, $\bm{\alpha}'! = \left(1 + \alpha_{\hat{x}'}\right) \cdot \bm{\alpha}!$.
	Proceeding with the calculation,
	\begin{multline}
		= \frac{(-1)^{|\bm{\alpha}_+|-1} (|\bm{\alpha}_+|-1)!}{ \inputmarginalVect^{\bm{\alpha}_+} } \cdot \bm{\alpha}! \; \intermediateencoder{}{}^{\bm{\alpha}_+ }  \cdot \left\{
		\alpha_{\hat{x}'} \cdot 
		\sum_{\bm{k}: \; |\bm{k}| = \alpha_0} \; \prod_{i=1}^{M} G\big( k_i, \alpha_i \big)_{(\hat{x}_i, x)}
		\right. \\ \left. - 
		|\bm{\alpha}_+| \cdot \left( 1 + \alpha_{\hat{x}'} \right) \cdot \intermediateencoder{'}{}
		\sum_{\bm{k}: \; |\bm{k}| = \alpha_0} \; \prod_{i=1}^{M} G\big( k_i, \alpha_i + \delta_{i, \hat{x}'} \big)_{(\hat{x}_i, x)}
		\right\} \\ =
		(-1)^{|\bm{\alpha}_+|-1} (|\bm{\alpha}_+|-1)! \; \bm{\alpha}! \cdot  \left(\frac{\intermediateencoder{}{}}{\inputmarginal{}}\right)^{\bm{\alpha}_+ } \cdot
		\sum_{\bm{k}: \; |\bm{k}| = \alpha_0} \left(\prod_{i\neq\hat{x}'} G\big( k_i, \alpha_i \big)_{(\hat{x}_i, x)} \right)
		\\ \cdot \left\{
		\alpha_{\hat{x}'} \cdot G\big( k_{\hat{x}'}, \alpha_{\hat{x}'} \big)_{(\hat{x}', x)} 
		- |\bm{\alpha}_+| \cdot \left( 1 + \alpha_{\hat{x}'} \right) \cdot \intermediateencoder{'}{} \cdot 
		G\big( k_{\hat{x}'}, 1 + \alpha_{\hat{x}'} \big)_{(\hat{x}', x)}
		\right\}
	\end{multline}
	Where, the last equality follows since the two products with $G$ are identical at all but the $\hat{x}'$ multiplicand, so that $\prod_{i\neq\hat{x}'} G\big( k_i, \alpha_i \big)_{(\hat{x}_i, x)}$ can be taken out of the curly brackets.
	This is formula \eqref{eq:mixed-high-order-enc-deriv-in-prop}, completing the proof.
\end{proof}

\medskip
\subsection{Proof of Proposition \ref{prop:Jacobian-of-BA-in-direct-encoder-coordinates}, Blahut-Arimoto's Jacobian in encoder coordinates}
\label{sub:proof-of-prop:Jacobian-of-BA-in-direct-encoder-coordinates}

Unlike the previous subsections in Section \ref{sec:appendix-derivations-of-BA-high-order-deriv-in-marginal-coords} which take the input marginal $\inputmarginalVect$ as the variable, here we consider the encoder $\intermediateencoderVect$ as the variable.

\begin{proof}[Proof of Proposition \ref{prop:Jacobian-of-BA-in-direct-encoder-coordinates}]
	We re-state the Blahut-Arimoto equations, with the encoder now playing the role of input and output distributions. Starting at the $i$-th iteration with $p_i(\hat{x}|x)$, set
	\begin{equation}		\label{eq:marginal-eq-in-BA-enc-proof}
		p_i(\hat{x}) := \sum_x p_X(x) p_i(\hat{x}|x) \;.
	\end{equation}
	Then, output
	\begin{equation}		\label{eq:encoder-eq-in-BA-enc-proof}
		p_{i+1}(\hat{x}|x) := \frac{p_{i}(\hat{x}) e^{-\beta d(x, \hat{x})}}{\sum_{\hat{x}'} p_{i}(\hat{x}') e^{-\beta d(x, \hat{x}')}} \;.
	\end{equation}
	For particular input coordinates $\hat{x}', x'$ and output coordinates $\hat{x}, x$, we would like to calculate
	\begin{equation}		\label{eq:multiv-chain-rule-in-BA-enc-proof}
		\frac{\partial p_{i+1}(\hat{x}|x)}{\partial p_{i}(\hat{x}'|x')} =
		\sum_{\hat{x}''} \frac{\partial p_{i+1}(\hat{x}|x)}{\partial p_{i}(\hat{x}'')}
		\frac{\partial p_{i}(\hat{x}'')}{\partial p_{i}(\hat{x}'|x')} \;,
	\end{equation}
	where the equality is due to the multivariate chain rule.
	
	For the first integrand at \eqref{eq:multiv-chain-rule-in-BA-enc-proof},
	\begin{multline}		\label{eq:enc-by-marginal-deriv}
		\frac{\partial p_{i+1}(\hat{x}|x)}{\partial p_{i}(\hat{x}'')} \overset{(\ref{eq:encoder-eq-in-BA-enc-proof})}{=}
		\frac{\partial }{\partial p_{i}(\hat{x}'')}
		\frac{p_{i}(\hat{x}) e^{-\beta d(x, \hat{x})}}{\sum_{\hat{x}'} p_{i}(\hat{x}') e^{-\beta d(x, \hat{x}')}} = \\ =
		\frac{\delta_{\hat{x}, \hat{x}''} \; e^{-\beta d(x, \hat{x})}}{\sum_{\hat{x}'} p_{i}(\hat{x}') e^{-\beta d(x, \hat{x}')}} -
		\frac{p_{i}(\hat{x}) e^{-\beta d(x, \hat{x})}}{(\sum_{\hat{x}'} p_{i}(\hat{x}') e^{-\beta d(x, \hat{x}')})^2} \cdot
		\overset{e^{-\beta d(x, \hat{x}'')}}{\overbrace{\frac{\partial }{\partial p_{i}(\hat{x}'')} \sum_{\hat{x}'} p_{i}(\hat{x}') e^{-\beta d(x, \hat{x}')}}} = \\ =
		\frac{e^{-\beta d(x, \hat{x}'')}}{\sum_{\hat{x}'} p_{i}(\hat{x}') e^{-\beta d(x, \hat{x}')}} 
		\left[ \delta_{\hat{x}, \hat{x}''} - \frac{p_{i}(\hat{x}) e^{-\beta d(x, \hat{x})}}{\sum_{\hat{x}'} p_{i}(\hat{x}') e^{-\beta d(x, \hat{x}')}}	\right]
		\overset{(\ref{eq:encoder-eq-in-BA-enc-proof})}{=} \\ =
		\frac{e^{-\beta d(x, \hat{x}'')}}{\sum_{\hat{x}'} p_{i}(\hat{x}') e^{-\beta d(x, \hat{x}')}} 
		\left[ \delta_{\hat{x}, \hat{x}''} - p_{i+1}(\hat{x}|x) \right]
	\end{multline}
	For the second,
	\begin{equation}		\label{eq:marginal-by-enc-deriv}
		\frac{\partial p_{i}(\hat{x}'')}{\partial p_i(\hat{x}'|x')} \overset{(\ref{eq:marginal-eq-in-BA-enc-proof})}{=}
		\frac{\partial}{\partial p_i(\hat{x}'|x')} \sum_{x} p_X(x) p_i(\hat{x}''|x) =
		\delta_{\hat{x}', \hat{x}''} \; p_X(x')
	\end{equation}
	Combining the results,
	\begin{multline}		\label{eq:BA-derived-as-enc-by-enc}
		\frac{\partial p_{i+1}(\hat{x}|x)}{\partial p_{i}(\hat{x}'|x')} \overset{\eqref{eq:multiv-chain-rule-in-BA-enc-proof}}{=}
		\sum_{\hat{x}''} \frac{\partial p_{i+1}(\hat{x}|x)}{\partial p_{i}(\hat{x}'')}
		\frac{\partial p_{i}(\hat{x}'')}{\partial p_{i}(\hat{x}'|x')} \overset{(\ref{eq:enc-by-marginal-deriv})}{\underset{(\ref{eq:marginal-by-enc-deriv})}{=}} \\ =
		\sum_{\hat{x}''} \frac{e^{-\beta d(x, \hat{x}'')}}{\sum_{\hat{x}'''} p_{i}(\hat{x}''') e^{-\beta d(x, \hat{x}''')}}  \left[ \delta_{\hat{x}, \hat{x}''} - p_{i+1}(\hat{x}|x) \right] \cdot \delta_{\hat{x}', \hat{x}''} \; p_X(x') = \\ =
		\frac{e^{-\beta d(x, \hat{x}')}}{\sum_{\hat{x}''} p_{i}(\hat{x}'') e^{-\beta d(x, \hat{x}'')}} 
		\left[ \delta_{\hat{x}, \hat{x}'} - p_{i+1}(\hat{x}|x) \right] p_X(x')
	\end{multline}
	As $\frac{p_{i+1}(\hat{x}'|x)}{p_i(\hat{x}')} = \frac{e^{-\beta d(x, \hat{x}')}}{\sum_{\hat{x}''} p_{i}(\hat{x}'') e^{-\beta d(x, \hat{x}'')}}$ by \eqref{eq:encoder-eq-in-BA-enc-proof} whenever $p_{i}(\hat{x}') \neq 0$ this yields the result.
\end{proof}

\medskip
\section{Proofs for error analysis}
\label{sec:proof-for-error-analysis}

The following lemma is used for error analysis and for assessing computational costs of RD derivative tensors.
While, the rest of this section contains results used only for error analysis.

\begin{lem}[Bounds on the complexity of $P_k$]			\label{lem:bound-on-the-complexity-of-P_K}
	Each polynomial $P_k$ \eqref{eq:P_0-def}-\eqref{eq:P_k-inductive-def} is of degree $k$ at most, and can be written as a sum of at most $2^k k!$ monomials in the coefficients $1$ and $-1$. In particular, $P_k$ has no more than $2^k k!$ monomials.
\end{lem}

\begin{proof}[Proof of Lemma \ref{lem:bound-on-the-complexity-of-P_K}]
	First, for the degree, note that the first addend $(x_1 - x_0) \cdot P_k$ in the inductive definition \eqref{eq:P_k-inductive-def} of $P_{k+1}$ increases the degree by 1, while by the definition \eqref{eq:variable-deriv-def-for-recursive-beta-formula} of $\dbar$, deriving $P_k$ increases its degree by 1 at most, $\deg \dbar P_k \leq 1 + \deg P_k$. 
	
	Second, the derivative of a monomial $x_{i_1} \cdot x_{i_2} \cdot \dots \cdot x_{i_k}$ of degree $k$ is a sum of at most $k$ differentiations, each of which is a sum of two monomials. Explicitly, if neither variable $x_{i_j}$ is $x_0$, then
	\begin{equation}			\label{eq:derivation-of-a-monic-monomial-in-lemmas-proof}
		\dbar \Big( \prod_{j=1}^k x_{i_j} \Big) =
		\sum_{l=1}^k \Big(\prod_{j\neq l}^k x_{i_j}\Big) \dbar x_{i_l} 
		\overset{\eqref{eq:variable-deriv-def-for-recursive-beta-formula}}{=}
		\sum_{l=1}^k \Big(\prod_{j\neq l}^k x_{i_j}\Big) \left( x_1 \cdot x_{i_l} - x_{i_l+1} \right)
	\end{equation}
	which is a sum of at most $2k$ monomials. The sum is shorter if any $i_j$ is $0$, or if the monomial we've started with is of degree a smaller than $k$.
	Denote by $l_k$ the minimal number of monomials in $P_k$, when represented as a sum with coefficients $\pm 1$. By the inductive definition \eqref{eq:P_k-inductive-def} of $P_k$,
	\begin{equation}		\label{eq:inductive-formula-for-length-of-P_k-in-proof}
		l_{k+1} \leq 2\cdot l_k + 2k\cdot l_k  = 2(k+1) \cdot l_k \;.
	\end{equation}
	Where, the first term $2\cdot l_k$ bounds the number of monomials in $(x_1 - x_0) \cdot P_k$, and by \eqref{eq:derivation-of-a-monic-monomial-in-lemmas-proof}, $2k\cdot l_k$ bounds that in $\dbar P_k$.
	Applying \eqref{eq:inductive-formula-for-length-of-P_k-in-proof} inductively starting at $l_0 = 1$ \eqref{eq:P_0-def}, we obtain $l_{k} \leq \left(2\cdot 1\right) \cdot \left(2 \cdot 2\right) \cdots \left(2\cdot k\right) = 2^k k!$.
	
	When the coefficients are not restricted to $\pm 1$, then identical monomials can be grouped together, showing that the minimal number of monomials in $P_k$ is smaller.
\end{proof}

\medskip
\subsection{Proof that RD derivative tensors are bounded uniformly, Lemma \ref{lem:uniform-bound-on-BA-derivative-tensors}}
\label{sub:proof-of-lem:uniform-bound-on-BA-derivative-tensors}

\begin{proof}[Proof of Lemma \ref{lem:uniform-bound-on-BA-derivative-tensors}]
	We show that each of the quantities in Section \ref{sub:high-order-deriv-tensors-of-BA} is bounded uniformly on the closed $\delta$-interior of the simplex $\Delta[ \hat{\mathcal{X}} ]$, by a bound which depends only on the orders of differentiation $b$ and $m$, and on the problem's properties, via $d$ and $p_X$.
	
	To synchronize the Lemma's notation at \eqref{eq:uniform-upper-bound-on-deriv-tensor} with the explicit forms \eqref{eq:repeated-beta-deriv-in-thm} and \eqref{eq:mixed-BA-deriv-in-thm} of the derivative tensors of $Id - BA_\beta$ \eqref{eq:RD-operator-def} (Theorem \ref{thm:high-order-derivs-of-BA-in-main-text}), define $\bm{\alpha} := \left( \alpha_0, \bm{\alpha}_+ \right) \in \bb{N}_0^{M+1}$ as following. Set $\alpha_0 := b$, the number of differentiations with respect to $\beta$. Next, using the tensor indices $(i_1, i_2, \dots, i_m)$ at \eqref{eq:uniform-upper-bound-on-deriv-tensor}, define $\bm{\alpha}_+ \in \bb{N}_0^M$ by $\bm{\alpha}_+ := \bm{e}_{i_1} + \bm{e}_{i_2} + \dots + \bm{e}_{i_m}$, where $\bm{e}_j$ is the standard $j$-th basis vector. Carefully note that $|\bm{\alpha}_+| \overset{\eqref{eq:multivariate-notation-defs}}{=} m$ by definition, and so $|\bm{\alpha}| = m + b$.
	cf., the comments after Equation \eqref{eq:mixed-deriv-def-evaluated-applied-to-vectors}, for the two different notations of high-order derivatives.
	
	Next, from definition \eqref{eq:expected-k-th-power-distortion-def},
	\begin{equation}			\label{eq:bound-on-variables-max-value-in-proof}
		\expectedDxWRTencoderK{k} = 
		\sum_{\hat{x}'} \intermediateencoder{'}{} d(x, \hat{x}')^k \leq
		\sum_{\hat{x}'} \intermediateencoder{'}{} d_{max}^k = d_{max}^k
	\end{equation}
	for \emph{any} conditional distribution $\intermediateencoder{}{}$.
	Therefore, by the definition \eqref{eq:P_k-by-abuse-of-notation} of the matrices $P_k[\bm{q}; d]$, each of its entries $(\hat{x}, x)$ is bounded,
	\begin{equation}		\label{eq:uniform-bound-on-P_k-in-proof}
		| P_k[\bm{q}; d](\hat{x}, x) | \leq 2^k k! \cdot (d_{max}^k)^k \;.
	\end{equation}
	For, by Lemma \ref{lem:bound-on-the-complexity-of-P_K}, each $P_k$ can be written as a sum of no more than $2^k k!$ monomials, each of degree $k$ at most, with the value of each variable bounded by \eqref{eq:bound-on-variables-max-value-in-proof}.	
	This immediately shows that repeated partial $\beta$-derivatives are uniformly bounded on the entire simplex,
	\begin{equation}		\label{eq:upper-bound-for-repeated-beta-derivs-in-proof}
		\left| \partialbetaK{}{\alpha_0} \big( Id - BA_\beta \big)[\inputmarginalVect](\hat{x}) \right| \overset{\eqref{eq:repeated-beta-deriv-in-thm} }{\leq}
		\sum_x p_X(x) \intermediateencoder{}{} \cdot \left| P_{\alpha_0}(\hat{x}, x) \right| \overset{\eqref{eq:uniform-bound-on-P_k-in-proof}}{\leq}
		2^{\alpha_0} {\alpha_0}! \cdot d_{max}^{\alpha_0^2}
	\end{equation}
	
	For mixed derivatives, first note that for each $k, a$,
	\begin{equation}		\label{eq:upper-bound-on-G}
		\left| G\big( k, a\big)_{(\hat{x}, x)} \right| \overset{\eqref{eq:combinatorial-G-in-terms-of-polynomials}}{\leq}
		\sum_{\substack{\bm{t}: \;|\bm{t}|\leq a, \\ \sum_{j} j\cdot t_{j} = k}} 1 \cdot \prod_{j=1}^{k} \left( \frac{ \left| P_j(\hat{x}, x) \right| }{j!} \right)^{t_{j}}
		\overset{\eqref{eq:uniform-bound-on-P_k-in-proof}}{\leq}
		p(k) \cdot \left( 2^k k! \cdot d_{max}^{k^2} \right)^a 
	\end{equation}
	Where, $p(k)$ stands for the partition function (the number of integer partitions of $k$), and we have discarded the factorials at each denominator. For the last inequality, note that $\sum_j t_j \leq a$.
	
	Next, by our assumption that $\inputmarginal{'}\geq \delta$ for all $\hat{x}'$ and the definition \eqref{eq:multivariate-notation-defs} of vector power, it follows that $\left(\nicefrac{1}{\inputmarginal{'}}\right)^{\bm{\alpha}_+} \leq \frac{1}{\delta^{|\bm{\alpha}_+|}}$. Each coordinate $\alpha_{\hat{x}}$ of $\bm{\alpha}_+$ is bounded by $|\bm{\alpha}_+|$, and so $\bm{\alpha}!$ is bounded by $\alpha_0 \cdot |\bm{\alpha}_+|!^M$. This allows to bound $\left|(-1)^{|\bm{\alpha_+}|-1} (|\bm{\alpha_+}|-1)! \; \bm{\alpha}! \sum_x p_X(x) \left(\frac{\intermediateencoder{'}{}}{\inputmarginal{'}}\right)^{\bm{\alpha_+} }\right|$ in \eqref{eq:mixed-BA-deriv-in-thm} from above by $\tfrac{(|\bm{\alpha_+}|-1)! \; \alpha_0 \; |\bm{\alpha}_+|!^M}{\delta^{|\bm{\alpha}_+|}}$.
	Similarly, each coordinate $k_{\hat{x}}$ of $\bm{k} \in \bb{N}_0^M$ with $|\bm{k}| = \alpha_0$ is bounded by $\alpha_0$, and so $|G(k, 1+a)|$ can be bounded by evaluating the upper bound \eqref{eq:upper-bound-on-G} at $(\alpha_0, 1+|\bm{\alpha}_+|)$.
	From the combinatorial definition of the binomial coefficient, there are $\binom{\alpha_0 + M - 1}{\alpha_0}$ integral vectors $\bm{k} \in \bb{N}_0^M$ with $|\bm{k}| = \alpha_0$.
	Using these, the formula \eqref{eq:mixed-BA-deriv-in-thm} for the mixed derivatives can be bounded by,
	\begin{multline}		\label{eq:upper-bound-for-mixed-deriv-tensor-in-proof}
		\left| \frac{\partial^{|\bm{\alpha}|} }{\partial \beta^{\alpha_0} \partial \inputmarginalVect^{\bm{\alpha_+}}} \left(Id - BA_\beta \right)\left[\inputmarginalVect\right](\hat{x}) \right| 
		\\ \overset{\eqref{eq:mixed-BA-deriv-in-thm}}{\leq}
		1 + \frac{2(|\bm{\alpha_+}| + 1)! \; \alpha_0}{\delta^{|\bm{\alpha}_+|}} 
		\binom{\alpha_0 + M - 1}{\alpha_0} \left[ 
		|\bm{\alpha}_+|! \; p(\alpha_0) \cdot \left( 2^{\alpha_0} \alpha_0! \cdot d_{max}^{\alpha_0^2} \right)^{1 + |\bm{\alpha}_+|} 
		\right]^{M}
	\end{multline}
	
	Since \eqref{eq:upper-bound-for-mixed-deriv-tensor-in-proof} bounds the right-hand side of \eqref{eq:upper-bound-for-repeated-beta-derivs-in-proof} from above, this completes the proof.
\end{proof}

\medskip
\subsection{Proof that Taylor method converges between RD bifurcations, Theorem \ref{thm:taylor-method-converges-for-RD-root-tracking-away-of-bifurcation}}

\label{sub:proof-of-thm:taylor-method-converges-for-RD-root-tracking-away-of-bifurcation}

\begin{proof}[Proof of Theorem \ref{thm:taylor-method-converges-for-RD-root-tracking-away-of-bifurcation}]
	In this proof, denote $\partial_\delta \Delta[\hat{\mathcal{X}}] := \{p\in \Delta[\hat{\mathcal{X}}]: \; \exists \hat{x} \; p(\hat{x}) \leq \delta \}$ for the closed $\delta$-boundary of the simplex, and	$\overline{\Delta[\hat{\mathcal{X}}] \setminus \partial_\delta \Delta[\hat{\mathcal{X}}]}$ for the closed $\delta$-interior.
	
	For the first claim of this Theorem, set $\beta_f(\delta)$ to be the $\beta$ value of the first time that $\inputmarginalVect_\beta$ reaches the $\delta$-boundary. That is, the	largest $\beta$ such that $\beta < \beta_0$ and $\inputmarginalVect_\beta \in \partial_\delta \Delta[\hat{\mathcal{X}}]$.
	This is well defined, since $\partial_\delta \Delta[\hat{\mathcal{X}}]$ is compact and $\inputmarginalVect_\beta$ is a continuous function of $\beta$, by Assumption \ref{assumption:solution-is-smooth-in-beta} in Section \ref{part:how-and-what}.\ref{sub:beta-derivs-at-an-operator-root}. If $\inputmarginalVect_\beta$ never reaches the $\delta$-boundary of $\Delta[\hat{\mathcal{X}}]$ for $\beta \in [0, \beta_0]$ then set $\beta_f(\delta) = 0$.
	
	For the second claim, we set to prove that the conditions of Theorem \ref{thm:error-analysis-for-euler-method} in Section \ref{part:details}.\ref{sub:error-analysis-of-Taylor-method-background} are met; namely, of error analysis for an $l$-th order Taylor method.
	To invoke that Theorem, it suffices to show $(i)$ that the derivative's norm $\left\|\dbetaK{\inputmarginalVect_\beta}{l+1}\right\|_\infty$ at the true solution $\inputmarginalVect_\beta$ is bounded uniformly on $[\beta_f(\delta), \beta_0]$, and $(ii)$ that the Taylor polynomial $T_l$ \eqref{eq:Taylor-poly-def-for-Taylor-method} has a finite Lipschitz constant $L_l$.
	The assumptions of that Theorem require that the Lipschitz condition holds for any $\beta\in [\beta_f(\delta), \beta_0]$ and $\tilde{\inputmarginalVect} \in \bb{R}^M$. However, the condition $\tilde{\inputmarginalVect} \in \bb{R}^M$ may be relaxed to only requiring that $\tilde{\inputmarginalVect}$ is not too far away from a true solution $\inputmarginalVect_\beta$, so long that the step-size will eventually be taken to be small enough, \cite[Working assumption 210A]{butcher2016numerical}.
	We may choose $\delta'$ with $0 < \delta' < \delta$, and prove that the Lipschitz condition $(ii)$ holds for $(\tilde{\inputmarginalVect}, \beta)$ with $\beta \in [\beta_f(\delta), \beta_0]$ and $\|\tilde{\inputmarginalVect} - \inputmarginalVect_\beta \|_\infty \leq \eta := \delta - \delta'$. Namely, we may consider only points in a ``tube'' $T_\eta := \left\{(\tilde{\inputmarginalVect}, \beta): \; \beta \in [\beta_f(\delta), \beta_0] \text{ and }\|\tilde{\inputmarginalVect} - \inputmarginalVect_\beta \|_\infty \leq \eta \right\}$ around the true solution $\inputmarginalVect_\beta$.
	For $\eta = 0$, $T_0$ is simply the graph of $\inputmarginalVect_\beta$.
	Note that by the first claim, if $0 < \eta < \delta$, then the $\tilde{\inputmarginalVect}$ coordinate of $T_\eta$ is contained in the closed $\delta'$-interior, and thus in the interior of the simplex.
	$T_\eta$ is then essentially the product of two compact spaces, the closed $\eta$-ball (around $\inputmarginalVect_\beta \in \Delta[\hat{\mathcal{X}}]$) and the interval $[\beta_f(\delta), \beta_0]$.
	
	Write $J(\tilde{\inputmarginalVect}, \beta) := D_{\inputmarginalVect} (Id - BA_\beta)\rvert_{(\tilde{\inputmarginalVect}, \beta)}$ for the Jacobian matrix.
	Its general form (\textit{not} necessarily at a fixed point) at a distribution $\tilde{\inputmarginalVect}$ of full support is given by Corollary \ref{cor:BA-jacobian} in Section \ref{part:details}.\ref{sub:encoders-marginal-derivatives} (formula \eqref{eq:BA-jacobian} there). If $\inputmarginalVect_\beta$ is in addition a fixed point of $BA_\beta$, then the basic properties of $J(\inputmarginalVect_\beta, \beta)$ are given by Theorem \ref{thm:properties-of-BA-jacobian-from-ISIT-CSD-paper} there.
	In particular, it is non-singular so long that $\inputmarginalVect_\beta$ is in the simplex interior, which by the first claim holds for $\beta \in [\beta_f(\delta), \beta_0]$.
	
	Before proving $(i)$ and $(ii)$, we shall show that $\eta > 0$ can be chosen small enough, such that $J(\tilde{\inputmarginalVect}, \beta)$ is non-singular for every $(\tilde{\inputmarginalVect}, \beta) \in T_\eta$. From this, it shall follow that the matrix norm $\|J(\tilde{\inputmarginalVect},\beta)^{-1}\|_\infty$ of its inverse is well-defined on $T_\eta$, and thus bounded uniformly. For, it is a continuous real-valued function on the compact set $T_\eta$, and so obtains a maximal value.
	e.g., \cite[1.3]{ortega1990numerical} for matrix norms.
	When $\tilde{\inputmarginalVect}$ is set to the true solution $\inputmarginalVect_\beta$, this is rather straightforward from Theorem \ref{thm:properties-of-BA-jacobian-from-ISIT-CSD-paper}. While for $\tilde{\inputmarginalVect}$ slightly off a true solution, this follows from continuity and compactness, as we show next.
	e.g., \cite[Chapter 3]{munkres2000topology} on compactness.
	
	First, we prove that $J$ is non-singular when evaluated at the true solution $\inputmarginalVect_\beta$.
	This follows since the composition
	\begin{equation}		\label{eq:J-composition-in-proof}
		\beta \mapsto \inputmarginalVect_\beta \mapsto J(\inputmarginalVect_\beta, \beta) \mapsto 
		|\det J(\inputmarginalVect_\beta, \beta) |
	\end{equation}
	is continuous on $[\beta_f(\delta), \beta_0]$ (shown below), and so obtains a minimal value $d' \geq 0$ at some point $\beta'$ there. 
	$J(\inputmarginalVect_{\beta'}, \beta')$ is non-singular by the Jacobian's properties mentioned above, and so $d'$ is strictly positive. In particular, $d'$ does not depend on the value of $\eta \geq 0$.
	This shows that $J(\inputmarginalVect_\beta, \beta)$ is non-singular on $[\beta_f(\delta), \beta_0]$. 
	Each function in the composition \eqref{eq:J-composition-in-proof} is indeed continuous. For the first $\beta \mapsto \inputmarginalVect_\beta$ this is by Assumption \ref{assumption:solution-is-smooth-in-beta}.
	For the second $\inputmarginalVect_\beta \mapsto J(\inputmarginalVect_\beta, \beta)$, by formula \eqref{eq:BA-jacobian} of Corollary \ref{cor:BA-jacobian}, the entries of $J(\tilde{\inputmarginalVect}, \beta)$ are continuous in both $\beta$ and $\tilde{\inputmarginalVect}$, so long that $\tilde{\inputmarginalVect}$ is in the interior of $\Delta[\hat{\mathcal{X}}]$, which holds at $\inputmarginalVect_\beta$ by the first claim.
	Finally, the determinant of a matrix is continuous, as it is a sum of products of matrix entries.
	
	Second, we show that $\eta > 0$ can be chosen small enough that $J$ is non-singular also at points in $T_\eta$ other than the true solution.
	Fix some $\eta$ with $0 < \eta < \delta$, and define a function $f$ on $T_\eta$ by
	\begin{equation}
		f(\tilde{\inputmarginalVect}, \beta) :=
		|\det J(\tilde{\inputmarginalVect}, \beta) |
	\end{equation}
	By the note after $T_\eta$'s definition, its projection onto $\Delta[\hat{\mathcal{X}}]$ is in the interior of the simplex. Hence, $f$ is well-defined and continuous in $(\tilde{\inputmarginalVect}, \beta)$, by formula \eqref{eq:BA-jacobian}. 
	It satisfies $f(\inputmarginalVect_\beta, \beta) \geq d' > 0$ by the argument above.
	By the definition of continuity, \cite[$\mathsection$18]{munkres2000topology}, the inverse image $f^{-1}((\nicefrac{d'}{2}, \infty])$ is open in $T_\eta$; it contains $T_0$ (the graph of $\inputmarginalVect_\beta$).
	As $T_\eta$ is a product of compact spaces, then by the tube lemma \cite[Lemma 26.8]{munkres2000topology} there is $0 < \eta' \leq \eta$ such that $T_0 \subset T_{\eta'} \subseteq f^{-1}((\nicefrac{d'}{2}, \infty])$. To see this, note that any open neighborhood of 0 in the $\eta$-ball contains an open $\eta'$-ball around 0, from the definition of a basis for a topology.
	Therefore, $|\det J|$ is at least $\nicefrac{d'}{2} > 0$ on $T_{\eta'}$, as argued.
	
	Summarizing the above, we have shown that $0 < \eta' < \delta$ can be chosen such that $J(\tilde{\inputmarginalVect}, \beta)$ is non-singular on $T_{\eta'}$. Therefore, matrix inversion is well-defined on it. As matrix inversion and norm are continuous, then $\|J(\tilde{\inputmarginalVect}, \beta)^{-1}\|_\infty$ is continuous on the compact set $T_{\eta'}$, and so obtains a maximum value on it. That is, the Jacobian's inverse is of uniformly bounded matrix norm, for distributions $\tilde{\inputmarginalVect}$ at most $\eta'$-far from $\inputmarginalVect_\beta$.
	
	\medskip 
	Aided by the above, we turn to prove $(i)$ and $(ii)$.
	For $(i)$, let $\dbetaK{\tilde{\inputmarginalVect}}{l}$ denote the numerical derivative calculated from formula \eqref{eq:formula-for-high-order-beta-derivatives} (Theorem \ref{thm:formula-for-high-order-expansion-of-F-in-main-result-sect} in Section \ref{part:how-and-what}.\ref{sub:high-order-beta-derivatives-at-an-operator-root}) at an approximation $\tilde{\inputmarginalVect}$ of the true solution $\inputmarginalVect_\beta$, such that $\tilde{\inputmarginalVect}$ is at most $\eta'$-far from $\inputmarginalVect_\beta$.
	We prove by induction on $l > 0$ that its norm $\left\|\dbetaK{\tilde{\inputmarginalVect}}{l}\right\|_{\infty}$ is bounded uniformly on $T_{\eta'}$.
	Since $T_{\eta'}$ contains the graph of $\inputmarginalVect_\beta$, this shall suffice to prove $(i)$.
	Assume that it holds for any $0 < k < l$, for the norms $\left\|\dbetaK{\tilde{\inputmarginalVect}}{k}\right\|_{\infty}$. We would like to prove that so is the $l$-th derivative $\left\|\dbetaK{\tilde{\inputmarginalVect}}{l}\right\|_{\infty}$.
	The assumption is of course vacuous when $l = 1$. Indeed, the first-order implicit derivative $\tfrac{d \tilde{\inputmarginalVect}}{d\beta}$ is the only one that does not involve implicit derivatives of lower order in its calculation.
	
	For the induction step, note that the $l$-th derivative $\dbetaK{\tilde{\inputmarginalVect}}{l}$ involves derivative tensors $D^m_{\beta^b, \inputmarginalVect^{m-b}} (Id - BA_\beta)\big\rvert_{(\tilde{\inputmarginalVect}, \beta)}$, with $0 \leq b \leq m \leq l$; see Equation \eqref{eq:formula-for-high-order-beta-derivatives} (Theorem \ref{thm:formula-for-high-order-expansion-of-F-in-main-result-sect}). 
	By Lemma \ref{lem:uniform-bound-on-BA-derivative-tensors} in Section \ref{part:details}.\ref{sub:computational-difficulty-of-RTRD}, the entries of these tensors are bounded uniformly on $\overline{\Delta[\hat{\mathcal{X}}] \setminus \partial_{\delta'} \Delta[\hat{\mathcal{X}}]}$ (regardless of $\beta$), for $\delta' = \delta - \eta' > 0$, which in turn contains $\tilde{\inputmarginalVect} \in T_{\eta'}$. 
	By the induction's hypothesis, the lower-order derivatives are also bounded uniformly on $T_{\eta'}$.
	Therefore, the right-hand side of \eqref{eq:formula-for-high-order-beta-derivatives} is bounded uniformly, as it involves only a (particular) finite sum of these tensors, which are evaluated at implicit derivatives of bounded coordinates.
	That is, it is a sum of products of quantities that are bounded uniformly on $T_{\eta'}$.
	To complete the induction step, both sides of \eqref{eq:formula-for-high-order-beta-derivatives} need to be multiplied by the inverse Jacobian. Yet, $\|(D_{\inputmarginalVect} (Id - BA_\beta)\rvert_{(\tilde{\inputmarginalVect}, \beta)})^{-1}\|_{\infty}$ is bounded uniformly on $T_{\eta'}$ (as shown before), and thus so is $\left\|\dbetaK{\tilde{\inputmarginalVect}}{l}\right\|_{\infty}$, as required.
	
	For $(ii)$, Lipschitz continuity can be established by means of a supremum over the derivative's matrix norm, e.g., Equation (3.9) at \cite{atkinson2011numerical},
	\begin{equation}		\label{eq:Lipschitz-constant-of-T_l-in-proof}
		L_l := \sup \| D_{\inputmarginalVect } T_l \|_\infty \;,
	\end{equation}
	where the supremum is over $T_{\eta'}$. A differentiable functions $T_l$ is Lipschitz continuous if the supremum \eqref{eq:Lipschitz-constant-of-T_l-in-proof} is finite. 
	
	From the definition \eqref{eq:Taylor-poly-def-for-Taylor-method} of Taylor method,
	\begin{equation}
		D_{\inputmarginalVect } T_l(\tilde{\inputmarginalVect}, \beta, \Delta \beta) =
		\frac{1}{1!} \cdot D_{\inputmarginalVect } \frac{d \tilde{\inputmarginalVect}}{d\beta} +
		\frac{\Delta \beta}{2!} \cdot D_{\inputmarginalVect } \frac{d^2 \tilde{\inputmarginalVect}}{d\beta^2} + \dots +
		\frac{\Delta \beta^{l-1}}{l!} \cdot D_{\inputmarginalVect } \frac{d^l \tilde{\inputmarginalVect}}{d\beta^l}
	\end{equation}
	As matrix norms are sub-additive,
	\begin{equation}			\label{eq:bounding-operator-norm-as-a-function-of-delta-beta}
		\left\| D_{\inputmarginalVect } T_l(\tilde{\inputmarginalVect}, \beta, \Delta\beta) \right\|_{\infty} \leq 
		\frac{1}{1!} \cdot \left\| D_{\inputmarginalVect } \frac{d \tilde{\inputmarginalVect}}{d\beta} \right\|_{\infty} +
		\frac{|\Delta \beta|}{2!} \cdot \left\| D_{\inputmarginalVect } \frac{d^2 \tilde{\inputmarginalVect}}{d\beta^2} \right\|_{\infty} + \dots +
		\frac{|\Delta \beta|^{l-1}}{l!} \cdot \left\| D_{\inputmarginalVect } \frac{d^l \tilde{\inputmarginalVect}}{d\beta^l} \right\|_{\infty}
	\end{equation}
	Thus, to prove that $T_l$ has a finite Lipschitz constant, it suffices to show that the matrix norms of the matrices $D_{\inputmarginalVect } \tfrac{d^k \tilde{\inputmarginalVect}}{d\beta^k}$ are bounded uniformly on $T_{\eta'}$. By Proposition \ref{prop:Jacobian-of-high-order-beta-derivative} in Section \ref{part:details}.\ref{sub:local-error-estimates-for-beta-derivs}, the Jacobian $D_{\inputmarginalVect } \tfrac{d^l \tilde{\inputmarginalVect}}{d\beta^l}$ for $l > 0$ can be expressed using derivative tensors $D^m_{\beta^b, \inputmarginalVect^{m-b}} (Id - BA_\beta)\big\rvert_{(\tilde{\inputmarginalVect}, \beta)}$ with $0 \leq b \leq m \leq l$, derivatives $\tfrac{d^k \tilde{\inputmarginalVect}}{d\beta^k}$ of lesser or equal degree $0 < k \leq l$, Jacobians $D_{\inputmarginalVect } \tfrac{d^k \tilde{\inputmarginalVect}}{d\beta^k}$ of strictly lower degree $0 < k < l$, and the inverse-Jacobian $\left( D_{\inputmarginalVect} (Id - BA_\beta)\rvert_{(\tilde{\inputmarginalVect}, \beta)} \right)^{-1}$. As in the proof of $(i)$, all these quantities are bounded uniformly on $T_{\eta'}$, showing that \eqref{eq:bounding-operator-norm-as-a-function-of-delta-beta} is indeed bounded uniformly, as required.
\end{proof}

\medskip
\section{Proof for the complexity of root-tracking and of RD root-tracking}
\label{sec:proofs-for-complexity-at-appendix}

We analyze the computational and memory complexities of root tracking in general and of root tracking for RD problems in particular. As the complexity of Taylor's method is determined by its complexity at a point, this boils down to analyzing the complexity of the Algorithm \ref{algo:high-order-derivs-of-operator-roots} --- both of its general form and of its specialization to RD. 

This section is structured as follows. In Section \ref{sub:complexity-of-high-order-derivs-for-root-tracking}, we analyze the cost of the recursive formula \eqref{eq:formula-for-high-order-beta-derivatives} for implicit derivatives (Algorithm \ref{algo:high-order-derivs-of-operator-roots}), assuming that the computational costs of the derivative tensors $D^{m}_{\beta^b, \bm{x}^{m-b}} F$ are given. 
When derivative tensors are computationally expensive (as in RD), it might be preferable to memorize them, so as to avoid computing a tensor more than once. The complexity is discussed both with and without tensor memorization, at Propositions \ref{prop:computational-cost-of-high-order-derivs-when-all-tensors-are-evaluated}, \ref{prop:computational-cost-of-high-order-derivs-when-tensors-are-memorized} and \ref{prop:total-computational-complexity-with-memorization}.
In Section \ref{sub:complexity-of-deriv-tensors-for-RD}, we analyze the complexity of the various quantities needed to compute RD derivative tensors, those of $Id - BA_\beta$ \eqref{eq:RD-operator-def}.
The results are combined in Section \ref{sub:complexity-of-RD-root-tracking}, proving the upper bounds for the complexity of root tracking for RD, Theorem \ref{thm:complexity-of-RD-root-tracking}.

The upper bounds we provide are often loose. In practice, some of the memory and computational costs grow at a much lower rate. cf., Figures \ref{fig:length-of-string-repr-of-P_k} and \ref{fig:number-of-monomials-in-P_k} for example.
For actual computational measures timed with our implementation see Figure \ref{fig:err-to-computational-cost-tradeoff}.

\medskip
\subsection{Complexity of root tracking: implicit derivatives of operator roots}
\label{sub:complexity-of-high-order-derivs-for-root-tracking}

We consider the complexity of the recursive formula \eqref{eq:formula-for-high-order-beta-derivatives} for $\tfrac{d^l \bm{x}}{d\beta^l}$ (Theorem \ref{thm:formula-for-high-order-expansion-of-F-in-main-result-sect}). Note that its right-hand side is comprised of an outer and of an inner summation. The outer summation is over the $p(l) - 1$ non-trivial partitions, while the inner summation is over the number of parts of size 1 in a given partition, $M_1 \cdot \delta(p_1 = 1)$. 
That is, the inner one is over the multiplicity of 1 in a partition. 
For example, 1 is of multiplicity 3 in the partition $1 + 1 + 1$ of 3, of multiplicity 1 in $2 + 1$, and of multiplicity 0 in the trivial partition 3.

Integer partitions can be grouped by their multiplicity of 1. For, ``set aside'' a single part of size 1, and consider partitions of $l-1$:
\begin{equation}		\label{eq:setting-aside-a-part-of-size-1-in-a-partition}
	l = 1 + \underset{\text{a partition of } l-1}{\underbrace{\left(\; \dots \quad \dots \; \right)}}
\end{equation}
Any partition of $l$ with at least one part of size 1 can be written as ``$1 + $  a partition of $(l-1)$'', as in \eqref{eq:setting-aside-a-part-of-size-1-in-a-partition}.
Thus, there are exactly $p(l-1)$ partitions of $l$ in which 1 is of multiplicity $\geq 1$. Hence, there are $p(l) - p(l-1)$ partitions with no part of size 1.
Setting aside two parts of size 1 shows that there are exactly $p(l-2)$ partitions of $l$ in which 1 is of multiplicity $\geq 2$. Hence, there are $p(l-1) - p(l-2)$ partitions which have exactly one part of size 1.
Proceeding in this manner, we have
\begin{lem}
	For $0 \leq j \leq l$, the number of partitions of $l$ with exactly $j$ parts of size 1 is
	\begin{equation}
		p(l - j) - p(l - j - 1) \;,
	\end{equation}
	where $p(-1)$ is defined to be 0, and $p(0) := 1$.
\end{lem}
A partition with $j$ parts of size 1 contributes $j+1$ summands to the inner summation at \eqref{eq:formula-for-high-order-beta-derivatives}. Grouping partitions by the multiplicity $j$ of 1 in them, the total number of summands at the right-hand side of \eqref{eq:formula-for-high-order-beta-derivatives} is
\begin{multline}
	1\cdot \Big(p(l) - p(l-1)\Big) + 2 \cdot \Big( p(l-1) - p(l-2) \Big) + \dots + l \cdot \Big(p(1) - p(0)\Big) + (l+1)\cdot p(0)- 1 \\
	= p(l) + p(l-1) + \dots + p(1) + p(0) - 1 \;.
\end{multline}
\begin{lem}			\label{lem:number-of-summands-at-formula-for-high-order-beta-derivs}
	The number of summands at the right-hand side of \eqref{eq:formula-for-high-order-beta-derivatives} is $\sum_{j=0}^l p(j) - 1$.
\end{lem}

At least for our case of interest $F := Id - BA_\beta$ \eqref{eq:RD-operator-def}, the complexity of a derivative tensor $D^m_{\beta^b, \bm{x}^{m-b}} F$ is determined by $m$ and $b$; see Section \ref{sub:complexity-of-deriv-tensors-for-RD}. Thus, we would like to group partitions not only by the multiplicity of 1 in them, but also by their total multiplicity $m$.
Denote $p_k(n)$ for the number of partitions of $n$ to exactly $k$ parts. From its definition, $p(n) = \sum_{k=0}^n p_k(n)$. 
For small $k, n$ values, it can be calculated using the recurrence relation \cite[Section 1.7]{stanley2011enumerative} 
\begin{equation}		\label{eq:recurrence-formula-for-partition-w-exact-part-num}
	p_k(n) = p_{k-1}(n-1) + p_k(n-k)
\end{equation}
and $p_0(0) = 1$; see Table \ref{tab:p_k-n-for-several-values} for example.
\begin{table}[h!]
	\begin{center}
		\setlength{\tabcolsep}{12pt}
		\begin{tabular}{c|ccccccc}
			$k, n$	&	0	&	1	&	2	&	3	&	4	&	5	&	6 \\
			\hline
			0	&	 1	&	0	&	0	&	0	&	0	&	0	&	0	\\
			1	&	  	&	1	&	1	&	1	&	1	&	1	&	1	\\
			2	&	  	&	 	&	1	&	1	&	2	&	2	&	3	\\
			3	&	  	&	 	&	 	&	1	&	1	&	2	&	3	\\
			4	&	  	&	 	&	 	&	 	&	1	&	1	&	2	\\
			5	&	  	&	 	&	 	&	 	&	 	&	1	&	1	\\
			6	&	  	&	 	&	 	&	 	&	 	&	 	&	1
		\end{tabular}
	\end{center}
	\caption{$p_k(n)$ for several small $n, k$ values.}
	\label{tab:p_k-n-for-several-values}
\end{table}

By the same reasoning as at \eqref{eq:setting-aside-a-part-of-size-1-in-a-partition}, there are $p_{k-1}(l-1)$ partitions of $l$ to $k$ parts in which 1 is of multiplicity $\geq 1$, $p_{k-2}(l-2)$ partitions of $l$ to $k$ parts in which 1 is of multiplicity $\geq 2$, and so forth. Therefore, we have as before:
\begin{lem}		\label{lem:num-of-partitions-with-restricted-parts-and-multiplicity-of-1}
	For $0 \leq j \leq k \leq l$, the number of partitions of $l$ to $k$ parts of which exactly $j$ are of size 1 is
	\begin{equation}
		p_{k-j}(l - j) - p_{k-j-1}(l - j - 1) \overset{\eqref{eq:recurrence-formula-for-partition-w-exact-part-num}}{=}
		p_{k-j}(l-k)\;,
	\end{equation}
	where $p_k(n) := 0$ if either $k < 0$ or $n < 0$.
\end{lem}

Write $C(b, m)$ for the complexity of calculating a derivative tensor $D^{m}_{\beta^b, \bm{x}^{m-b}} F$, and $p_{\leq k}(n)$ for the number of partitions of $n$ to at most $k$ parts. 
By definition, $p_{\leq k}(n) = \sum_{i=0}^k p_i(n)$. Using Lemma \ref{lem:num-of-partitions-with-restricted-parts-and-multiplicity-of-1}, we rearrange the outer summation over partitions at \eqref{eq:formula-for-high-order-beta-derivatives}. First, by the partitions' total multiplicity $m$ (number of parts), and then by the number $m_1$ of parts of size 1. With this, the complexity of calculating the derivative tensors is
\begin{multline}		\label{eq:complexity-of-recursive-formula-for-high-order-beta-derivs}
	\sum_{m=0}^l \sum_{m_1=0}^m p_{m-m_1}(l-m) \sum_{b=0}^{m_1} C(b, m) =
	\sum_{m=0}^l \sum_{m_1=0}^m \sum_{b=0}^{m_1} p_{m-m_1}(l-m) C(b, m) \\ =
	\sum_{m=0}^l \sum_{b=0}^{m} \sum_{m_1=b}^m p_{m-m_1}(l-m) C(b, m) =
	\sum_{m=0}^l \sum_{b=0}^m C(b, m) \sum_{m_1=b}^m p_{m-m_1}(l-m) \\ =
	\sum_{m=0}^l \sum_{b=0}^m C(b, m) \sum_{j=0}^{m-b} p_j(l-m) =
	\sum_{m=0}^l \sum_{b=0}^m p_{\leq m-b}(l-m) C(b, m) 
\end{multline}
The second equality above follows since, for a given $m$, we are summing over all the integers $b$ and $m_1$ with $0 \leq b \leq m_1 \leq m$. At the fourth equality we exchange $m_1$ with $j := m - m_1$. The last equality follows from the definition of $p_{\leq k}(n)$. 

Since \eqref{eq:complexity-of-recursive-formula-for-high-order-beta-derivs} adds up the complexity $C(b, m)$ once for each summand at the right-hand side of \eqref{eq:formula-for-high-order-beta-derivatives}, then setting $C(b, m) := 1$ gives an alternative formula for the number of summands there. Combined with Lemma \ref{lem:number-of-summands-at-formula-for-high-order-beta-derivs}, we have
\begin{cor}			\label{cor:two-forms-for-the-total-number-of-summands-at-recursive-formula-for-beta-derivs}
	\begin{equation}
		\sum_{m=0}^l \sum_{b=0}^m p_{\leq m-b}(l-m) = \sum_{j=1}^l p(j) \;.
	\end{equation}
\end{cor}

Write $F = (F_1, F_2, \dots, F_T)$ for the operator's coordinates, as at \eqref{eq:mixed-deriv-def-evaluated-applied-to-vectors}.
On top of the complexity \eqref{eq:complexity-of-recursive-formula-for-high-order-beta-derivs} for calculating the derivative tensors, we need to account for the complexity of evaluating the multilinear forms \eqref{eq:mixed-deriv-def-evaluated-applied-to-vectors} they define, of summing the evaluated forms, and of finding a linear pre-image under $D_{\bm{x}} F$. 
Indeed, a tensor $D^{m}_{\beta^b, \bm{x}^{m-b}} F$ has $T^{m-b+1}$ entries, and so its evaluation is $O((m-b+1) T^{m-b+1})$ operations; this is multiplied by the number of tensors of each rank, as in \eqref{eq:complexity-of-recursive-formula-for-high-order-beta-derivs}.
An evaluated multilinear form has $T$ entries. Thus, by Lemma \ref{lem:number-of-summands-at-formula-for-high-order-beta-derivs}, the complexity of their summation at the right of \eqref{eq:formula-for-high-order-beta-derivatives} is $O(T \cdot \sum_{j=1}^l p(j))$. Finding a linear pre-image is no more than $O(T^3)$ operations, e.g., using Gaussian elimination.
This is summarized by Proposition \ref{prop:computational-cost-of-high-order-derivs-when-all-tensors-are-evaluated} in Section \ref{part:details}.\ref{sec:computational-complexities}, on the complexity of $l$-th order derivative.

The computational cost is reduced drastically (at the expense of memory complexity) if every derivative tensor is computed only once and then memorized. 
By taking the all-ones partition $1 + 1 + \dots + 1$ of $l$, the $l$-th order formula \eqref{eq:formula-for-high-order-beta-derivatives} can be seen to contain \textit{all} the $l$-th order derivative tensors: $D^{l}_{\beta^0, \bm{x}^{l}} F, D^{l}_{\beta^1, \bm{x}^{l-1}} F, \dots, D^{l}_{\beta^l, \bm{x}^{0}} F$. These appear in the $l$-th order formula for the first time. For, the total multiplicity $m$ cannot exceed the partitioned integer $l$, and so neither of the $l$-th order tensors can appear when formula \eqref{eq:formula-for-high-order-beta-derivatives} is used with orders $< l$.
As $D^{m}_{\beta^b, \bm{x}^{m-b}} F$ has $T^{m-b+1}$ coordinates, memorizing all the tensors used up to the $l$-th order of the recursive formula \eqref{eq:formula-for-high-order-beta-derivatives} boils down to memorizing
\begin{equation}
	\sum_{m=0}^{l-1} \sum_{t=0}^m T^{t+1} 
	\leq l \cdot T^l
\end{equation}
coordinates. Where, we've written $t$ for $m-b$.
We thus obtain Proposition \ref{prop:computational-cost-of-high-order-derivs-when-tensors-are-memorized} in Section \ref{part:details}.\ref{sec:computational-complexities}, on the complexity of the $l$-th order derivative with tensor memorization.

The cumulative computational and memory costs with memorization are summarized by Proposition \ref{prop:total-computational-complexity-with-memorization} in Section \ref{part:details}.\ref{sec:computational-complexities}. Its proof is provided below.

\begin{proof}[Proof of Proposition \ref{prop:total-computational-complexity-with-memorization}]
	We sum the computational complexity \eqref{eq:computational-costs-in-prop-for-recursive-formula} per $l$ value (Proposition \ref{prop:computational-cost-of-high-order-derivs-when-tensors-are-memorized}) over $l = 1, \dots, L$.
	
	For the first summand at \eqref{eq:computational-costs-in-prop-for-recursive-formula}, $\sum_{l=1}^L O(T^3) = O(L\cdot T^3)$. For the second,
	\begin{equation}
		\sum_{l=1}^L T\cdot \sum_{j=1}^l p(j) = 
		T\cdot \sum_{l=1}^L \sum_{j=1}^l p(j) \leq
		L^2T \cdot p(L) \leq
		L^{\nicefrac{5}{4}} T e^{\pi\sqrt{\tfrac{2}{3}L}}
	\end{equation}
	Where, the first inequality follows by replacing $p(j)$ with the maximal summand $p(L)$, and
	the last follows from the upper bound $p(n) < \frac{e^{\pi\sqrt{\nicefrac{2n}{3}}}}{n^{\nicefrac{3}{4}}}$ at \cite{pribitkin2009paritit_bound}.
	
	Next, note that $p_{\leq k}(n)$ is non-decreasing in $n$. For, to a partition $\lambda_1 \leq \lambda_2 \leq \dots \leq \lambda_j$ of $n$ to with $j \leq k$ parts, one can injectively match the partition of $n+1$ with the same number of parts, defined by replacing $\lambda_j$ with $\lambda_j + 1$. Thus,
	\begin{equation}		\label{eq:bounding-number-of-int-partits-up-to-in-proof}
		p_{\leq m-b}(l-m) \leq p_{\leq m-b}(L) \leq p(L)
	\end{equation}
	where the last inequality follow from the definition of $p_{\leq k}(n)$. So, for the third summand at \eqref{eq:computational-costs-in-prop-for-recursive-formula},
	\begin{multline}
		\sum_{l=1}^L \sum_{m=0}^l \sum_{b=0}^m p_{\leq m-b}(l-m) O((m-b+1) T^{m-b+1}) =
		L \; p(L) \cdot \sum_{m=0}^L \sum_{b=0}^m O((m-b+1) T^{m-b+1}) \\ =
		L \; p(L) \cdot \sum_{0\leq t \leq L} (L - t + 1) O((t + 1) T^{t + 1}) =
		O(L^3 T^{L+1} p(L) ) =
		O(L^{\nicefrac{9}{4}} T^{L+1} e^{\pi\sqrt{\tfrac{2}{3}L}} )
	\end{multline}
	Where \eqref{eq:bounding-number-of-int-partits-up-to-in-proof} was used at the first equality, carrying out the summation over $l$. Next, $m-b$ was replaced by $t$. The last equality again uses the bound at \cite{pribitkin2009paritit_bound}.
	
	For $L = 1$, the cost of the terms above adds up to
	\begin{equation}
		O(L\cdot T^3) +
		O(L^{\nicefrac{5}{4}} T e^{\pi\sqrt{\tfrac{2}{3}L}}) +
		O(L^{\nicefrac{9}{4}} T^{L+1} e^{\pi\sqrt{\tfrac{2}{3}L}} ) =
		O(T^3 ) \;,
	\end{equation}
	proving \eqref{eq:total-computational-complexity-in-prop-for-first-order}.
	While for $L \geq 2$,
	\begin{multline}
		O(L\cdot T^3) +
		O(L^{\nicefrac{5}{4}} T e^{\pi\sqrt{\tfrac{2}{3}L}}) +
		O(L^{\nicefrac{9}{4}} T^{L+1} e^{\pi\sqrt{\tfrac{2}{3}L}} ) \\ =
		O(L^{\nicefrac{9}{4}} T^{L+1} e^{\pi\sqrt{\tfrac{2}{3}L}} ) =
		O\left( e^{\nicefrac{9}{4} \cdot \ln L + (L+1) \ln T + \pi\sqrt{\nicefrac{2L}{3}}} \right)
	\end{multline}
	which proves \eqref{eq:total-computational-complexity-in-prop}.
\end{proof}

\medskip 
\subsection{Complexity of high-order derivative tensors of $Id - BA_\beta$}
\label{sub:complexity-of-deriv-tensors-for-RD}

We analyze the complexity of the formulae for the derivative tensors $D^{b+m}_{\beta^b, \bm{x}^m} \left(Id - BA_\beta \right)\left[\inputmarginalVect\right](\hat{x})$ (Theorem \ref{thm:high-order-derivs-of-BA-in-main-text} in Section \ref{sub:high-order-deriv-tensors-of-BA}.). To that end, we distinguish between three kinds of computations: initial computations \ref{subsub:problem-ind-initial-computations} which depend neither on the RD problem at hand nor on the point of evaluation; initial computations \ref{subsub:initial computations-at-a-point} at a particular point; and the cost of a derivative tensor \ref{subsub:complexity-of-RD-deriv-tensor} given the above initial computations.
The analysis roughly follows the order of the results' presentation in Section \ref{part:how-and-what}.\ref{sub:high-order-deriv-tensors-of-BA}, and the comments in \ref{part:details}.\ref{sub:computing-high-order-derivatives-efficiently} on efficient computation of RD derivative tensors.
The results below are summarized in Table \ref{tab:complexities-of-RD-deriv-tensors}, in Section \ref{part:details}.\ref{sec:computational-complexities}.

\medskip
\subsubsection{Problem-independent initial-computations}
\label{subsub:problem-ind-initial-computations}

\begin{figure}[h!]
	\vspace{-15pt}
	\centering
	\includegraphics[trim={0 0 0 1cm}, clip, width=0.5\textwidth]{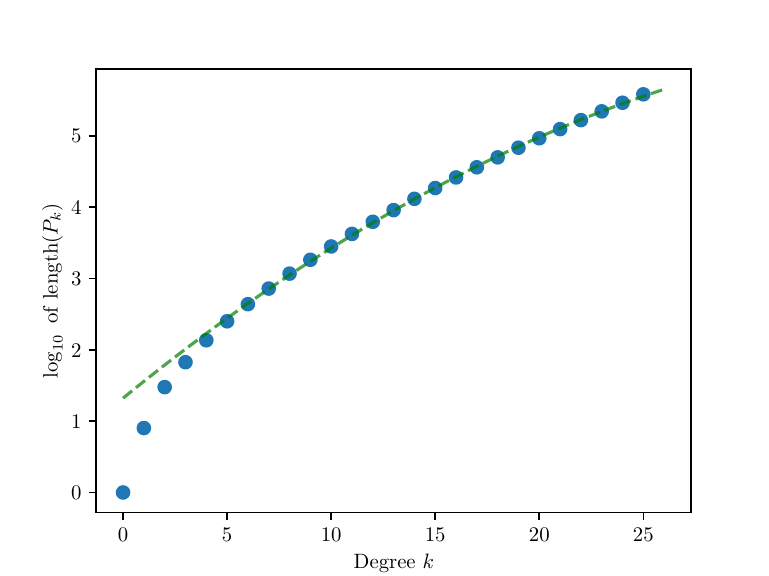}
	\caption{Length of the default string representation of $P_k$ used by our implementation.
		An approximation $-0.002739 k^2 + 0.2377 k + 1.323$ of the log-length is in dashed green.
	}
	\label{fig:length-of-string-repr-of-P_k}
	\vspace{-5pt}
\end{figure}

The algebraic form of the polynomials $P_k$ \eqref{eq:P_0-def}-\eqref{eq:P_k-inductive-def} needs to be computed only once and for all, as it depends neither on the problem details nor on the point of evaluation. We therefore ignore its computational complexity, and consider only its memory complexity.

By Lemma \ref{lem:bound-on-the-complexity-of-P_K} (in Section \ref{sec:proof-for-error-analysis}), each $P_k$ is of degree $k$ at most, and can be written as a sum of at most $2^k k!$ monomials, even when only $1$ and $-1$ are allowed as coefficients. 
Each of these monomials can be encoded using $1 + (k+1) \log_2 (k+1)$ bits: each of the $k+1$ variables is of degree $0, \dots, k$, and one bit for the coefficient. Hence, the memory needed to store $P_k$ is $O\left((k+1) \log_2 (k+1) 2^k k!\right)$.

In practice, the length of $P_k$ is roughly $1.73^k$, as demonstrated by Figure \ref{fig:length-of-string-repr-of-P_k}. Our implementation stores the first 20 polynomials in about 110 kilobytes in compressed form (LZMA) and the first 25 polynomials in 544 KBs.

\medskip
\subsubsection{Initial computations at a point}
\label{subsub:initial computations-at-a-point}

Recall, we write $N := |\mathcal{X}|$ and $M := |\hat{\mathcal{X}}|$ for the source and reproduction alphabet sizes, respectively.
Given an encoder $\intermediateencoder{}{}$ of an RD problem defined by $\big(d(x, \hat{x}), p(x)\big)$, we evaluate the cost of the components needed to compute the $M$-by-$N$ matrices $G\big(k, a; \intermediateencoderVect \big)$ \eqref{eq:combinatorial-G-in-terms-of-polynomials}. These are indexed by $0 \leq k \leq L$ and $0 \leq a \leq 1 + L$, where $L$ is the maximal order of derivation. Clearly, its memory complexity is $O(MN (1 + L) (2 + L) )$.

\begin{figure}[h!]
	\vspace{-15pt}
	\centering
	\includegraphics[trim={0 0 0 1cm}, clip, width=0.5\textwidth]{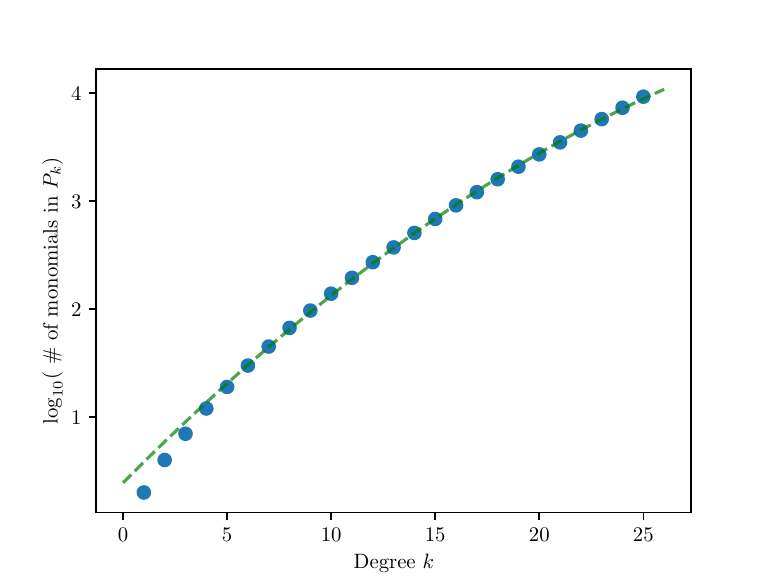}
	\caption{The number of monomials in $P_k$.
		An approximation $-0.002084 k^2 + 0.1944 k + 0.3911$ of $\log_{10} (\# \text{ monomials})$  is in dashed green.
	}
	\label{fig:number-of-monomials-in-P_k}
	\vspace{-10pt}
\end{figure}

The $P_k[\intermediateencoderVect; d]$ matrices \eqref{eq:P_k-by-abuse-of-notation} are needed to compute $G$; these require the expectations $\expectedDxWRTencoderK{k}$ \eqref{eq:expected-k-th-power-distortion-def}. The powers $d(x, \hat{x})^k$ may be computed only once per RD problem, and so we neglect their computational cost. Thus, the computation of $\expectedDxWRTencoderK{k}$ costs $O(MN)$ operations, as it involves only multiplication by $\intermediateencoder{}{}$ and a summation over $\hat{x}$. 
The zeroth variable $d(x, \hat{x})$ of $P_k$ has $MN$ coordinates, while the others $\expectedDxWRTencoderK{k}$ have $N$ coordinates, for $k \geq 1$.
As $P_k$ is of degree $k$ at most, Lemma \ref{lem:bound-on-the-complexity-of-P_K} in Section \ref{sec:proof-for-error-analysis}, evaluating a single monomial at $P_k[\intermediateencoderVect; d]$ \eqref{eq:P_k-by-abuse-of-notation} is $O(MNk)$ operations.
By the Lemma, $P_k$ has no more than $2^k k!$ monomials, so the total cost is $O(MNk \; 2^k k!)$.
In practice, there are roughly $1.56^k$ monomials in $P_k$, as seen in Figure \ref{fig:number-of-monomials-in-P_k} for $k\leq 25$, and so the actual cost is much smaller.

For $G\big(k, a; \intermediateencoderVect \big)$, computing $\left( \frac{P_j(\hat{x}, x)}{j!} \right)^{t_{j}}$ at \eqref{eq:combinatorial-G-in-terms-of-polynomials} for a given $\bm{t}$ and $j$ is $O(MN)$ operations. 
Accounting for the pointwise product over $j=1,\dots,k$, the cost for each $\bm{t}$ is $O(MNk)$. 
The entries of $G$ can be computed by iterating only once over the partitions $\bm{t}$ of all the integers $k$ with $0 \leq k \leq L$. For, given $\bm{t}$, an $M$-by-$N$ integrand $\prod_j \left( \frac{P_j(\hat{x}, x)}{j!} \right)^{t_{j}}$ can be added to $G(k, a)$ for all $a \geq |\bm{t}|$, at $O(MNL)$ operations. That is, the entire cost for a particular $\bm{t}$ is $O(MNk)+O(MNL) = O(MNL)$, and so $O(MNL \sum_{k=0}^L p(k))$ when iterating over all the integer partitions $\bm{t}$.


\medskip
\subsubsection{Complexity of a derivative tensor}
\label{subsub:complexity-of-RD-deriv-tensor}

We analyze the cost of the various derivative tensors of $Id - BA_\beta$ \eqref{eq:RD-operator-def}, assuming that the initial computations in \ref{subsub:problem-ind-initial-computations} and \ref{subsub:initial computations-at-a-point} were already done.

The memory complexity of $D^{b+m}_{\beta^b, \bm{x}^m} (Id - BA_\beta)[\inputmarginalVect](\hat{x})$ \eqref{eq:mixed-BA-deriv-in-thm} is $O(M^{m+1})$. 
The computational complexity of a derivative tensor $D^b_{\beta^b} (Id - BA_\beta)[\inputmarginalVect](\hat{x})$ \eqref{eq:repeated-beta-deriv-in-thm} with respect to $\beta$ alone is $O(MN)$.
For a mixed derivatives tensor $D^{b+m}_{\beta^b, \bm{x}^m} (Id - BA_\beta)[\inputmarginalVect](\hat{x})$ with $m > 0$, fix $\bm{\alpha} = (b, \bm{\alpha}_+) \in \bb{N}_0^{M+1}$ with $\bm{\alpha}_+ \neq 0$. For any $\hat{x}$ and $\bm{k} \in \bb{N}_0^M$ with $|\bm{k}| = b$, the complexity of the integrand under the summation $\sum$ at \eqref{eq:mixed-BA-deriv-in-thm} is $O(MN)$ operations.
It is an $N$-vector. Summing over the $\binom{b + M - 1}{b}$ choices for $\bm{k}$, the sum has a complexity of $O(\binom{b + M - 1}{b} \cdot MN)$ operations. 
The other operations at \eqref{eq:mixed-BA-deriv-in-thm} have a comparatively negligible cost.
The above needs to be performed for each $\hat{x}$ and $\bm{\alpha}_+$ with $|\bm{\alpha}_+|=m$, which is $M \cdot \binom{m + M - 1}{m}$ times.
Thus, the complexity of the mixed partial derivatives formula \eqref{eq:mixed-BA-deriv-in-thm} is $O(\binom{m + M - 1}{m} \binom{b + M - 1}{b} \cdot M^2 N)$ operations.

To a multi-index $\bm{\alpha}_+ \in \bb{N}_0^{M}$ may correspond multiple coordinates of a derivative tensor. 
For, a multi-index $\bm{\alpha}_+$ counts the number of derivations with respect to each of the $M$ coordinates of $\inputmarginalVect$, while a tensor entry stands for a particular order of the derivatives.
e.g., for a reproduction alphabet of size $M = 2$, $\bm{\alpha}_+ = (2, 1)$ represents taking two derivatives with respect to $\inputmarginalSymbol_1$, and one with respect to $\inputmarginalSymbol_2$. 
To it, correspond the tensor entries 
$\tfrac{\partial^3}{\partial \inputmarginalSymbol_2 \partial \inputmarginalSymbol_1 \partial \inputmarginalSymbol_1}$, $\tfrac{\partial^3}{\partial \inputmarginalSymbol_1 \partial \inputmarginalSymbol_2 \partial \inputmarginalSymbol_1}$ and $\tfrac{\partial^3}{\partial \inputmarginalSymbol_1 \partial \inputmarginalSymbol_1 \partial \inputmarginalSymbol_2}$ at \eqref{eq:mixed-deriv-def-evaluated-applied-to-vectors}, standing respectively for the $(2, 1, 1), (1, 2, 1)$ and $(1, 1, 2)$ entries of a derivative tensor with $m=3$. 
The number of tensor entries corresponding to a multi-index $\bm{\alpha}_+$ is $m!$ at most, with $O(M)$ copy operations for each tensor entry. Doing so for each $\bm{\alpha}_+$ is $O(\binom{m + M - 1}{m} m! M)$ operations.
While copying out partial derivatives to the various tensor entries need not be the most efficient solution, it is simple and straightforward to implement.

All in all, bounds on the complexities of RD derivative tensors are summarized in Table \ref{tab:complexities-of-RD-deriv-tensors} in Section \ref{part:details}.\ref{sec:computational-complexities}.

\medskip
\subsection{Complexity of root tracking for RD}
\label{sub:complexity-of-RD-root-tracking}

In this subsection we prove Theorem \ref{thm:complexity-of-RD-root-tracking} on the complexities of root tracking for RD. We do so by compiling the results of Sections \ref{sub:complexity-of-high-order-derivs-for-root-tracking} on the complexity of root-tracking and of \ref{sub:complexity-of-deriv-tensors-for-RD} on the complexity of RD derivative tensors (Table \ref{tab:complexities-of-RD-deriv-tensors} in Section \ref{part:details}.\ref{sec:computational-complexities}). The proof boils down to adding the various costs associated with the derivatives at a point, up to order $L$.

\begin{proof}[Proof of Theorem \ref{thm:complexity-of-RD-root-tracking}]
	In Section \ref{sub:complexity-of-deriv-tensors-for-RD}, we divided the cost of RD derivative tensors to initial computations and the calculation of the derivative tensors themselves.
	Using Table \ref{tab:complexities-of-RD-deriv-tensors} in Section \ref{part:details}.\ref{sec:computational-complexities}, initial computations at a point require the algebraic form of $P_k$, the expectations $\expectedDxWRTencoderK{k}$ and the matrices $P_k[\intermediateencoderVect; d]$ for $k = 0, \dots, L$; and the matrices $G(k,a)$. Summing over the respective table elements, their memory cost is:
	\begin{multline}			\label{eq:memory-cost-of-initial-computations}
		\sum_{k=0}^L \Big[ O\left( 2^k k! (k+1) \log_2 (k+1) \right) +	O(N) +	O(MN) \Big] + O(MN (L + 1) (L + 2) ) \\ =
		O\left( 2^L L! L^2 \ln L \right) + O(MNL^2) <
		O\left( e^{(L + \nicefrac{5}{2})\ln L - c_1 L + \ln \ln L + c_2} \right) + O(MNL^2) \\ <
		O\left( L^{(L + \nicefrac{5}{2}) } \ln L  \right) + O(MNL^2) 
	\end{multline}	
	Where, the first expression at the first line is bounded by its maximal value at $k = L$, the first inequality follows from $n! \leq 
	e^{(n + \nicefrac{1}{2})\ln n - n + c_2}$ \cite{robbins1955remark}, with $c_2 := \tfrac{1}{12} + \tfrac{1}{2} \ln (2\pi) \approx 1$, and by setting $c_1 := 1 - \ln 2 \approx 0.31$.
	
	For the initial computations,
	\begin{multline}			\label{eq:computational-cost-of-initial-computations}
		\sum_{k=0}^L \Big[ O(MN) + O(MNk \; 2^k k!) + O(MNL \; p(k)) \Big] \\ =
		\sum_{k=0}^L \Big[ O(MNL \; 2^k k!) + O(MNL \; p(k)) \Big] =
		\sum_{k=0}^L O(MNL \; 2^k k!) \\ =
		O(MN 2^L L! L^2) =
		O\left( MN \cdot e^{(L + \nicefrac{5}{2})\ln L - c_1 L + c_2} \right)
	\end{multline}
	Where, $k$ is bounded by $L$ at the first equality, $O(2^k k!) + O(p(k)) = O(2^k k!)$ by the bound $p(k) < \frac{e^{\pi \sqrt{2k/3}}}{k^{3/4}}$ \cite{pribitkin2009paritit_bound} at the second equality, the maximal element is taken at the third, and the above bound of \citeauthor{robbins1955remark} for factorial is used at the last.
	We comment that the tighter $O(MN L^2 1.56^L)$ can be seen to approximate the initial complexity accurately for $L \leq 25$. 
	cf., Table \ref{tab:complexities-of-RD-deriv-tensors} in Section \ref{part:details}.\ref{sec:computational-complexities}. 
	
	The cumulative memory cost of tensor memorization is given by \eqref{eq:memory-costs-in-prop-for-recursive-formula} (Proposition \ref{prop:computational-cost-of-high-order-derivs-when-tensors-are-memorized}).
	By \eqref{eq:total-computational-complexity-in-prop} at Proposition \ref{prop:total-computational-complexity-with-memorization}, for the computational complexity it suffices to calculate $\sum_{m=1}^L \sum_{b=0}^m C(b, m)$, where $C(b, m)$ is the complexity of $D^{m}_{\beta^b, \bm{x}^{m-b}} (Id - BA_\beta)$. We write,
	\begin{equation}			\label{eq:upper-bound-on-deriv-tensor-complexities}
		\sum_{m=1}^L \sum_{b=0}^{m-1} C(b, m) =
		\sum_{m=1}^L \sum_{b=0}^{m-1} \Big[ O(\genfrac{(}{)}{0pt}{1}{m - b + M - 1}{m - b} \genfrac{(}{)}{0pt}{1}{b + M - 1}{b} \cdot M^2 N)
		+ O(\genfrac{(}{)}{0pt}{1}{m - b + M - 1}{m - b} (m-b)! M) \Big]
	\end{equation}
	Where, we have ignored the complexity $O(MN)$ of the $L$ derivative tensors with respect to $\beta$ alone, those with $b = m$, as they are negligible compared to the below.
	We comment that the first summand in \eqref{eq:upper-bound-on-deriv-tensor-complexities} pertains to the cost of computing mixed partial derivatives, while the second summand corresponds to copy operations; see Section \ref{part:proofs}.\ref{subsub:complexity-of-RD-deriv-tensor}. 
	
	For the first summand at \eqref{eq:upper-bound-on-deriv-tensor-complexities}, by the formula $\sum_{k=0}^r \binom{r-k}{m} \binom{s+k}{n} = \binom{r+s+1}{m+n+1}$ \cite[Eq. (25) in 1.2.6]{knuth97art_vol1}, we have
	\begin{multline}			\label{eq:upper-bound-on-complexity-of-tensor-calculation}
		\sum_{m=1}^L \sum_{b=0}^{m-1} \binom{m+M-1-b}{M-1} \binom{M-1+b}{M-1} \leq 
		\sum_{m=1}^L \binom{m+2M-1}{2M-1} \\ \leq
		L \cdot \binom{2M+L-1}{2M-1} =
		L \cdot \frac{2M}{2M+L} \cdot \binom{2M+L}{2M} \\ \leq
		e^{(2M+L) H_e\left(\rho\right) + \tfrac{1}{2} \ln \tfrac{\rho L}{\pi}} 
	\end{multline}
	Where the first inequality follows by setting $r := m + M - 1$ and $s := M-1$ in the above formula, summing only over $b = 0, \dots, m-1$ rather than up to $r$. The second inequality follows by bounding with the largest summand $m = L$. For the last inequality, set $\rho := \tfrac{2M}{2M+L}$. It then follows from \cite[Lemma 17.5.1]{Cover2006}, up to straightforward modifications to measure entropy $H_e(\cdot)$ in nats instead of bits.
	
	For the second summand at \eqref{eq:upper-bound-on-deriv-tensor-complexities}, 
	\begin{multline}			\label{eq:upper-bound-on-spreading-tensor-coords}
		\sum_{m=1}^L \sum_{b=0}^{m-1} \binom{m - b + M - 1}{m-b} (m-b)!
		= \frac{1}{(M-1)!} \sum_{m=1}^L \sum_{b=0}^{m-1} (m - b + M - 1)! \\
		\leq
		L^2 \cdot \frac{M}{M+L} \frac{(M + L )!}{M!} \leq
		M L^2 \cdot (M+L)^{L-1} =
		e^{(L-1) \ln (M+L) + \ln(ML^2)}
	\end{multline}
	Where the inequalities follow by bounding $m - b$ by its largest value $L$, and $\tfrac{(M + L )!}{M!} = (M+1)\cdot (M+2) \cdots (M+L)$ by $(M+L)^L$.
	
	Plugging \eqref{eq:upper-bound-on-complexity-of-tensor-calculation} and \eqref{eq:upper-bound-on-spreading-tensor-coords} into \eqref{eq:upper-bound-on-deriv-tensor-complexities} and then back into formula \eqref{eq:total-computational-complexity-in-prop} for the computational complexity of root-tracking (Proposition \ref{prop:total-computational-complexity-with-memorization}), and adding the cost \eqref{eq:computational-cost-of-initial-computations} of initial computational, we obtain
	\begin{multline}		\label{eq:collecting-computational-costs}
		O\left( e^{(2M+L) H_e\left(\rho\right) + \tfrac{1}{2} \ln \tfrac{\rho L}{\pi} + \ln (M^2N) } \right) +
		O\left( e^{(L-1) \ln (M+L) + 2\ln(ML)} \right) \\ +
		O\left( e^{\nicefrac{9}{4} \cdot \ln L + (L+1) \ln M + \pi\sqrt{\nicefrac{2L}{3}}} \right) +
		O\left( e^{(L + \nicefrac{5}{2})\ln L - c_1 L + c_2 + \ln MN } \right)
	\end{multline}
	for the total computational cost when $L \geq 2$. These are bounded from above by,
	\begin{multline}
		O\left( e^{(2M+L) H_e(\rho) + (L + \nicefrac{5}{2}) \ln L + (L+1) \ln M + (L - 1) \ln (M+L) + \ln N + \pi \sqrt{\nicefrac{2L}{3}} } \right) \\ <
		O\left( e^{(2M+L) H_e(\rho) + (3L + \nicefrac{5}{2}) \ln (M + L) + \ln N + \pi \sqrt{\nicefrac{2L}{3}} } \right) \\ =
		O\left( N \cdot (M + L)^{(3L + \nicefrac{5}{2})} e^{(2M+L) H_e(\rho) + \pi \sqrt{\nicefrac{2L}{3}} } \right) 
	\end{multline}
	which completes the proof of \eqref{eq:computational-complexity-of-RD-root-tracking-for-L-geq-2-in-thm} in Theorem \ref{thm:complexity-of-RD-root-tracking}.
	
	When $L = 1$, one can see directly from \eqref{eq:total-computational-complexity-in-prop-for-first-order} of Proposition \ref{prop:total-computational-complexity-with-memorization} and from Table \ref{tab:complexities-of-RD-deriv-tensors} in Section \ref{part:details}.\ref{sec:computational-complexities} that the computational complexity in this case is
	\begin{equation}		\label{eq:collecting-computational-costs-for-L=1}
		O\left( M^3 \right) +
		O\left( MN \right) +
		O\left( M^3 N \right) +
		O\left( M^2 \right) =
		O\left( M^3 N \right) \;,
	\end{equation}
	proving \eqref{eq:computational-complexity-of-RD-root-tracking-for-L-=1-in-thm} at the Theorem.
	
	For the total memory cost \eqref{eq:memory-complexity-of-RD-root-tracking-in-thm}, we add the memory cost \eqref{eq:memory-cost-of-initial-computations} of the initial computations to that needed for storing the derivative tensors,	\eqref{eq:memory-costs-in-prop-for-recursive-formula} in Proposition \ref{prop:computational-cost-of-high-order-derivs-when-tensors-are-memorized}. This yields
	\begin{equation}
		O\left( L^{(L + \nicefrac{5}{2}) } \ln L \right) + O\left(MNL^2\right) +
		O\left( M^L L \right) \;,
	\end{equation}
	completing the proof.
\end{proof}

\medskip
\section{Binary Source with a Hamming distortion measure: an analytical solution}
\label{sec:binary-source-with-hamming-dist-appendix}

In this section, we develop the explicit equations of one of the simplest rate-distortion problems: a binary source with a Hamming distortion measure. As it admits an analytical solution, this problem can be used to verify the correctness of the theory and of its implementation. cf., Figure \ref{fig:derivative-calculation-loses-accuracy-near-bifurcation} for example.
This problem was used throughout with parameters $p = 0.3$ and $2^{-1} \leq \beta \leq 2^5$, unless stated otherwise.
\medskip

The \textit{Hamming distortion} is defined by
\begin{equation}		\label{eq:hamming-distortion-measure}
	d(x, \hat{x}) := \begin{cases}
		0 & \quad \text{if } x= \hat{x}, \\
		1 & \quad \text{if } x\neq \hat{x} \;.
	\end{cases}
\end{equation}
It is also called a \textit{probability of error distortion}, as $\bb{E}[d(X, \hat{X})] = Pr(X \neq \hat{X})$.
The proof of \cite[Theorem 10.3.1]{Cover2006} shows that the achieving distribution of a binary source $X\sim$ Bernoulli$(p)$, $p < \nicefrac{1}{2}$, with a Hamming distortion measure \eqref{eq:hamming-distortion-measure} is
\begin{equation}			\label{eq:analytical-solution-for-binary-source-with-hamming-at-appendix}
	Pr(\hat{X} = 1) = \frac{p-D}{1-2D} 
\end{equation}
when $D \leq p$, and $Pr(\hat{X} = 1) = 0$ otherwise (Eq. (10.21) there).
The rate-distortion function for this problem is
\begin{equation}			\label{eq:RD-func-for-binary-source-with-hamming-dist}
	R(D) = H(p) - H(D) \quad \text{for } 0 \leq D \leq \min\{p, 1-p\} \;,
\end{equation}
and zero-rate otherwise.

To exchange variables from $D$ to $\beta$, we use the relation $R'(D) = -\beta$ \cite[Theorem 2.5.1]{berger71}, when information is expressed in nats (logarithms taken in the natural basis). Plugging \eqref{eq:RD-func-for-binary-source-with-hamming-dist} in implies,
\begin{equation}
	-\beta = \ln \frac{D}{1-D}		\quad \Longleftrightarrow \quad
	D = \frac{1}{1 + e^{\beta}}
\end{equation}
for $\beta > 0$. Plugging this back into \eqref{eq:analytical-solution-for-binary-source-with-hamming-at-appendix} yields an analytical solution in terms of $\beta$,
\begin{equation}		\label{eq:explicit-solution-by-beta}
	Pr(\hat{X} = 1) = \frac{1 - p\cdot (1 + e^{\beta})}{1 - e^{\beta}}
\end{equation}
This problem has a unique bifurcation, occurring when $Pr(\hat{X} = 1)$ first hits 0, at
\begin{equation}		\label{eq:bifurcation-beta-of-binary-source-with-Hamming-distortion}
	\beta_c = \ln \frac{1 - p}{p} \;.
\end{equation}
So long that $\beta \geq \beta_c$, the exact solution is given by \eqref{eq:explicit-solution-by-beta}, and is otherwise constant.

Having an analytical solution \eqref{eq:explicit-solution-by-beta} in terms of $\beta$, one can easily evaluate its derivatives with respect to $\beta$ of any order, at any point.

\fi

\newpage
\bibliographystyle{plain}
\bibliography{my_bib}

\end{document}